\documentclass[11pt,a4paper,twoside,dvipsnames,table]{article}

\usepackage{fullpage}
\usepackage[utf8]{inputenc}
\usepackage[T1]{fontenc}
\usepackage{float}

\usepackage{xcolor}              
\usepackage{hyperref}
\hypersetup{
 pdfauthor={Athmane Bakhta, Thomas Boiveau, Yvon Maday, Olga Mula},
 pdftitle={Epidemiological Forecasting with Model Reduction of Compartmental Models. Application to the COVID-19 pandemic.},
 pdfsubject={},
 pdfkeywords={COVID-19, epidemiology, forecasting, model reduction, reduced basis}
}

\newcommand\myshade{85}
\colorlet{mylinkcolor}{magenta}
\colorlet{mycitecolor}{blue}
\colorlet{myurlcolor}{Aquamarine}

\hypersetup{
  linkcolor  = mylinkcolor!\myshade!black,
  citecolor  = mycitecolor!\myshade!black,
  urlcolor   = myurlcolor!\myshade!black,
  colorlinks = true,
}

\usepackage{authblk}               
\usepackage[normalem]{ulem} 
\usepackage{amsmath, amsthm, amssymb, amsfonts}	
\usepackage{thmtools, thm-restate}				         
\usepackage{bm}              
\usepackage{mathtools}	
\usepackage{xspace}	
\usepackage{nicefrac}                                                  
\usepackage{paralist}
\usepackage{ifthen}
\newboolean{pictures}
\setboolean{pictures}{false}

\usepackage{graphicx}
\usepackage{subcaption}
\usepackage{tikz}
\usepackage{pgfplots}
\usepackage{caption}
\mathtoolsset{showonlyrefs} 
\usepackage{dsfont}              

\graphicspath{{figs/}}

\usepackage{algorithm, algcompatible}

\algnewcommand\algorithmicto{\textbf{to}}
\algnewcommand\RETURN{\State \textbf{return} }

\newtheorem{theorem}{Theorem}[section]
\newtheorem{Proposition}[theorem]{Proposition}
 \newtheorem{Remark}[theorem]{Remark}

\DeclareMathOperator*{\argmax}{arg\,max}
\DeclareMathOperator*{\argmin}{arg\,min}
\DeclareMathOperator*{\arginf}{arg\,inf}

\newcommand{\e}{\varepsilon}

\renewcommand{\d}{\mathrm{d}}
\newcommand{\dd}{\ensuremath{\mathrm d}}

\newcommand{\dt}{\ensuremath{\mathrm dt}}

\newcommand{\vspan}{\ensuremath{\mathrm{span}}}

\newcommand{\cB}{\ensuremath{\mathcal{B}}}

\newcommand{\cG}{\ensuremath{\mathcal{G}}}

\newcommand{\cI}{\ensuremath{\mathcal{I}}}
\newcommand{\cJ}{\ensuremath{\mathcal{J}}}

\newcommand{\cP}{\ensuremath{\mathcal{P}}}

\newcommand{\cR}{\ensuremath{\mathcal{R}}}
\newcommand{\cS}{\ensuremath{\mathcal{S}}}

\newcommand{\cU}{\ensuremath{\mathcal{U}}}

\newcommand{\bN}{\ensuremath{\mathbb{N}}}

\newcommand{\bR}{\ensuremath{\mathbb{R}}}

\newcommand{\rB}{\ensuremath{\mathrm{B}}}

\newcommand{\rG}{\ensuremath{\mathrm{G}}}

\newcommand{\Linf}{\ensuremath{L^\infty([0,T])}}

\newcommand{\Linfdp}{\ensuremath{L^\infty([0,\widetilde{T}], \bR^d)}}

\newcommand{\Linftaud}{\ensuremath{L^\infty([0,T+\tau], \bR^d)}}
\newcommand{\bfbeta}{\ensuremath{{\bm{\beta}}}}
\newcommand{\bfpsi}{\ensuremath{{\bm{\psi}}}}

\newcommand{\bfgamma}{\ensuremath{{\bm{\gamma}}}}
\newcommand{\bff}{\ensuremath{{\textbf{f}}}}
\newcommand{\bfu}{\ensuremath{{\textbf{u}}}}
\newcommand{\bfS}{\ensuremath{{\textbf{S}}}}
\newcommand{\bfI}{\ensuremath{{\textbf{I}}}}
\newcommand{\bfR}{\ensuremath{{\textbf{R}}}}
\newcommand{\NMF}{\ensuremath{\texttt{NMF}}}
\newcommand{\SIR}{\ensuremath{\texttt{SIR}}}
\newcommand{\MRSIR}{\ensuremath{\texttt{Multiregional\_SIR}}}
\newcommand{\DM}{\ensuremath{\texttt{Detailed\_Model}}}
\newcommand{\ENLARGECONE}{\ensuremath{\texttt{Enlarge\_Cone}}}
\newcommand{\Colizza}{\ensuremath{\texttt{SEI$5$CHRD}}}
\newcommand{\Magal}{\ensuremath{\texttt{SE$2$IUR}}}
\newcommand{\fitIR}{\ensuremath{\texttt{routine-IR}}}
\newcommand{\fitbg}{\ensuremath{\texttt{routine-$\bfbeta\bfgamma$}}}

\newcommand{\bfE}{\ensuremath{{\textbf{E}}}}
\newcommand{\bfEun}{\ensuremath{{\textbf{E1}}}}
\newcommand{\bfEdeux}{\ensuremath{{\textbf{E2}}}}
\newcommand{\bfIp}{\ensuremath{{\textbf{I$_p$}}}}
\newcommand{\bfIa}{\ensuremath{{\textbf{I$_a$}}}}
\newcommand{\bfIps}{\ensuremath{{\textbf{I$_{ps}$}}}}
\newcommand{\bfIms}{\ensuremath{{\textbf{I$_{ms}$}}}}
\newcommand{\bfIss}{\ensuremath{{\textbf{I$_{ss}$}}}}
\newcommand{\bfC}{\ensuremath{{\textbf{C}}}}
\newcommand{\bfH}{\ensuremath{{\textbf{H}}}}
\newcommand{\bfD}{\ensuremath{{\textbf{D}}}}
\newcommand{\bfU}{\ensuremath{{\textbf{U}}}}

\newcommand{\bfb}{\ensuremath{{\textbf{b}}}}

\newcommand{\diag}{\operatorname{diag}}

\newcommand{\collapse}{\ensuremath{{\text{col}}}}
\newcommand{\tr}{\ensuremath{{\text{tr}}}}
\newcommand{\obs}{\ensuremath{{\text{obs}}}}

\newcommand{\col}[1]{{\color{black}{#1}}}

\newcolumntype{L}[1]{>{\raggedright\let\newline\\\arraybackslash\hspace{0pt}}m{#1}}
\newcolumntype{C}[1]{>{\centering\let\newline\\\arraybackslash\hspace{0pt}}m{#1}}
\newcolumntype{R}[1]{>{\raggedleft\let\newline\\\arraybackslash\hspace{0pt}}m{#1}}

\usepackage{tikz}
\usetikzlibrary{arrows,snakes,backgrounds}
\usetikzlibrary{positioning}

\title{
Epidemiological Forecasting with Model Reduction of Compartmental Models. \\Application to the COVID-19 pandemic.
}
\author{Athmane Bakhta, Thomas Boiveau, Yvon Maday, Olga Mula}
\date{}

\begin{document}
\maketitle

\begin{abstract}
We propose a forecasting method for predicting epidemiological health series on a two-week horizon at the regional and interregional resolution. The approach is based on model order reduction of parametric compartmental models, and is designed to accommodate small amount of sanitary data. The efficiency of the method is shown in the case of the prediction of the number of \col{infected and removed people} during the \col{two} pandemic waves of COVID-19 in France, which \col{have} taken place approximately between \col{February and November 2020}. Numerical results illustrate the promising potential of the approach.
\end{abstract}

\section{Introduction}
\label{sec:Motivation}
Providing reliable epidemiological forecasts during an ongoing pandemic is crucial to mitigate the potentially disastrous consequences on global public health and economy. As the ongoing pandemic on COVID-19 sadly illustrates, this is a daunting task in the case of new diseases due to the incomplete knowledge of the behavior of the disease, and the heterogeneities and uncertainties in the health data count. Despite these difficulties, many forecasting strategies exist and we could cast them into two main categories. The first one is purely data-based, and involves statistical and learning methods such as time series analysis, multivariate linear regression, grey forecasting or neural networks \cite{Ceylan2020,anastassopoulou2020data,fang2020forecasting,roda2020difficult, liu2020predicting}. The second approach uses epidemiological models which are appealing since they provide an interpretable insight of the mechanisms of the outbreak. They also provide high flexibility in the level of detail to describe the evolution of the pandemic, ranging from simple compartmental models that divide the population into a few exclusive categories, to highly detailed descriptions involving numerous compartments or even agent-based models \col{(see, e.g., \cite{KR2011,martcheva2015introduction,brauer2017mathematical} for  general references on mathematical epidemiological models and \cite{magal2020predicting, di9impact, flaxman2020estimating} for some models focused on Covid-19)}. One salient drawback of using epidemiological models for forecasting purposes lies in the very high uncertainty in the estimation of the involved parameters. This is due to the fact that most of the time the parameters cannot be inferred from real observations and the available data are insufficient or too noisy to provide any reliable estimation. The situation is aggravated by the fact that the number of parameters can quickly become large even in moderately simple compartmental models \cite{di9impact}. As a result, forecasting with these models involves making numerous a priori hypothesis which can sometimes be hardly justified by data observations.

\medskip 
In this paper, our goal is to forecast the time-series of \col{infected,} removed and dead patients with compartmental models that involve as few parameters as possible in order to infer them solely from the data. \col{The available data are only given for hospitalized people, one can nevertheless estimate the total number of infected people through an adjustment factor hinted from the literature. Such a factor takes into account the proportion of asymptomatic people and infected people who do not go to hospital}. The model involving the least number of parameters is probably the SIR model \cite{kermack1927contribution} which is based on a partition of the population into:
\begin{itemize}
\item Uninfected people, called susceptible (S),
\item Infected and contagious people (I), with more or less marked symptoms,
\item People removed (R) from the infectious process, either because they are cured or unfortunately died after being infected.
\end{itemize}
If $N$ denotes the total population size that we assume to be constant over a certain time interval $[0, T]$, we have
$$
N = S(t)+I(t)+R(t), \forall t \in [0, T],
$$
and the evolution from $S$ to $I$ and from $I$ to $R$ is given for all $t\in [0,T]$ by
\begin{align*}
\frac{d S}{dt}(t) &= - \frac{\beta I(t) S(t)}{N} \\
\frac{d I}{dt}(t) &= \frac{\beta I(t) S(t)}{N} - \gamma I(t) \\
\frac{d R}{dt}(t) &=  \gamma I(t) .
\end{align*}
The SIR model has only two parameters:
\begin{itemize}
\item $\gamma>0$ represents the recovery rate. In other words,
its inverse $\gamma^{-1}$ can be interpreted as the length (in days)
of the contagious period.
\item $\beta>0$ is the transmission
rate of the disease. It essentially depends on two factors: the
contagiousness of the disease and the
contact rate within the population. The larger this second parameter is, the
faster the transition from susceptible to infectious will be. As a
consequence, the number of hospitalized patients may increase very
fast, and may lead to a collapse of the health system \cite{armocida2020italian}. Strong
distancing measures like confinement can effectively act on this
parameter \cite{roques2020impact, ferguson2020}, helping to keep it low.
\end{itemize}
Our forecasting strategy is motivated by the following observation: by allowing the parameters $\beta$ and $\gamma$ to be time-dependent, we can find optimal coefficients $\beta^*(t)$ and $\gamma^*(t)$ that exactly fit any series of \col{infected and removed} patients. In other words, we can perfectly fit any observed health series with a SIR model having time-dependent coefficients.

As we explain later on in the paper, the high fitting power stems from the fact that the parameters $\beta$ and $\gamma$ are searched in $L^\infty([0, T], \bR_+)$, the space of essentially bounded measurable functions. For our forecasting purposes, this space is however too large to give any predictive power and we need to find a smaller {manifold} that has simultaneously good fitting and forecasting properties. To this aim, we develop a method based on model order reduction. The idea is to find a space of reduced dimension that can host the dynamics of the current epidemic. This reduced space is learnt from a series of detailed compartmental models based on precise underlying mechanisms of the disease. One major difficulty in these models is to fit the correct parameters. In our approach, we do not seek to estimate them. Instead, we consider a large range of possible parameter configurations with a uniform sampling that allows us to simulate virtual epidemic scenarios on a longer range than the fitting window $[0, T]$. We next {cast} each virtual epidemics from the family of detailed compartmental models into the family of SIR models with time-dependent coefficients in a sense that we explain later on. This procedure yields time-dependent parameters $\beta$ and $\gamma$ for each detailed virtual epidemics. The set of all such $\beta$ (resp. $\gamma$) is then condensed into a reduced basis of small dimension. We finally use the available health data on the time window $[0, T]$ to find the functions $\beta$ and $\gamma$ from the reduced space that fit at best the current epidemics {over $[0,T]$}. Since the reduced basis functions are defined over a longer time range  $[0, T+\tau]$ with $\tau>0$ (say, e.g., two weeks), the strategy automatically provides forecasts from $T$ to $T+\tau$. Its accuracy will be related to the pertinence of the mechanistic mathematical models that have been used in the learning process.

We note that an important feature of our approach is that all virtual simulations are considered equally important on a first stage and the procedure automatically learns what are the best scenarios (or linear combinations of scenarios) to describe the available data. Moreover, the approach even mixes different compartmental models to accommodate these available data. This is in contrast to other existing approaches which introduce a strong a priori belief on the quality of a certain particular model. Note also that we can add models related to other illness, and use the large manifold to fit a possible new epidemic.  \col{It is also possible to mix the current approach with other purely statistical or learning strategies by means of expert aggregration. One salient difference with these approaches which is important to emphasize is that our method hinges on highly detailed compartmental models which have clear epidemiological interpretations. Our collapsing methodology into the time-dependent SIR is a way of “summarizing” the dynamics into a few terms}. \col{One may expect that reducing the number of parameters comes at the cost of losing interpretability of parameters, and this is in general true. Nevertheless, the numerical results of the present work show that a reasonable tradeoff  between "reduction of the number of parameters" and "interpretability of these few parameters" can be achieved.}

The paper is organized as follows. In Section \ref{sec:mono-region}, we present the forecasting method in the case of one single region with constant population. For this, we briefly introduce in Section \ref{sec:models} the epidemiological models involved in the procedure, namely the SIR model with time-dependent coefficients and more detailed compartmental models used for the training step. In Section \ref{sec:methodo}, after proving that SIR models with time-dependent coefficients in $\Linf$ \col{is able to fit any admissible epidemiological evolution (in a sense that will be made rigorous)}, we present the main steps of the forecasting method. The method involves a collapsing step from detailed models to SIR models with time-dependent coefficients, and  model reduction techniques. We detail these points in Sections \ref{sec:collapse} and \ref{sec:MOR}. In Section \ref{sec:multiregional} we explain how the method can easily be extended to a multi-regional context involving population mobility, and regional health data observations (provided, of course, that mobility data is available). We start in Section \ref{sec:basics-multiregion} by clarifying that the nature of the mobility data will dictate the kind of multi-regional SIR model to use in this context. In Section \ref{sec:forecast-multiregion} we outline how to adapt the main steps of the method to the multi-regional case. Finally, in Section \ref{sec:numerical-Paris}, we present numerical results for the the \col{two} pandemic waves of COVID-19 in France in year \col{2020}, \col{which have taken place approximately between February and November 2020}. Concluding comments are given in Section~\ref{sec:conclusion} followed by two appendices \ref{appendix:noise} and \ref{appendix:forecasting-error} that present details about the processing of the data noise, and the forecasting error.

\section{Methodology for one single region}
\label{sec:mono-region}
For the sake of clarity, we first consider the case of one single region with constant population and no population fluxes with other regions. Here, the term region is generic and may be applied to very different geographical scales, ranging from a full country, to a department within the country, or even smaller partitions of a territory.

\subsection{Compartmental models}
\label{sec:models}
The final output of our method is a mono-regional SIR model with time dependent coefficients {as exposed below}. 
{This SIR model with time-dependent coefficients will be evaluated} with reduced modeling techniques involving {a large family of }models with finer compartments {proposed in the literature}. Before presenting the method in the next section, we introduce here all the models {that we use in this paper} along with useful notations for the rest of the paper. 

\paragraph{SIR models with time-dependent parameters:} We will fit and forecast the series of \col{infected and removed patients (dead and recovered)} with SIR models where the coefficients $\beta$ and $\gamma$ are time-dependent:
\begin{align*}
\frac{d S}{dt}(t) &= - \frac{\beta(t) I(t) S(t)}{N} \\
\frac{d I}{dt}(t) &= \frac{\beta(t) I(t) S(t)}{N} - \gamma(t) I(t) \\
\frac{d R}{dt}(t) &=  \gamma(t) I(t) .
\end{align*}
In the following, we use bold-faced letters for past-time quantities. For example, $\bff \coloneqq \{ f(t) \,:\, 0\leq t \leq T \}$ for any function $\bff\in\Linf$. Using this notation, for any given $\bfbeta$ and $\bfgamma\in \Linf$ we denote by
$$
(\bfS, \bfI, \bfR) = \SIR(\bfbeta, \bfgamma, [0, T])
$$
the solution of the associated SIR dynamics in $[0, T]$.

\paragraph{Detailed compartmental models:} Models involving many compartments offer a detailed description of epidemiological mechanisms at the expense of involving many parameters. In our approach, we use them to generate virtual scenarios. One of the initial motivations behind the present work is to provide forecasts for the COVID-19 pandemic, thus we have selected the two following models which are specific for this disease, but note that any other compartmental model \cite{liu2020covid,di9impact,magal2020predicting} or agent-based simulation could also be used .
\begin{itemize}
\item \textbf{First model}, $\Colizza$: This model is inspired from the one proposed in \cite{di9impact}. It involves 11 different compartments and a set of 19 parameters. The dynamics of the model is illustrated in Figure \ref{fig:Colizza} and the system of equations reads 
\begin{align*}
\frac{dS}{dt}(t)&=  - \frac{1}{N}  S(t)  \left( \beta_ p I_p(t)  +  \beta_a I_a(t)  +  \beta_{ps} I_{ps}(t)  +  \beta_{ms} I_{ms}(t)  + \beta_{ss} I_{ss}(t)+ \beta_H H(t) + \beta_C C(t) \right) \\
\frac{dE}{dt}(t)&=  \frac{1}{N} S(t)\left( \beta_ p I_p(t)  +  \beta_a I_a(t)  +  \beta_{ps} I_{ps}(t)  +  \beta_{ms} I_{ms}(t)  + \beta_{ss} I_{ss}(t) + \beta_H H(t) + \beta_C C(t) \right)   - \varepsilon E(t)\\
\frac{dI_p}{dt}(t)&=  \varepsilon E(t)  -  \mu_p I_p(t) \\
\frac{dI_a}{dt}(t)&=      p_a \mu_p I_p(t)   - \mu I_a(t)\\
\frac{dI_{ps}}{dt}(t)&=  p_{ps}(1-p_a) \mu_p I_p(t)   - \mu I_{ps}(t) \\
\frac{dI_{ms}}{dt}(t)&= p_{ms}(1-p_a) \mu_p I_p(t)   - \mu I_{ms}(t) \\
\frac{dI_{ss}}{dt}(t)&=  p_{ss}(1-p_a )\mu_p I_p(t)   - \mu I_{ss}(t) \\
\frac{dC}{dt}(t)&=  p_c \mu I_{ss}(t)  - (\lambda_{C,R} +\lambda_{C,D} ) C(t)\\ 
\frac{dH}{dt}(t)&=    (1-p_c) \mu I_{ss}(t)  - (\lambda_{H,R} + \lambda_{H,D}) H(t)  \\ 
\frac{dR}{dt}(t)&= \lambda_{C,R} C(t)  + \lambda_{H,R} H(t)\\
\frac{dD}{dt}(t)&= \lambda_{C,D} C(t)  + \lambda_{H,D} H(t) \\ 
\end{align*}
The different parameters involved in the model are described in Table \ref{tab:params_Colizza} and detailed in the appendix of \cite{di9impact}.

\begin{figure}
\begin{center}
\resizebox{0.65\columnwidth}{!}{\tikzstyle{compartment}=[rectangle,draw=black,fill=black!20,thick,minimum size=10mm]
\begin{tikzpicture}[node distance=2cm]
\node[compartment] (Ips) {$I_{ps}$};
\node (ghost) [below=0.3cm of Ips] {};
\node[compartment] (Ims) [below=0.3cm of ghost] {$I_{ms}$};
\node[compartment] (Iss) [below=0.6cm of Ims] {$I_{ss}$};
\node[compartment] (Ia) [above=0.6cm of Ips] {$I_{a}$};
\node[compartment] (Ip) [left of=ghost] {$I_p$};
\node[compartment] (E) [left of=Ip] {$E$};
\node[compartment] (S) [left of=E] {$S$};
\node[compartment] (C) [right=1.5cm of Iss] {$C$};
\node[compartment] (D) [right=1.5cm of C] {$D$};
\node[compartment] (R) [above=1.5cm of D] {$R$};
\node[compartment] (H) [above=0.1cm of C] {$H$};
\draw [->, very thick] (Ip.east) to node {} node [above] {} (Ia.west);
\draw [->, very thick] (Ip.east) to node {} node [above] {} (Iss.west);
\draw [->, very thick] (Ia.east) to node {} node [above] {} (R.west);
\draw [->, very thick] (Iss.east) to node {} node [above] {} (H.west);
\draw [->, very thick] (Iss.east) to node {} node [above] {} (C.west);
\draw [->, very thick] (C.east) to node {} node [above] {} (D.west);
\draw [->, very thick] (C.east) to node {} node [above] {} (R.west);
\draw [->, very thick] (H.east) to node {} node [above] {} (D.west);
\draw [->, very thick] (H.east) to node {} node [above] {} (R.west);
\draw [->, very thick] (S.east) to node {} node [above] {} (E.west);
\draw [->, very thick] (E.east) to node {} node [above] {} (Ip.west);
\draw [->, very thick] (Ip.east) to node {} node [left] {} (Ips.west);
\draw [->, very thick] (Ip.east) to node {} node [left] {} (Ims.west);
\draw [->, very thick] (Ips.east) to node {} node [right] {} (R.west);
\draw [->, very thick] (Ims.east) to node {} node [right] {} (R.west);

\end{tikzpicture}}
\end{center}
\caption{$\Colizza$ model}
\label{fig:Colizza}
\end{figure}

\begin{table}
\begin{center}
\begin{tabular}{|c|c|}
\hline
Compartment & Description  \\ \hline
   $S$  & Susceptible \\
   $E$  &  Exposed (non infectious)\\
   $I_p$ & Infected and pre-symptomatic (already infectious) \\
   $I_a$ &  Infected and a-symptomatic (but infectious)\\
   $I_{ps}$  & Infected and paucisymptomatic \\
   $I_{ms}$  &  Infected with mild symptoms\\
   $I_{ss}$ &  Infected with severe symptoms\\
   $H$ &  Hospitalized\\
   $C$  &  Intensive Care Unit\\
   $R$  & Removed \\
   $D$  &  Dead\\ \hline
\end{tabular}
\end{center}
\caption{Description of the compartments in Model \Colizza}
\label{tab:compartments_Colizza}
\end{table}%

\begin{table}
\begin{center}
\begin{tabular}{|c|c|}
\hline
Parameter & Description  \\ \hline
   $ \beta_{p}$  & Relative infectiousness of $I_p$  \\
    $\beta_{a}$  &  Relative infectiousness of $I_a$ \\
   $ \beta_{ps} $ & Relative infectiousness of $I_{ps}$  \\
    $\beta_{ms} $ & Relative infectiousness of $I_{ms}$  \\
    $\beta_{ss} $ & Relative infectiousness of $I_{ss}$ \\
    $\beta_H  $ & Relative infectiousness of $I_{H}$ \\
    $\beta_C$   & Relative infectiousness of $I_{C}$ \\
    $\varepsilon^{-1}$    &Latency period \\
    $\mu_p^{-1} $  & Duration of prodromal phase \\
    $p_a $ & Probability of being asymptomatic   \\
    $\mu^{-1} $ & Infectious period of $I_a$, $I_{ps}$, $I_{ms}$, $I_{ss}$ \\
    $p_{ps}$  & If symptomatic, probability of being paucisymptomatic \\
    $p_{ms} $  &  If symptomatic, probability of developing mild symptoms\\
    $p_{ss}$  & If symptomatic, probability of developing severe symptoms {(note that $p_{ps}+p_{ms}+p_{ss}=1$)} \\
    $p_C $  & If severe symptoms, probability of going in C \\
   $\lambda_{CR}$   & If in C, daily rate entering in $R$ \\
  $\lambda_{CD} $  & If in C, daily rate entering in $D$ \\
   $\lambda_{HR}$   & If hospitalized, daily rate entering in $R$ \\
   $\lambda_{HD}  $   &  If hospitalized, daily rate entering in $D$ \\ \hline
\end{tabular}
\end{center}
\caption{Description of the parameters involved in Model \Colizza}
\label{tab:params_Colizza}
\end{table}%
We denote by 
\begin{align}
\bfu = \Colizza(\bfu_0, \beta_{p}, & \beta_{a},  \beta_{ps},  \beta_{ms},  \beta_{ss},  \beta_{H},  \beta_{C},\\& \varepsilon,  \mu_p, p_a , \mu, p_{ps}, p_{ms}, p_{ss}, p_C,\\& \lambda_{CR}, \lambda_{CD} , \lambda_{HR} ,\lambda_{HD}, [0, T])
\end{align}
the parameter-to-solution map where $\bfu= (\bfS, \bfE, \bfIp, \bfIa, \bfIps, \bfIms, \bfIss, \bfC, \bfH, \bfR, \bfD)$.

\item \textbf{Second model}, $\Magal$: This model is a variant of the one proposed in \cite{magal2020predicting}. It involves 5 different compartments and a set of 6 parameters. The dynamics of the model is illustrated in Figure \ref{fig:Magal} and the set of equations is

\begin{align*}
\frac{dS}{dt}(t)&=  - \frac{1}{N} \beta S(t) (E_2(t) + U(t) + I(t)) \\
\frac{dE_1}{dt}(t)&=  \frac{1}{N} \beta S(t) (E_2(t) + U(t) + I(t))  - \delta E_1(t)\\
\frac{dE_2}{dt}(t)&=  \delta E_1(t) - \sigma E_2(t) \\
\frac{dI}{dt}(t)&=   \nu \sigma E_2(t) - \gamma_1 I(t) \\
\frac{dU}{dt}(t)&=   (1-\nu) \sigma E_2(t) - \gamma_2 U(t) \\
\frac{dR}{dt}(t)&=   \gamma_1 I(t)+\gamma_2 U(t) \\
\end{align*}

\begin{figure}
\begin{center}
\resizebox{0.5\columnwidth}{!}{\tikzstyle{compartment}=[rectangle,draw=black,fill=black!20,thick,minimum size=10mm]
\begin{tikzpicture}[node distance=2cm]
\node[compartment] (I) {$I$};
\node (ghost) [below=0.5cm of I] {};
\node[compartment] (U) [below=0.5cm of ghost] {$U$};
\node[compartment] (R) [right of=ghost] {$R$};
\node[compartment] (E2) [left of=ghost] {$E_2$};
\node[compartment] (E1) [left of=E2] {$E_1$};
\node[compartment] (S) [left of=E1] {$S$};
\draw [->, very thick] (S.east) to node {} node [above] {} (E1.west);
\draw [->, very thick] (E1.east) to node {} node [above] {} (E2.west);
\draw [->, very thick] (E2.east) to node {} node [left] {} (I.west);
\draw [->, very thick] (E2.east) to node {} node [left] {} (U.west);
\draw [->, very thick] (I.east) to node {} node [right] {} (R.west);
\draw [->, very thick] (U.east) to node {} node [right] {} (R.west);


\end{tikzpicture}}
\end{center}
\caption{$\Magal$ model} 
\label{fig:Magal}
\end{figure}
We denote by 
$$
\bfu = \Magal (\bfu_0,\beta, \delta, \sigma, \nu, \gamma_1, \gamma_2, [0, T])
$$
the parameter-to-solution map where $\bfu = (\bfS, \bfEun, \bfEdeux, \bfI, \bfU, \bfR)$. The different parameters involved in the model are described in Table \ref{tab:params_Magal}. 

\begin{table}
\begin{center}
\begin{tabular}{|c|c|}
\hline
Compartment & Description  \\ \hline
   $S$  & Susceptible \\
   $E_1$  &  Exposed (non infectious)\\
   $E_2$ & Infected and pre-symptomatic (already infectious) \\
   $I$ &  Infected and symptomatic\\
   $U$  & Un-noticed \\
   $R$  &  dead and removed\\ \hline
\end{tabular}
\end{center}
\caption{Description of the compartments in Model \Magal}
\label{tab:compartments_Magal}
\end{table}%

\begin{table}
\begin{center}
\begin{tabular}{|c|c|}
\hline
Parameter & Description  \\ \hline
   $ \beta$  & Relative infectiousness of $I$, $U$, $E_2$ \\
    $ \delta^{-1}$ & Latency period\\
    $ \sigma^{-1}$ & Duration of prodromal phase \\
    $ \nu$ & Proportion of $I$ among $I+U$\\
    $ \gamma_1$ & If $I$, daily rate entering in $R$\\
    $ \gamma_2$ & If $U$, daily rate entering in $R$ \\ \hline
\end{tabular}
\end{center}
\caption{Description of the parameters involved in Model \Magal}
\label{tab:params_Magal}
\end{table}

\item \textbf{Generalization:} In the following, we abstract the procedure as follows. For a any $\DM$ with $d$ compartments involving a vector $\mu\in\bR^p$ of $p$ parameters, we denote by 
$$
\bfu(\mu) = \DM(\mu, [0, \widetilde{T}]), \quad \bfu(\mu) \in \Linfdp
$$
the parameter-to-solution map, {where $\widetilde{T}$ is some given time simulation that can be as large as wished because this is a virtual scenario. Note that, because the initial condition of the illness at time 0 is not known, we include in the parameter set the initial condition $\bfu_0$.
}
\end{itemize}

\subsection{Forecasting based on model reduction of detailed models}
\label{sec:methodo}
We assume that we are given health data in a time window $[0, T]$, where $T>0$ is assumed to be the present time. The observed data is the series of \col{infected} people, denoted $I_\obs$, and removed people denoted $R_\obs$.  They are usually given at a national or a regional scale and on a daily basis. For our discussion, it will be useful to work with time-continuous functions and $t\to I_\obs(t)$ will denote the piecewise constant approximation in $[0, T]$ from the given data (and similarly for $R_\obs(t)$). Our goal is to give short-term forecasts of the series in a time window $\tau>0$ {whose size}  will be about two weeks. We denote by $I(t)$ and $R(t)$ the approximations to the series $I_\obs(t)$ and $R_\obs(t)$ at any time $t\in[0, T]$.

As already brought up, we propose to fit the data and provide forecasts with SIR models with time-dependent parameters $\bfbeta$ and $\bfgamma$. The main motivation for using such a simple family is because it possesses optimal fitting and forecasting properties for our purposes in the sense that we explain next. Defining the cost function
\begin{equation}
\cJ(\bfbeta, \bfgamma \, |\, I_\obs(t), R_\obs(t), [0,T]) \coloneqq \int_0^T \left( \vert I(t) - I_\obs(t)\vert^2 + \vert R(t) - R_\obs(t)\vert^2 \right)\dt
\label{eq:cJ}
\end{equation}
such that
$$
(\bfS, \bfI, \bfR)= \SIR(\bfbeta, \bfgamma, [0, T]),
$$
the fitting problem can be expressed at the continuous level as the optimal control problem of finding 
\begin{equation}
\label{eq:fit-1}
J^* = \inf_{(\bfbeta, \bfgamma) \in \Linf\times \Linf} \cJ(\bfbeta, \bfgamma \, |\, \bfI_\obs, \bfR_\obs, [0,T]).
\end{equation}
The following result ensures the existence of a unique minimizer under very mild constraints.

\begin{Proposition}
\label{prop:SIR}
Let $N \in \bN^*$ and $T>0$. For any real-valued functions $S_\obs,I_\obs,R_\obs$ {of class ${\mathcal C}^1$, }defined on $[0,T]$ satisfying 
\begin{enumerate}
\item[(i)] $S_\obs(t)+I_\obs(t)+R_\obs(t) = N$ for every $t \in [0,T]$,
\item[(ii)] $S_\obs$ in nonincreasing on $[0,T]$,
\item[(iii)]$R_\obs$ is nondeacreasing on $[0,T]$, 
\end{enumerate}
there exists a unique minimizer $(\bfbeta^*_\obs, \bfgamma^*_\obs)$ to problem \eqref{eq:fit-1}.
\end{Proposition}
\begin{proof}
One can set 
\begin{align}
\begin{cases}
\beta^*_\obs(t) &\coloneqq -  \dfrac{N}{ I_\obs(t) S_\obs(t)} \dfrac{d S_\obs}{dt}(t)  \\
\gamma^*_\obs(t) &\coloneqq \dfrac{1}{I_\obs(t)} \dfrac{d R_\obs}{dt}(t)
\end{cases}
\label{eq:bg-star}
\end{align}
so that
$$
(\bfS_\obs, \bfI_\obs, \bfR_\obs)= \SIR(\bfbeta^*, \bfgamma^*, [0, T])
$$
and
$$
\cJ(\bfbeta^*_\obs, \bfgamma^*_\obs,[0,T]) = 0
$$
which obviously implies that $J^* = 0$. 
\end{proof}

{Note that one can  equivalently set
\begin{align}
\begin{cases}
\beta^*_\obs(t) &\coloneqq -  \dfrac{N}{ I_\obs(t) S_\obs(t)} \dfrac{d S_\obs}{dt}(t)  \\
\gamma^*_\obs(t) &\coloneqq \dfrac{1}{I_\obs(t)} \left[\dfrac{d I_\obs}{dt}(t) - \dfrac{\beta^*_\obs(t) I_\obs(t) S_\obs(t)}{N} \right]
\end{cases}
\label{eq:bg-star-var}
\end{align}
or again
\begin{align}
\begin{cases}
\gamma^*_\obs(t) &\coloneqq \dfrac{1}{I_\obs(t)} \dfrac{d R_\obs}{dt}(t)\\
\beta^*_\obs(t) &\coloneqq \dfrac{N}{I_\obs(t)S_\obs(t)} \left[\dfrac{d I_\obs}{dt}(t) - \gamma^*_\obs(t)  I_\obs(t)  \right]
\end{cases}
\label{eq:bg-star-var}
\end{align}
}

This simple observation means that there exists a time-dependent SIR model which can  perfectly fit the data of any (epidemiological) \col{evolution} that satisfies properties (i), (ii) and (iii).  In particular, we can perfectly fit the series of {sick} people with a time-dependent SIR model (modulo a smoothing of the local oscillations due to noise). Since the health data are usually given on a daily basis, we can approximate  $\bfbeta^*_\obs, \bfgamma^*_\obs$ by approximating the derivatives by classical finite differences in equation \eqref{eq:bg-star}.

The fact that we can build $\bfbeta^*_\obs$ and $\bfgamma^*_\obs$ such that $\cJ(\bfbeta^*_\obs, \bfgamma^*_\obs) = J^* =0$ implies that the family of time-dependent SIR models is rich enough not only to fit the evolution of any epidemiological series, but also to deliver perfect predictions of the health data. However, this great approximation power comes at the cost of defining the parameters $\bfbeta$ and $\bfgamma$ in $\Linf$ which is a space that is too large in order to be able to define any feasible prediction strategy.

In order to pin down a smaller {manifold} where these parameters may vary without sacrificing much on the fitting and forecasting power, our strategy is the following: 
\begin{enumerate}
\item \textbf{Learning phase:} {The fundamental hypothesis  of our approach is the confidence that the specialists of epidemiology have well understood the \col{mechanisms of infection transmission}. The second hypothesis is that these accurate models suffer from the proper determination of the parameters they contain. We thus propose to generate a large number of virtual epidemics, following these mechanistic models with parameters that can be chosen in a vicinity of the suggested parameters values, including the  different initial conditions.}
\begin{enumerate}
\item \textbf{Generate virtual scenarios using detailed models with constant coefficients:}
\begin{itemize}
\item Define the notion of $\DM$ which is most appropriate for the epidemiological study. Several models could be considered. For our numerical application, the detailed models were defined in Section \ref{sec:models}.
\item Define an interval range $\cP\subset\bR^p$ where the parameters $\mu$ of $\DM$ will vary. {We call the solution manifold $\cU$ the set of virtual dynamics over $[0, T+\tau]$, namely
\begin{equation}
\cU \coloneqq \{ \bfu(\mu) = \DM(\mu, [0, T+\tau]) \; : \; \mu \in \cP \}.
\end{equation}
}
\item Draw a finite training set
\begin{equation}
\cP_{\tr} = \{\mu_1,\dots, \mu_K\} \subseteq \cP
\end{equation}
of $K\gg 1$ parameter instances and we compute $\bfu(\mu) = \DM(\mu, [0, T+\tau])$ for $\mu \in \cP_{\tr}$. Each $\bfu(\mu)$ is a virtual epidemiological scenario. An important detail for our prediction purposes is that the simulations are done in $[0, T+\tau]$, that is, we simulate not only in the fitting time interval but also in the prediction time interval. We call
$$
\cU_\tr = \{\bfu(\mu) \; : \; \mu \in \cP_{\tr}\}
$$
the training set of all virtual scenarios. 
\end{itemize}
\item \textbf{Collapse:} Collapse every detailed model $\bfu(\mu) \in \cU_\tr$ into a SIR model following the ideas which we explain in Section \ref{sec:collapse}. For every $\bfu(\mu)$, the procedure gives time-dependent parameters $\bfbeta(\mu)$ and $\bfgamma(\mu)$ and associated SIR solutions $(\bfS, \bfI, \bfR)(\mu)$ in $[0, T+\tau]$. {This yields the sets
\begin{equation}
\cB_\tr \coloneqq
\{\bfbeta(\mu) \; : \; \mu \in \cP_\tr \}
\quad  \text{and}\quad
\cG_\tr \coloneqq
\{\bfgamma(\mu) \; :\; \mu \in \cP_\tr\}.
\label{eq:training-sets-beta-gamma}
\end{equation}
}
\item \textbf{Compute reduced models:}

We apply model reduction techniques using $\cB_\tr$ and $\cG_\tr$ as training sets in order to build two basis
$$
\rB_n = \vspan \{b_1,\dots, b_n\},
\quad
\rG_n = \vspan \{g_1,\dots, g_n\} \subset L^\infty([0, T+\tau], \bR),
$$
which are defined over $[0, T+\tau]$.  {The space $\rB_n$ is such that it approximates well all functions $\bfbeta(\mu) \in \cB_\tr$ (resp.~all $\bfgamma(\mu)\in \cG_\tr$ can be well approximated by elements of  $\rG_n$)}. We outline in Section \ref{sec:MOR} the methods we have  explored in our numerical tests.
\end{enumerate}
\item \textbf{Fitting on the reduced spaces:} We next solve the fitting problem \eqref{eq:fit-1} in the interval $[0, T]$ by searching $\bfbeta$ (resp. $\bfgamma$) in $\rB_n$ (resp. in $\rG_n$) instead of in $\Linf$, that is,
\begin{equation}
\label{eq:JstarMOR}
J^*_{(\rB_n, \rG_n)}
= \min_{(\bfbeta, \bfgamma) \in \rB_n\times \rG_n}
\cJ(\bfbeta, \bfgamma \; |\; \bfI_\obs, \bfR_\obs, [0, T]).
\end{equation}
Note that the respective dimensions of $\rB_n$ and $\rG_n$ can be different, for simplicity we take them equal in the following.
Obviously, since $\rB_n$ and $\rG_n\subset \Linf$, we have that
$$
J^* \leq J^*_{(\rB_n, \rG_n)},
$$
but we numerically observe that the function $n\mapsto J^*_{(\rB_n, \rG_n)}$ decreases very rapidly as $n$ increases, which indirectly confirms the fact that the manifold generated by the two above models accommodates well the COVID-19 epidemics.

The solution of problem \eqref{eq:JstarMOR} gives us coefficients $(c^*_i)_{i=1}^n$ and $(\tilde c^*_i)_{i=1}^n\in \bR^n$ such that the time-dependent parameters
\begin{align*}
\beta^*_n(t) &= \sum_{i=1}^n c^*_i b_i(t), \quad \forall t \in [0, T+\tau], \\
\gamma^*_n(t) &= \sum_{i=1}^n \tilde c^*_i g_i(t).
\end{align*}
{achieve} the minimum \eqref{eq:JstarMOR}.
\item \textbf{Forecast:} For a given dimension $n$ of the reduced spaces, propagate in $[0, T+\tau]$ the associated SIR model
$$
(\bfS_n^*, \bfI_n^*, \bfR_n^*) = \SIR(\bfbeta^*_n, \bfgamma^*_n, [0, T+\tau])
$$
The values $I_n^*(t)$ and $R_n^*(t)$ for $t\in [0, T[$ are by construction close to the observed data $\bfI_\obs, \bfR_\obs$ (up to some numerical optimization error). The values $I_n^*(t)$ and $R_n^*(t)$ for $t\in [T, T+\tau]$ are then used for prediction.
\item \textbf{Forecast Combination/Aggregation of Experts (optional step):} By varying the dimension $n$ and using different model reduction approaches, we can easily produce a collection of different forecasts and the question of how to select the best predictive model arises. Alternatively, we can also resort to Forecast Combination techniques \cite{poncela2011forecast}: denoting $(I_1, R_1), \dots, (I_P, R_P)$ the different forecasts, the idea is to search for an appropriate linear combination
$$
I^{\text{FC}}(t) = \sum_{p=1}^P w_p I_p(t)
$$
and similarly for $R$. Note that these combinations do not need to involve forecasts from our methodology only. Other approaches like time series forecasts could also be included. One simple forecast combination is the average, in which all alternative forecasts are given the same weight $w_p = 1/P,\, p=1,\dots P$. More elaborate approaches consist in estimating the weights that minimize a loss function involving the forecast error.
\end{enumerate}
Before going into the details of some of the steps, three remarks are in order:
\begin{enumerate}
\item To bring out the essential mechanisms, we have idealized some elements in the above discussion by omitting certain unavoidable discretization aspects. To start with, the ODE solutions cannot be computed exactly but only up to some accuracy given by a numerical integration scheme. In addition, the optimal control problems \eqref{eq:fit-1} and \eqref{eq:JstarMOR} are non-convex. As a result, in practice we can only find a local minima. Note however that modern solvers find solutions which are very satisfactory for all practical purposes. In addition, note that solving the control problem in a reduced space as in \eqref{eq:JstarMOR} could be interpreted as introducing a regularizing effect with respect to the control problem \eqref{eq:fit-1} in the full $\Linf$ space. It is to be expected that the search of global minimizers is facilitated in the reduced landscape.
\item \textbf{\fitIR~and \fitbg:} A variant for the fitting problem \eqref{eq:JstarMOR} which we have studied in our numerical experiments is to replace the cost function $\cJ(\bfbeta, \bfgamma \; |\; \bfI_\obs, \bfR_\obs, [0, T])$ by the cost function 
\begin{equation}
\label{eq:cJ2}
\widetilde{\cJ}(\bfbeta, \bfgamma \, | \, \bfbeta^*_\obs, \bfgamma^*_\obs, [0,T]) \coloneqq \int_0^T \left( \vert \bfbeta - \bfbeta^*_\obs \vert^2 + \vert \bfgamma - \bfgamma^*_\obs)\vert^2 \right)\dt.
\end{equation}
In other words, we use the variables $\bfbeta^*_\obs$ and $\bfgamma^*_\obs$ from \eqref{eq:bg-star} as observed data instead of working with the observed health series $\bfI_\obs$, $\bfR_\obs$. In Section \ref{sec:numerical-Paris}, we will refer to the standard fitting method as \fitIR~and to this variant as \fitbg.
\item \col{The fitting procedure works both on the components of the reduced basis and the initial time of the epidemics to minimize the loss function but, for simplicity, this last optimization is not reported in the notation.}
\end{enumerate}

\subsection{Details on Step 1-(b): Collapsing the detailed models into SIR dynamics}
\label{sec:collapse}
Let
$$
\bfu(\mu) = \DM(\mu, [0, T+\tau]) \; \in \Linftaud
$$
be the solution in $[0, T+\tau]$ to a detailed model for a given vector of parameters $\mu \in \bR^d$. Here $d$ is possibly large ($d=11$ in the case of $\Colizza$ model and $d=5$ in the case of $\Magal$'s one). The goal of this section is to explain how to collapse the detailed dynamics $\bfu(\mu)$ into a \SIR~dynamics with time-dependent parameters. The procedure can be understood as a projection of a detailed dynamics into the family of dynamics given by \SIR~models with time-dependent parameters.

For the $\Colizza$ model, we collapse the variables by setting
\begin{align*}
\bfS^\collapse &= \bfS + \bfE \\
\bfI^\collapse &= \bfIp + \bfIa + \bfIps + \bfIms + \bfIss +\bfC + \bfH \\
\bfR^\collapse &= \bfR + \bfD
\end{align*}
Similarly, for the $\Magal$ model, we set
\begin{align*}
\bfS^\collapse &=  \bfS + \bfE_{1i} \\
\bfI^\collapse &=   \bfE_{2i} + \bfI_{i} + \bfU_{i} \\
\bfR^\collapse &= \bfR
\end{align*}
Note that $\bfS^\collapse $, $\bfI^\collapse $ and $\bfR^\collapse $ depend on $\mu$ but we have omitted the dependence to lighten the notation.

Once the collapsed variables are obtained, we identify the time-dependent parameters $\bfbeta$ and $\bfgamma$ of the $\SIR$ model by solving the fitting problem
\begin{equation}
\label{eq:fit-2}
(\bfbeta,\bfgamma)  \in \arginf_{(\bfbeta, \bfgamma) \in    L^\infty([0, T+\tau], \bR)  \times L^\infty([0, T+\tau], \bR)  } \cJ(\bfbeta, \bfgamma \, |\, \bfI^\collapse, \bfR^\collapse, [0,T+\tau])
\end{equation}
where
$$
(\bfS, \bfI, \bfR)= \SIR(\bfbeta, \bfgamma, [0, T+\tau]).
$$
{Note that problem \eqref{eq:fit-2} has the same structure as problem \eqref{eq:fit-1}, the difference coming from the fact that the collapsed variables $\bfI^\collapse$, $\bfR^\collapse$ in \eqref{eq:fit-2} play the  role of the health data $\bfI_\obs$, $\bfR_\obs$ in \eqref{eq:fit-1}. Therefore, it follows from Proposition \ref{prop:SIR} that problem \eqref{eq:fit-2}  has a very simple solution since it suffices to apply formula \eqref{eq:bg-star} for solving it. Note here that the exact derivatives of $\bfS^\collapse$, $\bfI^\collapse$ and $\bfR^\collapse$ can be obtained from the \DM.}

Since the solution $(\bfbeta,\bfgamma)$ to \eqref{eq:fit-2} depends on the parameter $\mu$ of the detailed model, repeating the above procedure for every detailed scenario $\bfu(\mu)$ {for any $\mu\in\cP_{\tr} $} yields the two families of time-dependent functions {$\cB_\tr=\{ \bfbeta(\mu)  \; : \; \mu \in \cP_\tr\} $  and $\cG_\tr=\{ \bfgamma(\mu)  \; : \; \mu \in \cP_\tr\} $ defined on the interval $[0, T+\tau]$ which were introduced in section \eqref{eq:training-sets-beta-gamma}.}

\subsection{Details on Model Order Reduction}
\label{sec:MOR}
Model Order Reduction is a family of methods aiming at approximating a set of solutions of parametrized PDEs or ODEs (or related quantities) with linear spaces, which are called reduced models or reduced spaces. In our case, the sets to approximate are
$$
\cB = \{ \bfbeta(\mu) \,:\, \mu \in \cP \}
\quad \text{and}\quad 
\cG = \{ \bfgamma(\mu) \,:\, \mu \in \cP \},
$$
where each $\mu$ is the vector of parameters of the detailed model which take values over $\cP$, and $\bfbeta(\mu)$ and $\bfgamma(\mu)$ are the associated time-dependent coefficients of the collapsed SIR evolution. In the following, we view $\cB $ and $\cG$ as subsets of $L^2([0,T])$, and we denote by $\Vert  \cdot \Vert$ and $\left<\cdot, \cdot\right>$ its norm and inner product. Indeed, in view of Proposition \ref{prop:SIR}, $\cB$ and $\cG\subset\Linf \subset L^2([0, T])$.

Let us carry the discussion for $\cB$ (the same will hold for $\cG$). If we measure performance in terms of the worst error in the set $\cB$, the best possible performance that reduced models of dimension $n$ can achieve is given by the Kolmogorov $n$-width
$$
d_n(\cB )_{L^2([0, T])} \coloneqq \inf_{\substack{Y\in L^2([0, T]) \\ \dim(Y)\leq n}} \max_{u\in \cB} || u - P_Y u \Vert
$$
where $P_Y$ is the orthogonal projection onto the subspace $Y$. In the case of measuring errors in an average sense, the benchmark is given by
$$
\delta^2_n(\cB, \nu )_{L^2([0, T])} \coloneqq \inf_{\substack{Y\in L^2([0, T])\\ \dim(Y)=n}} \int_{\cP} || u(y) - P_Y u(y) \Vert^2 \dd \nu(y)
$$
where $\nu$ is {some given} measure on $\cP$.

In practice, building spaces that meet these benchmarks is {generally} not possible. However, it is possible to build {sequences of} spaces {for which their error decay}  is comparable to the one given by $\left( d_n(\cB )_{L^2([0, T])} \right)_n$ or $\left( \delta_n(\cB )_{L^2([0, T])} \right)_n$. As a result, when the Kolmogorov width decays fast, the constructed reduced spaces will deliver a very good approximation of the set $\cB$ with few modes (see \cite{BCDDPW2011, DPW2013, CD2015, MMT2016}).

We next present the reduced models that we have used in our numerical experiments. Other methods could of course be considered and we refer to  \cite{QMN2015, BCOW2017, HRS2019, MP2020} for general references on model reduction. We carry the discussion in a fully discrete setting in order to simplify the presentation and keep it as close to the practical implementation as possible. All the claims below could be written in a fully continuous sense at the expense of introducing additional mathematical objects such as certain Hilbert-Schmidt operators to define the continuous version of the Singular Value Decomposition (SVD).

We build the reduced models using the two discrete training sets of functions $\cB_\tr = \{ \bfbeta_i \}_{i=1} ^{K}$ and $\cG_\tr = \{ \bfgamma_i \}_{i=1} ^{K}$ from $\cB$ and $\cG$ where $K$ denotes the number of virtual scenarios considered. The sets have been generated in step 1-(b) of our general pipeline (see Section \ref{sec:methodo}).  

We consider a discretization of the time interval  $[0, T+\tau]$ into a set of $Q \in \bN^*$ points as follows $\{t_1=0, \cdots, t_P =T, \cdots, t_Q= T+\tau\}$ where $P < Q$. Thus, we can represent each function $\bfbeta_i$ as a vector of $Q$ values  
$$
\bfbeta_i = (\beta_i(t_1), \cdots,  \beta_i(t_Q))^T \in \bR_+^Q.
$$
and hence assemble all the functions of the family $\{ \bfbeta_i \}_{i=1} ^{K}$ into a matrix $M_\cB \in \bR_+^{Q\times K}$. The same remark applies for the family $\{ \bfgamma_i \}_{i=1} ^{K}$ that gives a matrix $M_\cG \in \bR_+^{Q \times K}$.

\begin{enumerate}
\item \textbf{SVD:} The eigenvalue decomposition of the correlation matrix $M_\cB^T M_\cB\in \bR^{K \times K}$ writes
$$
M_\cB^T M_\cB = V \Lambda V^T,
$$ 
where $V=(v_{i,j})\in \bR^{K\times K}$ is an orthogonal matrix and $\Lambda\in \bR^{K\times K}$ is a  diagonal matrix with non-negative entries which we denote $\lambda_i$ and order them in decreasing order. The {$\ell^2(\bR^{Q})$}-orthogonal basis functions $\{\bfb_1,\dots, \bfb_K\}$ are then given by the linear combinations
$$
\bfb_i = \sum_{j=1}^K v_{j,i} \bfbeta_j,\quad 1\leq i \leq K.
$$
{For $n\le K$, } the space
$$
\rB_n = \vspan\{ \bfb_1,\dots \bfb_n \}
$$
is the best $n$-dimensional space to approximate the set $\{\bfbeta_i\}_{i=1}^K$ in the average sense. We have
$$
\delta_n( \{\bfbeta_i\}_{i=1}^K )_{\ell^2(\bR^{Q+1})} = 
\left(\frac{1}{K} \sum_{i=1}^K || \bfbeta_i - P_{\rB_n} \bfbeta_i \Vert_{\ell^2(\bR^{Q+1})}^2 \right)^{1/2}
=
\left( \sum_{i>n}^K \lambda_i \right)^{1/2}
$$
and the average approximation error is given by the sum of the tail of the eigenvalues.

Therefore the SVD method is particularly efficient if there is a fast decay of the eigenvalues, meaning that the set $\cB_\tr = \{\bfbeta_i\}_{i=1}^K$ can be approximated by only few modes. However, note that by construction, this method does not ensure positivity in the sense that $P_{\rB_n} \beta_i(t)$ may become negative for some $t\in [0,T]$ although the original function $\bfbeta_i(t)\geq 0$ for all $t\in [0,T]$. This is due to the fact that the vectors {$\bfb_i$} are not necessarily nonnegative. As we will see later, in our study, ensuring positivity especially for extrapolation (i.e., forecasting) is particularly important, and motivates the next methods.

\item \textbf{Nonnegative Matrix Factorization} (NMF, see \cite{PT1994,gillis2014and}): NMF is a variant of SVD involving nonnegative modes and expansion coefficients. In this approach, we build a family of non-negative functions $\{\bfb_i\}_{i=1}^n$ and we approximate each $\bfbeta_i$ with a linear combination 
\begin{equation}
\label{eq:LC-NMF}
\bfbeta^{\NMF}_i  = \sum_{j=1}^n  a_{i,j} \bfb_j,\quad 1\leq i \leq K,
\end{equation}
where for every $1\leq i\leq K$ and $1 \leq j \leq n$, the coefficients $a_{i,j} \geq 0$ and the basis function $\bfb_j \geq 0$.  In other words, we solve the following constrained optimization problem 
\begin{equation*}
{(W^*, B^*)  \in \argmin_{(W,B) \in  \bR_+^{K \times n}  \times  \bR_+^{n \times Q}  }  \| M_\cB - WB \|_F^2 .}
\end{equation*}
We refer to \cite{gillis2014and} for further details on the NMF and its numerical aspects.

\item \textbf{Enlarged Nonnegative Greedy (ENG): greedy algorithm with projection on an extended cone of positive functions:}
\col{We now present our novel model reduction method}, which is of interest in itself since it allows to build reduced models that preserve positivity and even other types of bounds. The method stems from the observation that NMF approximates functions in the cone of positive functions of $\vspan\{\bfb_i \geq 0\}_{i=1}^n$ since it imposes that $a_{i,j}\geq 0$ in the linear combination \eqref{eq:LC-NMF}. However, note that the positivity of the linear combination is not equivalent to the positivity of the coefficients $a_{i,j}$ since  there are obviously linear combinations involving very small $a_{i,j}< 0$ for some $j$ which may still deliver a nonnegative linear combination $\sum_{j=1}^n  a_{i,j} \bfb_j $. We can thus widen the cone of linear combinations yielding positive values by carefully including these negative coefficients $a_{i,j}$. One possible strategy for this is proposed in Algorithm \ref{alg:cone}, which describes a routine that we call $\ENLARGECONE$. The routine
$$
\{ \bfpsi_1,\dots, \bfpsi_n \} = \ENLARGECONE[\{ \bfb_1,\dots, \bfb_n \}, \e]
$$
 takes a  set of nonnegative functions $\{\bfb_1, \dots, \bfb_n\}$ as input, and modifies each function $\bfb_i$ by iteratively adding negative linear combinations of the other basis functions $\bfb_j$ for $j\neq i$ (see line 11 of the routine). The coefficients are chosen in an optimal way {so as to maintain the positivity of the final linear combination while minimizing the $L^\infty$-norm}. The algorithm returns a set of functions
$$
\bfpsi_i = \bfb_i - \sum_{j\neq i}\sigma^i_j \bfb_j,\quad \forall i\in \{1,\dots, n\}
$$
with $\sigma^i_j\geq 0$. Note that the algorithm requires to set a tolerance parameter $\e>0$ for the computation of the $\sigma^i_j$.

Once we have run $\ENLARGECONE$, the approximation of any function $\bfbeta$ is then sought as
\begin{equation}
\label{2K10b}
\bfbeta^{(EC)}  = \argmin_{c_1,\dots, c_n\geq 0} \Vert \bfbeta - \sum_{j=1}^n  c_{j} \bfpsi_j\Vert^2_{L^2([0, T+\tau])}
\end{equation}

We emphasize that the routine is valid for any set of nonnegative input functions. We can thus apply  $\ENLARGECONE$ to the functions $\{\bfb_i\geq 0\}_{i=1}^n$ from NMF, but also  to the functions selected by a greedy algorithm such as the following:
\begin{itemize}
\item For $n=1$, find
$$
\bfb_1 = \argmax_{\bfbeta \in \{ \bfbeta_i \}_{i=1} ^{K}} \Vert \bfbeta \Vert^2_{L^2([0, T+\tau])}
$$
\item At step $n>1$, we have selected the set of functions $\{\bfb_1,\dots, \bfb_{n-1}\}$. We next find
$$
\bfb_n = \argmax_{\bfbeta \in \{ \bfbeta_i \}_{i=1} ^{K}}  \min_{c_1,\dots, c_n\geq 0} \Vert \bfbeta  -\sum_{j=1}^{n-1} c_{j} \bfb_j\Vert^2_{L^2([0, T+\tau])}
$$

\end{itemize}
In our numerical tests, we call Enlarged Nonnegative Greedy (ENG) the routine involving the above greedy algorithm combined with our routine $\ENLARGECONE$.


\renewcommand{\algorithmicrequire}{\textbf{Input:}}
\renewcommand{\algorithmicensure}{\textbf{Output:}}
\begin{algorithm}[h]
  \begin{algorithmic}[]
  \REQUIRE Set of nonnegative functions $\{ \bfb_1,\dots, \bfb_n \}$. Tolerance $\e>0$.
    \FOR {$i\in \{1,\dots, n\}$}
      \State Set $\sigma^i_j =0,\quad \forall j\neq i.$
      \FOR {$\ell \in \{1,\dots, n\}$} 
      	\State $\alpha_\ell^{i,*} =   \argmax\{\alpha\geq 0 \; :\; \bigl[\bfb_i - \sum_{j\neq i}\sigma^i_j \bfb_j - \alpha \bfb_\ell \bigr](t) >0, \ \  \forall t \in [0,T+\tau]\}$ 
	\State $\sigma^i_\ell = \sigma^i_\ell + \frac{\alpha_\ell^{i,*}}{2}$ 
      	\WHILE {$\alpha_\ell^{i,*} \geq \e$ } 
      		\State $\alpha_\ell^{i,*} =   \argmax\{\alpha\geq 0 \; :\; \bigl[\bfb_i - \sum_{j\neq i}\sigma^i_j \bfb_j - \alpha \bfb_\ell \bigr](t) >0, \ \  \forall t \in [0,T+\tau]\}$ 
		\State $\sigma^i_\ell = \sigma^i_\ell + \frac{\alpha_\ell^{i,*}}{2}$ 
	\ENDWHILE 
	\ENDFOR 
	\State $\bfpsi_i = \bfb_i - \sum_{j\neq i}\sigma^i_j \bfb_j $ 
    \ENDFOR 
 \ENSURE $\{ \bfpsi_1,\dots, \bfpsi_n \}$ 
  \end{algorithmic}
\caption{$\ENLARGECONE[\{ \bfb_1,\dots, \bfb_n \}, \e]\to \{ \bfpsi_1,\dots, \bfpsi_n \}$}
  \label{alg:cone}
\end{algorithm}

We remark that if we work with positive functions that are upper bounded by a constant $L>0$, we can ensure that the approximations{, denoted as $\Psi$, and} written as a linear combination of basis functions, will also be between these bounds $0$ and $L$ by defining, on the one hand and as we have just done, a cone of positive functions  generated by the above family $\{\psi_i\}_i$, and, on the other hand, by considering the base of the functions $L-\varphi$, $\varphi$ being as above the set all greedy elements of the reduced basis to which we also apply an enlargement of these positive functions. We then impose that the approximation is written as a positive combination of the first (positive) functions and that $L-\Psi$  is also written as a combination with positive components in the second basis.

In this frame, the approximation appears under the form of a least square approximation with $2n$ linear constraints on the $n$ coefficients expressing the fact that in the two above transformed basis the coefficients are nonnegative.
\end{enumerate}

\col{In addition to the previous basis functions, it is possible to include more general/generic basis functions such as polynomial, radial, or wavelet functions in order to guarantee simple forecasting trends. For instance, one can add affine functions in order to include the possibility of predicting with simple linear extrapolation to the range of possible forecasts offered by the reduced model. Given the overall exponential behavior of the health data, we have added an exponential function of the form $b_0(t) = \xi \exp(-\xi' t)$ (respectively $g_0(t) = \psi \exp(- \psi' t)$) to the reduced basis functions $\{b_1,\dots, b_n\}$ and (respectively $\{g_1,\dots, g_n\}$) where the real-valued nonnegative parameters $\xi, \xi',\psi, \psi'$ are obtained through a standard exponential regression from $\bfbeta_\obs^*$ (respectively $\bfgamma_\obs^*$) associated to the targeted profiles of infectious people, that is the profiles defined on the time interval $[0,T]$ that one wants to extrapolate to $]T, \tau]$. In other words, the final reduced models that we use for forecasting are: 
$$
\rB_{n+1} = \vspan \{b_0, b_1,\dots, b_n\},
\quad
\rG_{n+1} = \vspan \{g_0, g_1,\dots, g_n\} \subset L^\infty([0, T+\tau], \bR),
$$
Indeed, including the exponential functions in the reduced models gives access cheaply to the overall behavior of the parameters $\beta$ and $\gamma$, the rest of basis functions generated from the training sets catch the high order approximations and allow then to improve the extrapolation. }

\begin{Remark}
{\textbf{Reduced models on $\cI=\{\textbf{I}(\mu) \,:\, \mu \in \cP\}$ and $\cR=\{\textbf{R}(\mu) \,:\, \mu \in \cP \}$} Instead of applying model reduction to the sets $\cB$ and $\cG$ as we do in our approach, we could apply the same above techniques directly to the sets of solutions $\cI$ and $\cR$ of the SIR models with time-dependent coefficients in $\cB$ and $\cG$. In this case, the resulting approximation would however not follow a SIR dynamics.}
\end{Remark}

\section{Methodology for multiple regions including population mobility data}
\label{sec:multiregional}

The forecasting method of Section \ref{sec:methodo} for one single region can be extended to {the treatment of }multiple regions involving population mobility. The prediction scheme will now be based on a multiregional SIR with time-dependent coefficients. Compared to other more detailed models, its main advantage is that it reduces drastically the number of parameters to be estimated. Indeed, detailed multiregional models such as multiregional extensions of the above $\Colizza$ and $\Magal$ models from Section \ref{sec:collapse} require a number of parameters which quickly grows with the number $P$ of regions involved. Their calibration thus requires large amounts of data which, in addition, may be unknown, very uncertain, or not available. In a forthcoming paper, we will apply the fully multi-regional setting for the post-lockdown period.

The structure of this section is the same as the previous one for the case of a single region. We start by introducing in Section \ref{sec:basics-multiregion} the multi-regional SIR model with time-dependent coefficients and associated detailed models. As any multiregional model, mobility data are required as an input information, and the nature and level of detail of the available data {motivates} certain choices on the modelling of the multiregional SIR (as well as the other detailed models). We next present in Section \ref{sec:forecast-multiregion} the general pipeline, in which we emphasize the high modularity of the approach.

\subsection{Multi-regional compartmental models}
\label{sec:basics-multiregion}
In the spirit of fluid flow modeling, there are essentially two ways of describing mobility between regions:
\begin{itemize}
\item In a Eulerian description, we take the regions as fixed references for which we record incoming and outgoing travels.
\item In a Lagrangian description, we follow the motion of people domiciled in a certain region and record their travels in the territory. We can expect this modeling to be more informative on the geographical spread of the disease but it comes at the cost of additional details on the people's domicile region.
\end{itemize}
Note that both descriptions hold at any coarse or fine geographical level in the sense that what we call the regions could be taken to be full countries, departments within a country, or very small geographical partitions of a territory.  We next describe the multi-regional SIR models with the Eulerian and Lagrangian description of population fluxes which will be the output of our methodology.

\subsubsection{Multi-regional SIR models with time-dependent parameters}

\textbf{Eulerian description of population flux:}
Assume that we have $P$ regions and the number of people in region $i$ is $N_i$ {for $i=1,\dots, P$}. Due to mobility, the population in each region varies, so $N_i$ depends on $t$. However, the total population is assumed to be constant and equal to $N$, that is
$$
N = \sum_{i=1}^P N_i(t), \quad \forall t \geq 0.
$$
For any $t\geq0$, let $\lambda_{i\to j}(t) \in [0,1]$ be the probability that people from $i$ travel to $j$ at time $t$. In other words, $\lambda_{i\to j}(t)N_i(t)\delta t$ is the number of people from region $i$ that have travelled to region $j$ between time $t$ and $t+\delta t$. Note that we have
$$
\sum_{j=1}^P \lambda_{i\to j}(t) = 1, \quad \forall t \geq 0.
$$
Since, for any $\delta t\geq 0$,
$$
N_i(t+\delta t) = N_i(t) - \sum_{j\neq i} \lambda_{i\to j}(t) N_i(t) \delta t + \sum_{j\neq i} \lambda_{j\to i}(t) N_j(t) \delta t
$$
dividing by $\delta t$ and taking the limit $\delta t\to 0$ yields
$$
\frac{\d N_i}{\dt}(t) = -\sum_{j\neq i} \lambda_{i\to j}(t) N_i(t) + \sum_{j\neq i} \lambda_{j\to i}(t) N_j(t).
$$
Note that we have
$$
\sum_{i=1}^P \frac{\d N_i}{\dt}(t) = 0, \quad \forall t \geq 0.
$$
Thus $\sum_i N_i(t) = \sum_i N_i(0) = N$, which is consistent with our assumption that the total population is constant.

The time evolution of the $N_i$ is known in this case if we are given the $\lambda_{i\to j}(t)$ from Eulerian mobility data. In addition to this mobility data, we also have the data of the evolution of infected and removed people and our goal is to fit a multi-regional SIR model that is in accordance with this data. We propose the following model.

Denoting $S_i$, $I_i$ an $R_i$ the number of Susceptible, Infectious and Removed people in  region $i$ at time $t$, we first have the relation
$$
N_i(t) = S_i(t) + I_i(t) + R_i(t) \quad \Leftrightarrow \quad 1 = \frac{ S_i(t) }{N_i(t)} + \frac{ I_i(t) }{N_i(t)} + \frac{ R_i(t) }{N_i(t)}.
$$
Note that from the second relation, it follows that
\begin{equation}
\label{eq:check}
0 = \frac{\d}{\dt}\frac{ S_i }{N_i} + \frac{\d}{\dt}\frac{ I_i }{N_i} + \frac{\d}{\dt}\frac{ R_i }{N_i}.
\end{equation}
To model the evolution between compartments, one possibility is the following SIR model
\begin{align}
\frac{\d}{\dt}\frac{ S_i }{N_i} &= - \left(\beta_i \lambda_{i\to i} \frac{I_i}{N_i} \label{eq:S}
+ \sum_{j\neq i}\beta_j \lambda_{j\to i} \frac{I_j}{N_j} \right)\frac{S_i}{N_i} \\
\frac{\d}{\dt}\frac{ I_i }{N_i} &= - \frac{\d}{\dt}\frac{ S_i }{N_i} - \gamma_i \frac{I_i}{N_i} \\
\frac{\d}{\dt}\frac{ R_i }{N_i} &=\gamma_i  \frac{I_i}{N_i} ,
\end{align}
The parameters $\beta_i$, $\gamma_i$, $N_i$ depend on $t$ but we have omitted the dependence to ease the reading. Introducing the compartmental densities
\begin{equation}
s_i = \frac{ S_i }{N_i}, \quad i_i = \frac{ I_i }{N_i}, \quad r_i = \frac{ R_i }{N_i},
\end{equation}
the system equivalently reads
\begin{align}
\label{eq:S}
\frac{\d}{\dt}s_i &= - \left(\beta_i \lambda_{i\to i} i_i
+ \sum_{j\neq i}\beta_j \lambda_{j\to i} i_j \right)s_i \\
\frac{\d}{\dt}i_i &= - \frac{\d}{\dt}s_i - \gamma_i i_i \\
\frac{\d}{\dt}r_i &=\gamma_i  i_i,
\end{align}
Before going further, some comments are in order:
\begin{itemize}
\item The model is consistent in the sense that it satisfies \eqref{eq:check} and when $P=1$ we recover the traditional SIR model.
\item Under lockdown measures, $\lambda_{i\to j} \approx \delta_{i,j}$ and the population $N_i(t)$ remains practically constant. As a result, the evolution of each region is decoupled from the others and each region can be addressed with the mono-regional approach.
\item The use of $\beta_j$ in equation \eqref{eq:S} is debatable. When the people from region $j$ arrive to region $i$, it may be reasonable to assume that the contact rate is $\beta_i$.
\item The use of $\lambda_{j\to i}$ in equation \eqref{eq:S} is also very debatable. The probability $\lambda_{j\to i}$ was originally defined to account for the mobility of people from region $j$ to region $i$ without specifying the compartment. However, in equation \eqref{eq:S}, we need the probability of mobility of infectious people from region $j$ to region $i$, which we  denote by $\mu_{j\to i}$ in the following. It seems reasonable to think that $\mu_{j\to i}$ may be smaller than $\lambda_{j\to i}$ because as soon as people become symptomatic and suspect their illness, they will probably not move. Two possible options would be:
\begin{itemize}
\item We could try to make a guess on $\mu_{j\to i}$. If the symptoms arise, say, 2 days after infection and if we recover in 15 days in average, then we could say that $\mu_{j\to i} = 2/15 \lambda_{j\to i}$.
\item Since, the above seems however pretty empirical, {another option is to use  $\lambda_{j\to i}$ and absorb the uncertainty in the values of the $\beta_j$ that one fits.}
\end{itemize}
\item \col{We choose not to add mobility in the $R$ compartment as this does not modify the dynamics of the epidemic spread and only adjustments in the population sizes are needed.}
\end{itemize}
\vspace{0.2cm}
\noindent
\textbf{Lagrangian description of population flux:} We call the above description Eulerian because we have fixed the {regions} as a fixed reference. Another point of view is to follow the trajectories of inhabitants of each region, in the same spirit as when we follow the trajectories of fluid particles.

Let now $S_i$, $I_i$, $R_i$ the number of susceptible, infectious and Removed people who are domiciled in region $i, i=1,\dots P$. It is reasonable to assume that $S_i(t)+I_i(t)+R_i(t)$ is constant in time. However, all the dwellers of region $i$ may not all be in that region at time $t$. Let $\lambda^{(i)}_{j\to k}(t)$ be the probability that susceptible people domiciled at $i$ travel from region $j$ to region $k$ at time $t$. With this notation, $\lambda^{(i)}_{i\to i}(t)$ is the probability that susceptible people domiciled at $i$ remain in region $i$ at time $t$. Similarly, let $\mu^{(i)}_{j\to k}(t)$ be the probability that infectious people domiciled at $i$ travel from region $j$ to $k$ at time $t$. Hence the total number of susceptible and infectious people that are in region $i$ at time $t$ is
\begin{align}
\cS_i(t) &= \sum_{k=1}^P \sum_{j=1}^P \left( \lambda^{(k)}_{j\to i}(t) - \lambda^{(k)}_{i\to j}(t) \right)S_{k}(t) \\
\cI_i(t) &= \sum_{k=1}^P \sum_{j=1}^P \left( \mu^{(k)}_{j\to i}(t) - \mu^{(k)}_{i\to j}(t) \right)S_{k}(t)
\end{align}
We can thus write the evolution over $S_i$, $I_i$, $R_i$ as
\begin{align}
\label{eq:SIR-lagrangian}
\frac{\d S_i}{\dt} &= - \sum_{j=1}^P \sum_{k=1}^P \beta_k(t)\lambda^{(i)}_{j\to k}(t) S_i(t) \cI_k(t)  \\
\frac{\d I_i}{\dt} &= -\frac{\d S_i}{\dt} - \gamma_i(t) I_i(t)  \\
\frac{\d R_i}{\dt} &=  \gamma_i(t) I_i(t) 
\end{align}
Note that $S_i(t)+I_i(t)+R_i(t)$ is constant, which is consistent with the fact that in our model
$$
\frac{\d}{\dt} (S_i+I_i+R_i)=0.
$$
We emphasize that, to implement this model, one needs the Lagrangian mobility data $\lambda^{(i)}_{j\to k}$ for all $(i,j,k)\in \{1,\dots,P\}^3$.
\medskip\\
\textbf{Notation:} In the following, we gather the compartmental variables in vectors
$$
\vec{\bfS} \coloneqq (\bfS)_{i=1}^P, \quad \vec{\bfI} \coloneqq (\bfI)_{i=1}^P, \quad \vec{\bfR} \coloneqq (\bfR)_{i=1}^P
$$
as well as the time-dependent coefficients
$$
\vec{\bfbeta}=(\bfbeta)_{i=1}^P,\quad \vec{\bfgamma}=(\bfgamma)_{i=1}^P.
$$
For any $\vec{\bfbeta}$ and $\vec{\bfgamma}\in \left( \Linf \right)^P$, we denote by
$$
( \vec{\bfS}, \vec{\bfI}, \vec{\bfR} ) = \MRSIR (\vec{\bfS}(0), \vec{\bfI}(0), \vec{\bfR}(0), \vec{\bfbeta}, \vec{\bfgamma}, [0, T] )
$$
the output of any of the above multiregional SIR models. For simplicity in what follows, we will omit the initial condition in the notation.

\subsubsection{Detailed multi-regional models with constant coefficients}
In the spirit of the multi-regional SIR, one can formulate detailed multi-regional versions of more detailed models such as the ones introduced in Section \ref{sec:models}. We omit the details for the sake of brevity.

\subsection{Forecasting for multiple regions with population mobility}
\label{sec:forecast-multiregion}
Similarly as in the mono-regional case, we assume that we are given health data in $[0, T]$ in all regions. The observed data in region $i$ is the series of \col{infected} people, denoted $I_i^\obs$, and recovered people denoted $R_i^\obs$.  They are usually given at a national or a regional scale and on a daily basis. 

We propose to fit the data and provide forecasts with SIR models with time-dependent parameters $\bfbeta_i$ and $\bfgamma_i$ for each region $i$. Like in the mono-regional case, we can prove that such a simple family possesses optimal fitting properties for our purposes. In the current case, the cost function reads
\begin{align*}
&{\cJ(\vec{\bfbeta}, \vec{\bfgamma} \, |\, \vec{\bfI}_\obs, \vec{\bfR}_\obs, [0,T]) }
\coloneqq \sum_{i=1}^P \int_0^T \left( \vert I_i(t) - I_i^\obs(t)\vert^2 + \vert R_i(t) - R_i^\obs(t)\vert^2 \right)\dt \\
&\qquad \text{such that } \left( \vec{\bfS}, \vec{\bfI}, \vec{\bfR} \right) = \MRSIR \left( \vec{\bfbeta}, \vec{\bfgamma}, [0, T] \right),
\end{align*}
and the fitting problem is the optimal control problem of finding 
\begin{equation}
J^* = \inf_{  \vec{\bfbeta}, \vec{\bfgamma} \in \left( \Linf \right)^P \times \left( \Linf \right)^P } {\cJ(\vec{\bfbeta}, \vec{\bfgamma} \, |\, \vec{\bfI}_\obs, \vec{\bfR}_\obs, [0,T]).}
\label{eq:fit-mr}
\end{equation}
The following proposition ensures the existence of a unique minimizer under certain conditions. To prove it, it will be useful to remark that any of the above multi-regional SIR models (see \eqref{eq:S}, \eqref{eq:SIR-lagrangian}) can be written in the general form
\begin{align*}
\frac{\dd \vec{\bfS}}{\dt} &= \mathrm{M}(\Lambda(t), \vec{\bfS}(t), \vec{\bfI}(t)) \vec{\bfbeta} \\
\frac{\dd \vec{\bfI}}{\dt} &= - \frac{\dd \vec{\bfS}}{\dt} - \diag(\bfI(t))\, \vec{\bfgamma} \\
\frac{\dd \vec{\bfR}}{\dt} &=  \diag(\bfI(t)) \vec{\bfgamma},
\end{align*}
where, by a slight abuse of notation, the vectors $\vec{\bfS}$, $\vec{\bfI}$ and $\vec{\bfR}$ are densities of population in the case of the Eulerian approach (see equation \eqref{eq:S}). They are classical population numbers in the case of the Lagrangian approach (see equation \eqref{eq:SIR-lagrangian}). $\diag(\bfI(t))$ is the $P\times P$ diagonal matrix with diagonal entries given by the vector $\bfI(t)$. $\mathrm{M}(\Lambda(t), \vec{\bfS}(t), \vec{\bfI}(t))$ is a matrix of size $P\times P$ which depends on the vectors of susceptible and infectious people $\vec{\bfS}(t)$, $\vec{\bfI}(t)$ and on the mobility matrix $\Lambda$. In the case of the Eulerian description, $\Lambda(t) = (\lambda_{i,j}(t))_{1\leq i, j \leq P}$ and in the case of the Lagrangian approach $\Lambda(t)$ is the $P\times P \times P$ tensor $\Lambda(t)= (\lambda^{(i)}_{j,k}(t))_{1\leq i, j, k \leq P}$. For example, in the case of the Eulerian model \eqref{eq:SIR-lagrangian}, the matrix $M$ reads 
\begin{equation}
\label{eq:M}
\mathrm{M}(\Lambda(t), \vec{\bfS}(t), \vec{\bfI}(t))
= - \diag( \vec{\bfS}(t)) \Lambda^T \diag( \vec{\bfI}(t))
= - ( S_i \lambda_{i\to j} I_j )_{1\leq i, j \leq P}
\end{equation}

\begin{Proposition}
\label{prop:MRSIR}
If $\mathrm{M}(\Lambda(t), \vec{\bfS}(t), \vec{\bfI}(t))$ is invertible for all $t\in [0, T]$, then
there exists a unique minimizer $( \vec{\bfbeta}^*, \vec{\bfgamma}^*)$ to problem \eqref{eq:fit-mr}.
\end{Proposition}
\begin{proof}
Since we assume that $\mathrm{M}(\Lambda(t), \vec{\bfS}(t), \vec{\bfI}(t))$ is invertible for every $t\in [0,T]$, we can set 
\begin{align}
\begin{cases}
\vec{\bfbeta}^*(t) &\coloneqq M^{-1}(t) \frac{\dd \vec{\bfS}}{\dt}  \\
\vec{\bfgamma}^*(t) &\coloneqq \diag^{-1}(I(t)) \frac{\dd \vec{\bfR}}{\dt}
\end{cases}
\end{align}
or equivalently 
\begin{align}
\begin{cases}
\vec{\bfbeta}^*(t) &\coloneqq M^{-1}(t) \frac{\dd \vec{\bfS}}{\dt}  \\
\vec{\bfgamma}^*(t) &\coloneqq - \diag^{-1}(I(t)) \left(  \frac{\dd \vec{\bfI}}{\dt} + \mathrm{M}(\Lambda(t), \vec{\bfS}(t), \vec{\bfI}(t)) \vec{\bfbeta}^*  \right)
\end{cases}
\end{align}
so that
$$
(\vec{\bfS}_\obs, \vec{\bfI}_\obs, \vec{\bfR}_\obs)= \MRSIR \left( \vec{\bfbeta}^*, \vec{\bfgamma}^*, [0, T] \right)
$$
and
$$
\cJ(\vec{\bfbeta}^*, \vec{\bfgamma}^* \, |\, \vec{\bfI}_\obs, \vec{\bfR}_\obs, [0,T]) = 0
$$
which implies that $J^* = 0$.
\end{proof}

Before going further, let us comment on the invertibility of $\mathrm{M}(\Lambda(t), \vec{\bfS}(t), \vec{\bfI}(t))$ which is necessary in Proposition \ref{prop:MRSIR}. A sufficient condition to ensure it is if the matrix is diagonally dominant row-wise or column-wise. This yields certain conditions on the mobility matrix $\Lambda(t)$ with respect to the values of $\vec{\bfS}(t)$, $\vec{\bfI}(t)$. For example, if $M$ is defined as in equation \eqref{eq:M}, the matrix is diagonally dominant per rows if for every $1\leq i\leq P$,
$$
\lambda_{i\to i}  > \sum_{j\neq i} \lambda_{i\to j} \frac{I_j}{I_i}.
$$
Similarly, if for every $1\leq j\leq P$,
$$
\lambda_{j\to j} > \sum_{i\neq j}  \lambda_{i\to j}  \frac{S_i}{S_j},
$$
then the matrix is diagonally dominant per columns, and guarantees invertibility. Note that any of the above conditions is satisfied in \col{situations} with little or no mobility where $\lambda_{i\to i} \approx \delta_{i,j}$.

Now that we have exactly defined the set up for the multi-regional case, we can follow the same steps of Section \ref{sec:methodo} to derive forecasts involving model reduction for the time-dependent variables $\vec{\bfbeta}$ and $\vec{\bfgamma}$.

\section{Numerical results}
\label{sec:numerical-Paris}

In this section we apply our forecasting method \col{to the ongoing COVID-19 pandemic spread in year 2020 in France which started approximately from February 2020.}  Particular emphasis is put on the first pandemic wave for which we consider the period going from March 19 to May 20, 2020. Due to the lockdown imposed between March 17 and May 11, inter-regional population mobility was drastically reduced in that period. Studies using anonymised Facebook data have estimated the reduction in 80\% (see \cite{FAV}). As a result, it is reasonable to treat each region independently from the rest, and we apply the mono-regional setting of Section \ref{sec:mono-region}. Here, we focus on the case of the Paris region, and we report different forecasting errors obtained using the methods described in Section \ref{sec:mono-region}. \col{Some forecasts are also shown for the second wave for the Paris region between September 24 and November 25.}

The numerical results are presented as follows. Section \ref{sec:health-data} explains the sources of health data. Sections \ref{sec:forecast-compare} and \ref{sec:forecast-ENG} go straight to the core of the results and present a study of the forecasting power of the methods in a two week horizon. Section \ref{sec:forecast-wave-2} displays forecasts on the second wave. \col{Section \ref{sec:forecast-28-days} aims to illustrate the robustness of the forecasting over longer periods of time.} \col{A discussion on the fitting errors of the methods is given in Appendix \ref{appendix:noise}. Additional results highlighting the accuracy of the forecasts are shown in Appendix \ref{appendix:forecasting-error}.} 

\subsection{Data}
\label{sec:health-data}
\col{We use public data from Santé Publique France\footnote{\url{https://www.data.gouv.fr/en/datasets/donnees-hospitalieres-relatives-a-lepidemie-de-covid-19/}} to get the number $I_\obs(t)$ of infected, and $R_\obs(t)$ of removed people. As shown in Figure \ref{fig:data_I_R}, the raw data present some oscillations at the scale of the week, which are due to administrative delays for the cases to be officially reported by hospitals}. For our methodology, we have smoothed the data by applying a 7 days moving average filter. In order to account for the total number of infected people, we also multiply the data by an adjustment factor $f_{\rm adj}=15$ as hinted from the literature\footnote{Indeed, it is said in \cite{mizumoto2020estimating} that "of the 634 confirmed cases, a total of 306 and 328 were reported to be symptomatic and asymptomatic" and in \cite{di9impact} claims that the probability of developing severe symptoms for a symptomatic patient is 0 for children, 0.1 for adults and 0.2 for seniors. Thus if one takes $p=0.13$ as an approximate value of these probabilities,  one may estimate the adjustment factor as $ f_{\rm adj} = \frac{634}{328} \times \frac{1}{p}  \approx 15.$ }. Obviously, this factor is uncertain and could be improved in the light of further retrospective studies of the outbreak. However, note that when {$S \simeq N$ which is the case} in the start of the epidemic, the impact of this factor is negligible in the dynamics as can be understood from \eqref{eq:bg-star}. \col{In addition, since we use the same factor down to provide a forecast on the hospitalized, the influence on the choice is minor.}

\begin{figure}
\centering
\begin{subfigure}{.45\textwidth}
\includegraphics[width=1\textwidth]{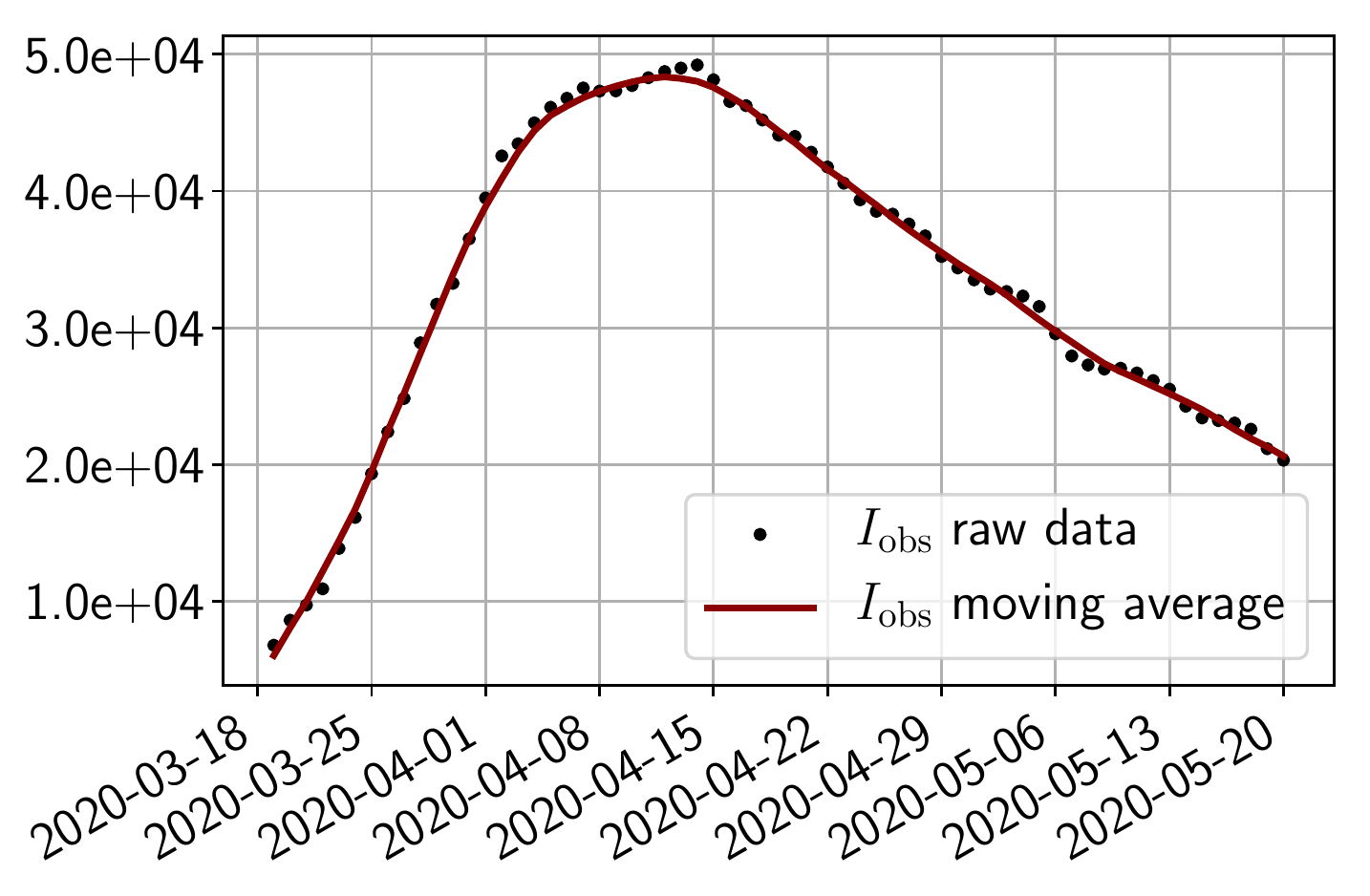}
\caption{\col{Infected}}
\end{subfigure}
\begin{subfigure}{.45\textwidth}
\includegraphics[width=1\textwidth]{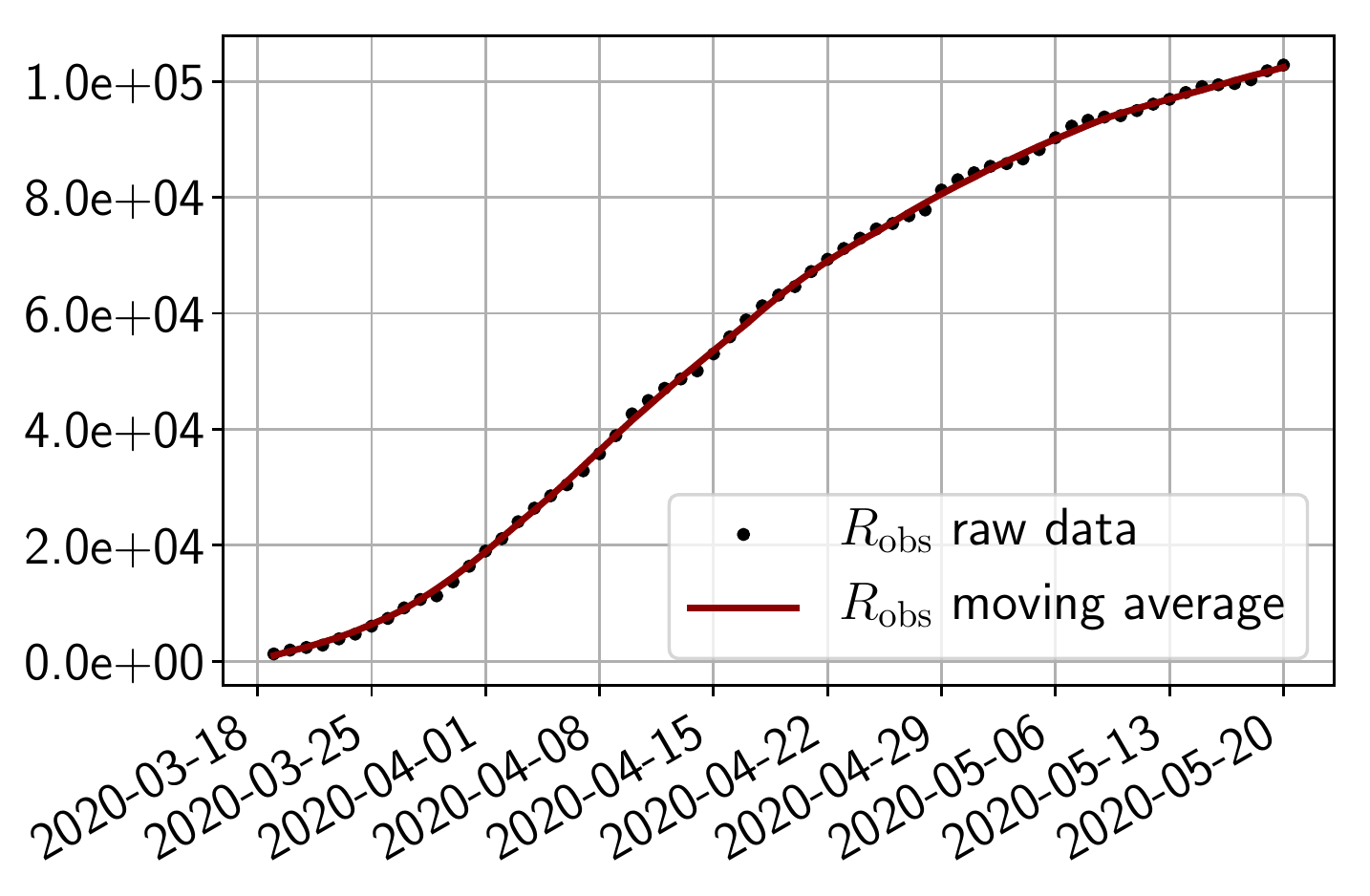}
\caption{\col{Removed}}
\end{subfigure}
\caption{\textbf{Data from $t_0=19/03/2020$ to $T=20/05/2020$}}
\label{fig:data_I_R}
\end{figure}

\subsection{Results}
\label{sec:results-paris}
Using the observations $I_\obs(t)$ and $R_\obs(t)$, we apply a finite difference scheme in formula \eqref{eq:bg-star} to derive $\beta^*_\obs(t)$ and $\gamma^*_\obs(t)$ for $t\in [0, T]$. Figure \ref{fig:beta_gamma_wave_1} shows the values of these parameters as well as the basic reproduction number $R^*_{0, \obs} = \beta_\obs^*/\gamma_\obs^*$ for the first pandemic wave in Paris. 

\begin{figure}
\centering
\begin{subfigure}{.32\textwidth}
\includegraphics[width=1\textwidth]{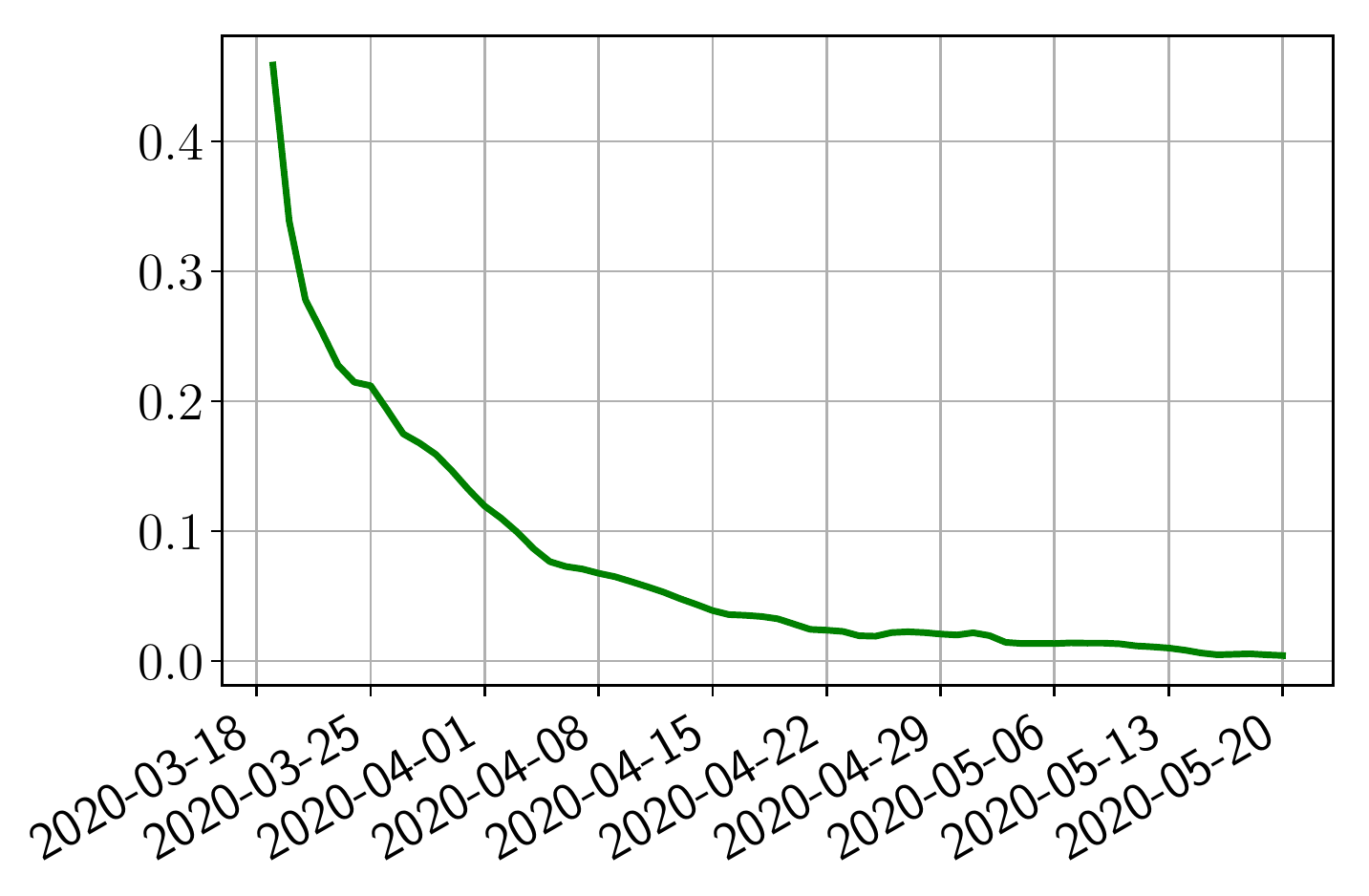}
\caption{$\beta_{\rm obs}^*$}
\end{subfigure}
\begin{subfigure}{.32\textwidth}
\includegraphics[width=1\textwidth]{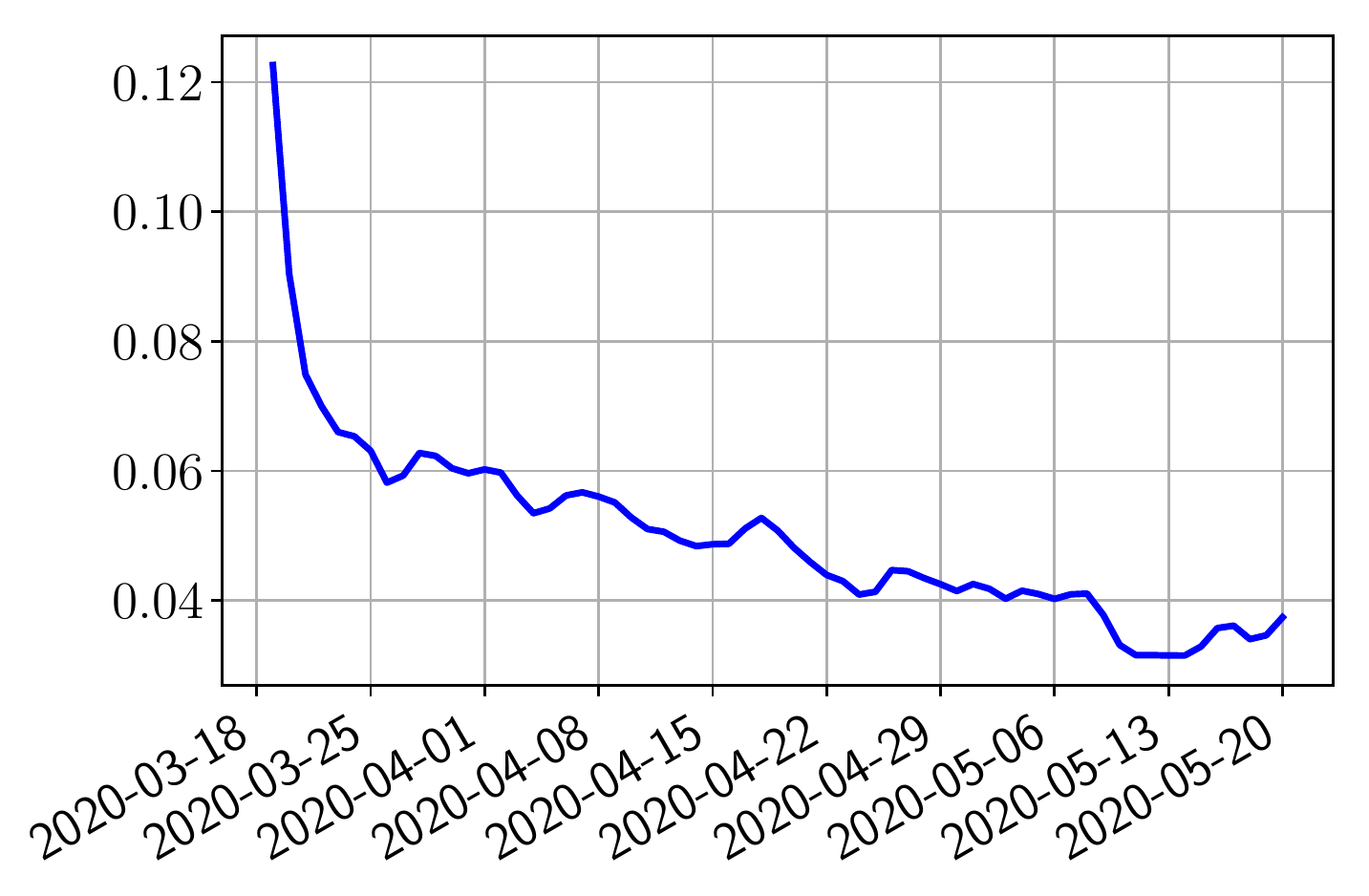}
\caption{$\gamma_{\rm obs}^*$}
\end{subfigure}
\begin{subfigure}{.32\textwidth}
\includegraphics[width=1\textwidth]{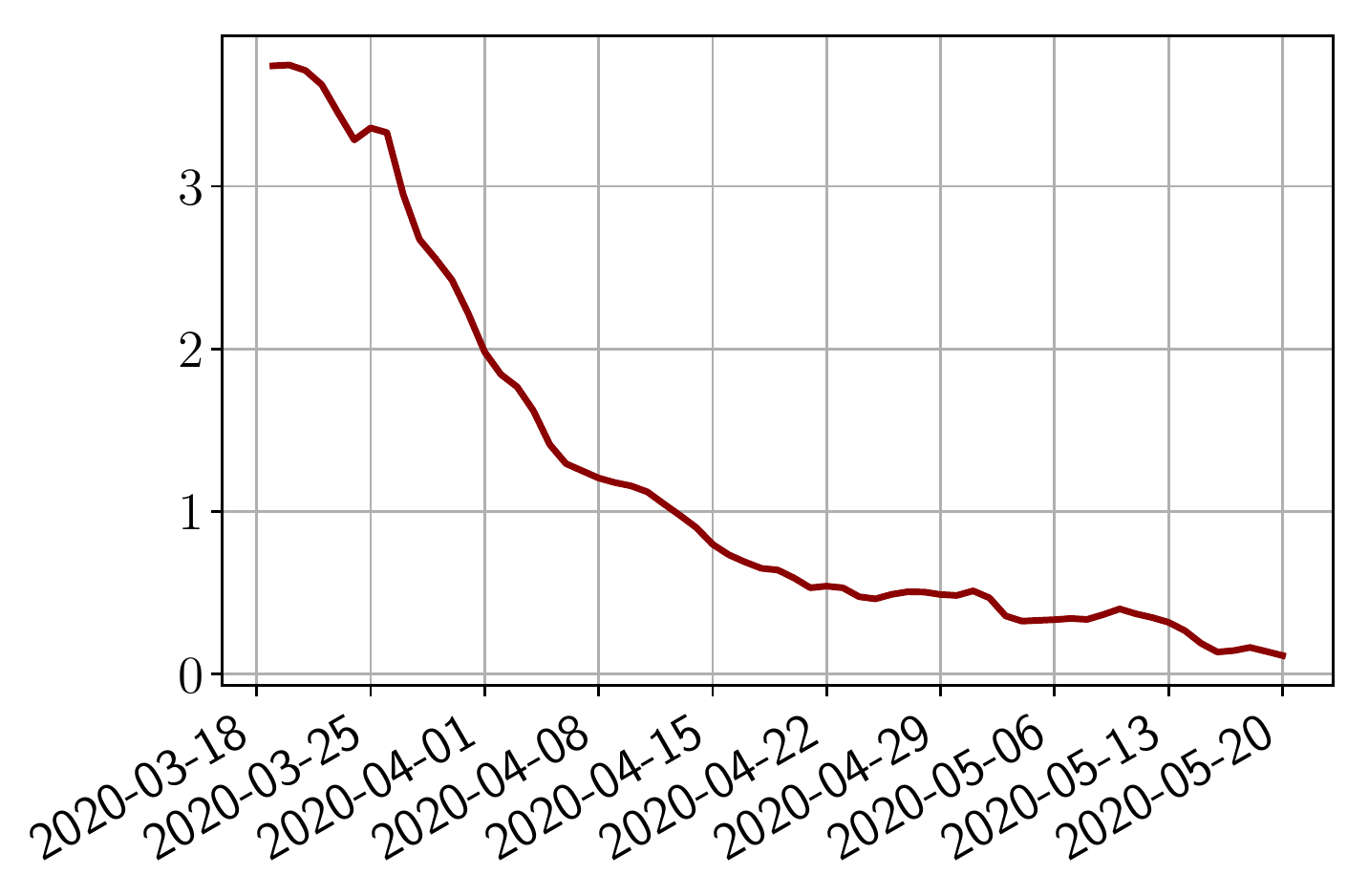}
\caption{$R^*_{0, \obs}$}
\end{subfigure}
\caption{\textbf{$\beta_\obs^*$, $\gamma_\obs^*$, $R^*_{0, \obs} = \beta_\obs^*/\gamma_\obs^*$ deduced from the data from $t_0=19/03/2020$ to $T=20/05/2020$}}
\label{fig:beta_gamma_wave_1}
\end{figure}

We next follow the steps of Section \ref{sec:methodo} to obtain the forecasts. In the learning phase (step 1), we use two parametric detailed models of type \Magal~and \Colizza~to generate training sets $\cB_\tr$ and $\cG_\tr$ composed of $K=2618$ training functions $\bfbeta(\mu)$ and $\bfgamma(\mu)$ where $\mu$ are uniformly sampled in the set of parameters $\cP_\tr$ in the vicinity of the parameter values suggested in the literature \cite{di9impact,magal2020predicting}. Based on these training sets, we finish step 1 by building three types of reduced models: SVD, NMF and ENG (see Section \ref{sec:MOR}).

Given the \col{reduced bases $\rB_{n}$ and $\rG_{n}$}, we next search for the optimal $\bfbeta^*_n\in \rB_{n}$ and $\bfgamma^*_n\in \rG_{n}$  that fit at best the observations (step 2 of our procedure). For this fitting step we consider two loss functions:
\begin{enumerate}
\item \fitIR: loss function $\cJ(\bfbeta, \bfgamma \; |\; \bfI_\obs, \bfR_\obs, [0, T])$ from \eqref{eq:cJ},
\item \fitbg: loss function $\widetilde{\cJ}(\bfbeta, \bfgamma \, | \, \bfbeta^*_\obs, \bfgamma^*_\obs, [0,T])$ from \eqref{eq:cJ2}
\end{enumerate}

We study the performance of each of the three reduced models and the impact of the dimension $n$ of the reduced model in terms of fitting error. \col{The presentation of these results is deferred to Appendix \ref{appendix:noise} in order not to overload the main discussion}. The main conclusion is that \col{the fitting strategy using SVD reduced bases provide smaller errors than NMF and ENG especially when we increase the number of modes $n$}. This is illustrated in Figure \ref{fig:fitting_b_g_1} where we show the fittings obtained with \fitbg~and $n=10$ for the first wave (from $t_0=19/03/2020$ to $T=20/05/2020$). We observe that SVD is the best at fitting $\beta^*_{obs}$ and $\gamma^*_{obs}$ while ENG produces a smoother fitting of the data. Although the smoother fitting with ENG yields larger fitting errors than SVD, we will see in the next section that it yields better forecasts. 
\begin{figure}
\centering
\begin{subfigure}{.45\textwidth}
\includegraphics[width=1\textwidth]{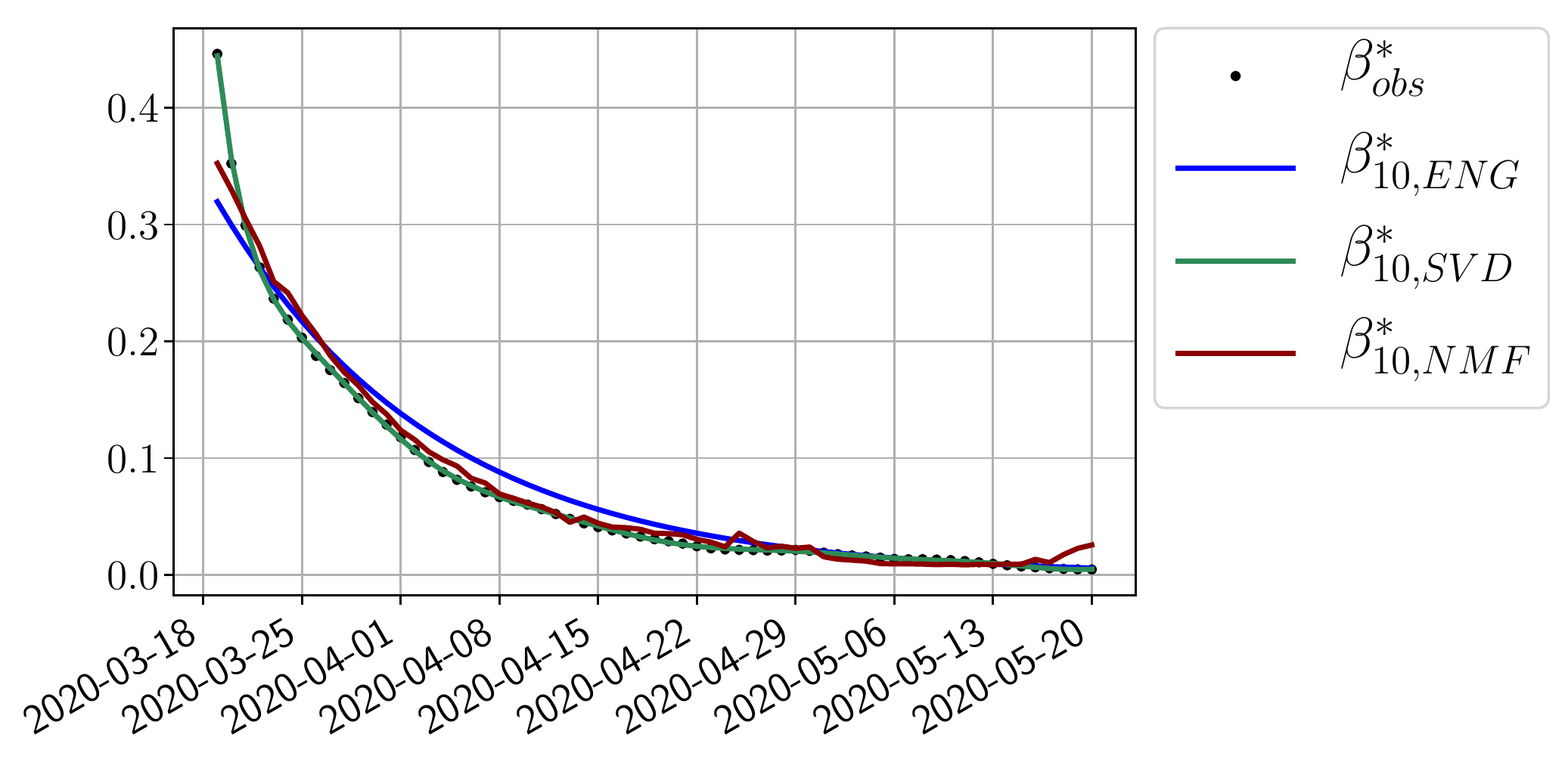}
\caption{$\beta$}
\end{subfigure}
\begin{subfigure}{.45\textwidth}
\includegraphics[width=1\textwidth]{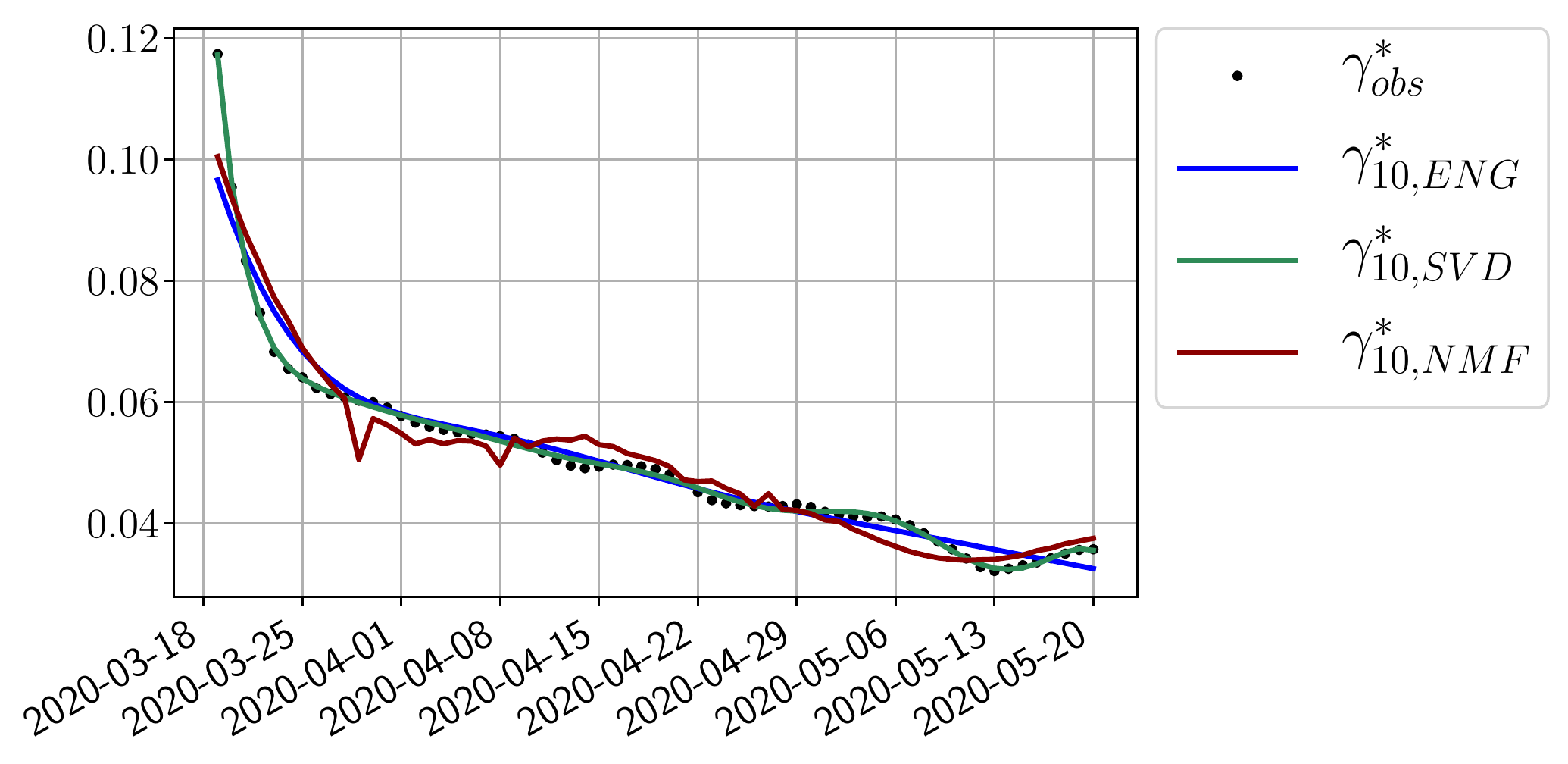}
\caption{$\gamma$}
\end{subfigure}
\caption{\textbf{Fitting from $t_0=19/03/2020$ to $T=20/05/2020$}}
\label{fig:fitting_b_g_1}
\end{figure}

\subsubsection{Forecasting for the first pandemic wave on a 14 day horizon}
\label{sec:forecast-compare}
In this section we illustrate \col{the short-term forecasting behavior of our method}. We consider a forecasting window of $\tau=14$ days and we examine several different starting days in the course of the first pandemic wave. The results are shown in figures \ref{fig:forecast_0104} to \ref{fig:forecast_0505}. Recall that that the forecasting uses the coefficients of the reduced basis obtained by the fitting procedure, but also the optimal initial condition of the forecast that minimizes the error on the three days prior to the start of the forecast. For each given fitting strategy (\fitIR, \fitbg), and each given type of reduced model (SVD, NMF, ENG), we have chosen to plot an average forecast computed with predictions using reduced dimensions $n \in \{5,6,7,8,9,10 \}$. This choice is a simple type of forecast combination but of course other more elaborate aggregation options could be considered. The labels of the plots correspond to the following:
\begin{itemize}
\item $I_{SVD}$, $I_{NMF}$, $I_{ENG}$, $R_{SVD}$, $R_{NMF}$, $R_{ENG}$ are the average forecasts obtained using \fitbg.
\item $I^*_{SVD}$, $I^*_{NMF}$, $I^*_{ENG}$, $R^*_{SVD}$, $R^*_{NMF}$, $R^*_{ENG}$ are the average forecasts obtained using \fitIR.
\end{itemize}
The main observation from Figures \ref{fig:forecast_0104} to \ref{fig:forecast_0505} is that the ENG reduced model is the most robust and accurate forecasting method. Fitting ENG with \fitIR~or \fitbg~does not seem to lead to large differences in the quality of the forecasts but \fitbg~seems to provide slightly better results. \col{This claim is further confirmed by the study of the numerical forecasting errors of the different methods that the reader finds in Appendix \ref{appendix:forecasting-error}}. 

Figures \ref{fig:forecast_0104} to \ref{fig:forecast_0505} also show that the SVD reduced model is very unstable and provides forecasts that blow up. \col{This behavior illustrates the dangers of overfitting}, namely that a method with high fitting power may present poor forecasting power due to instabilities. In the case of SVD, the instabilities stem from the fact that approximations are allowed to take negative values. This is the reason why NMF, which incorporates the nonnegative constraint, performs better than SVD. One of the reasons why ENG outperforms NMF relies in the enlargement of the cone of nonnegative functions (see Section \ref{sec:MOR}). \col{It is important to note that, with ENG the reduced bases are directly related to well chosen virtual scenarios, whereas SVD and NMF rely on matrix factorization techniques that provide purely artificial bases. This makes forecasts from ENG more realistic and therefore more reliable.}
\begin{figure}[H]
\centering
\begin{subfigure}{.45\textwidth}
\includegraphics[width=1\textwidth]{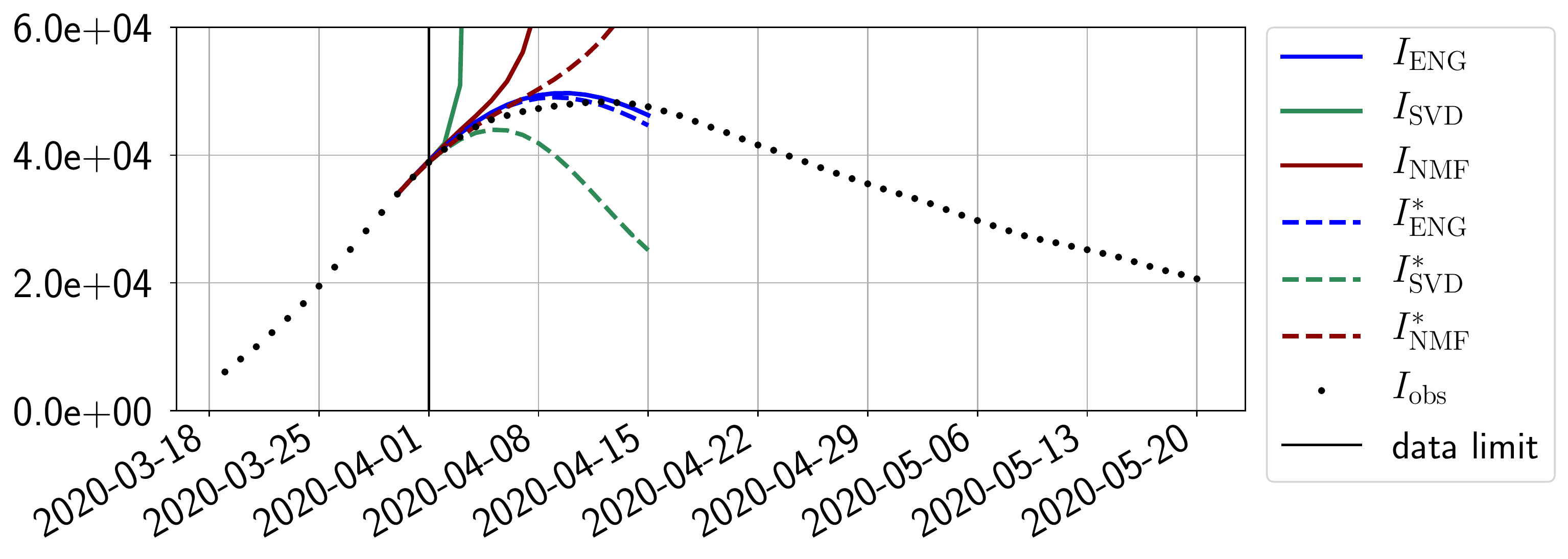}
\caption{Infected}
\end{subfigure}
\begin{subfigure}{.45\textwidth}
\includegraphics[width=1\textwidth]{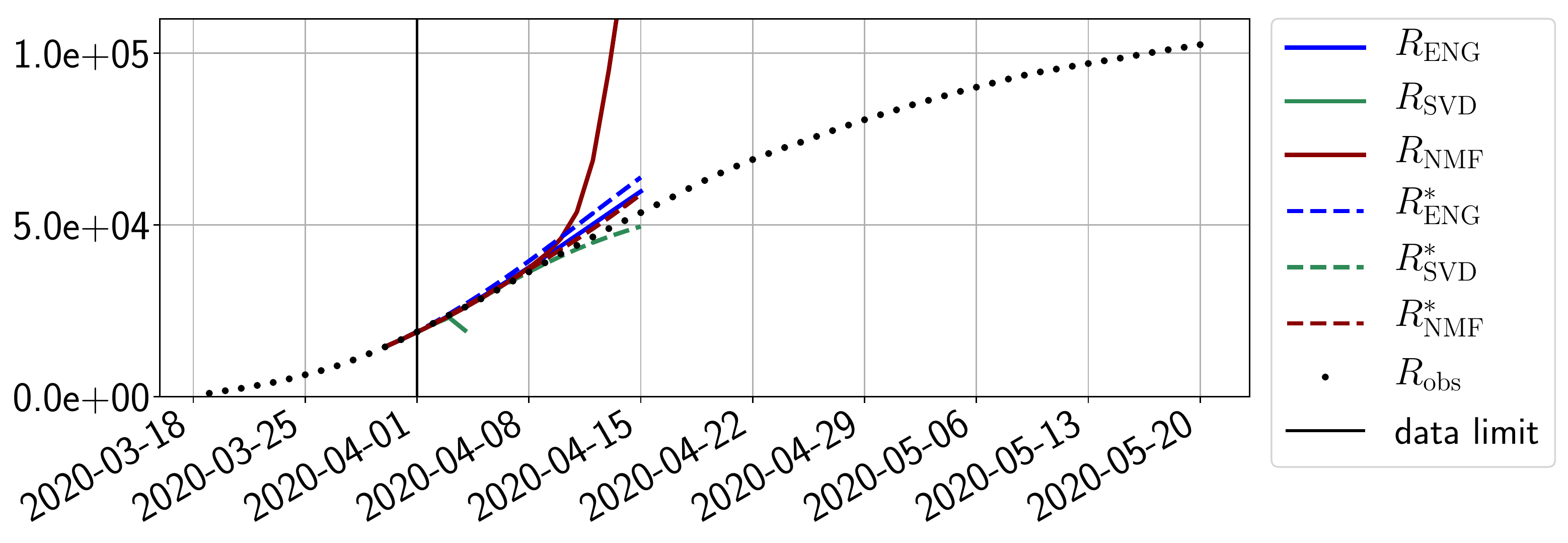}
\caption{Removed}
\end{subfigure}
\caption{14 day forecasts starting from $T=01/04$.}
\label{fig:forecast_0104}
\end{figure}

\begin{figure}[H]
\centering
\begin{subfigure}{.45\textwidth}
\includegraphics[width=1\textwidth]{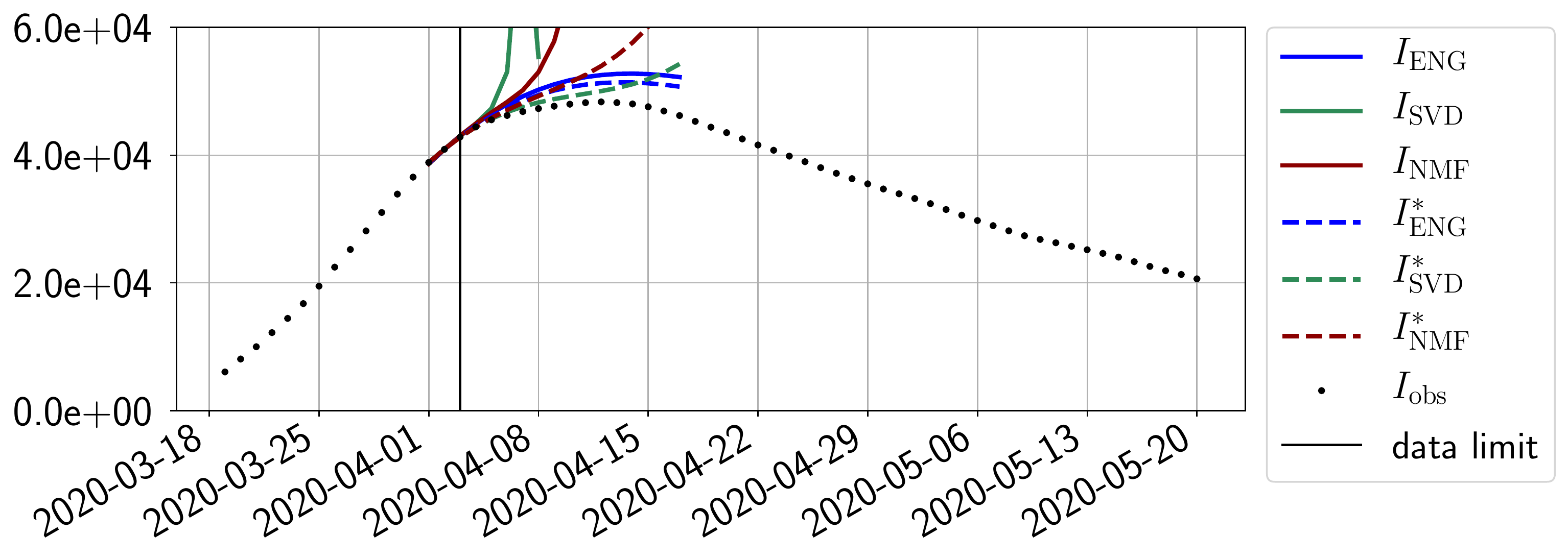}
\caption{Infected}
\end{subfigure}
\begin{subfigure}{.45\textwidth}
\includegraphics[width=1\textwidth]{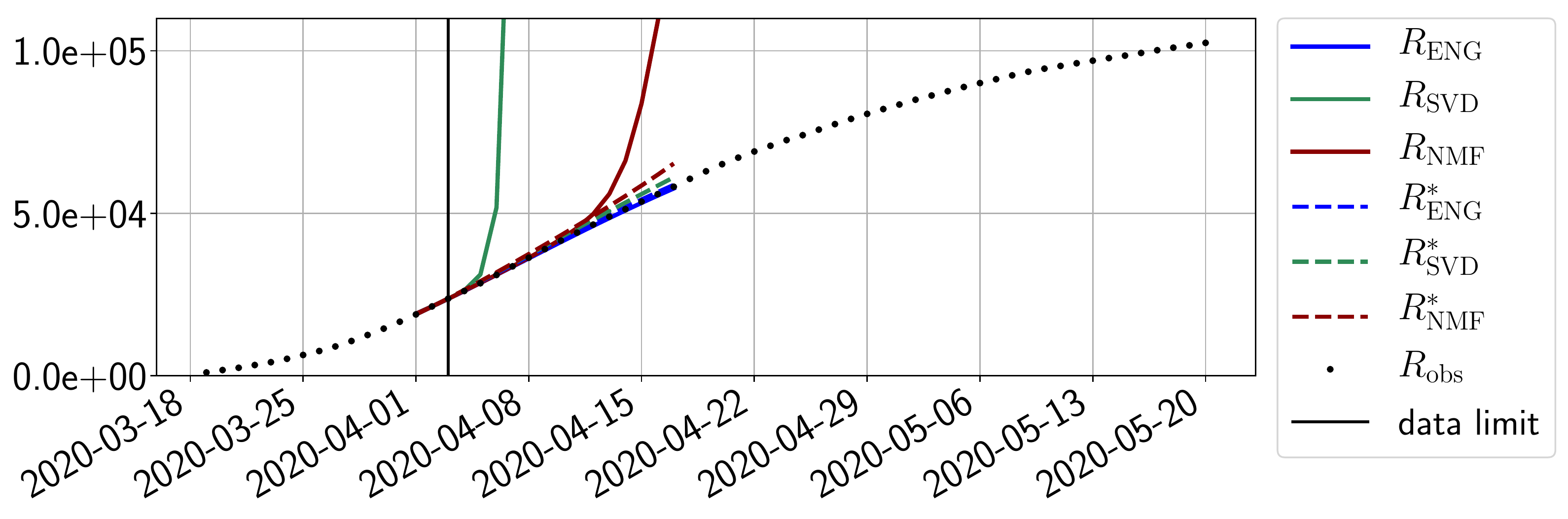}
\caption{Removed}
\end{subfigure}
\caption{14 day forecasts starting from $T=03/04$.}
\label{fig:forecast_0304}
\end{figure}

\begin{figure}[H]
\centering
\begin{subfigure}{.45\textwidth}
\includegraphics[width=1\textwidth]{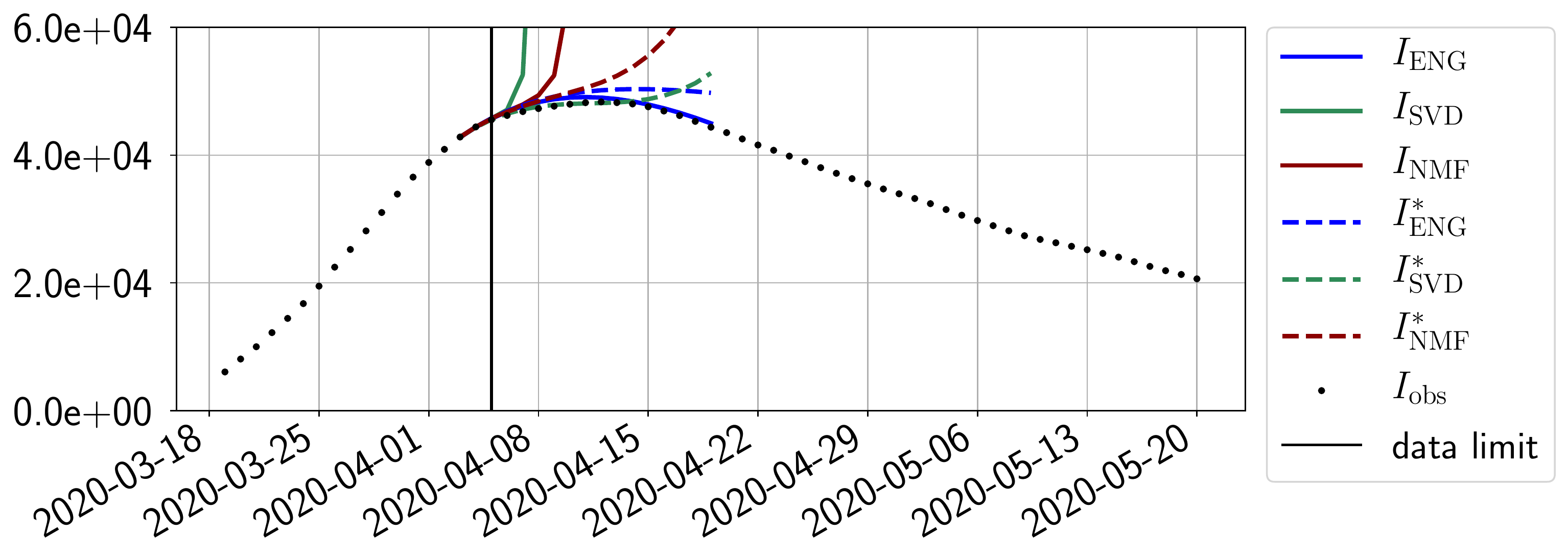}
\caption{Infected}
\end{subfigure}
\begin{subfigure}{.45\textwidth}
\includegraphics[width=1\textwidth]{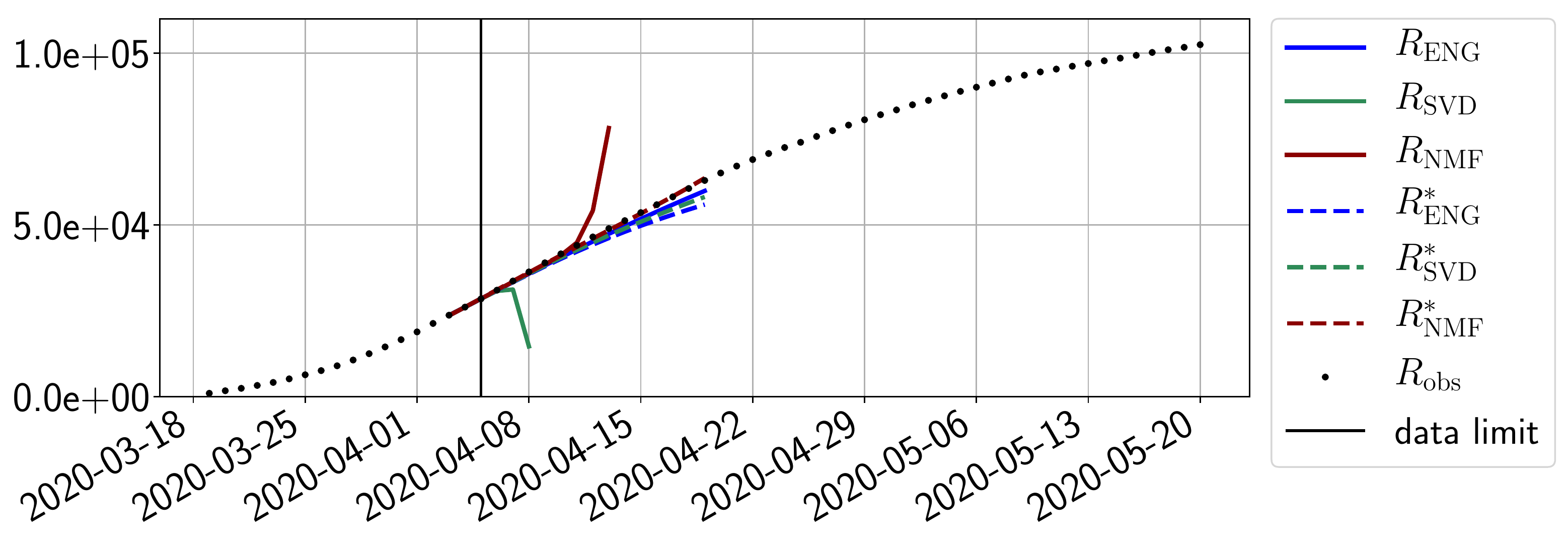}
\caption{Removed}
\end{subfigure}
\caption{14 day forecasts starting from $T=05/04$.}
\label{fig:forecast_0504}
\end{figure}

\begin{figure}[H]
\centering
\begin{subfigure}{.45\textwidth}
\includegraphics[width=1\textwidth]{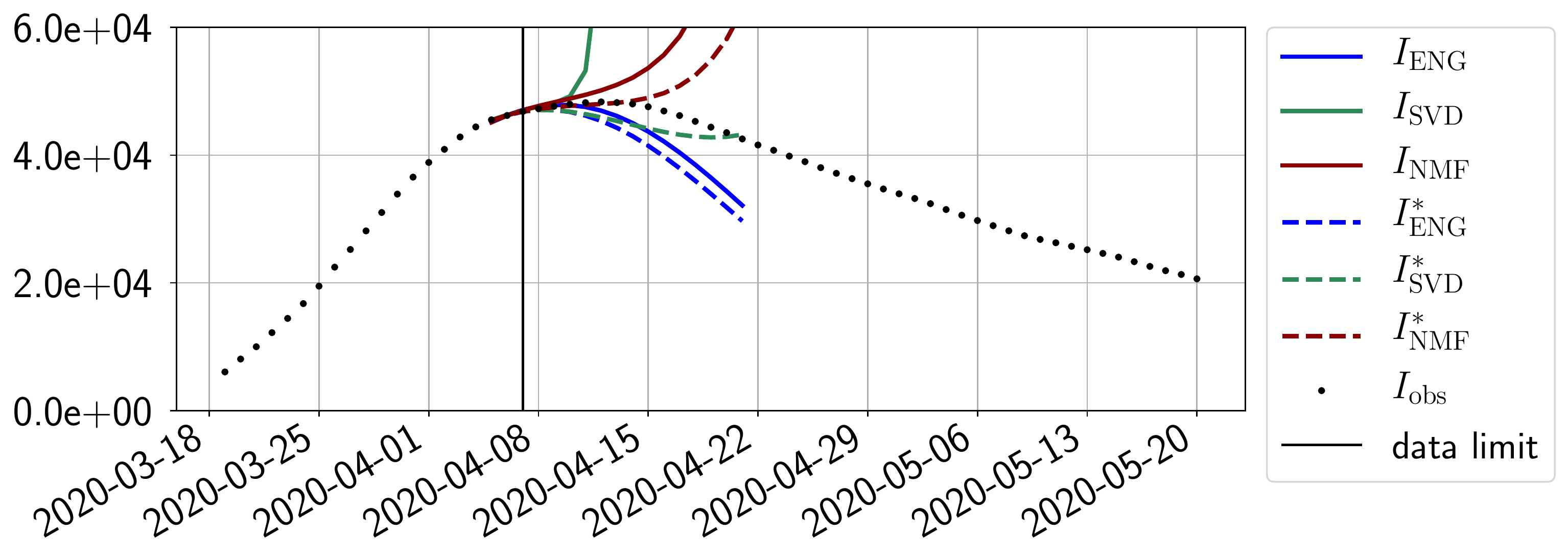}
\caption{Infected}
\end{subfigure}
\begin{subfigure}{.45\textwidth}
\includegraphics[width=1\textwidth]{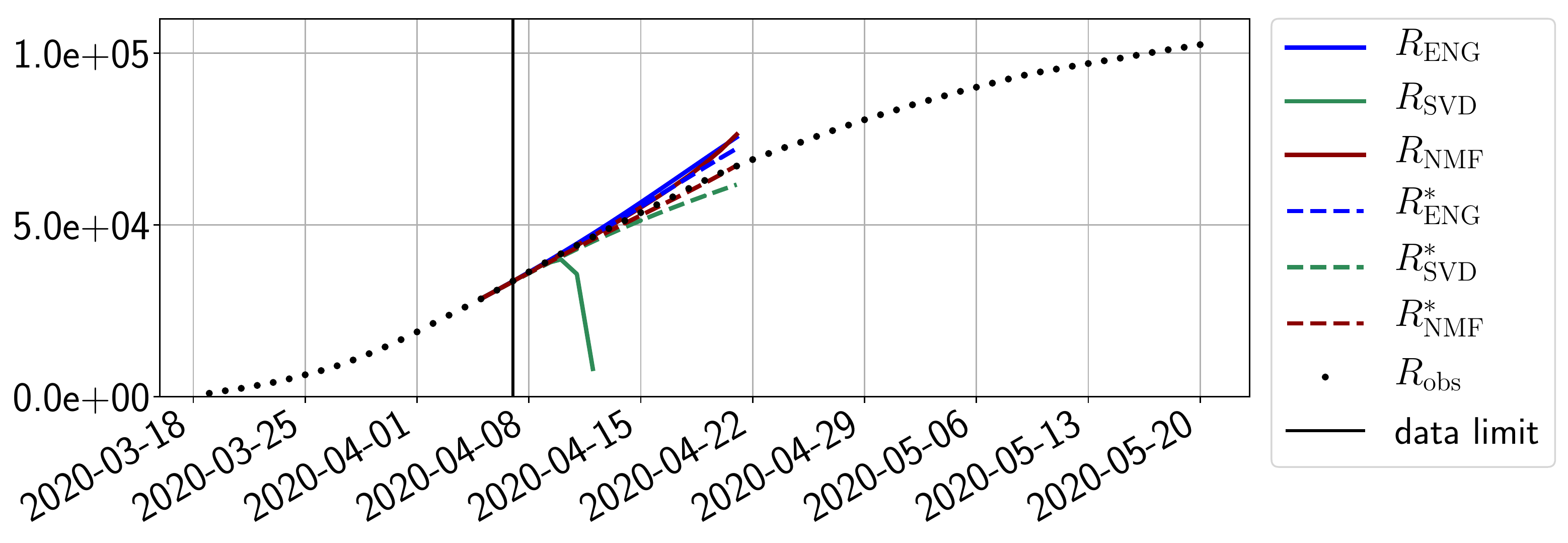}
\caption{Removed}
\end{subfigure}
\caption{14 day forecasts starting from $T=07/04$.}
\label{fig:forecast_0704}
\end{figure}

\begin{figure}[H]
\centering
\begin{subfigure}{.45\textwidth}
\includegraphics[width=1\textwidth]{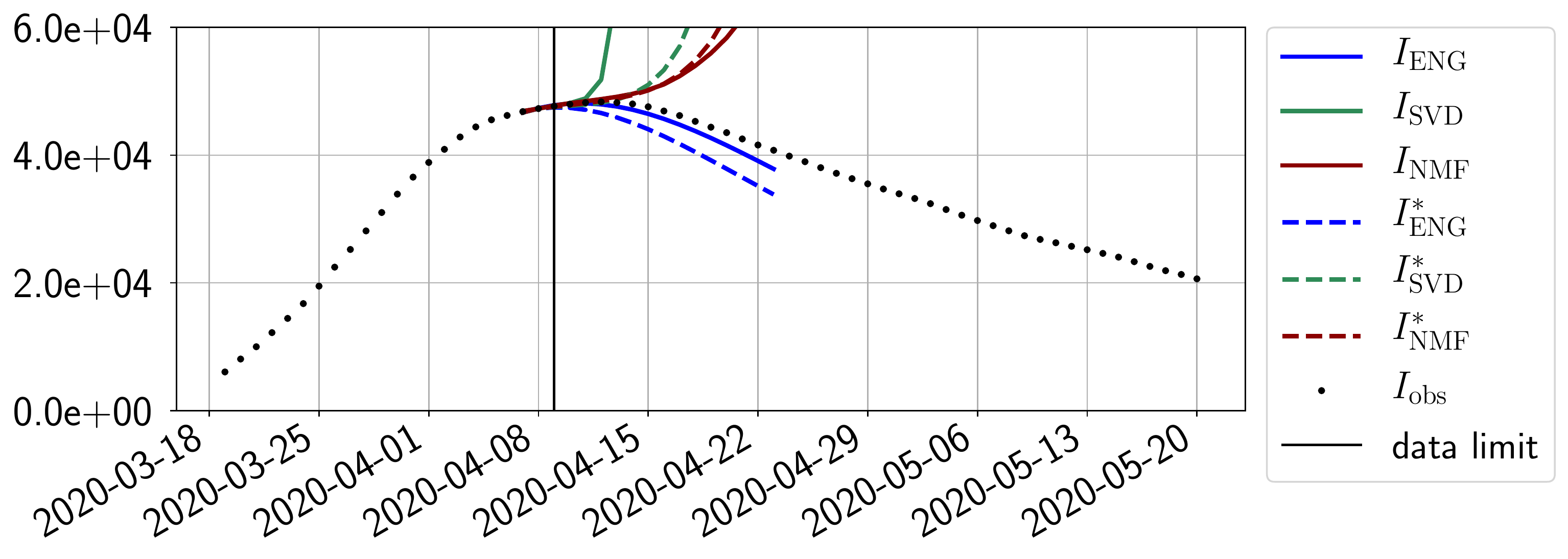}
\caption{Infected}
\end{subfigure}
\begin{subfigure}{.45\textwidth}
\includegraphics[width=1\textwidth]{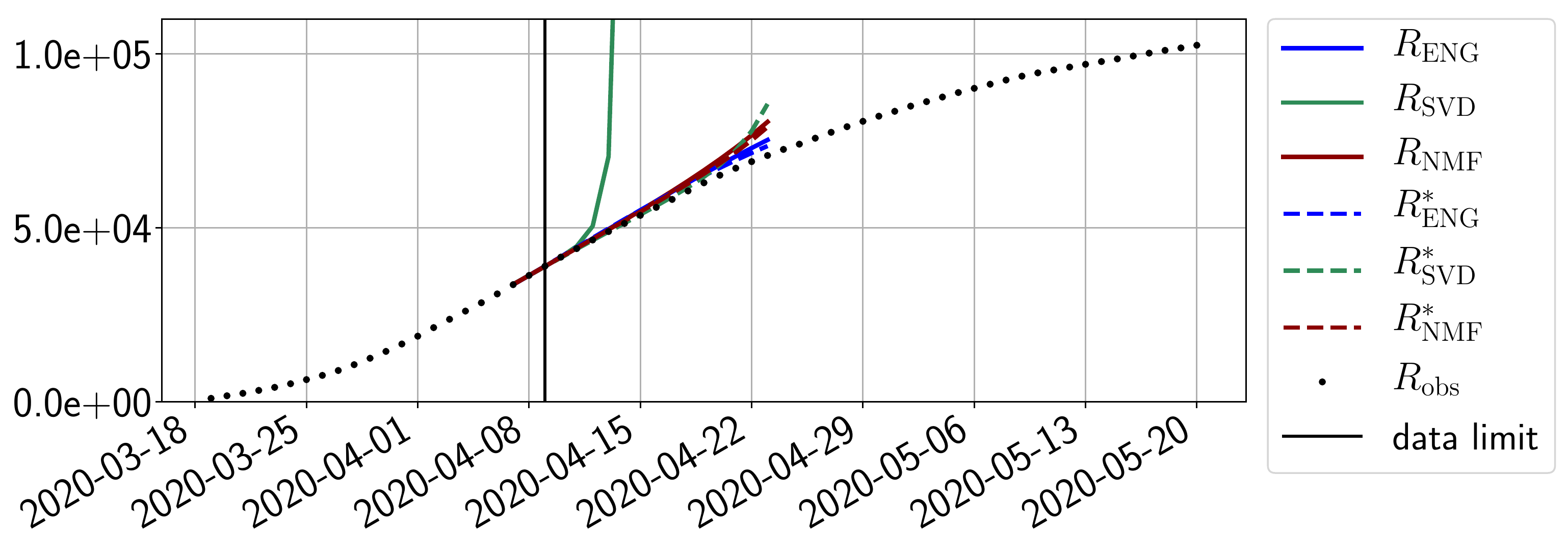}
\caption{Removed}
\end{subfigure}
\caption{14 day forecasts starting from $T=09/04$.}
\label{fig:forecast_0904}
\end{figure}

\begin{figure}[H]
\centering
\begin{subfigure}{.45\textwidth}
\includegraphics[width=1\textwidth]{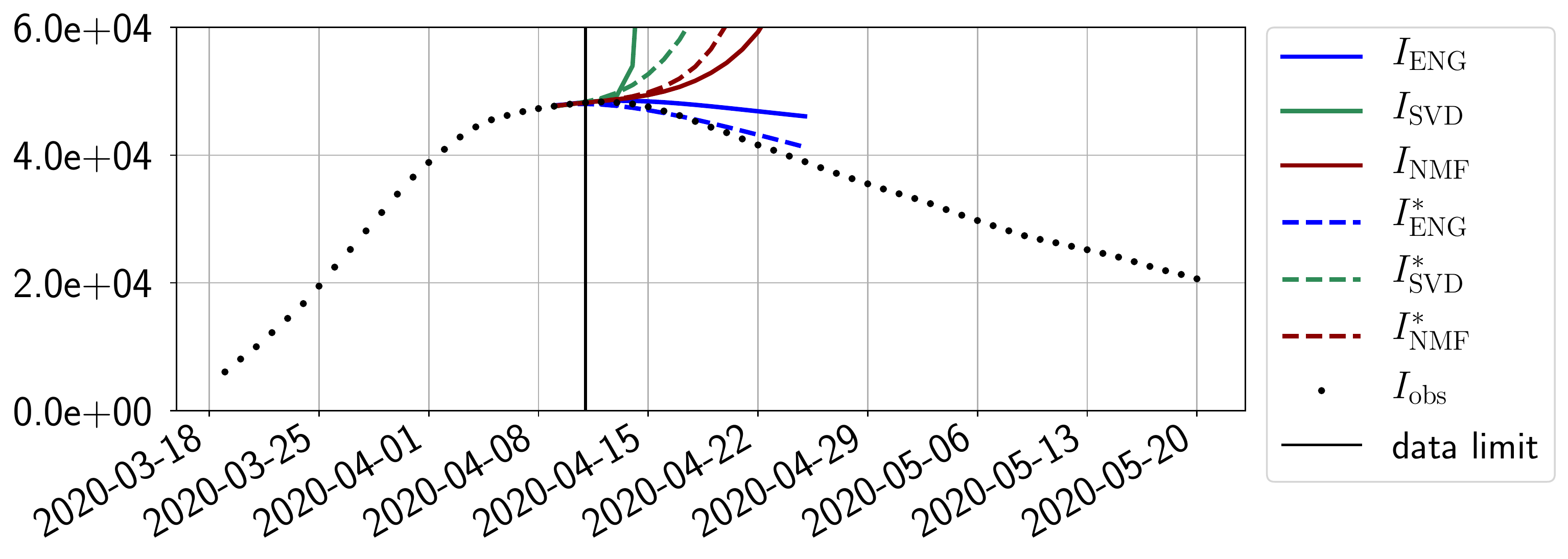}
\caption{Infected}
\end{subfigure}
\begin{subfigure}{.45\textwidth}
\includegraphics[width=1\textwidth]{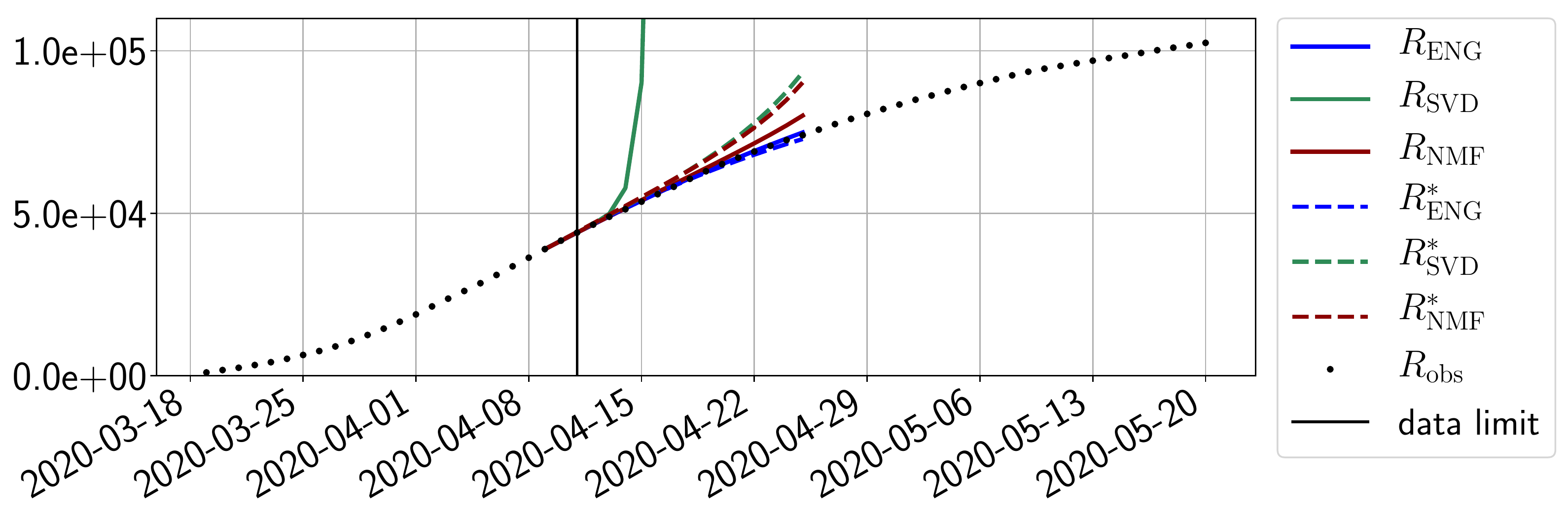}
\caption{Removed}
\end{subfigure}
\caption{14 day forecasts starting from $T=11/04$.}
\label{fig:forecast_1104}
\end{figure}

\begin{figure}[H]
\centering
\begin{subfigure}{.45\textwidth}
\includegraphics[width=1\textwidth]{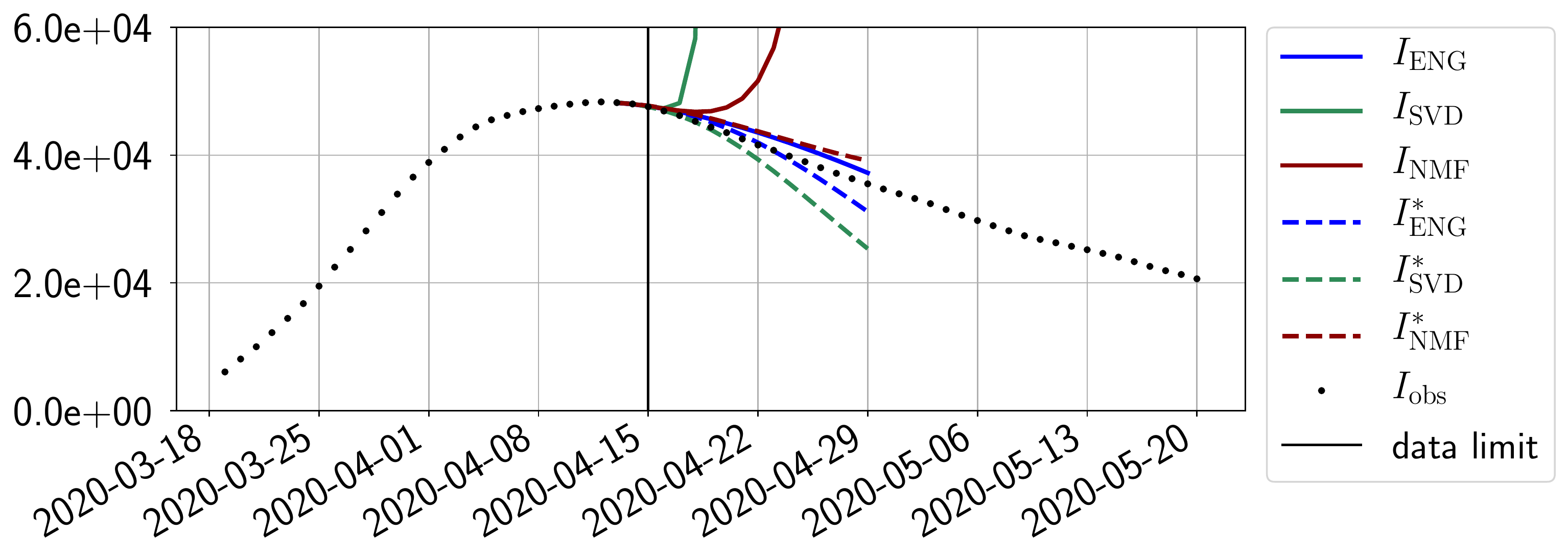}
\caption{Infected}
\end{subfigure}
\begin{subfigure}{.45\textwidth}
\includegraphics[width=1\textwidth]{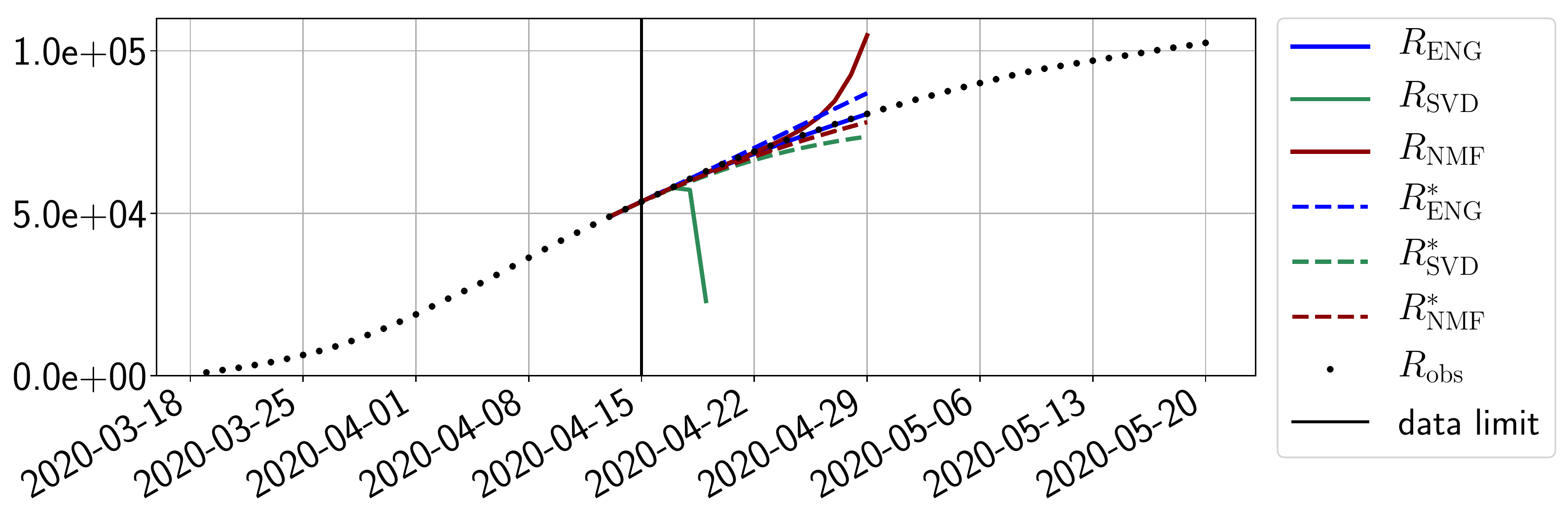}
\caption{Removed}
\end{subfigure}
\caption{14 day forecasts starting from $T=15/04$.}
\label{fig:forecast_1504}
\end{figure}

\begin{figure}[H]
\centering
\begin{subfigure}{.45\textwidth}
\includegraphics[width=1\textwidth]{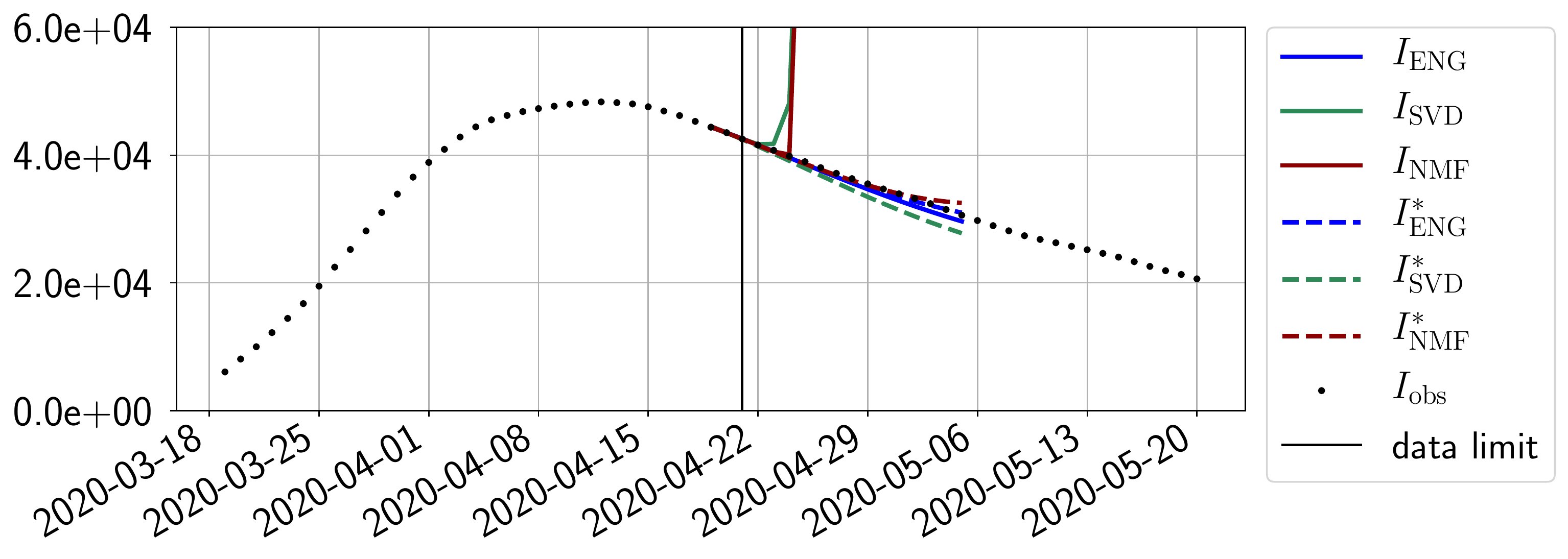}
\caption{Infected}
\end{subfigure}
\begin{subfigure}{.45\textwidth}
\includegraphics[width=1\textwidth]{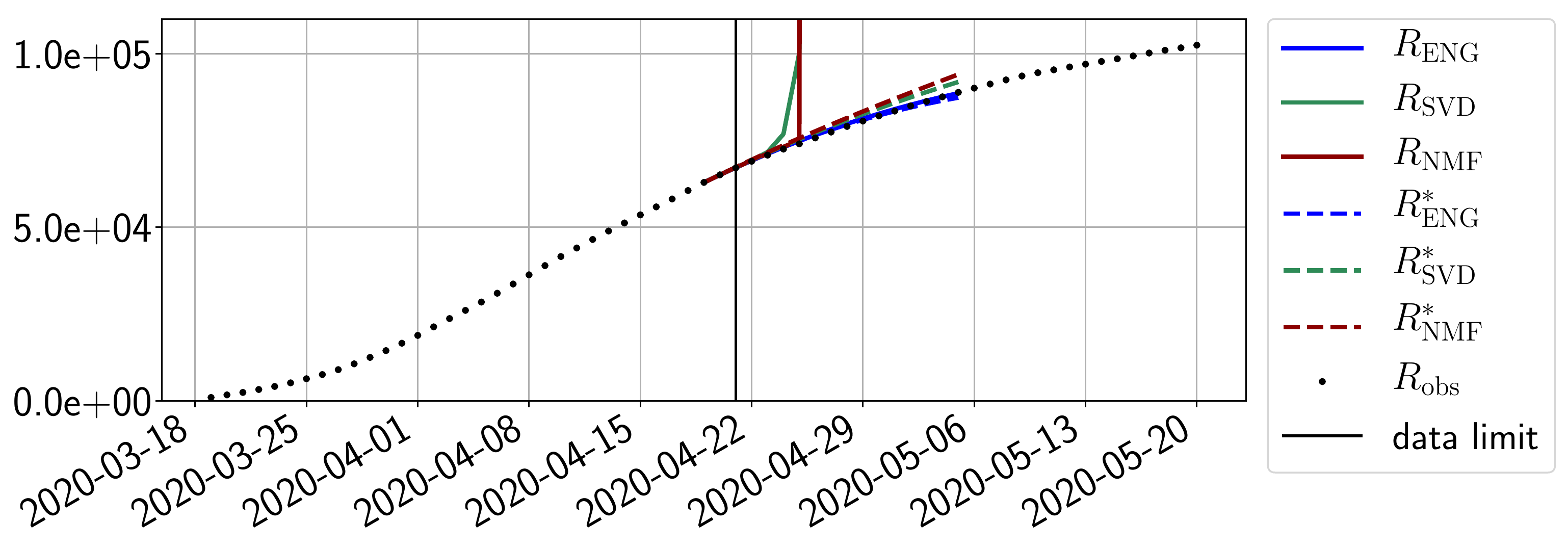}
\caption{Removed}
\end{subfigure}
\caption{14 day forecasts starting from $T=21/04$.}
\label{fig:forecast_2104}
\end{figure}

\begin{figure}[H]
\centering
\begin{subfigure}{.45\textwidth}
\includegraphics[width=1\textwidth]{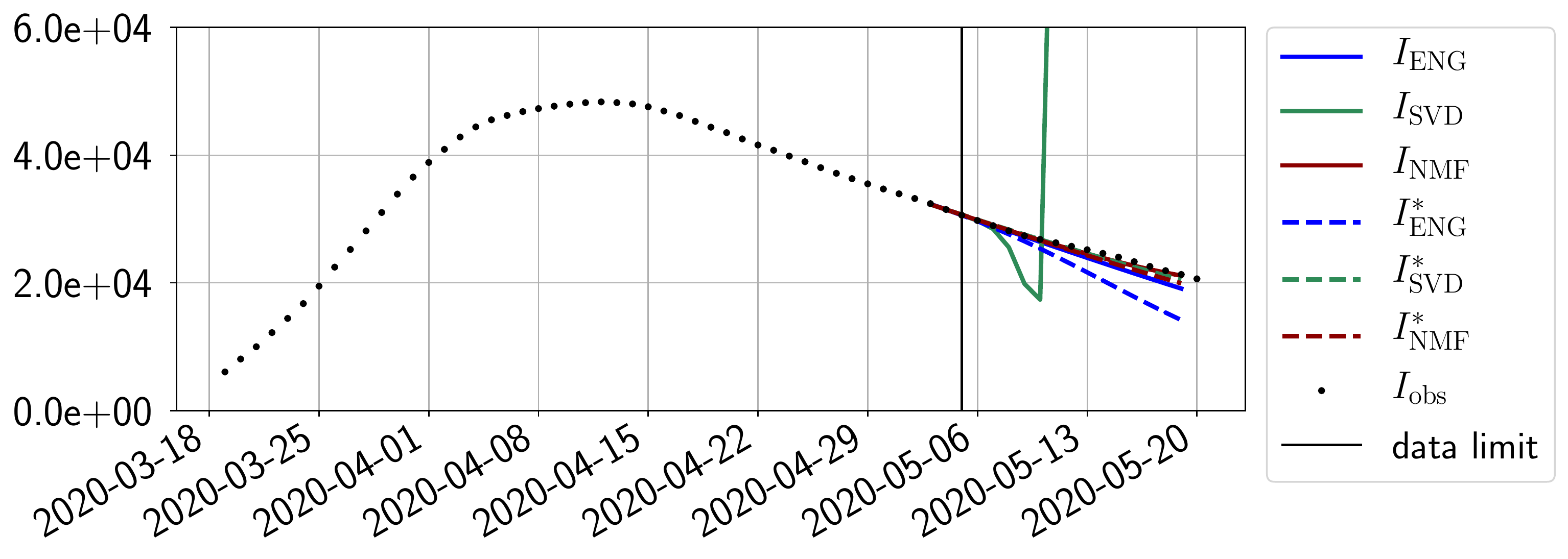}
\caption{Infected}
\end{subfigure}
\begin{subfigure}{.45\textwidth}
\includegraphics[width=1\textwidth]{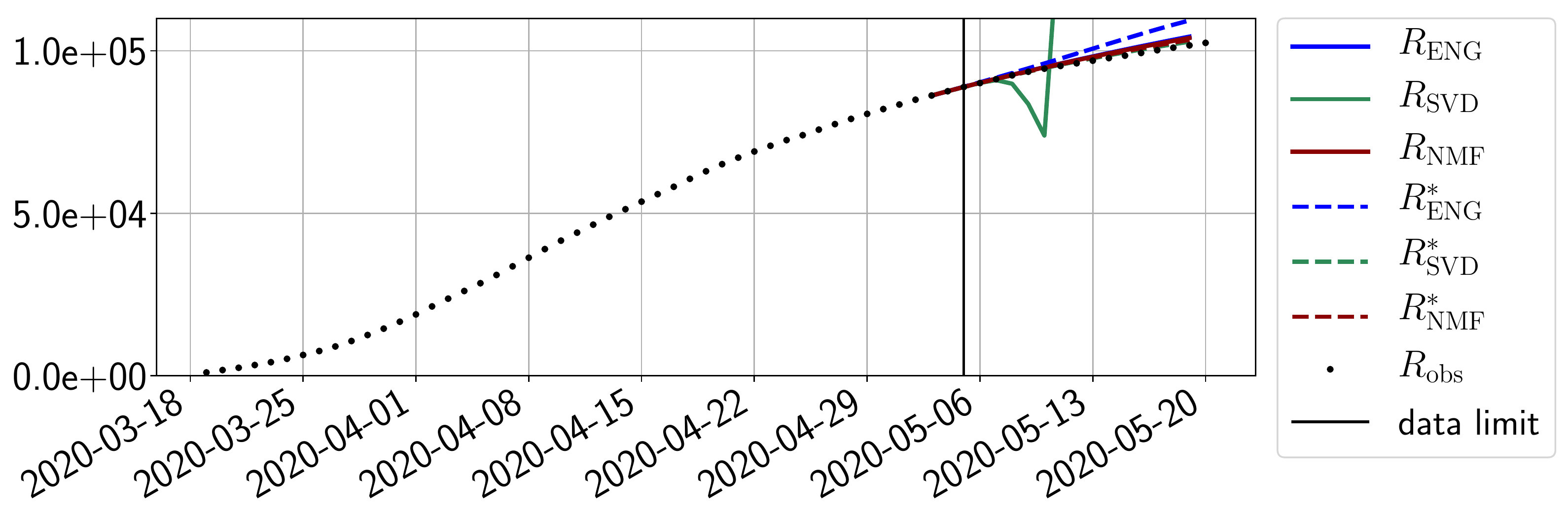}
\caption{Removed}
\end{subfigure}
\caption{14 day forecasts starting from $T=05/05$.}
\label{fig:forecast_0505}
\end{figure}

\subsubsection{Focus on the forecasting with ENG }
\label{sec:forecast-ENG}
For our best forecasting method (\fitbg~using ENG), we plot in Figures \ref{fig:forecast_modes_0104} to \ref{fig:forecast_modes_0505} the forecasts for each dimension $n=5$ to $10$. The plots give the \col{forecasts} on a 14 day ahead window for $\beta$, $\gamma$ and the resulting evolution of the infected $I$ and removed $R$. We see that the method performs reasonably well for all values of $n$, and proves that the results of the previous section with the averaged forecast are not compensating spurious effects which could occur for certain values of $n$. \col{We intentionally chose to display the inaccurate forecasts from 03/04, 07/04 and 11/04 as they are among of the worst predictions obtained using this method, however, it is important to mention that despite the lack of accuracy for these cases it remains plausible epidemic behaviors with different but realistic evolutions for $\beta$ and $\gamma$ compared to the actual evolution.} Note that the method was able to predict the peak of the epidemic several days in advance (see Figure \ref{fig:forecast_modes_0104}). \col{We also observe that the prediction on $\gamma$ is difficult at all times due to the fact that $\gamma^*_\obs$ presents an oscillatory behavior.} Despite this difficulty, the resulting forecasts for $I$ and $R$ are very satisfactory in general.

\begin{figure}[H]
\centering
\begin{subfigure}{.45\textwidth}
\includegraphics[width=1\textwidth]{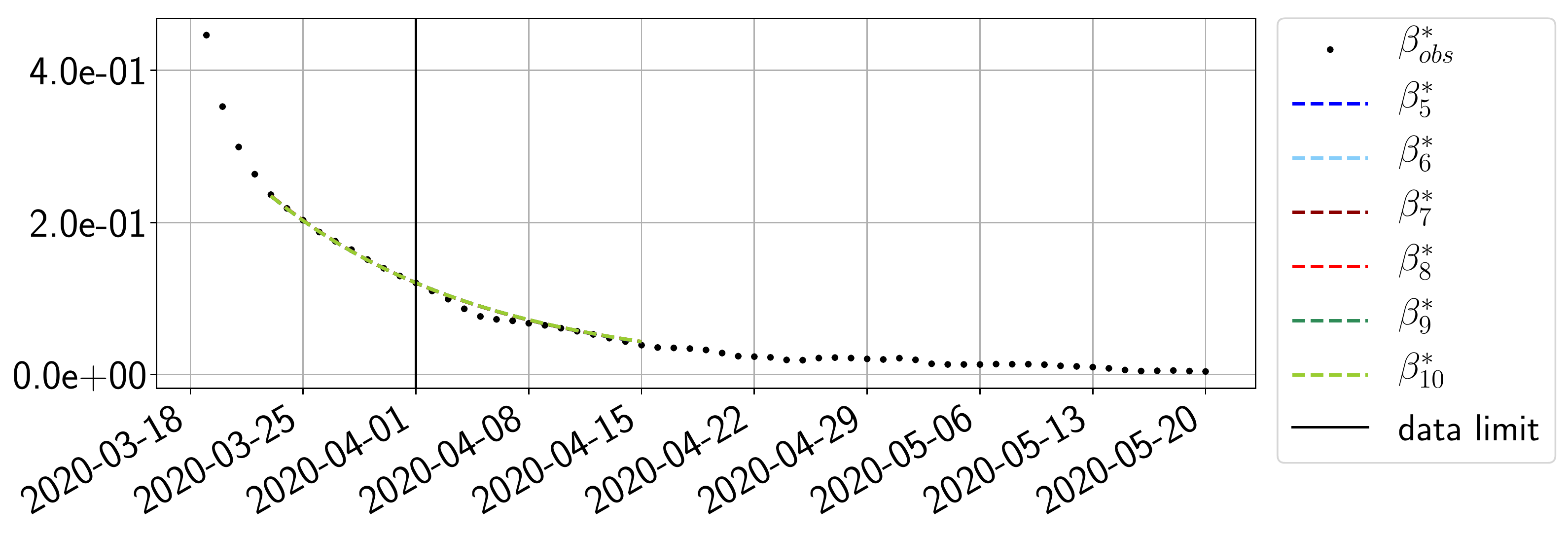}
\caption{$\beta$}
\end{subfigure}
\begin{subfigure}{.45\textwidth}
\includegraphics[width=1\textwidth]{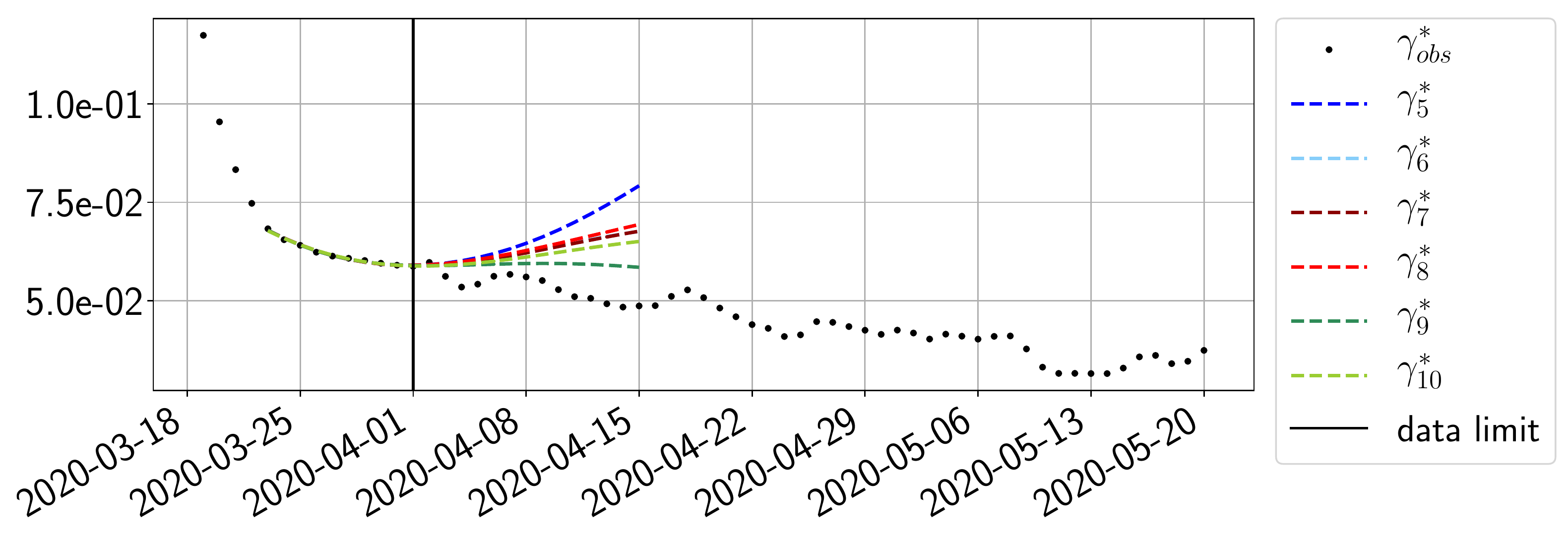}
\caption{$\gamma$}
\end{subfigure}

\vspace{0.4cm}

\begin{subfigure}{.45\textwidth}
\includegraphics[width=1\textwidth]{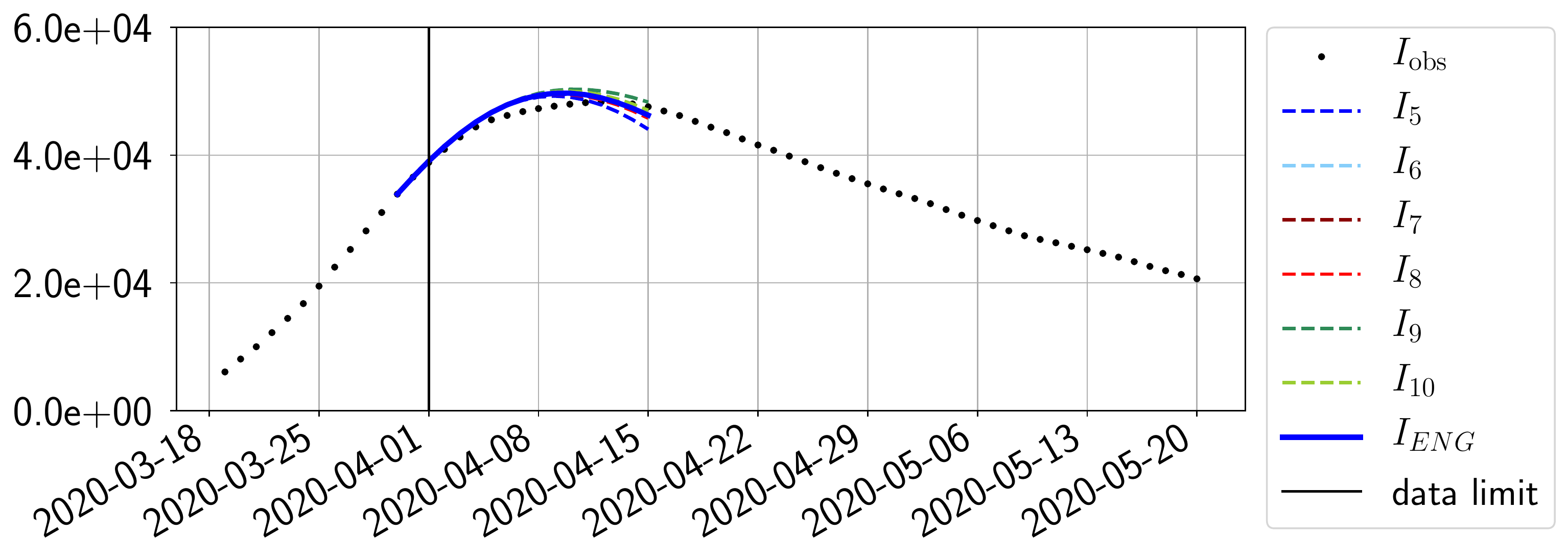}
\caption{Infected}
\end{subfigure}
\begin{subfigure}{.45\textwidth}
\includegraphics[width=1\textwidth]{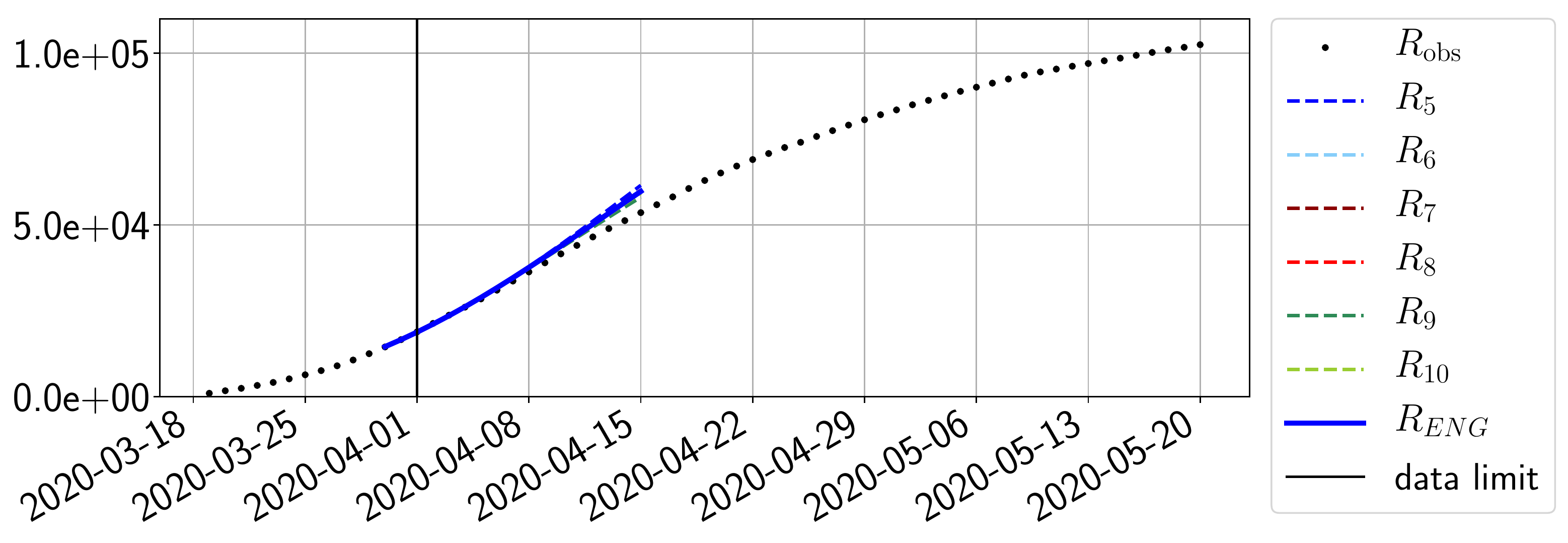}
\caption{Removed}
\end{subfigure}
\caption{ENG forecast from $T=01/04$}
\label{fig:forecast_modes_0104}
\end{figure}

\begin{figure}[H]
\centering
\begin{subfigure}{.45\textwidth}
\includegraphics[width=1\textwidth]{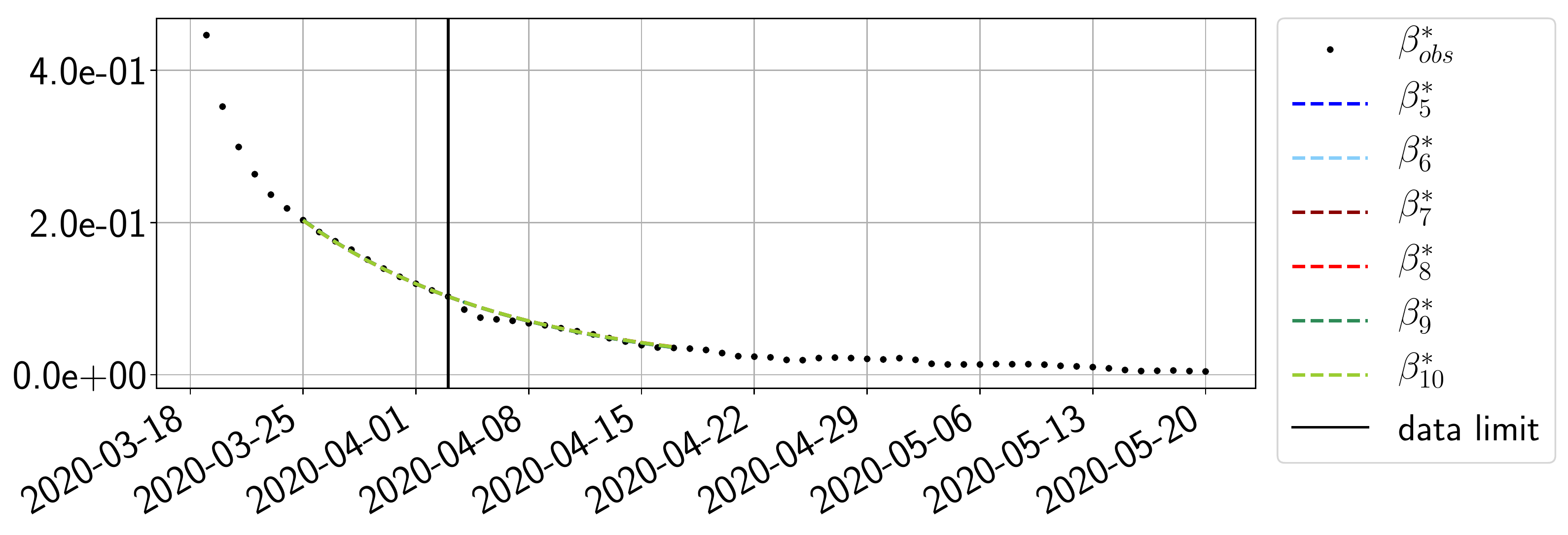}
\caption{$\beta$}
\end{subfigure}
\begin{subfigure}{.45\textwidth}
\includegraphics[width=1\textwidth]{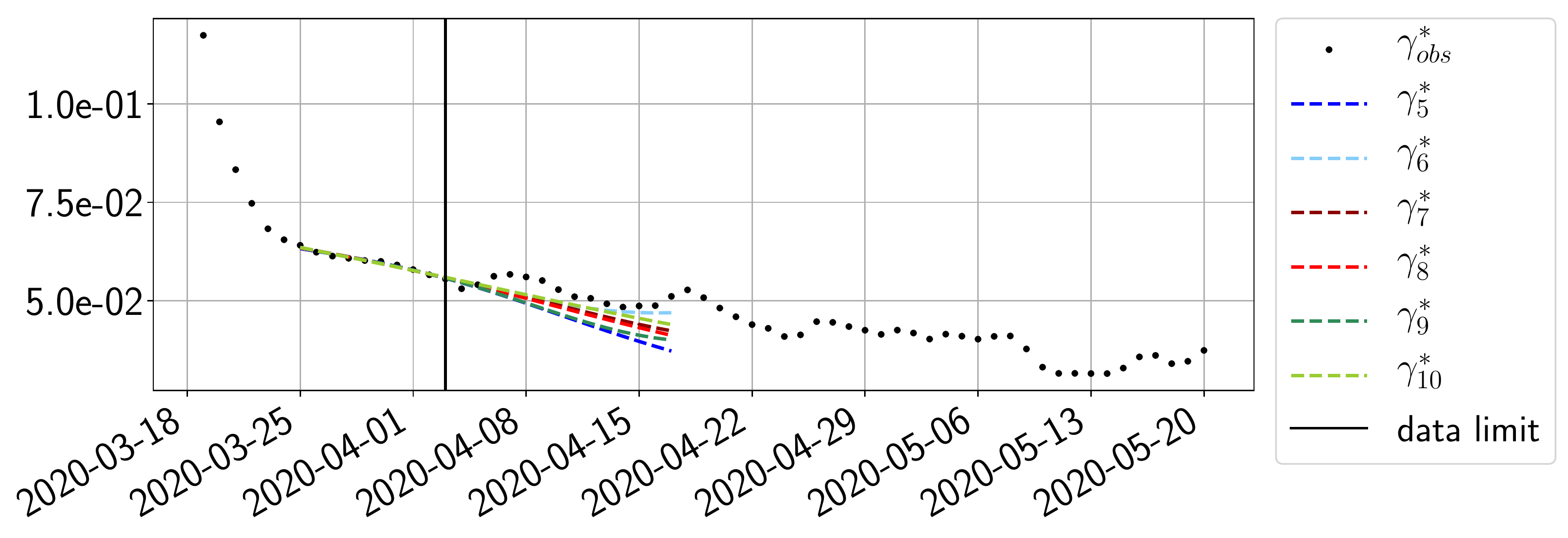}
\caption{$\gamma$}
\end{subfigure}

\vspace{0.4cm}

\begin{subfigure}{.45\textwidth}
\includegraphics[width=1\textwidth]{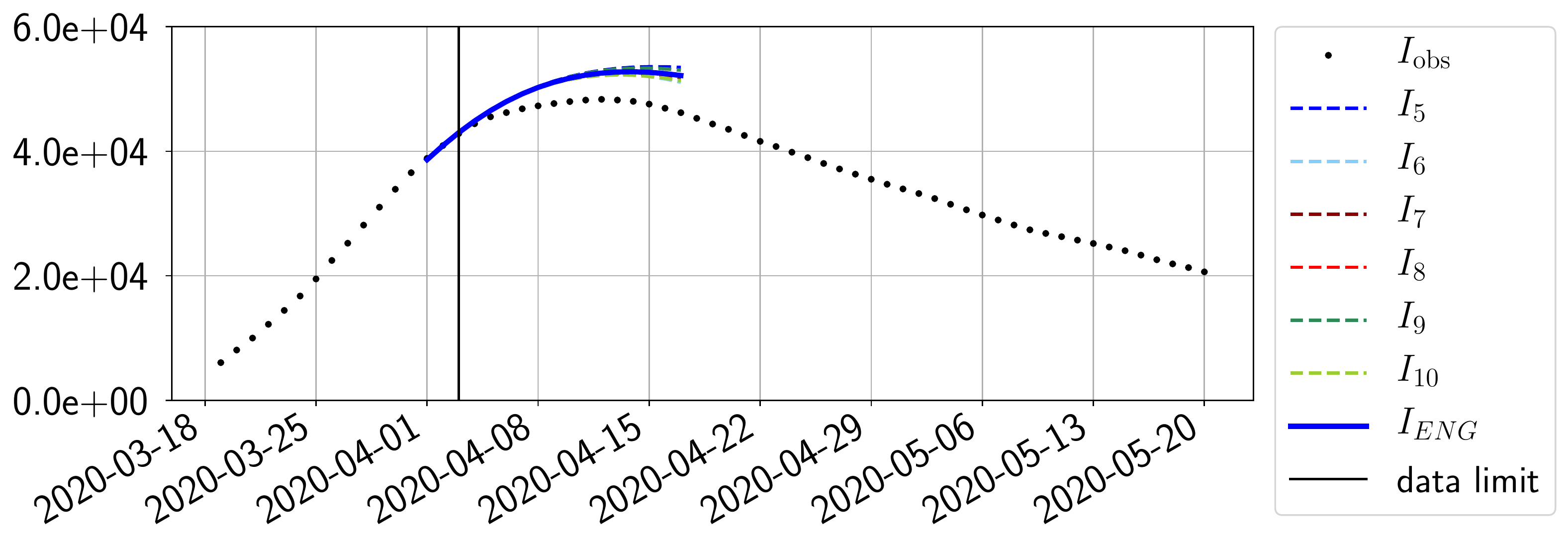}
\caption{Infected}
\end{subfigure}
\begin{subfigure}{.45\textwidth}
\includegraphics[width=1\textwidth]{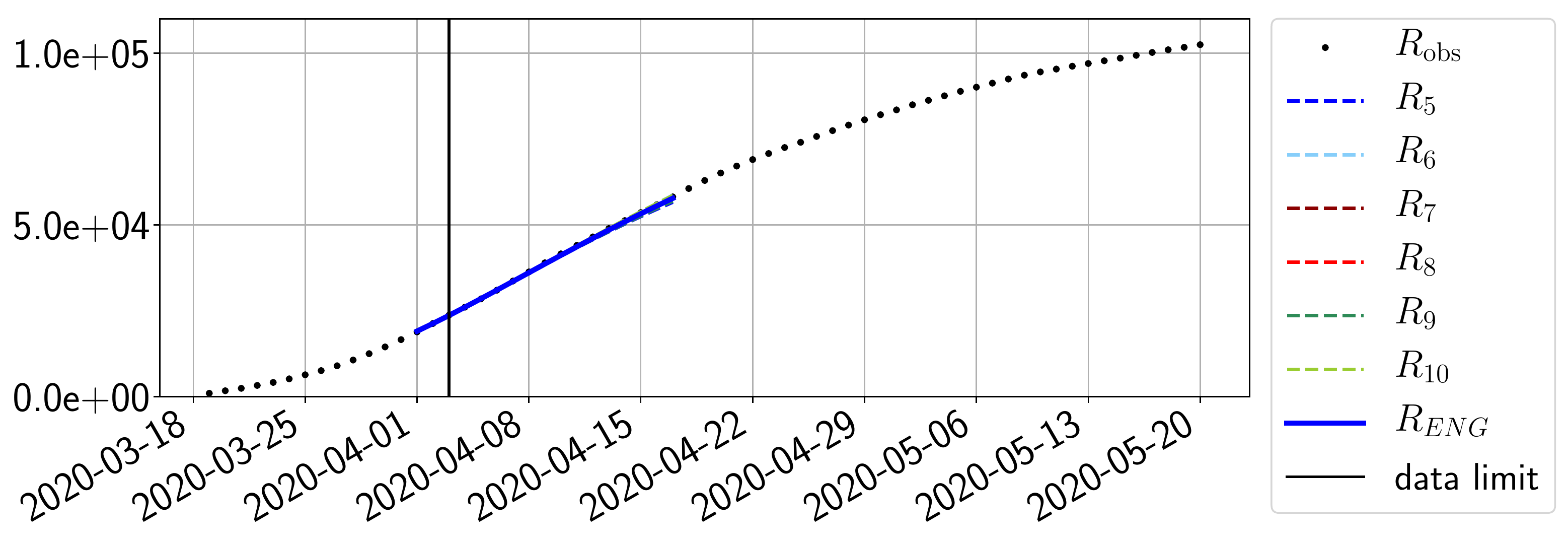}
\caption{Removed}
\end{subfigure}
\caption{ENG forecast from $T=03/04$}
\label{fig:forecast_modes_0304}
\end{figure}

\begin{figure}[H]
\centering
\begin{subfigure}{.45\textwidth}
\includegraphics[width=1\textwidth]{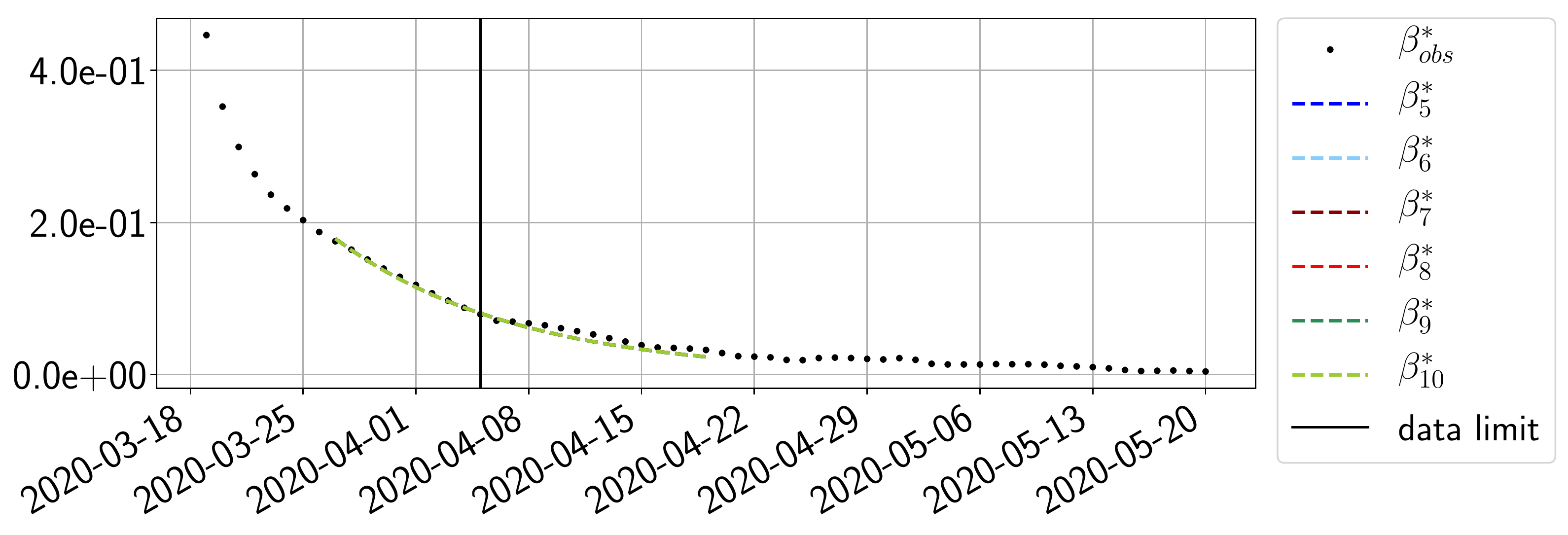}
\caption{$\beta$}
\end{subfigure}
\begin{subfigure}{.45\textwidth}
\includegraphics[width=1\textwidth]{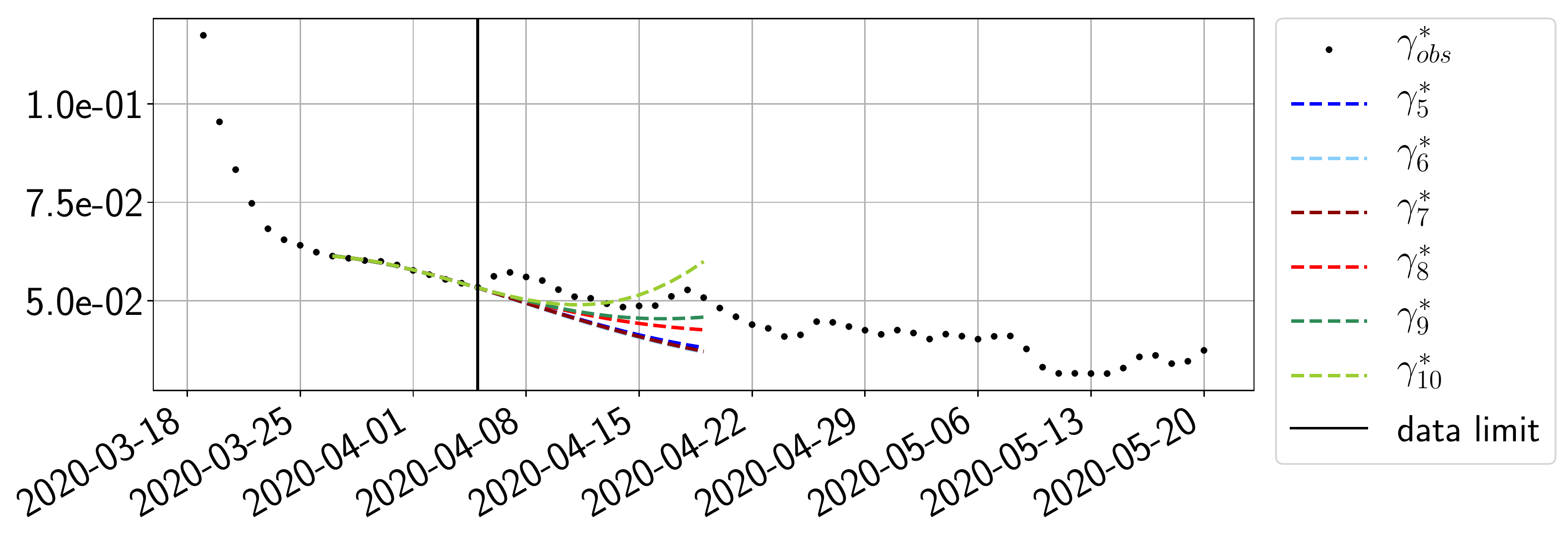}
\caption{$\gamma$}
\end{subfigure}

\vspace{0.4cm}

\begin{subfigure}{.45\textwidth}
\includegraphics[width=1\textwidth]{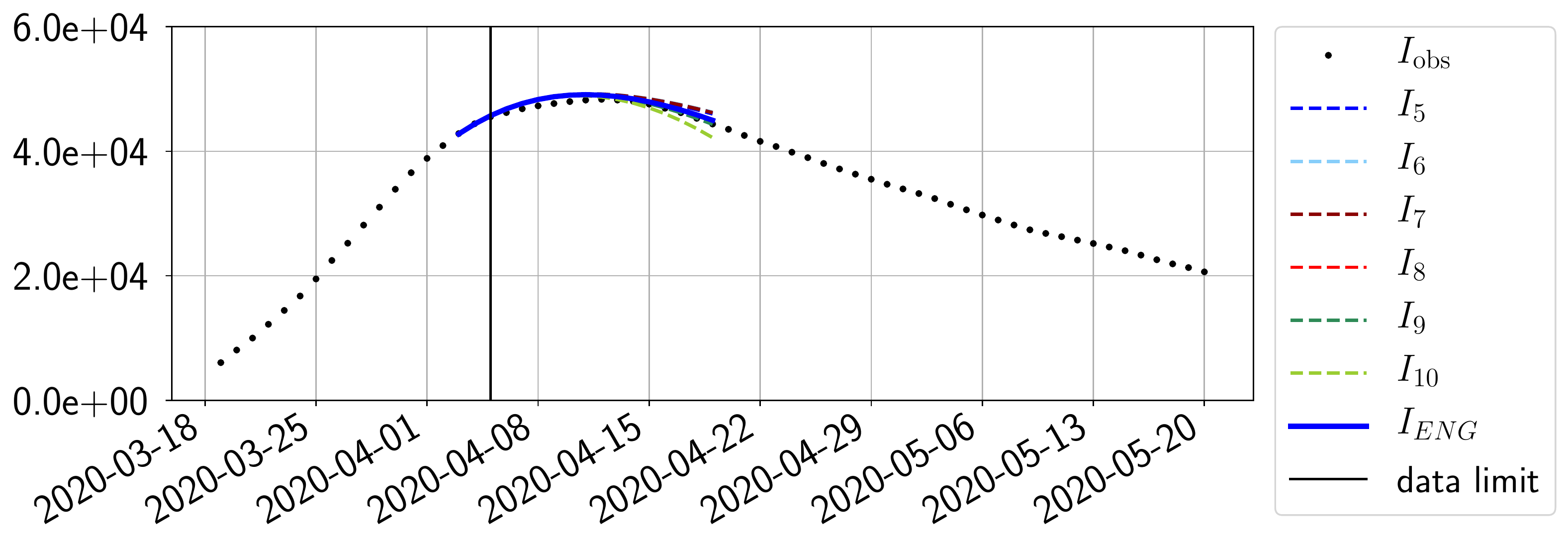}
\caption{Infected}
\end{subfigure}
\begin{subfigure}{.45\textwidth}
\includegraphics[width=1\textwidth]{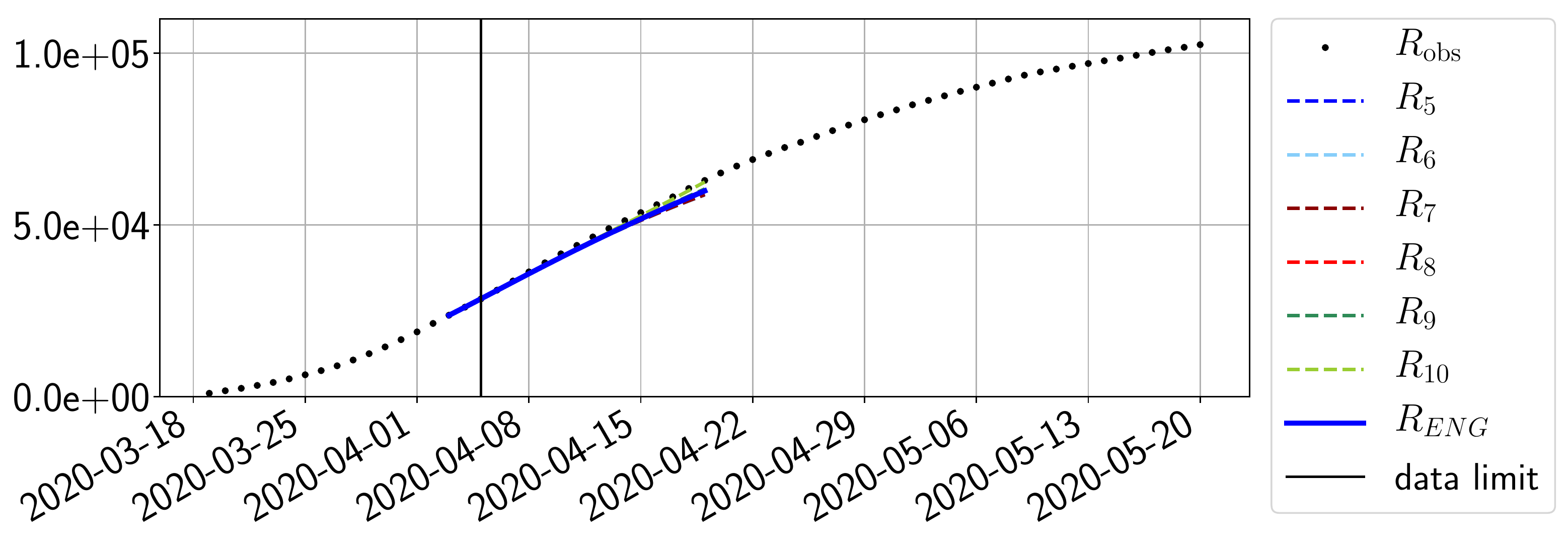}
\caption{Removed}
\end{subfigure}
\caption{ENG forecast from $T=05/04$}
\label{fig:forecast_modes_0504}
\end{figure}

\begin{figure}[H]
\centering
\begin{subfigure}{.45\textwidth}
\includegraphics[width=1\textwidth]{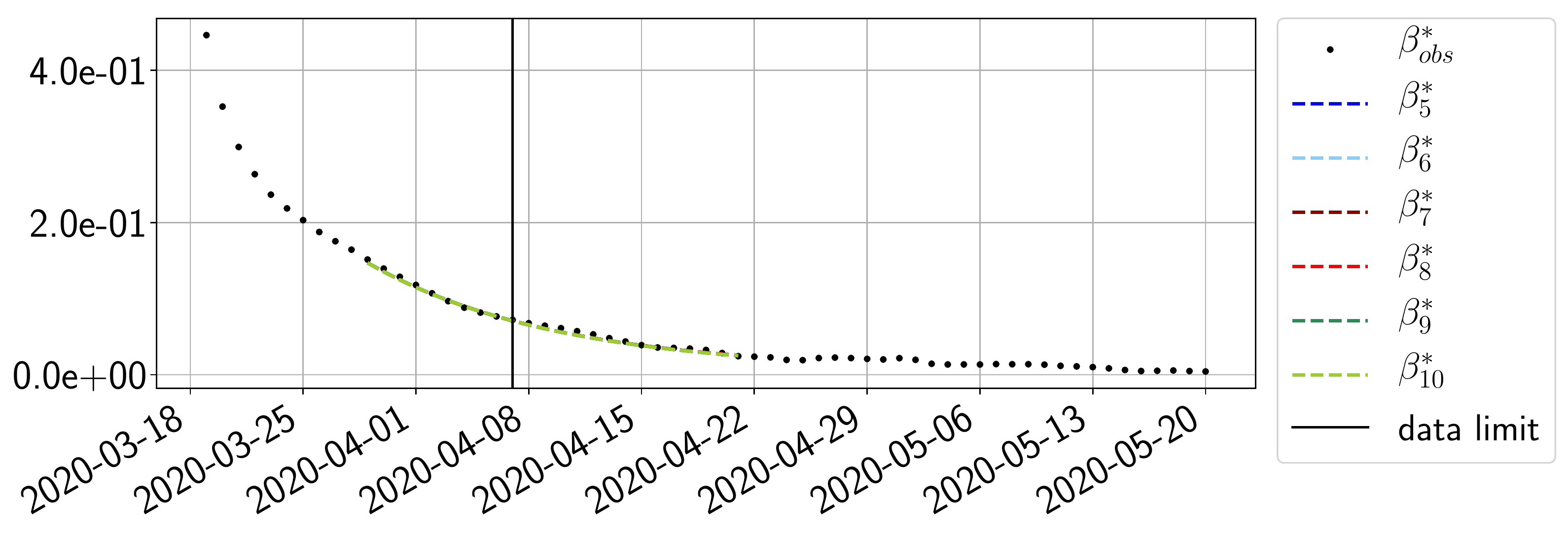}
\caption{$\beta$}
\end{subfigure}
\begin{subfigure}{.45\textwidth}
\includegraphics[width=1\textwidth]{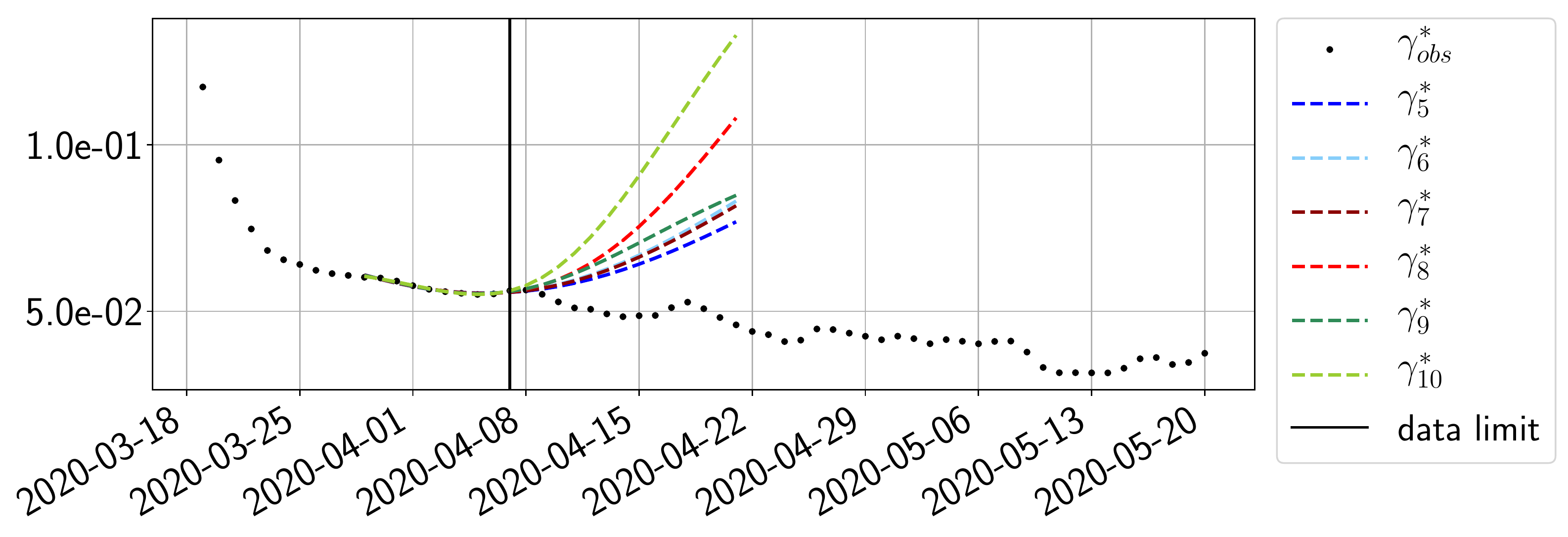}
\caption{$\gamma$}
\end{subfigure}

\vspace{0.4cm}

\begin{subfigure}{.45\textwidth}
\includegraphics[width=1\textwidth]{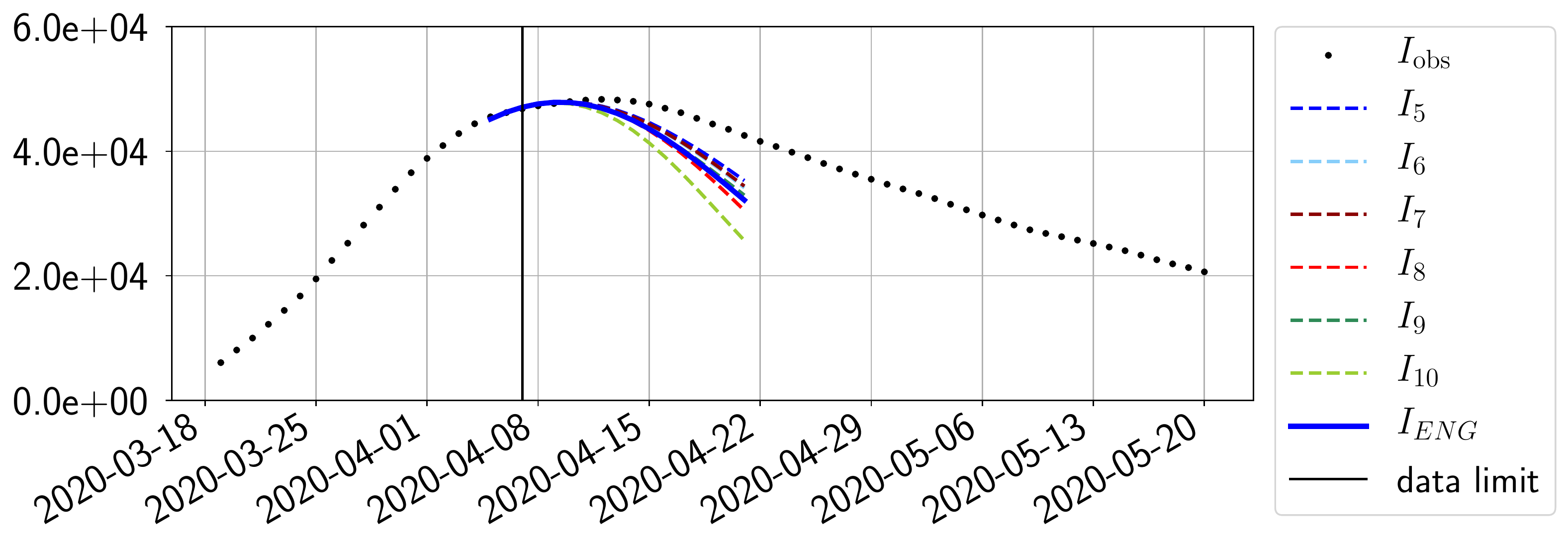}
\caption{Infected}
\end{subfigure}
\begin{subfigure}{.45\textwidth}
\includegraphics[width=1\textwidth]{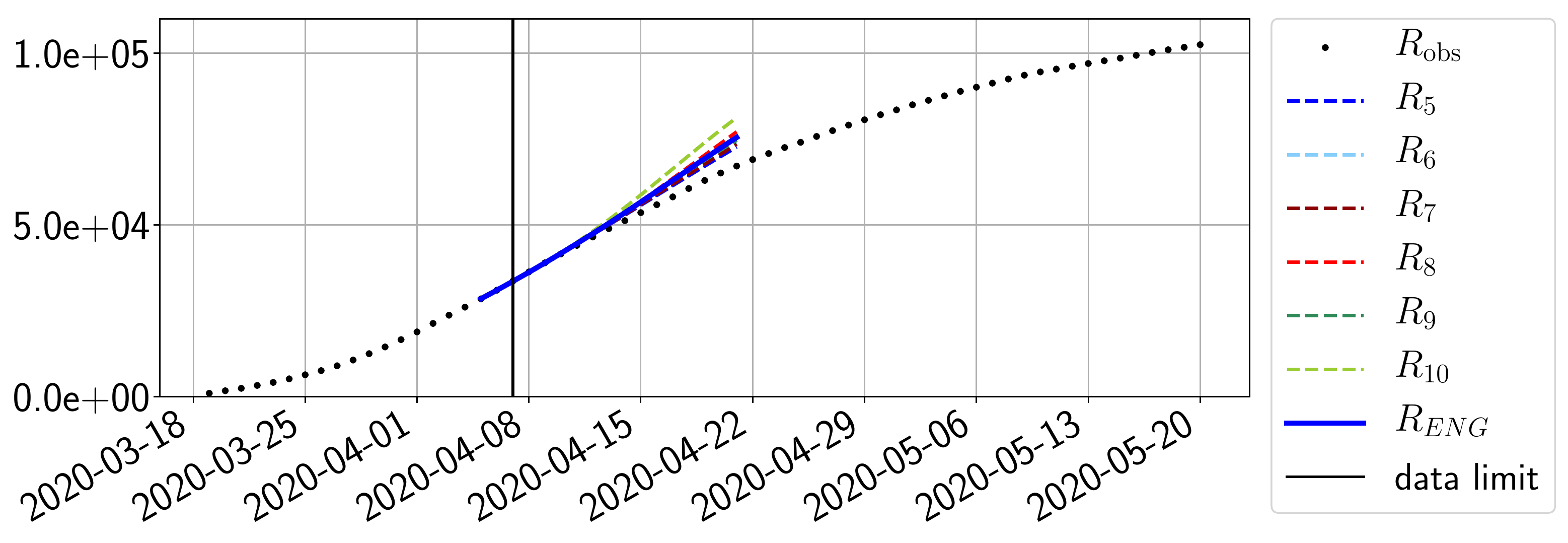}
\caption{Removed}
\end{subfigure}
\caption{ENG forecast from $T=07/04$}
\label{fig:forecast_modes_0704}
\end{figure}

\begin{figure}[H]
\centering
\begin{subfigure}{.45\textwidth}
\includegraphics[width=1\textwidth]{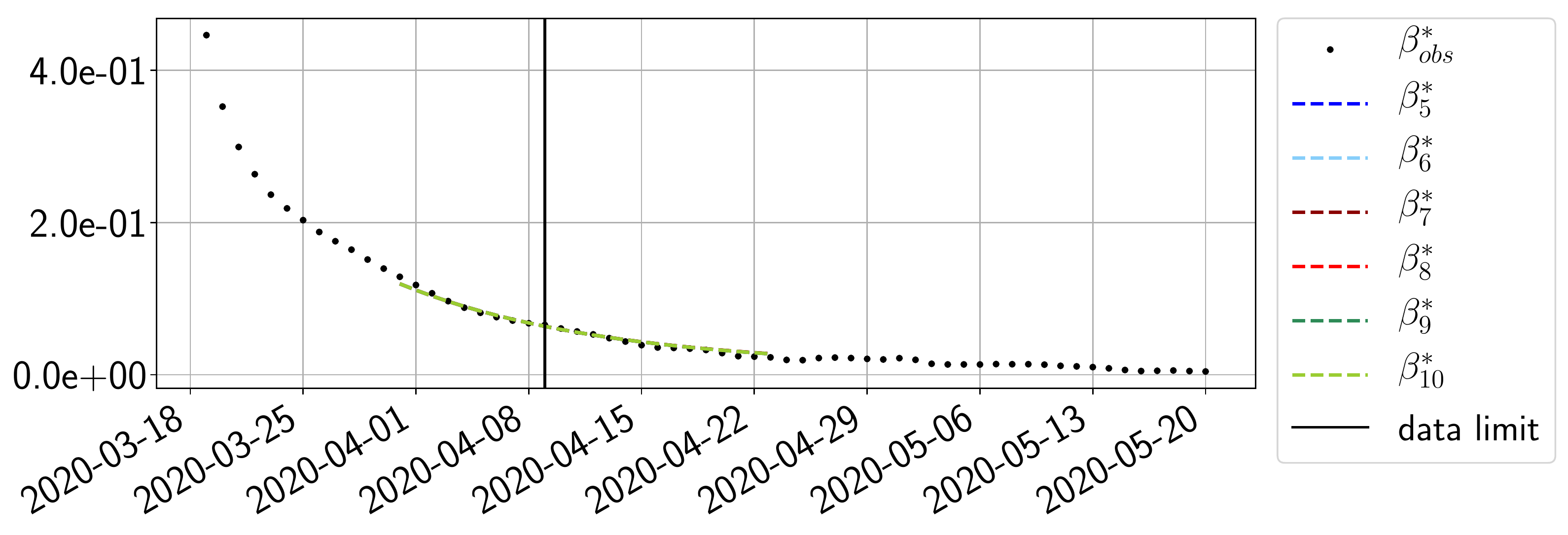}
\caption{$\beta$}
\end{subfigure}
\begin{subfigure}{.45\textwidth}
\includegraphics[width=1\textwidth]{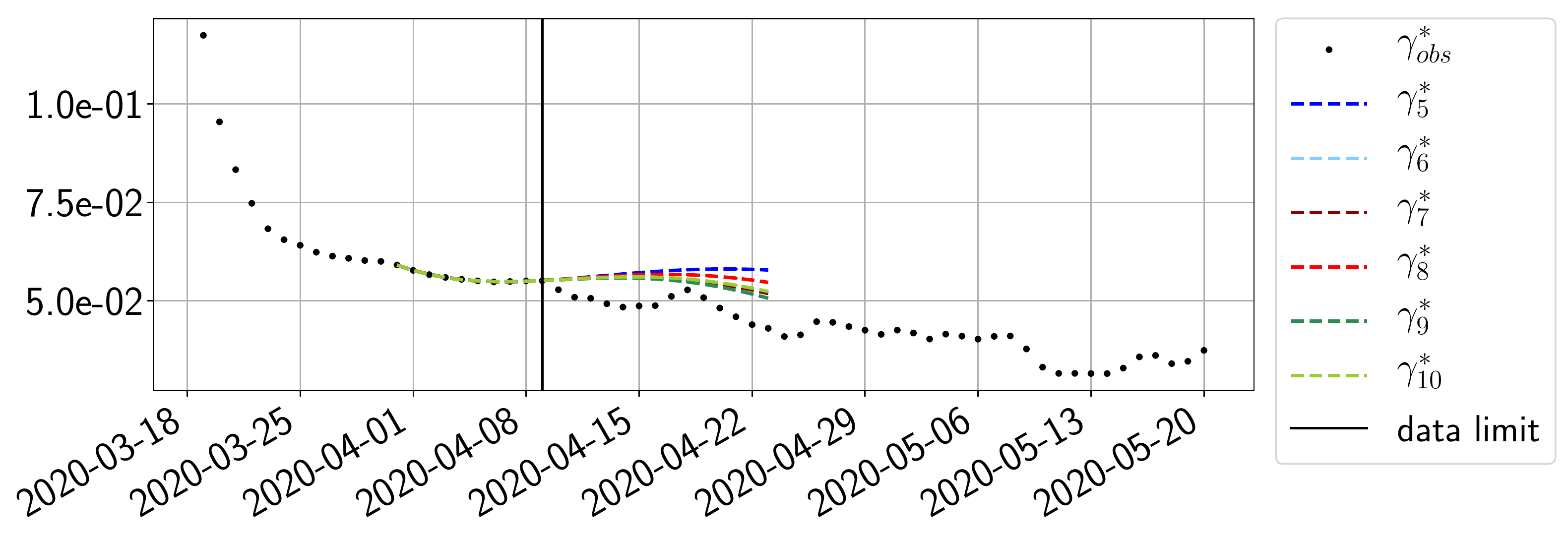}
\caption{$\gamma$}
\end{subfigure}

\vspace{0.4cm}

\begin{subfigure}{.45\textwidth}
\includegraphics[width=1\textwidth]{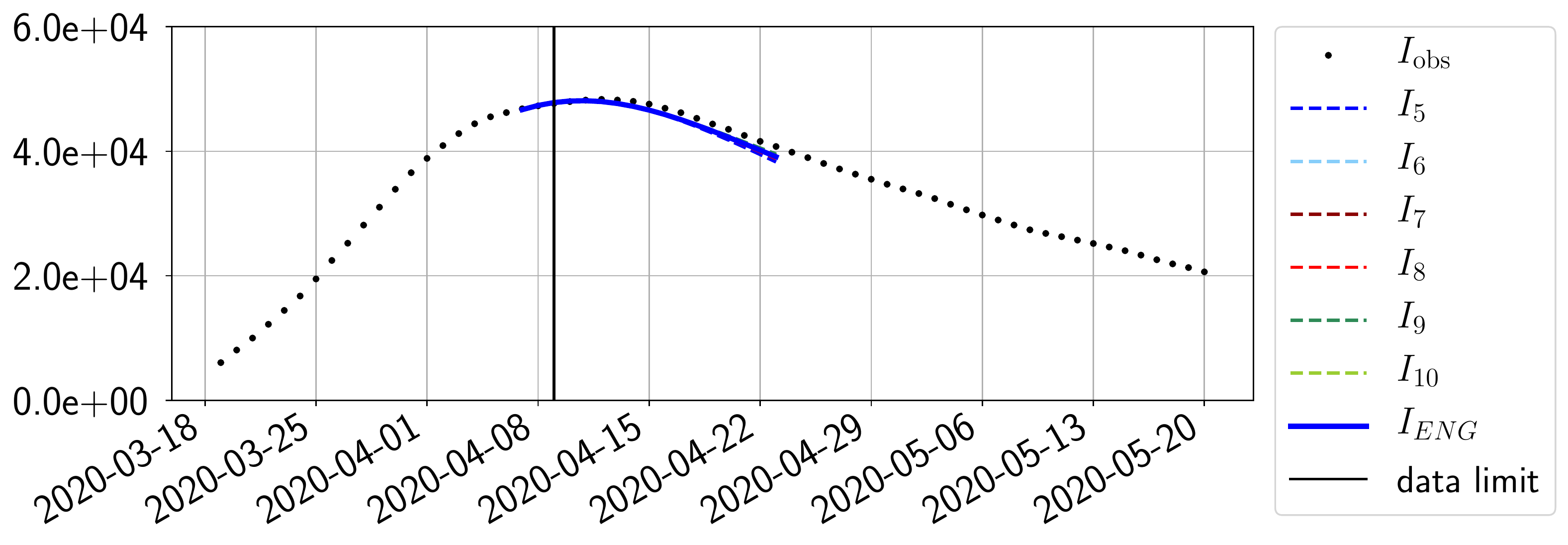}
\caption{Infected}
\end{subfigure}
\begin{subfigure}{.45\textwidth}
\includegraphics[width=1\textwidth]{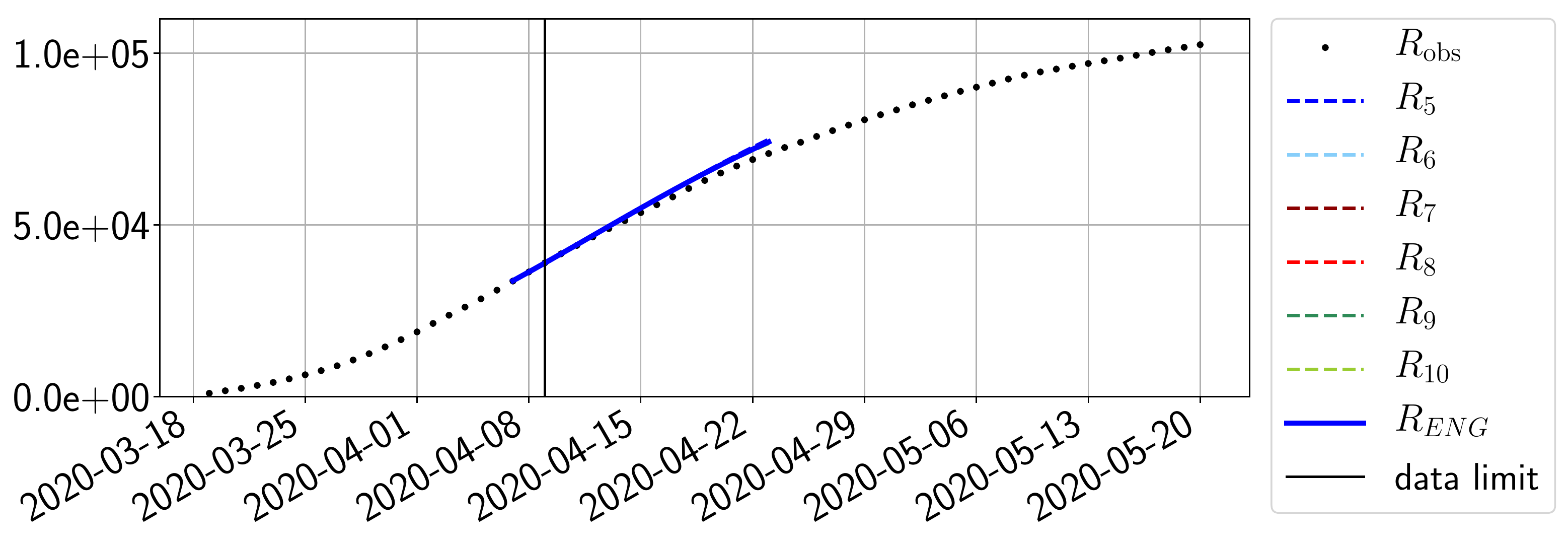}
\caption{Removed}
\end{subfigure}
\caption{ENG forecast from $T=09/04$}
\label{fig:forecast_modes_0904}
\end{figure}

\begin{figure}[H]
\centering
\begin{subfigure}{.45\textwidth}
\includegraphics[width=1\textwidth]{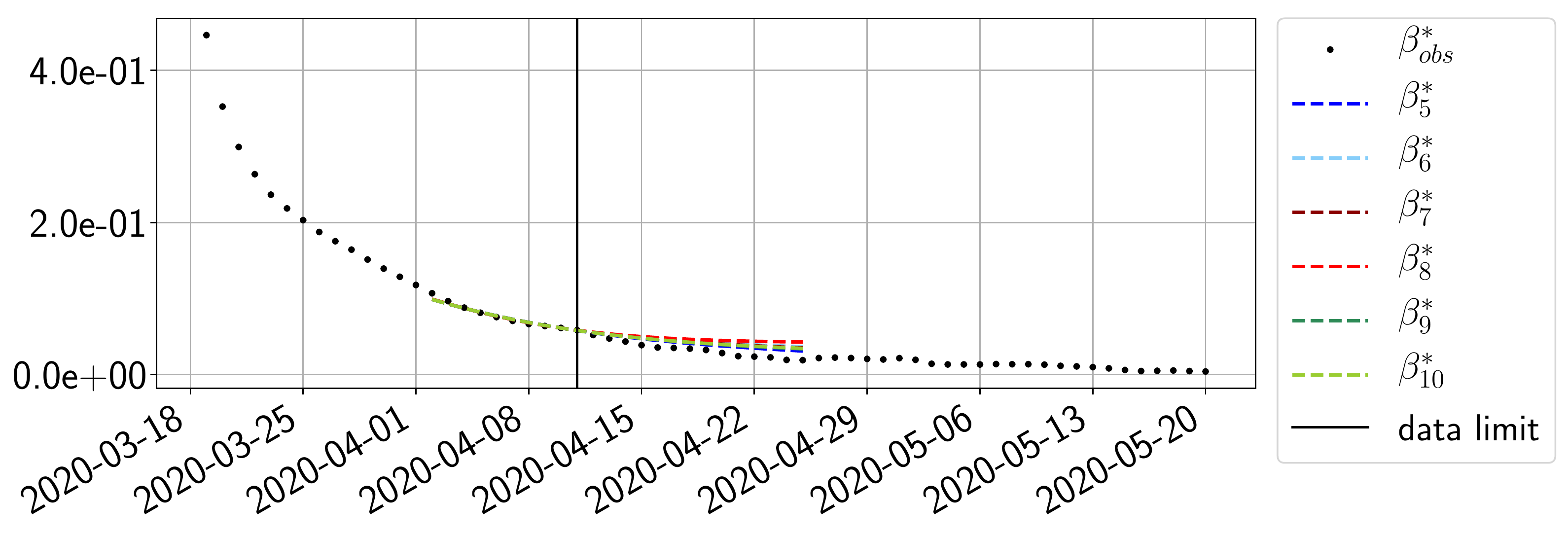}
\caption{$\beta$}
\end{subfigure}
\begin{subfigure}{.45\textwidth}
\includegraphics[width=1\textwidth]{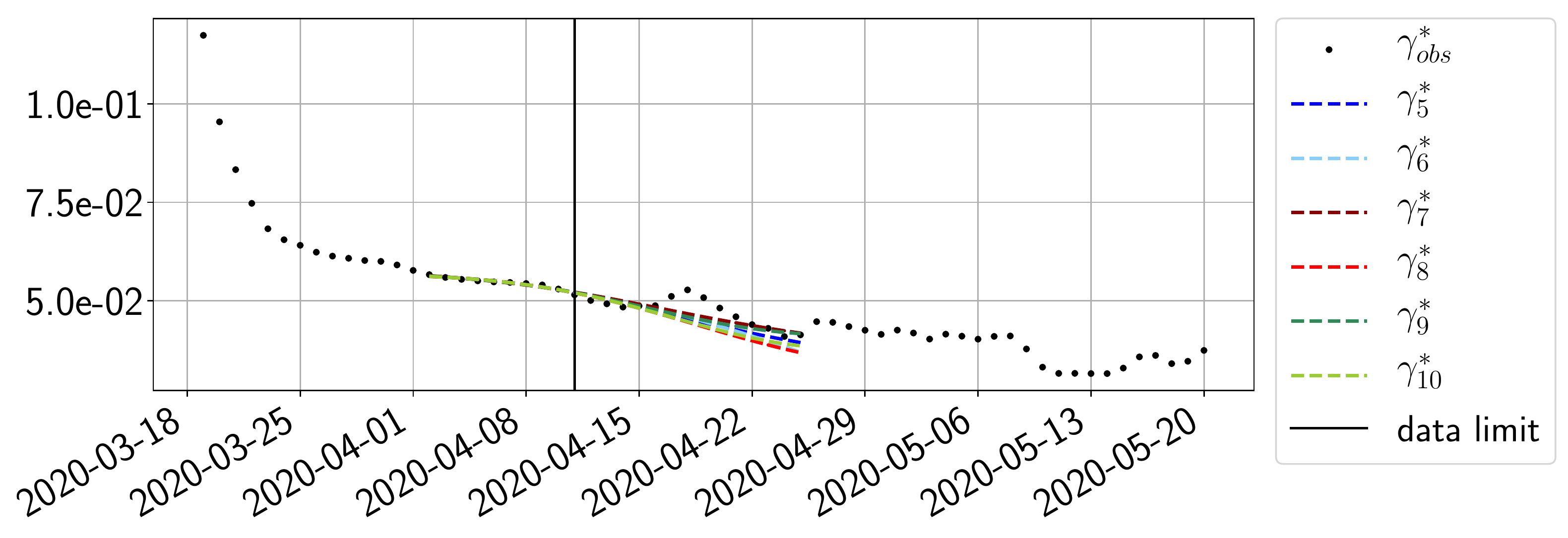}
\caption{$\gamma$}
\end{subfigure}

\vspace{0.4cm}

\begin{subfigure}{.45\textwidth}
\includegraphics[width=1\textwidth]{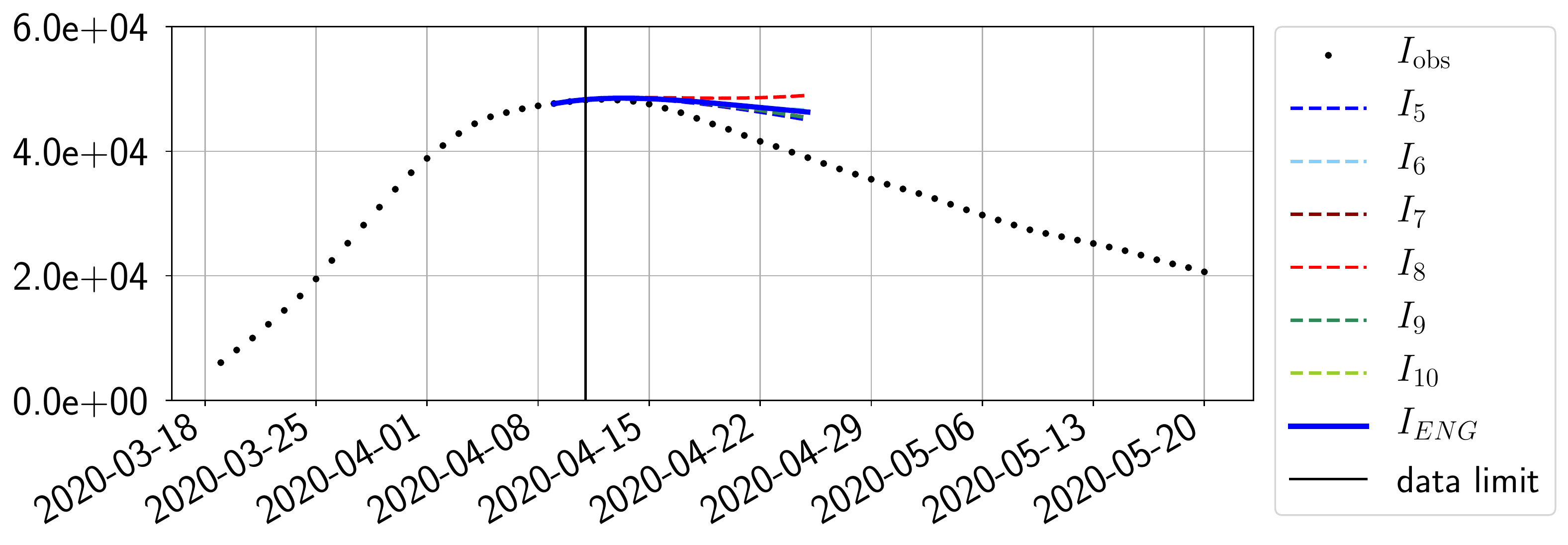}
\caption{Infected}
\end{subfigure}
\begin{subfigure}{.45\textwidth}
\includegraphics[width=1\textwidth]{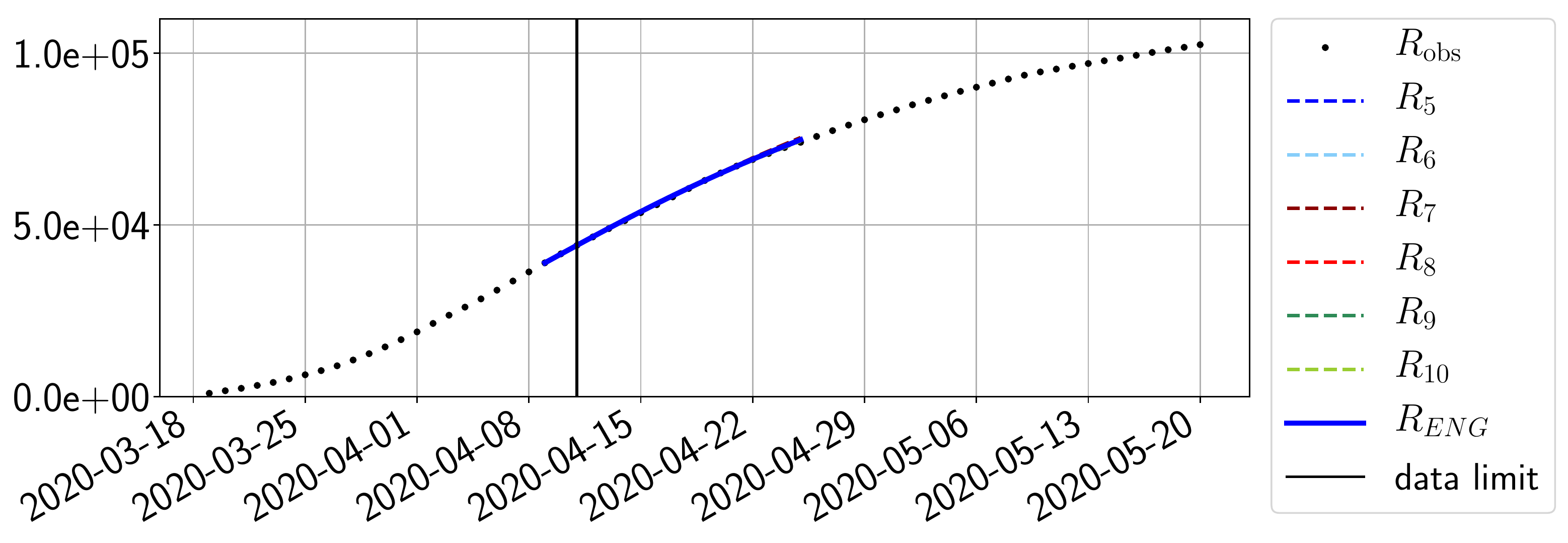}
\caption{Removed}
\end{subfigure}
\caption{ENG forecast from $T=11/04$}
\label{fig:forecast_modes_1104}
\end{figure}

\begin{figure}[H]
\centering
\begin{subfigure}{.45\textwidth}
\includegraphics[width=1\textwidth]{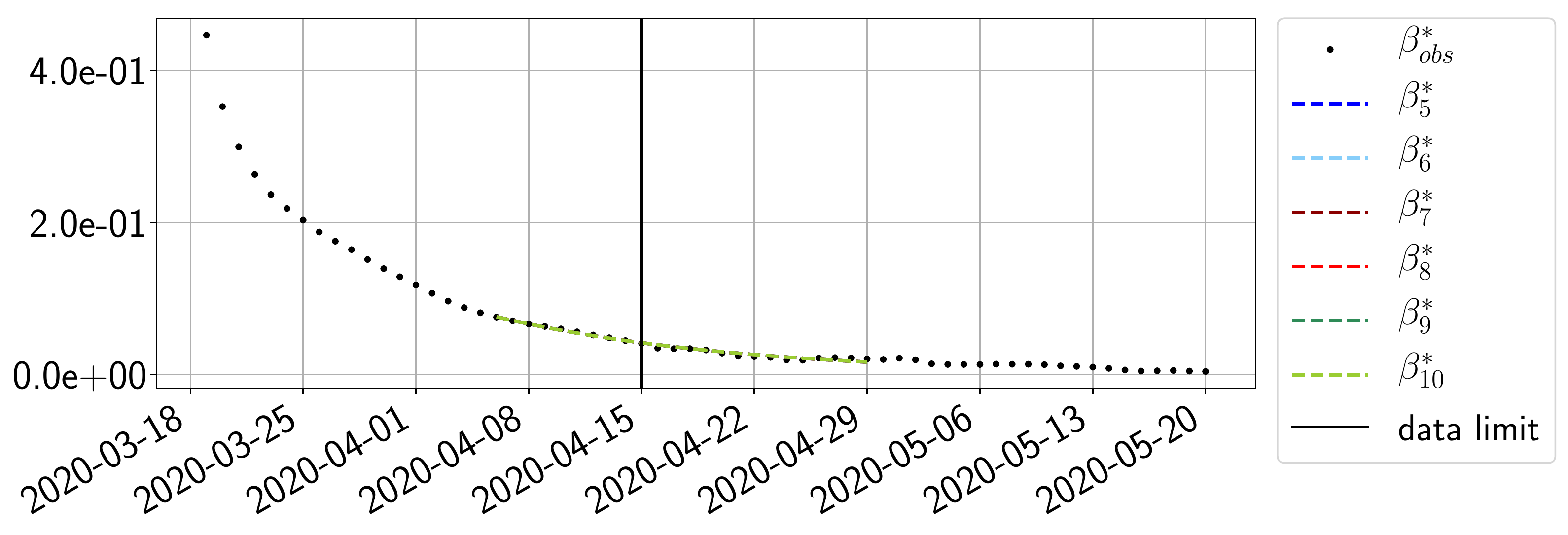}
\caption{$\beta$}
\end{subfigure}
\begin{subfigure}{.45\textwidth}
\includegraphics[width=1\textwidth]{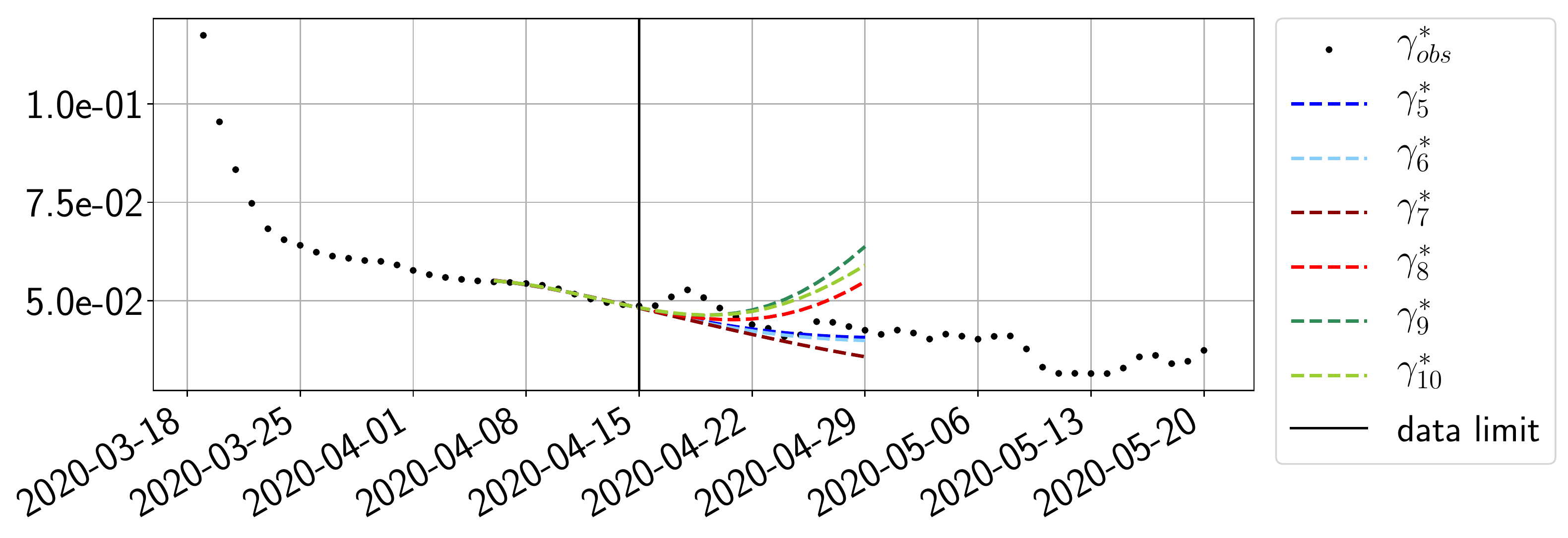}
\caption{$\gamma$}
\end{subfigure}

\vspace{0.4cm}

\begin{subfigure}{.45\textwidth}
\includegraphics[width=1\textwidth]{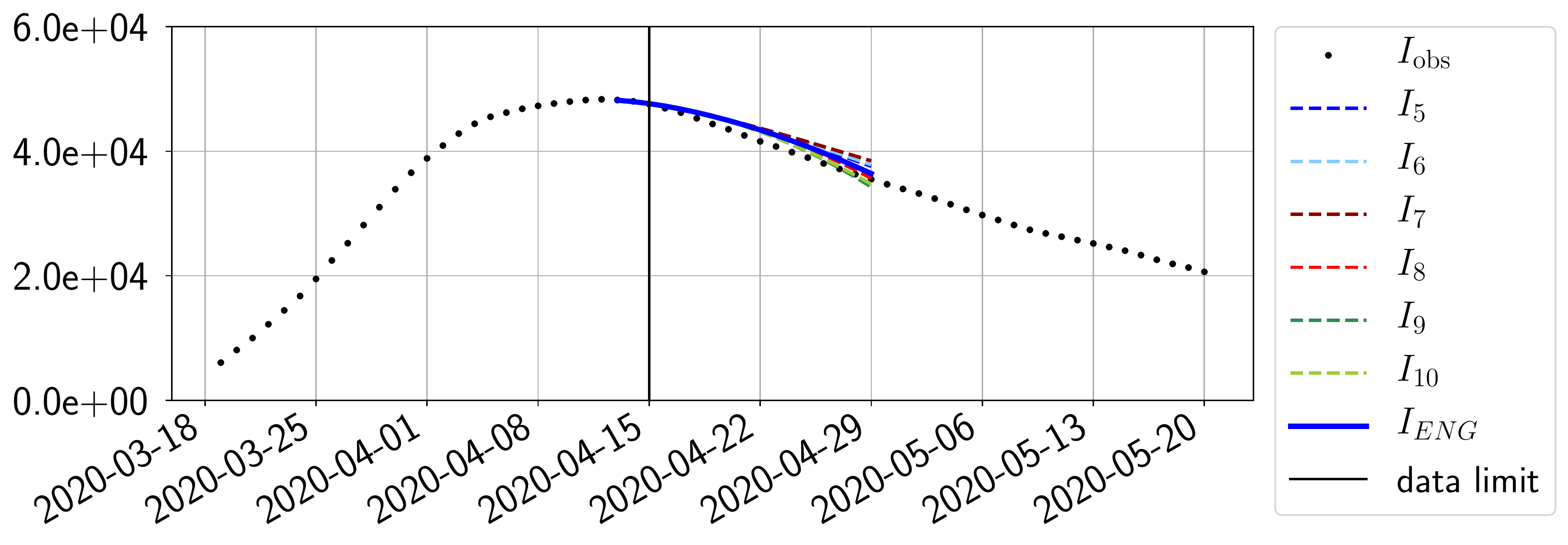}
\caption{Infected}
\end{subfigure}
\begin{subfigure}{.45\textwidth}
\includegraphics[width=1\textwidth]{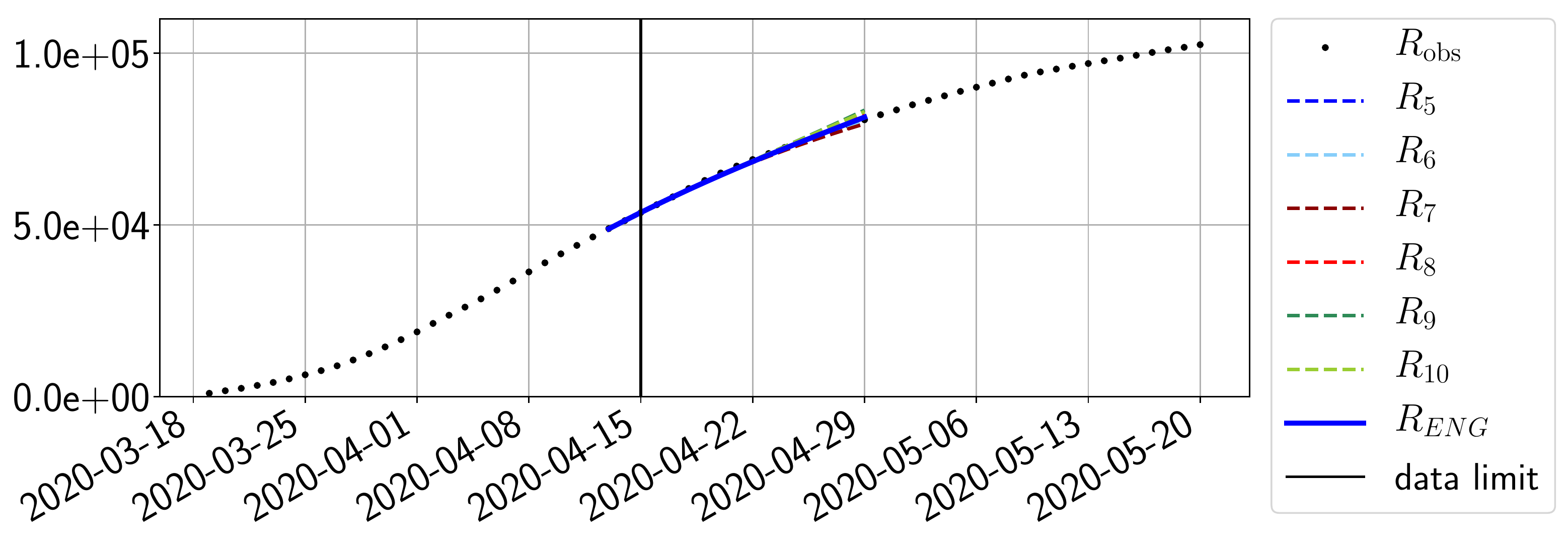}
\caption{Removed}
\end{subfigure}
\caption{ENG forecast from $T=15/04$}
\label{fig:forecast_modes_1504}
\end{figure}

\begin{figure}[H]
\centering
\begin{subfigure}{.45\textwidth}
\includegraphics[width=1\textwidth]{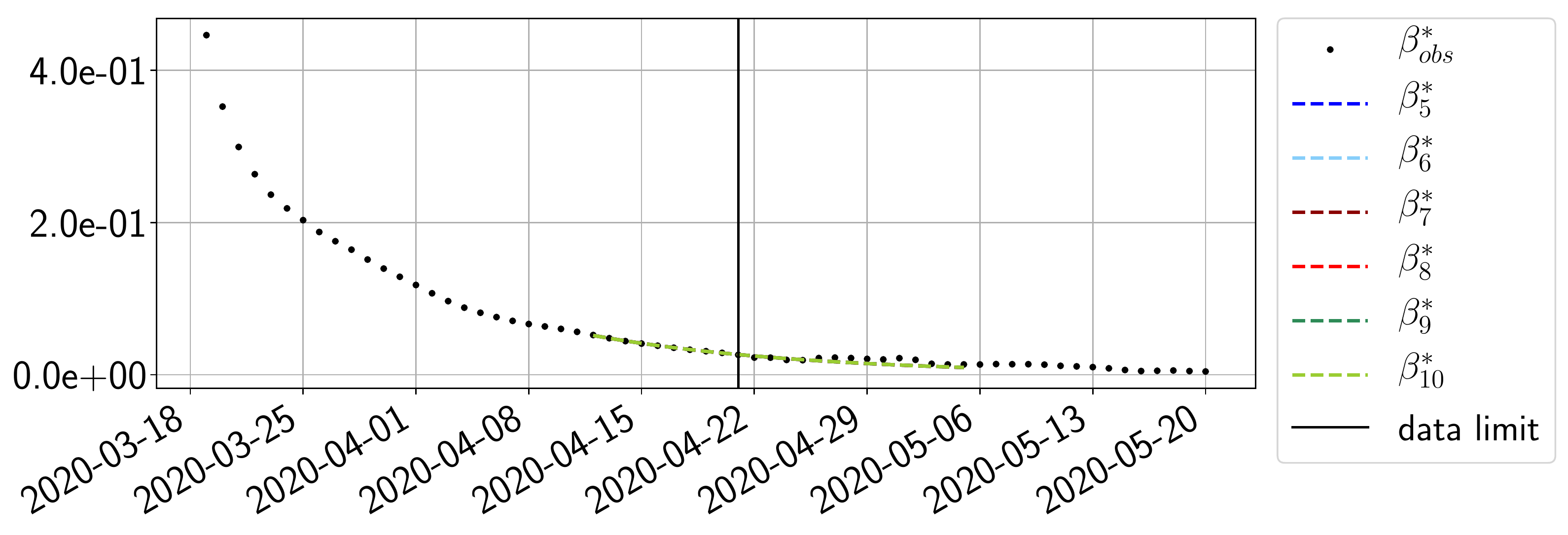}
\caption{$\beta$}
\end{subfigure}
\begin{subfigure}{.45\textwidth}
\includegraphics[width=1\textwidth]{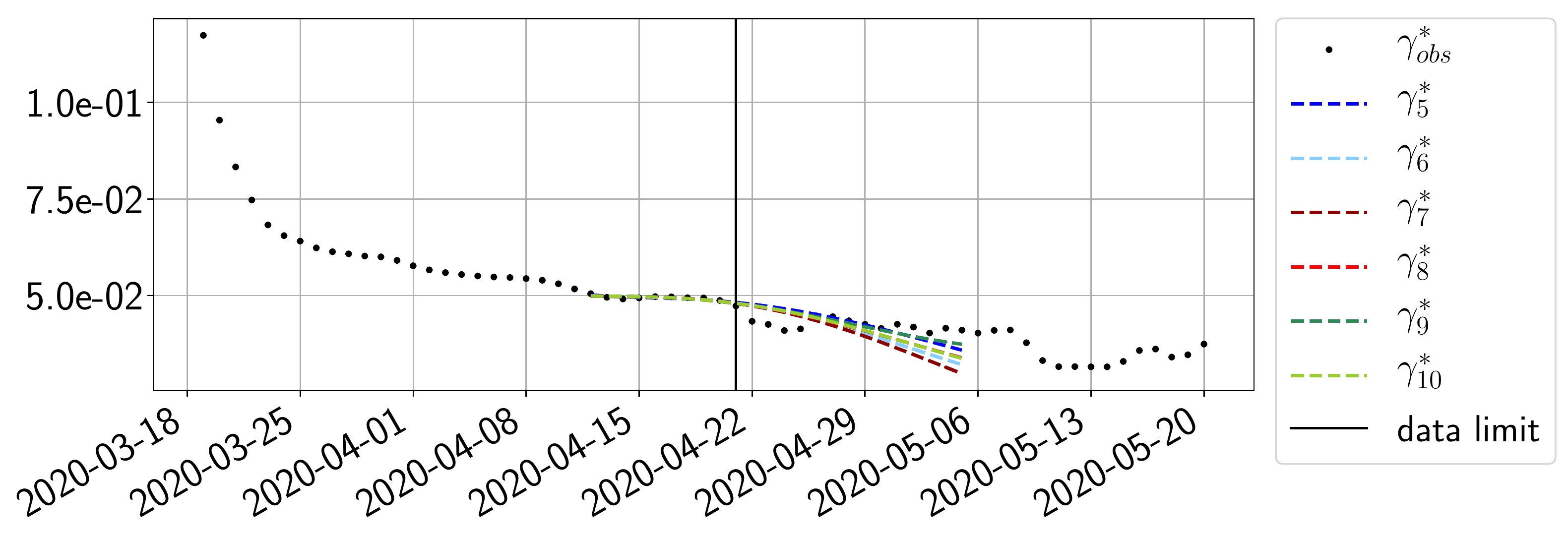}
\caption{$\gamma$}
\end{subfigure}

\vspace{0.4cm}

\begin{subfigure}{.45\textwidth}
\includegraphics[width=1\textwidth]{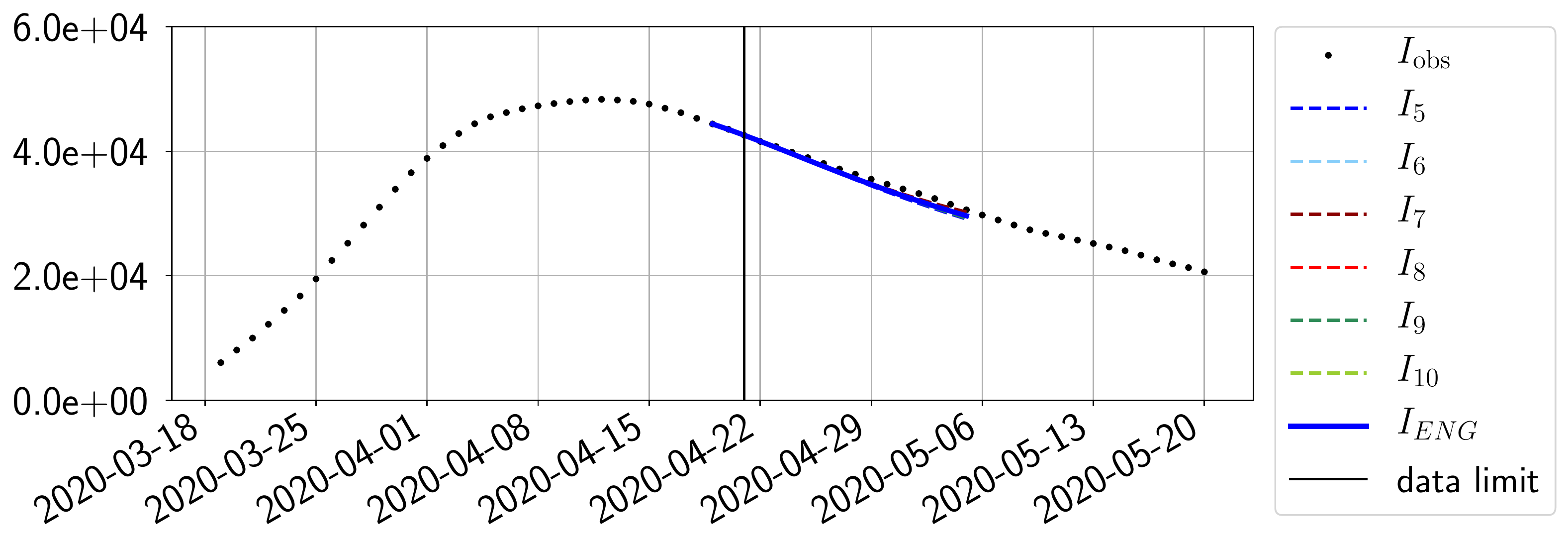}
\caption{Infected}
\end{subfigure}
\begin{subfigure}{.45\textwidth}
\includegraphics[width=1\textwidth]{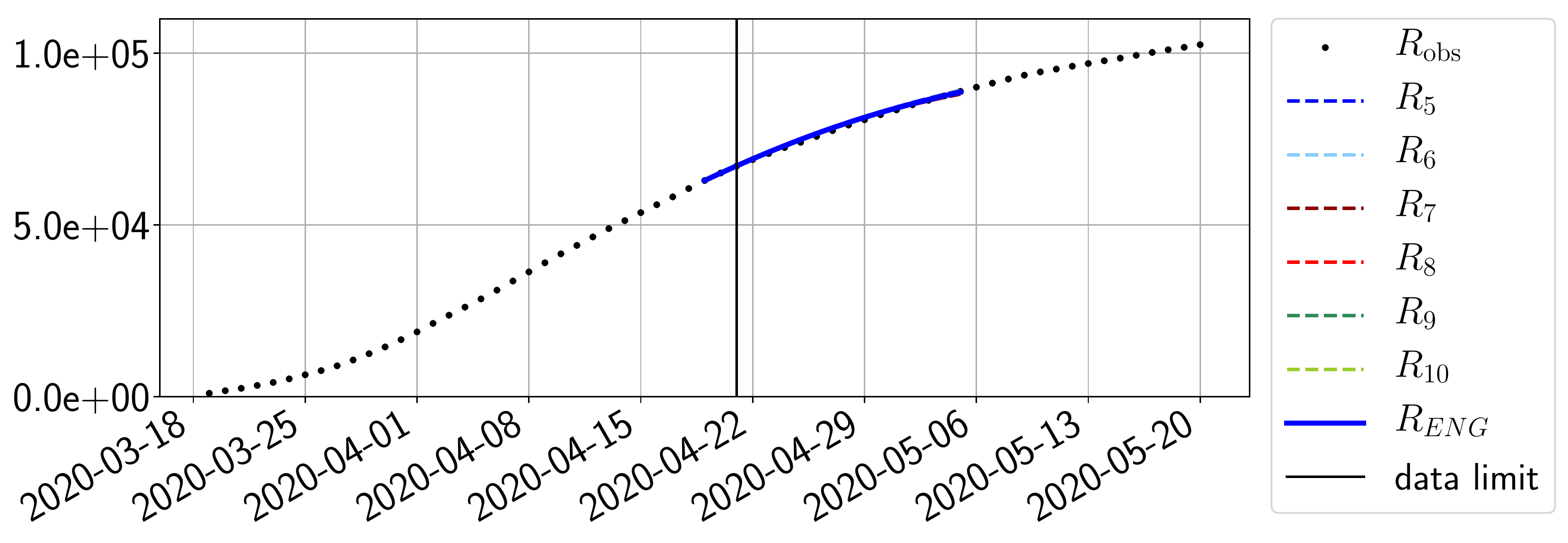}
\caption{Removed}
\end{subfigure}
\caption{ENG forecast from $T=21/04$}
\label{fig:forecast_modes_2104}
\end{figure}

\begin{figure}[H]
\centering
\begin{subfigure}{.45\textwidth}
\includegraphics[width=1\textwidth]{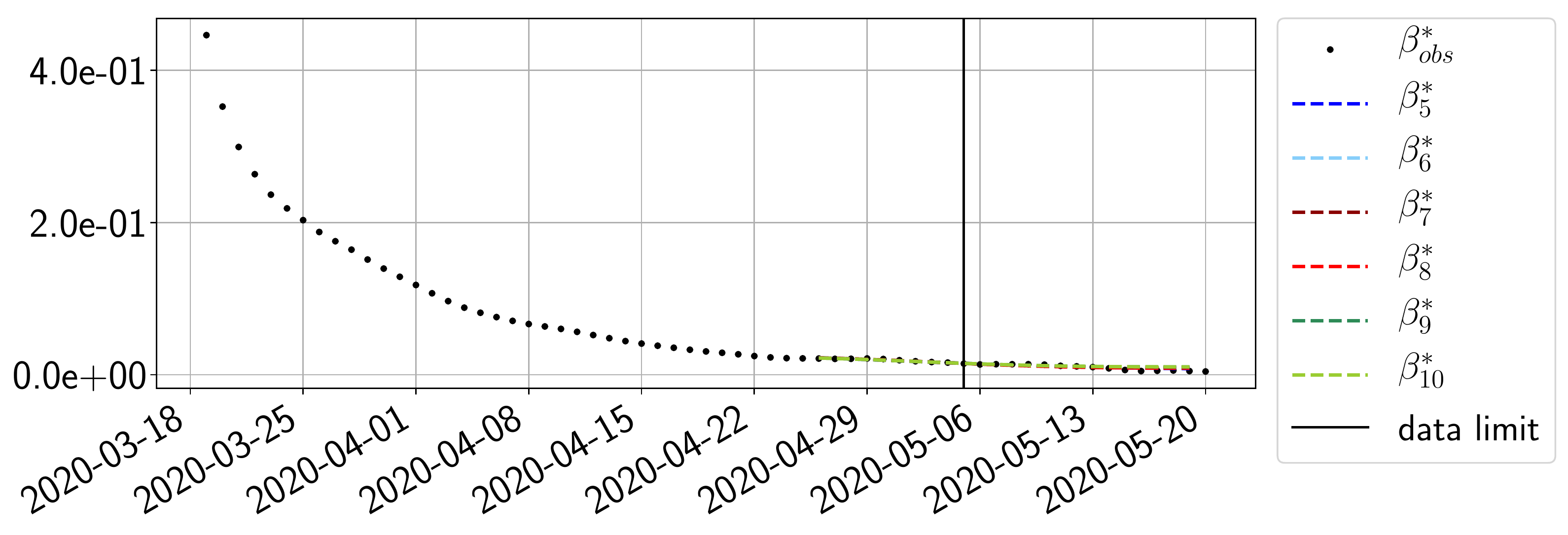}
\caption{$\beta$}
\end{subfigure}
\begin{subfigure}{.45\textwidth}
\includegraphics[width=1\textwidth]{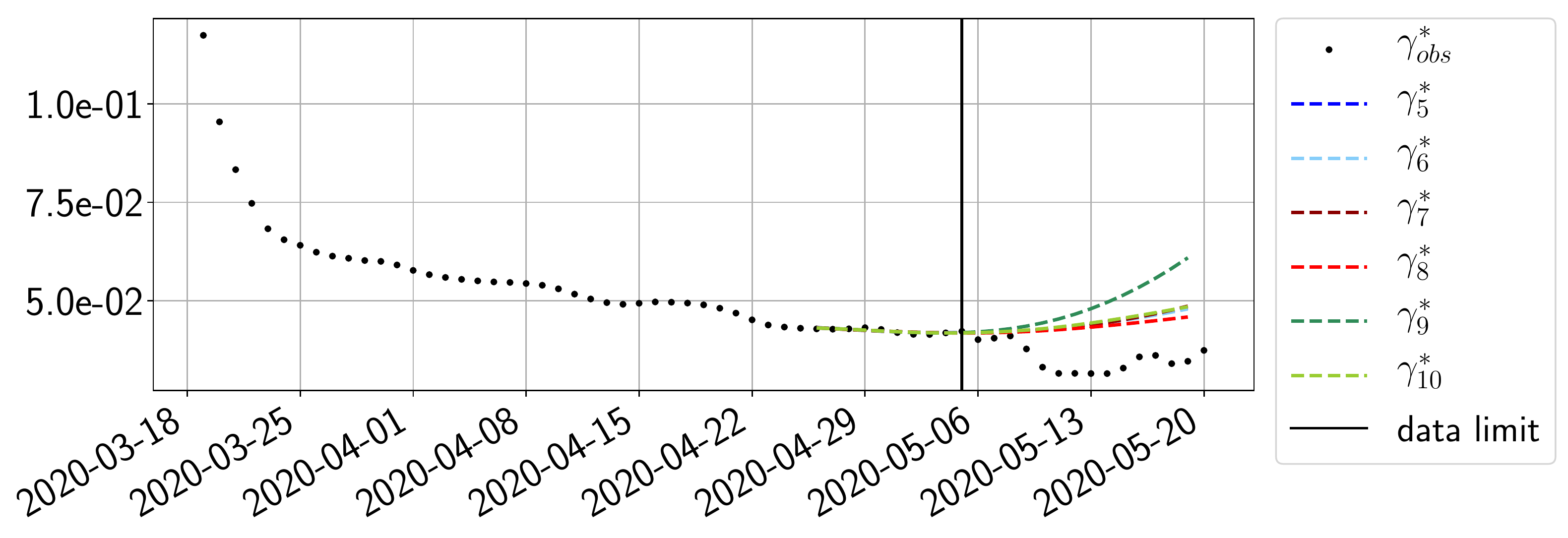}
\caption{$\gamma$}
\end{subfigure}

\vspace{0.4cm}

\begin{subfigure}{.45\textwidth}
\includegraphics[width=1\textwidth]{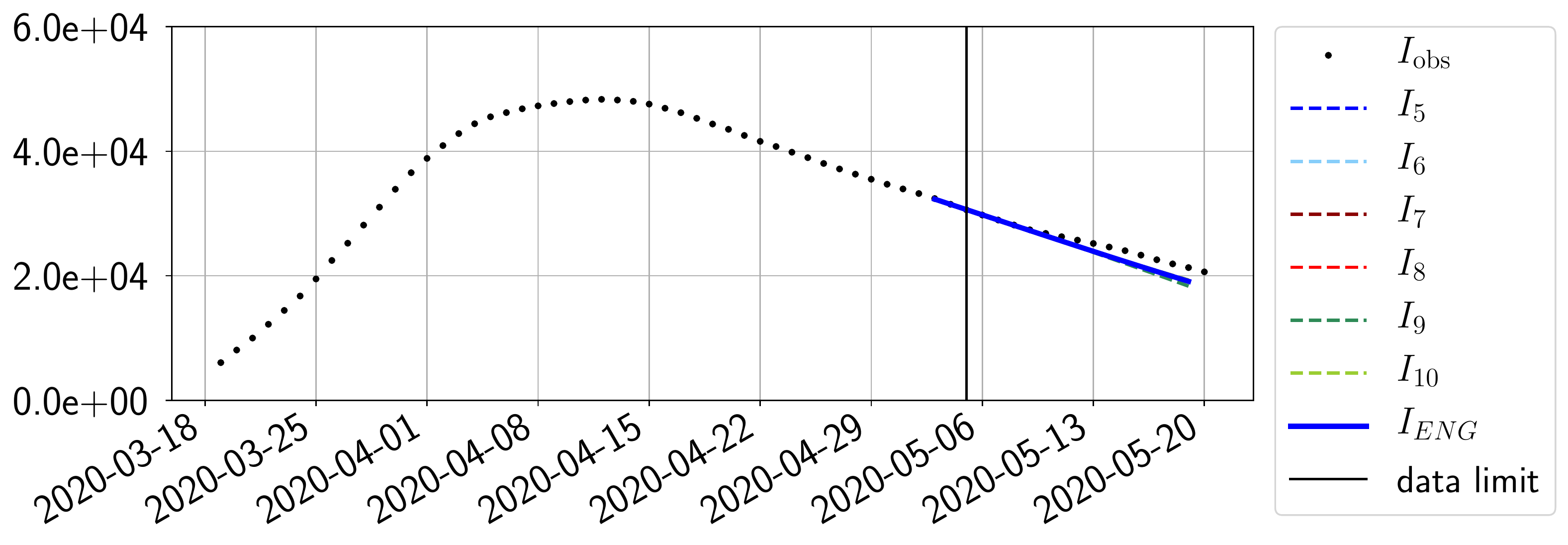}
\caption{Infected}
\end{subfigure}
\begin{subfigure}{.45\textwidth}
\includegraphics[width=1\textwidth]{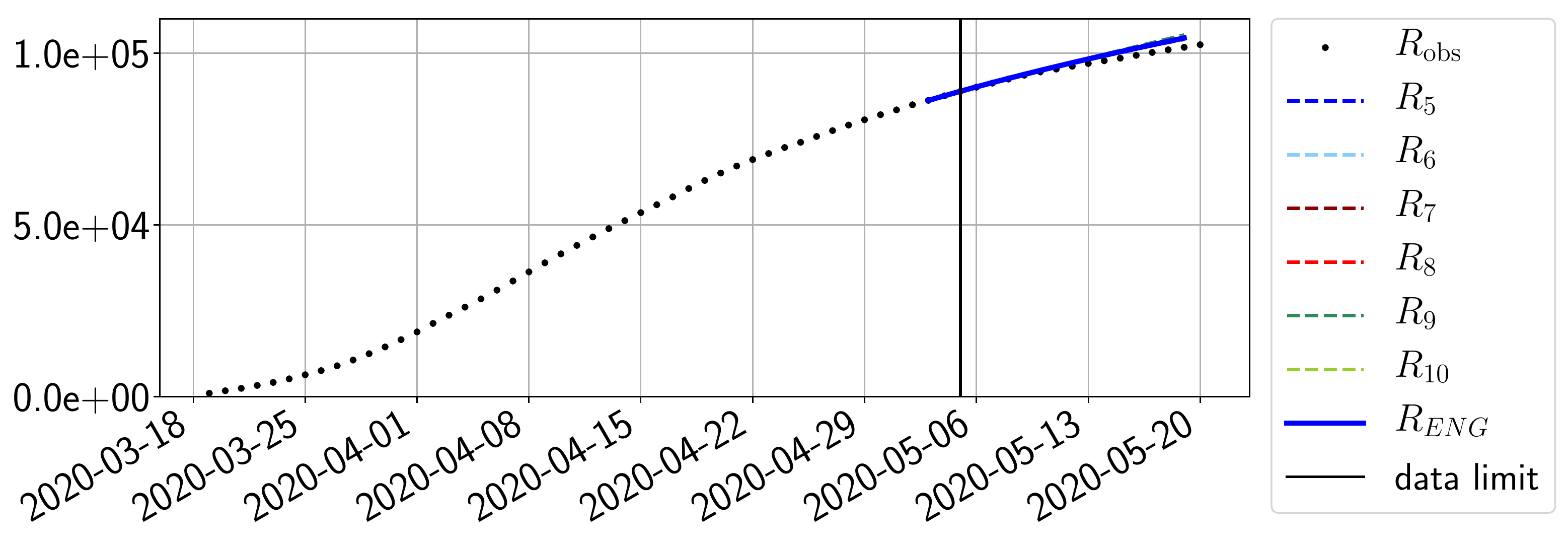}
\caption{Removed}
\end{subfigure}
\caption{ENG forecast from $T=05/05$}
\label{fig:forecast_modes_0505}
\end{figure}

\col{
\subsubsection{Forecasting of the second wave with ENG}
\label{sec:forecast-wave-2}
The review of the present paper took place during the month of November 2020 while the second COVID-19 pandemic wave hit France. We have taken the chance to enlarge the body of numerical results and we provide some example forecasts with ENG on this wave in figures \ref{fig:forecast_modes_2810} to \ref{fig:forecast_modes_0911}. As the figures illustrate, the method provides very satisfactory forecasts on a 14 day ahead window. We again observe a satisfactory prediction of the second peak (figures \ref{fig:forecast_modes_2810}, \ref{fig:forecast_modes_0311}  and  \ref{fig:forecast_modes_0911}) and the same difficulty in forecasting $\gamma$ due to the oscillations in $\gamma_\obs$ but this has not greatly impacted on the quality of the forecasts for $I$ and $R$.
}

\begin{figure}[H]
\centering
\begin{subfigure}{.45\textwidth}
\includegraphics[width=1\textwidth]{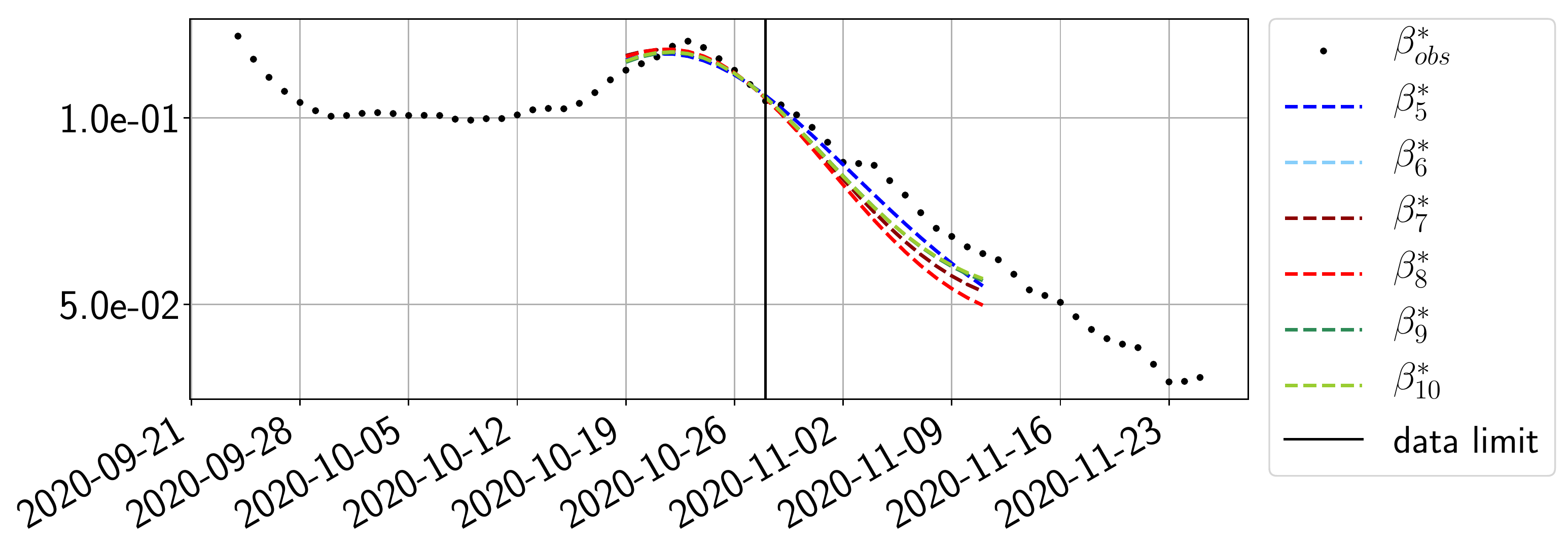}
\caption{$\beta$}
\end{subfigure}
\begin{subfigure}{.45\textwidth}
\includegraphics[width=1\textwidth]{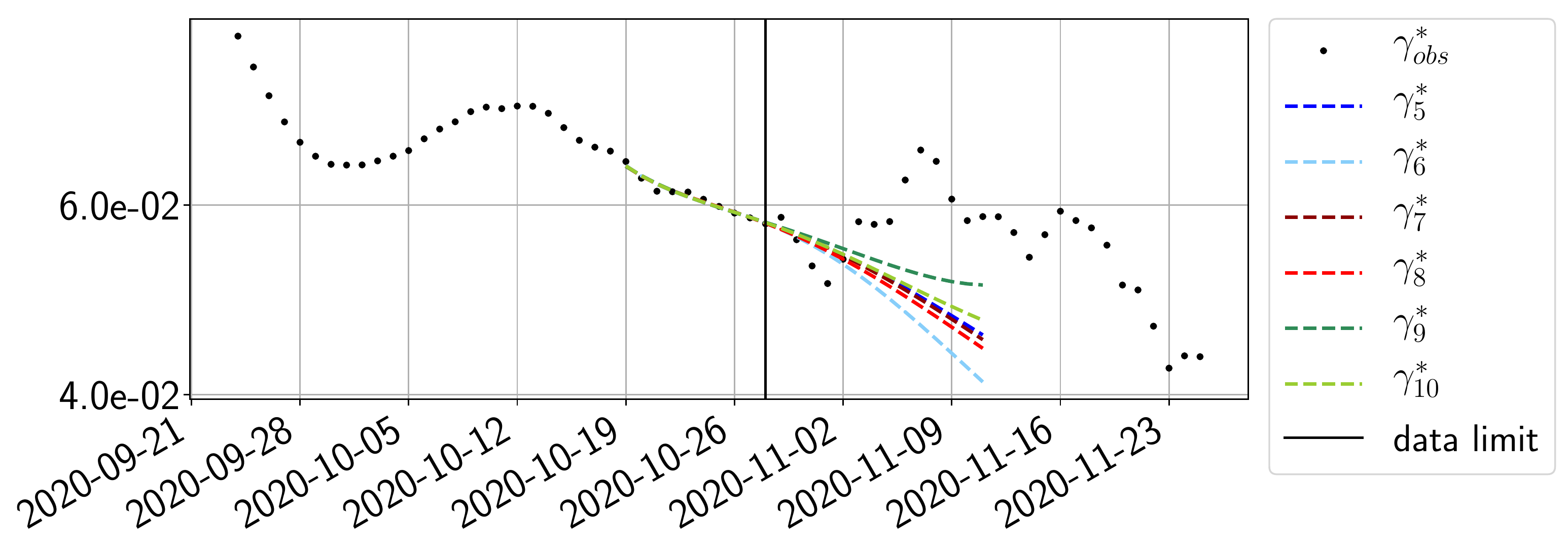}
\caption{$\gamma$}
\end{subfigure}

\vspace{0.4cm}

\begin{subfigure}{.45\textwidth}
\includegraphics[width=1\textwidth]{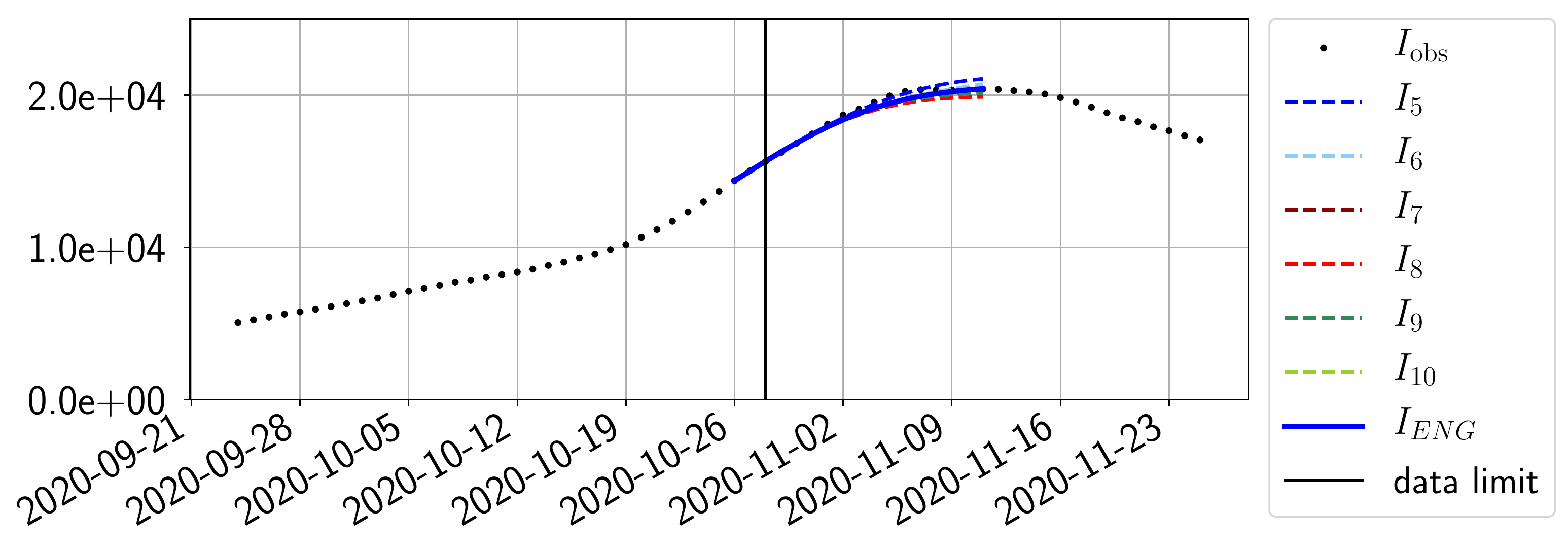}
\caption{Infected}
\end{subfigure}
\begin{subfigure}{.45\textwidth}
\includegraphics[width=1\textwidth]{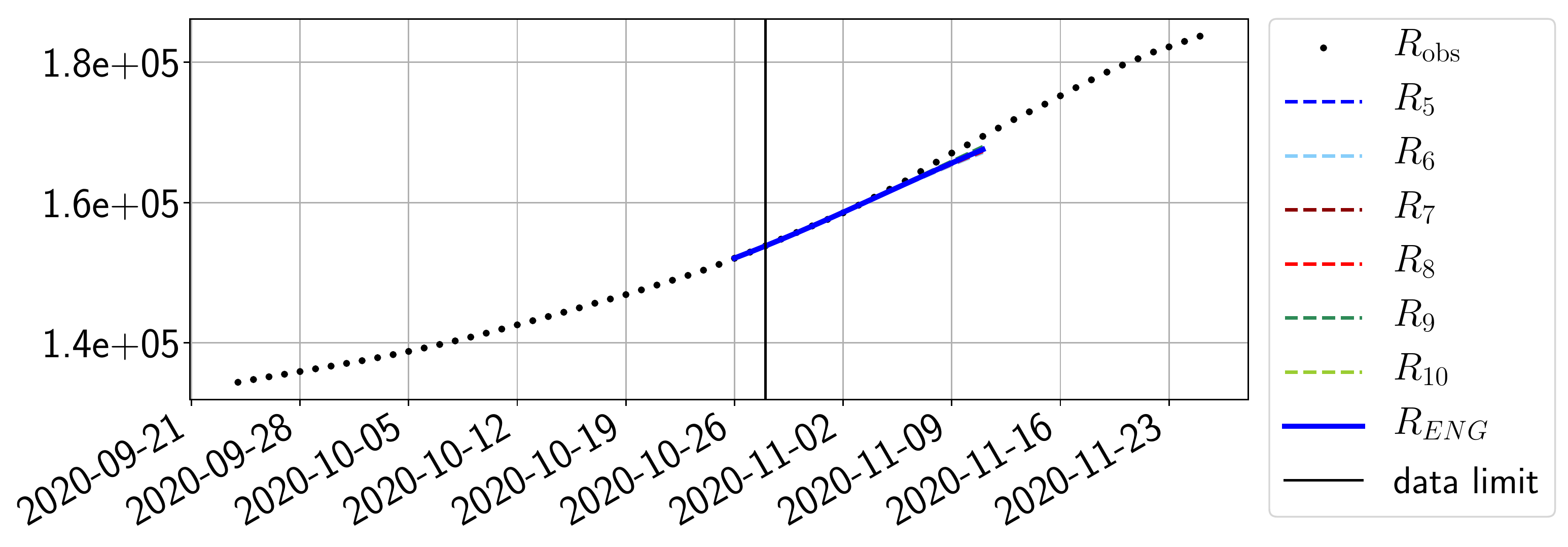}
\caption{Removed}
\end{subfigure}
\caption{ENG forecast from $T=28/10$}
\label{fig:forecast_modes_2810}
\end{figure}

\begin{figure}[H]
\centering
\begin{subfigure}{.45\textwidth}
\includegraphics[width=1\textwidth]{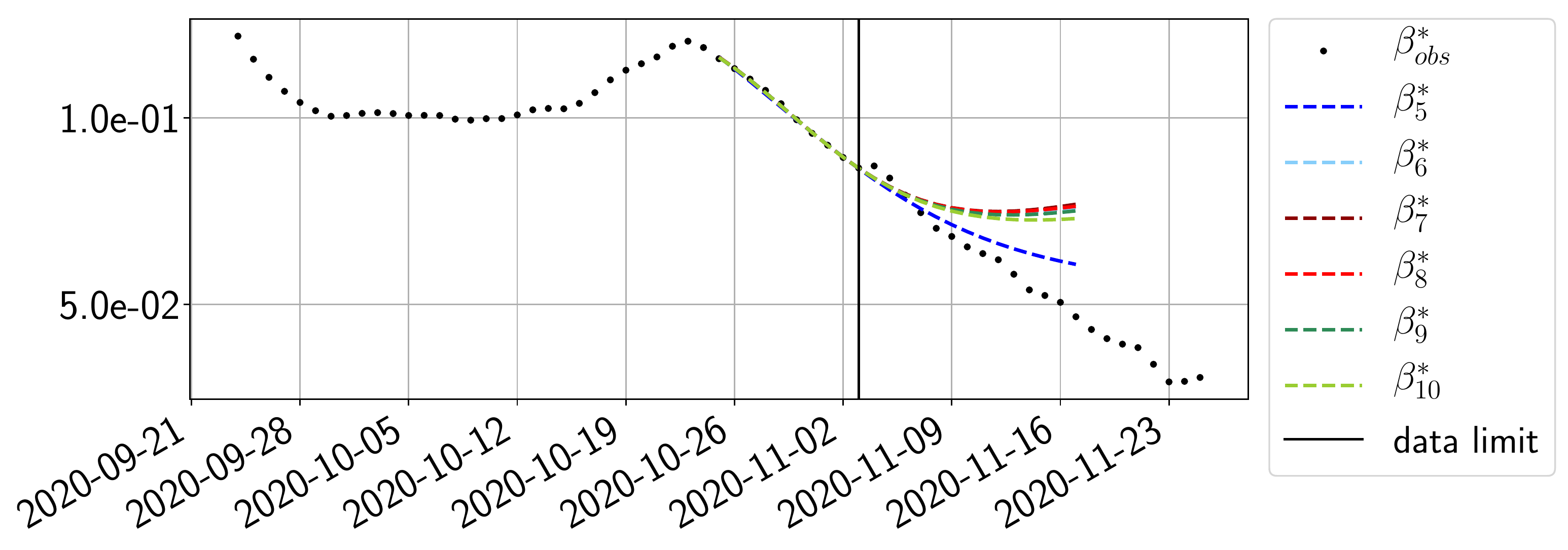}
\caption{$\beta$}
\end{subfigure}
\begin{subfigure}{.45\textwidth}
\includegraphics[width=1\textwidth]{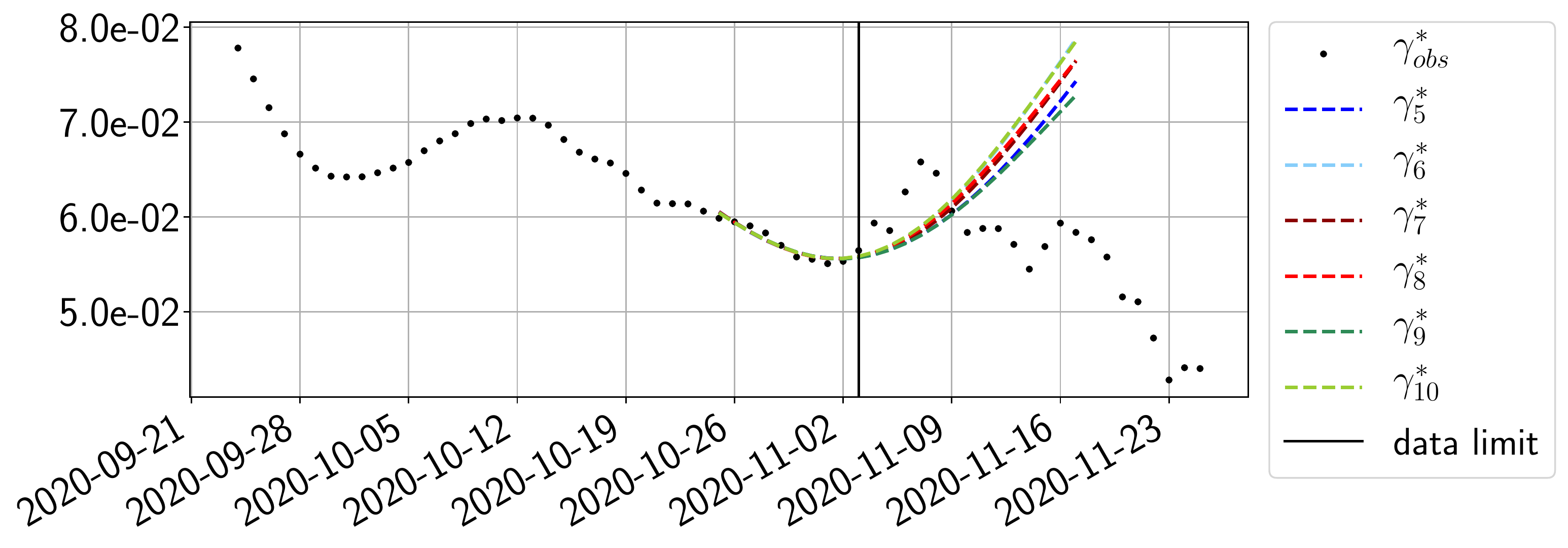}
\caption{$\gamma$}
\end{subfigure}

\vspace{0.4cm}

\begin{subfigure}{.45\textwidth}
\includegraphics[width=1\textwidth]{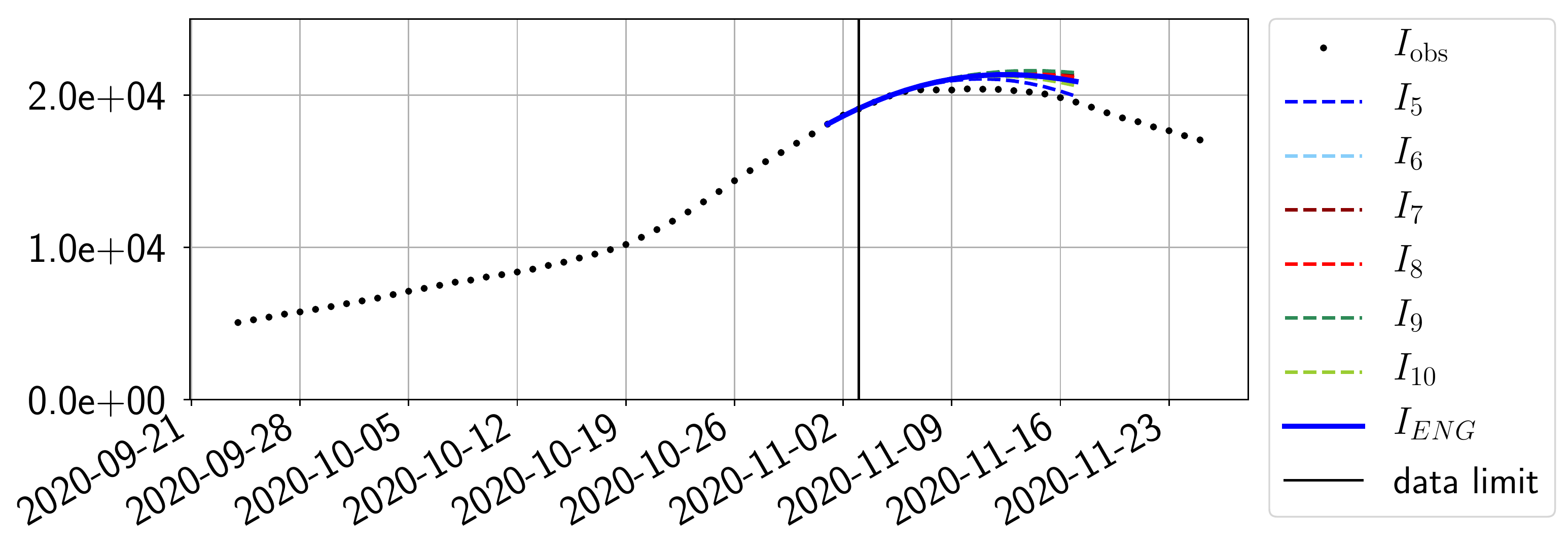}
\caption{Infected}
\end{subfigure}
\begin{subfigure}{.45\textwidth}
\includegraphics[width=1\textwidth]{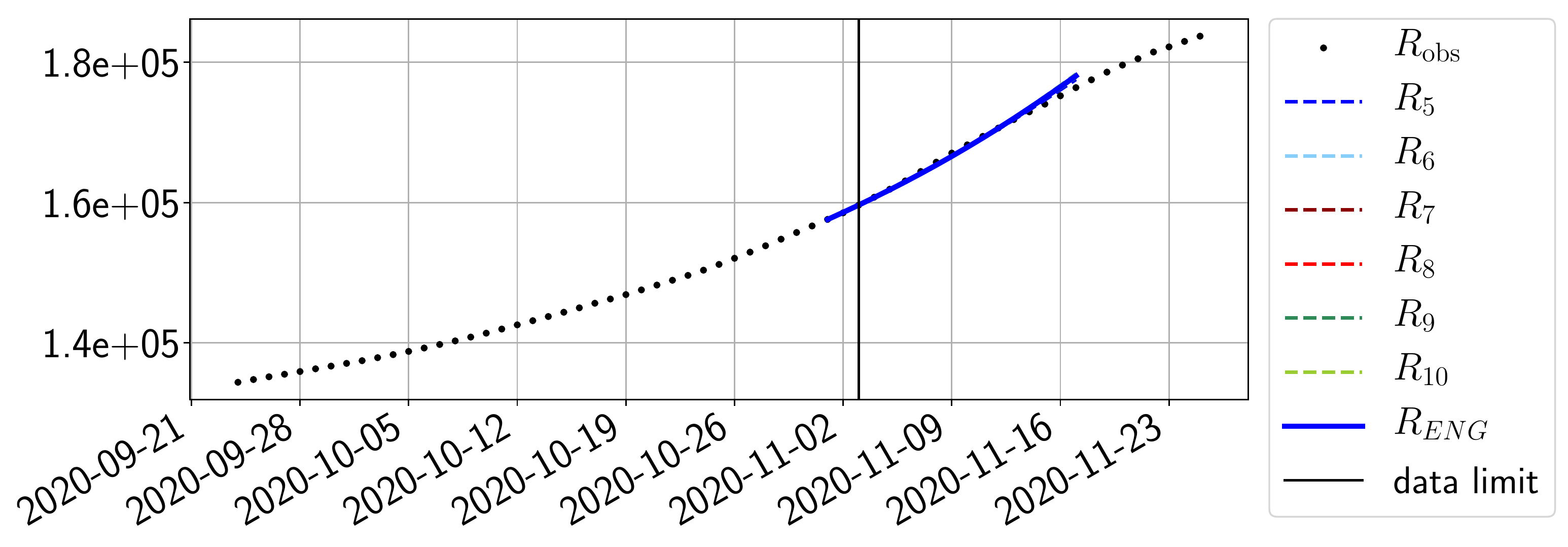}
\caption{Removed}
\end{subfigure}
\caption{ENG forecast from $T=03/11$}
\label{fig:forecast_modes_0311}
\end{figure}

\begin{figure}[H]
\centering
\begin{subfigure}{.45\textwidth}
\includegraphics[width=1\textwidth]{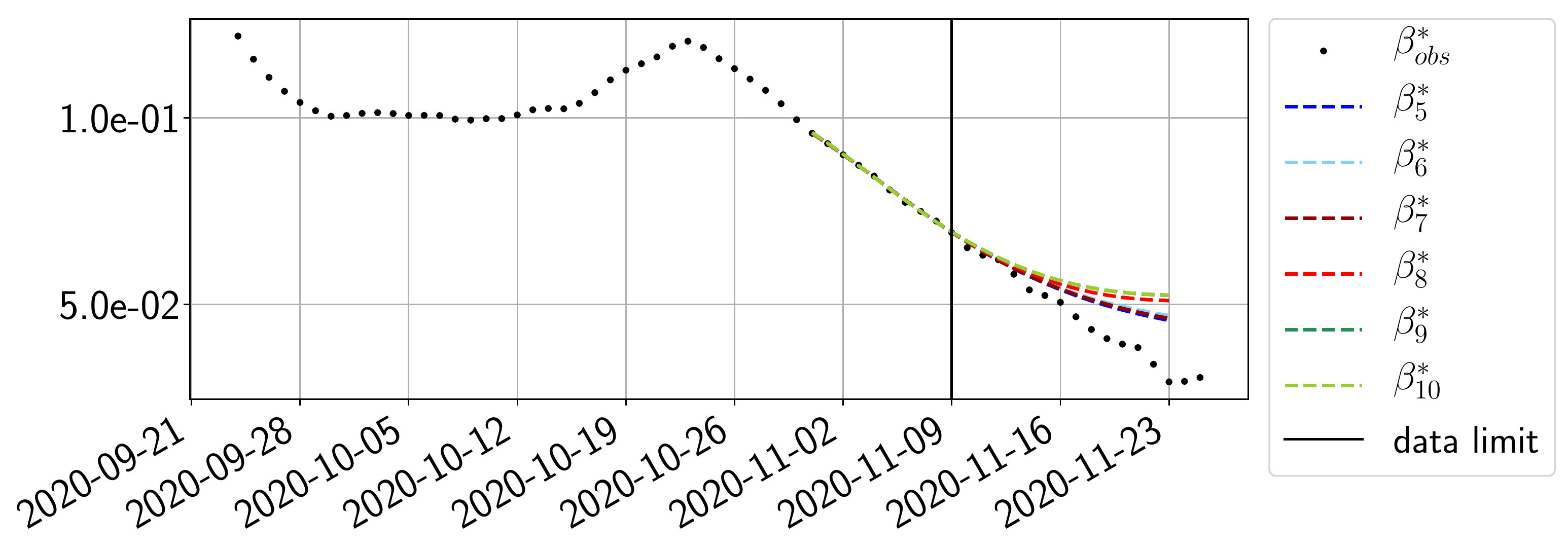}
\caption{$\beta$}
\end{subfigure}
\begin{subfigure}{.45\textwidth}
\includegraphics[width=1\textwidth]{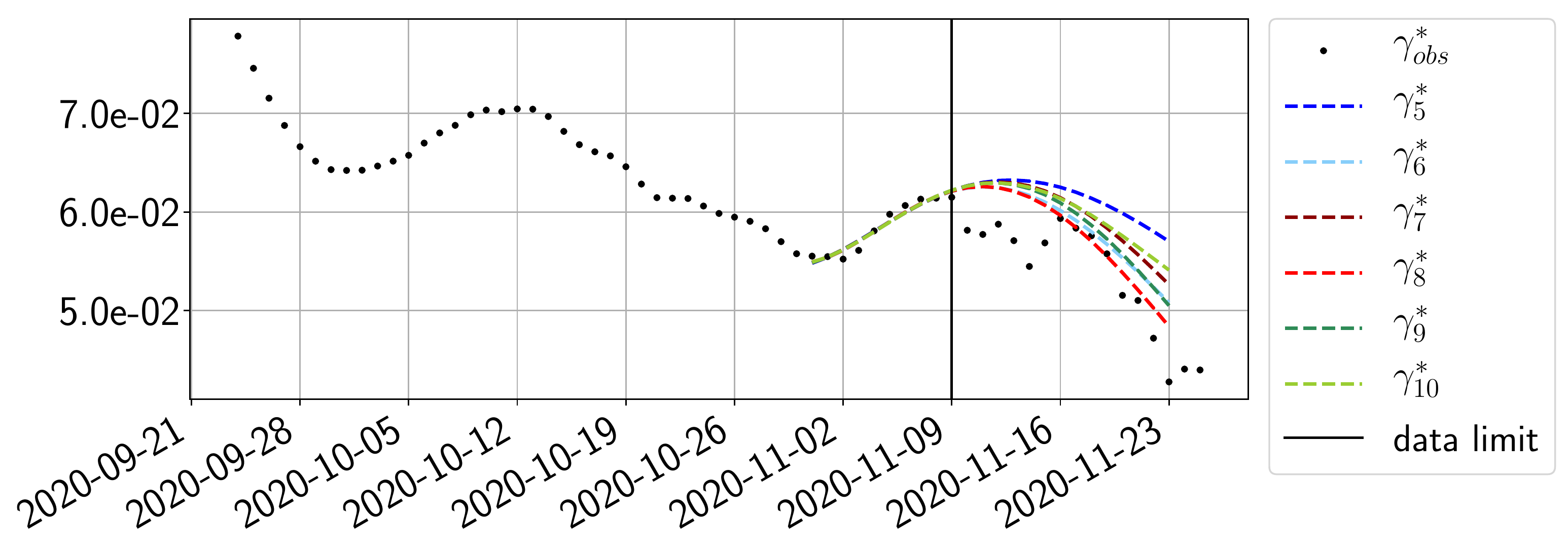}
\caption{$\gamma$}
\end{subfigure}

\vspace{0.4cm}

\begin{subfigure}{.45\textwidth}
\includegraphics[width=1\textwidth]{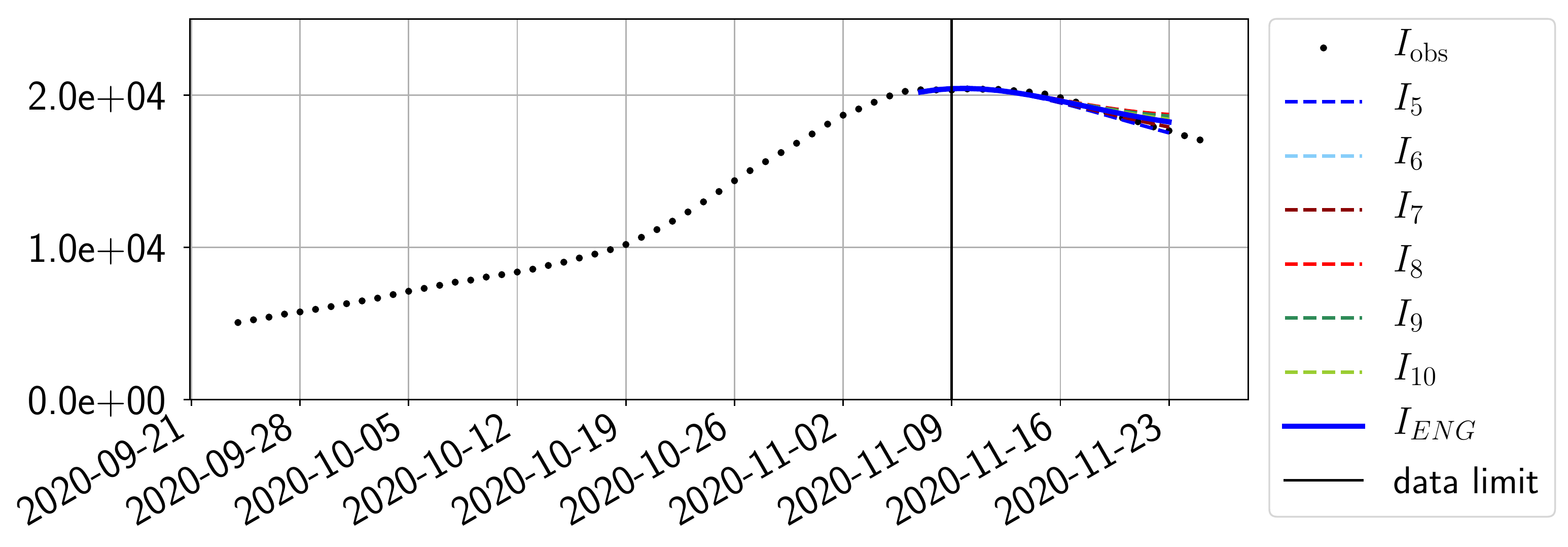}
\caption{Infected}
\end{subfigure}
\begin{subfigure}{.45\textwidth}
\includegraphics[width=1\textwidth]{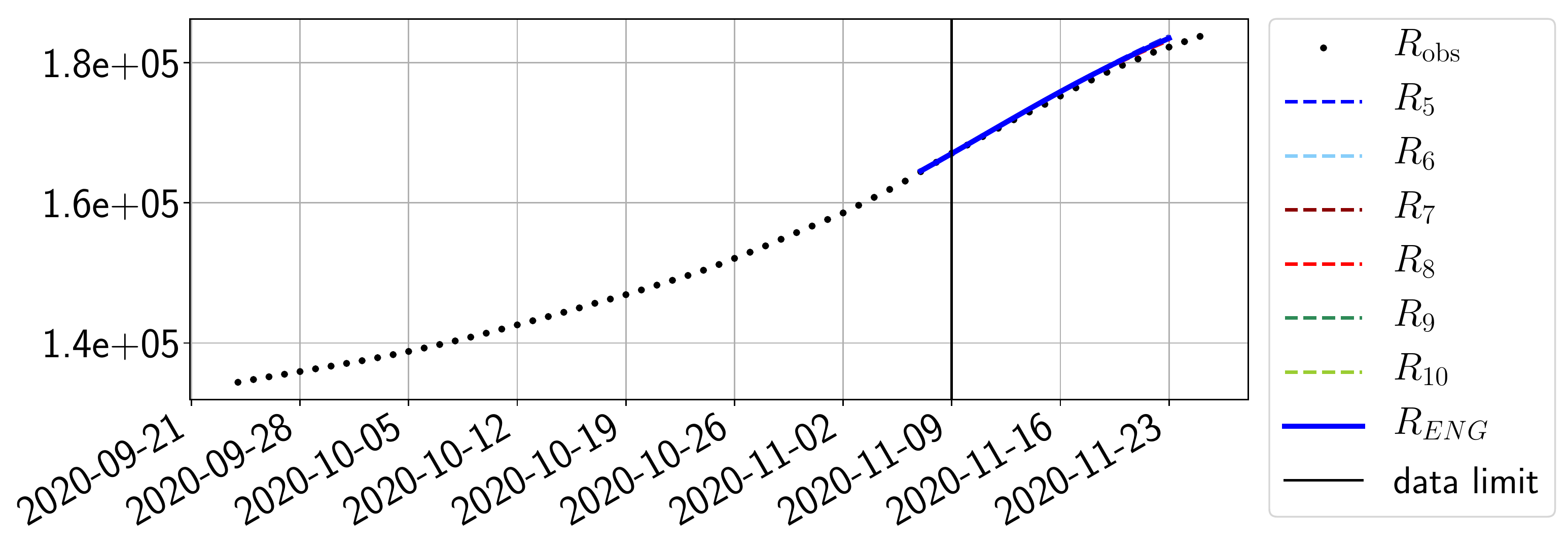}
\caption{Removed}
\end{subfigure}
\caption{ENG forecast from $T=09/11$}
\label{fig:forecast_modes_0911}
\end{figure}

\col{
\subsubsection{Forecasts with ENG on a 28-day-ahead window}
\label{sec:forecast-28-days}
To conclude this results section, we extend the forecasting window to 28 days instead of 14 and study whether the introduced ENG method still provides satisfactory \col{forecasts}. As shown in figures \ref{fig:forecast_28_days_modes_0104} to \ref{fig:forecast_28_days_modes_2810}), the results of the methods are quite stable for large windows. This shows that in contrast to standard extrapolation methods using classical linear or affine regressions, the reduced basis catches the dynamics of $\bfbeta$ and $\bfgamma$ not only locally but also at extended time intervals.
} 

\begin{figure}[H]
\centering
\begin{subfigure}{.45\textwidth}
\includegraphics[width=1\textwidth]{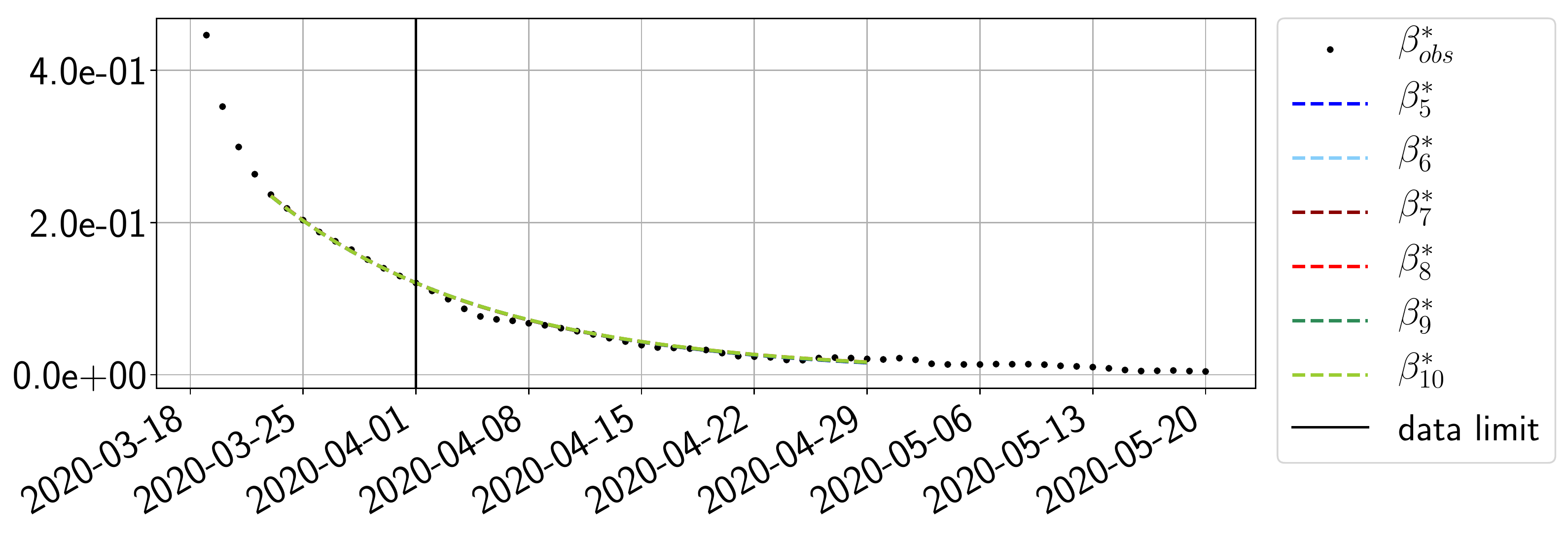}
\caption{$\beta$}
\end{subfigure}
\begin{subfigure}{.45\textwidth}
\includegraphics[width=1\textwidth]{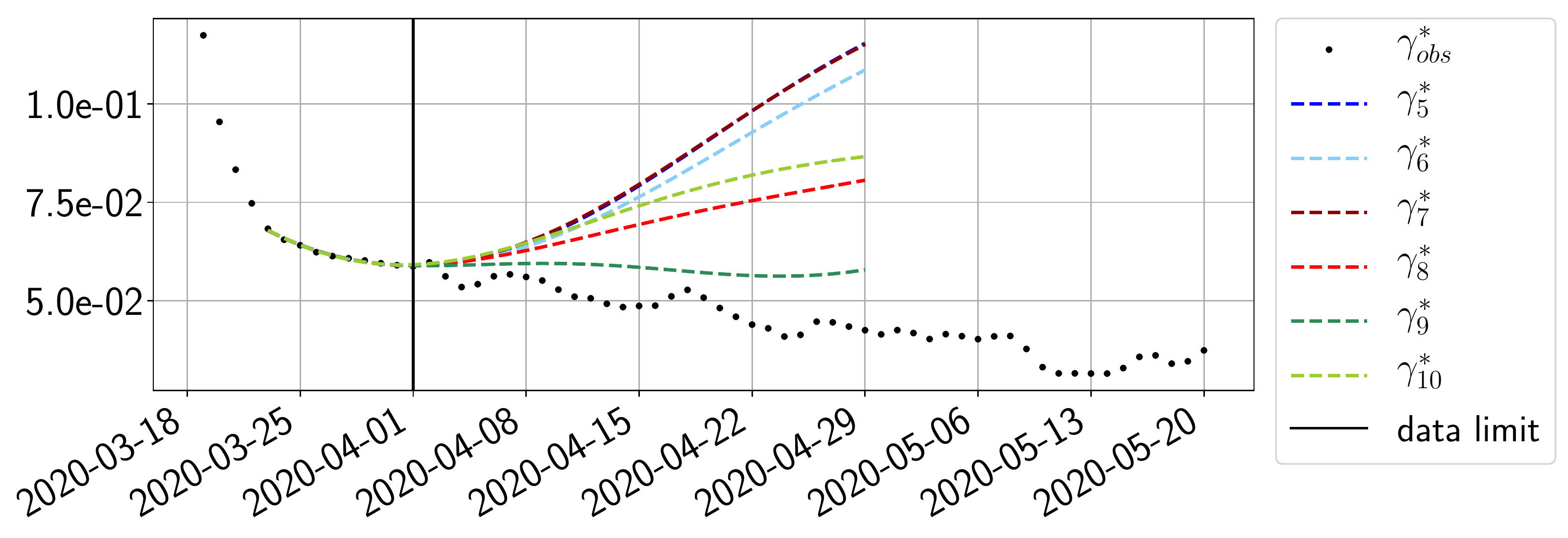}
\caption{$\gamma$}
\end{subfigure}

\vspace{0.4cm}

\begin{subfigure}{.45\textwidth}
\includegraphics[width=1\textwidth]{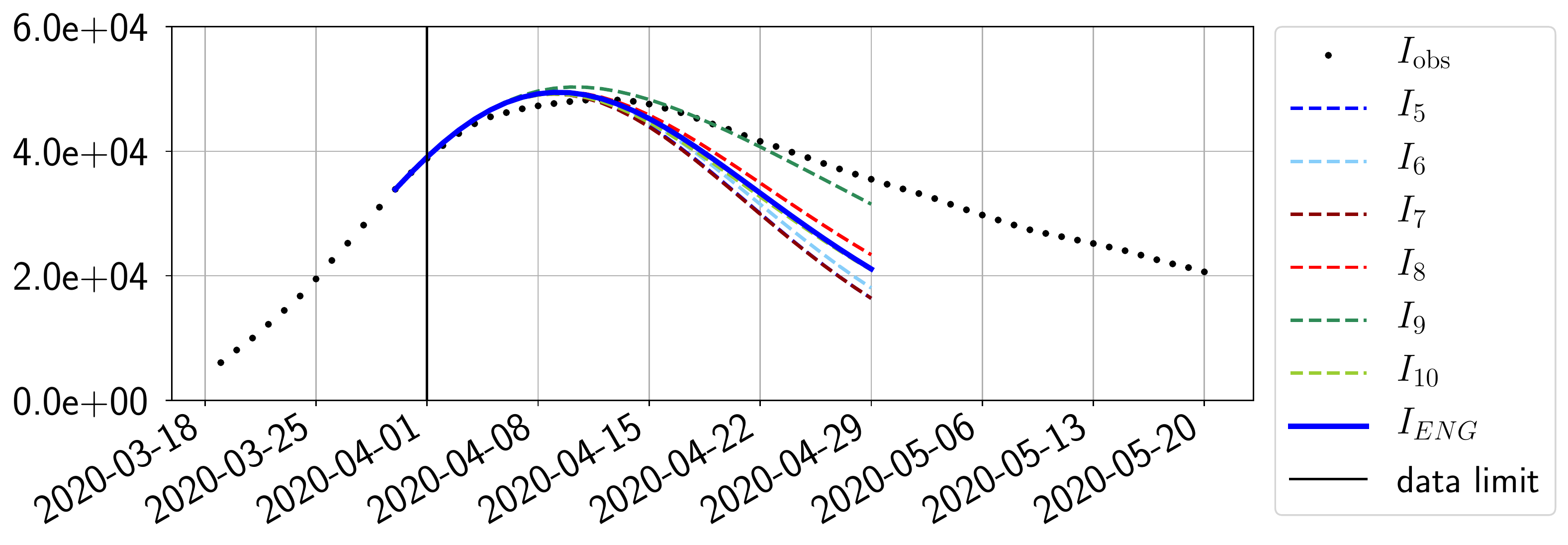}
\caption{Infected}
\end{subfigure}
\begin{subfigure}{.45\textwidth}
\includegraphics[width=1\textwidth]{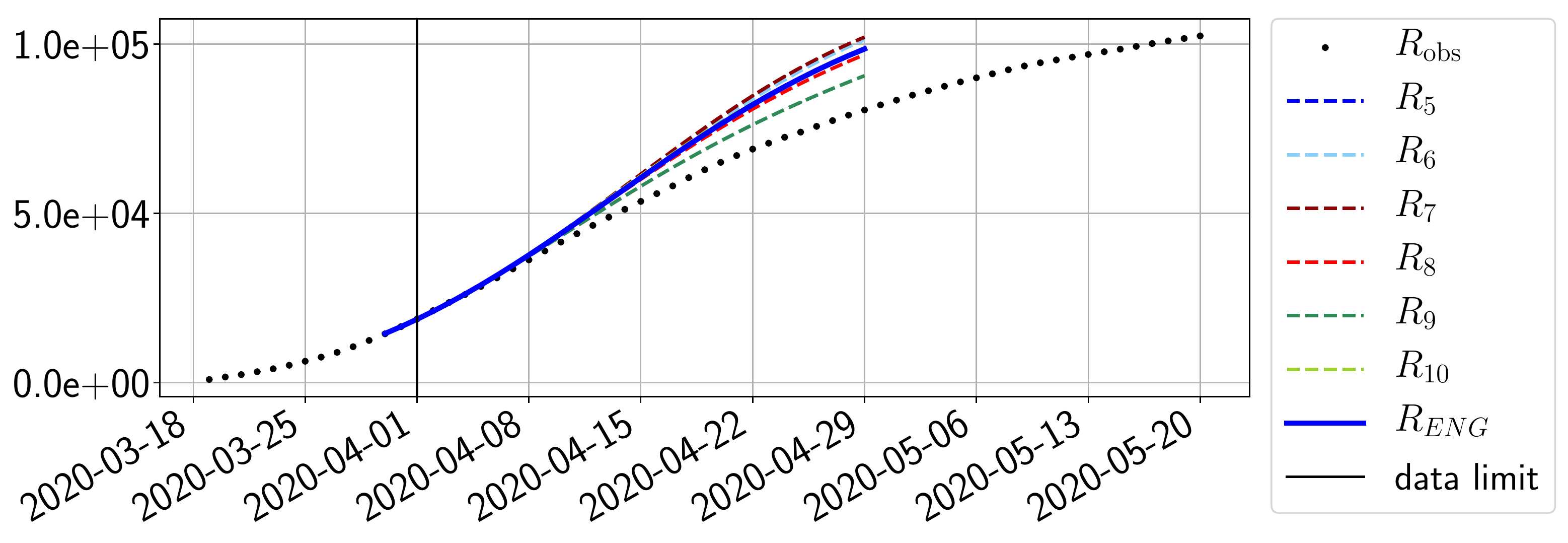}
\caption{Removed}
\end{subfigure}
\caption{ENG forecast from $T=01/04$}
\label{fig:forecast_28_days_modes_0104}
\end{figure}

\begin{figure}[H]
\centering
\begin{subfigure}{.45\textwidth}
\includegraphics[width=1\textwidth]{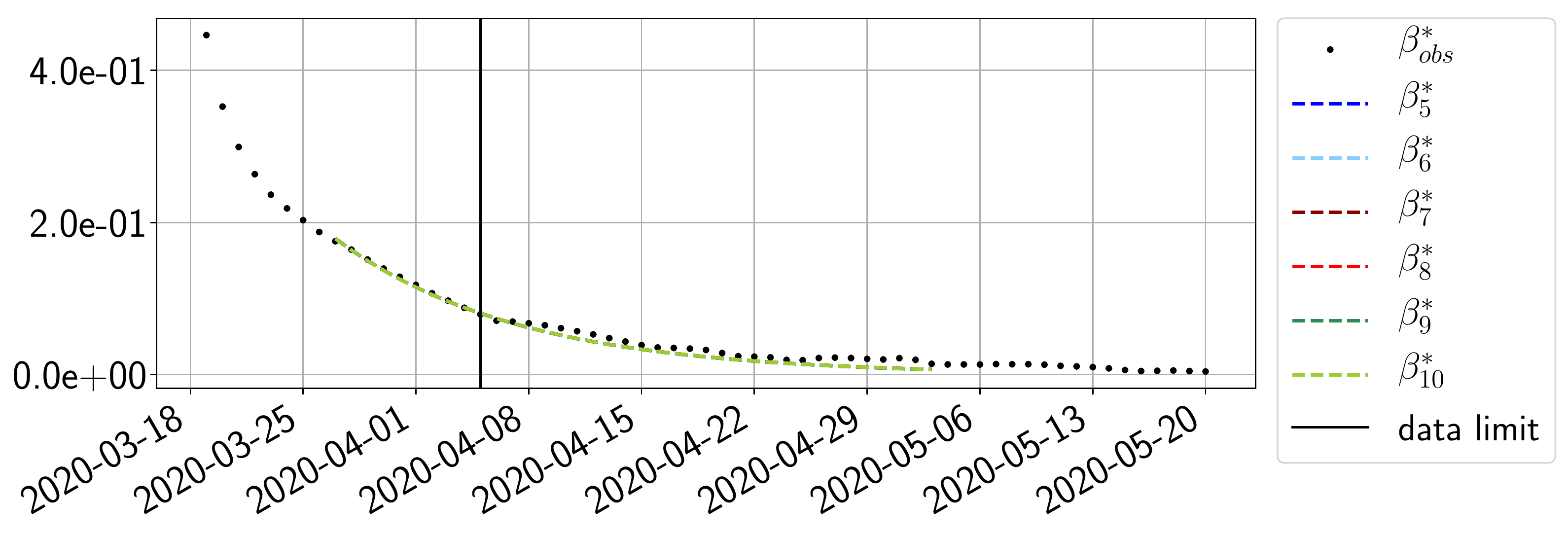}
\caption{$\beta$}
\end{subfigure}
\begin{subfigure}{.45\textwidth}
\includegraphics[width=1\textwidth]{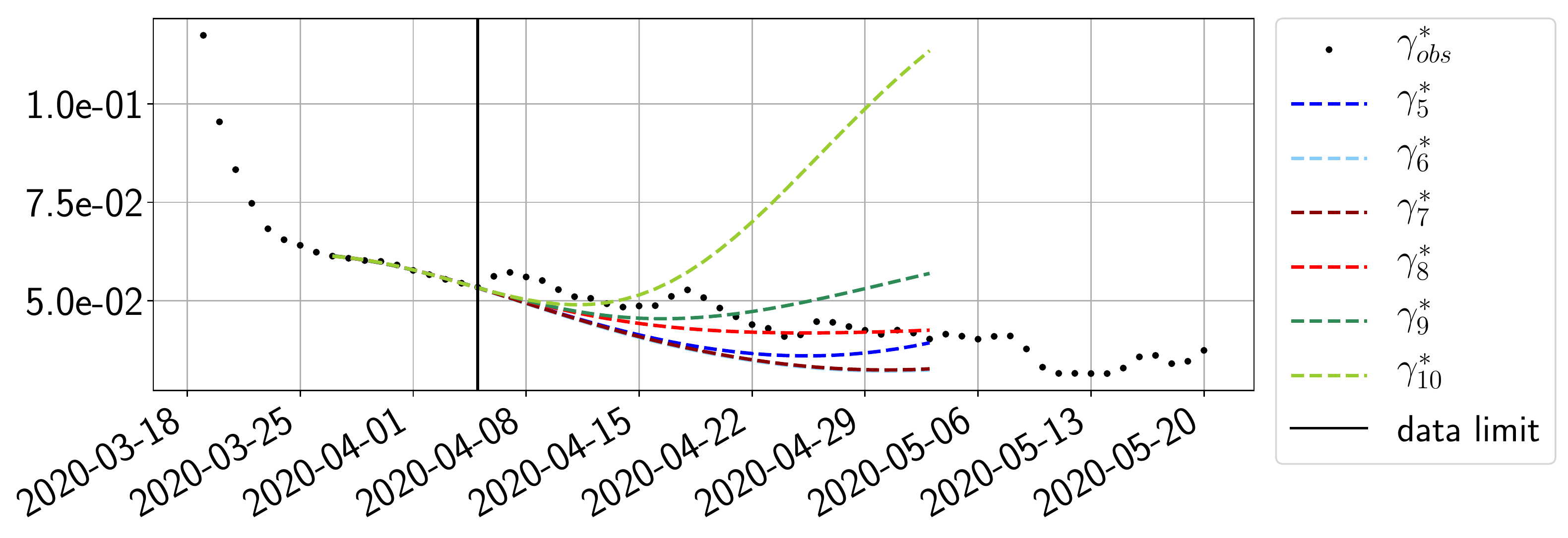}
\caption{$\gamma$}
\end{subfigure}

\vspace{0.4cm}

\begin{subfigure}{.45\textwidth}
\includegraphics[width=1\textwidth]{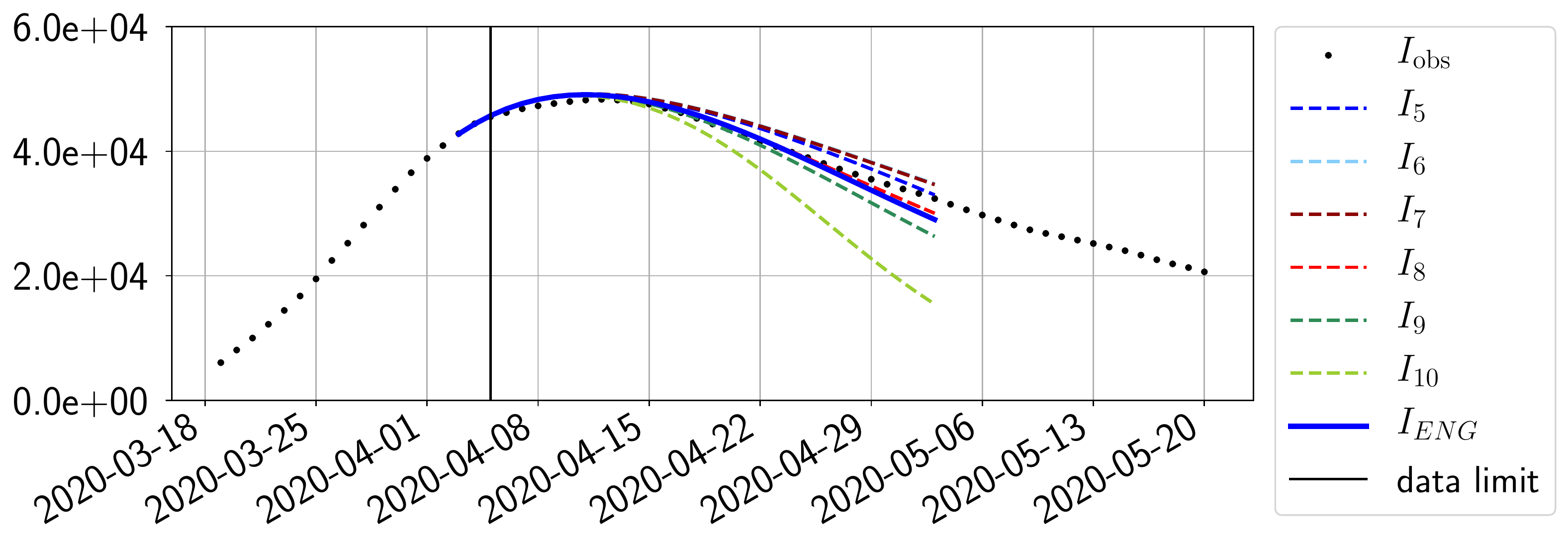}
\caption{Infected}
\end{subfigure}
\begin{subfigure}{.45\textwidth}
\includegraphics[width=1\textwidth]{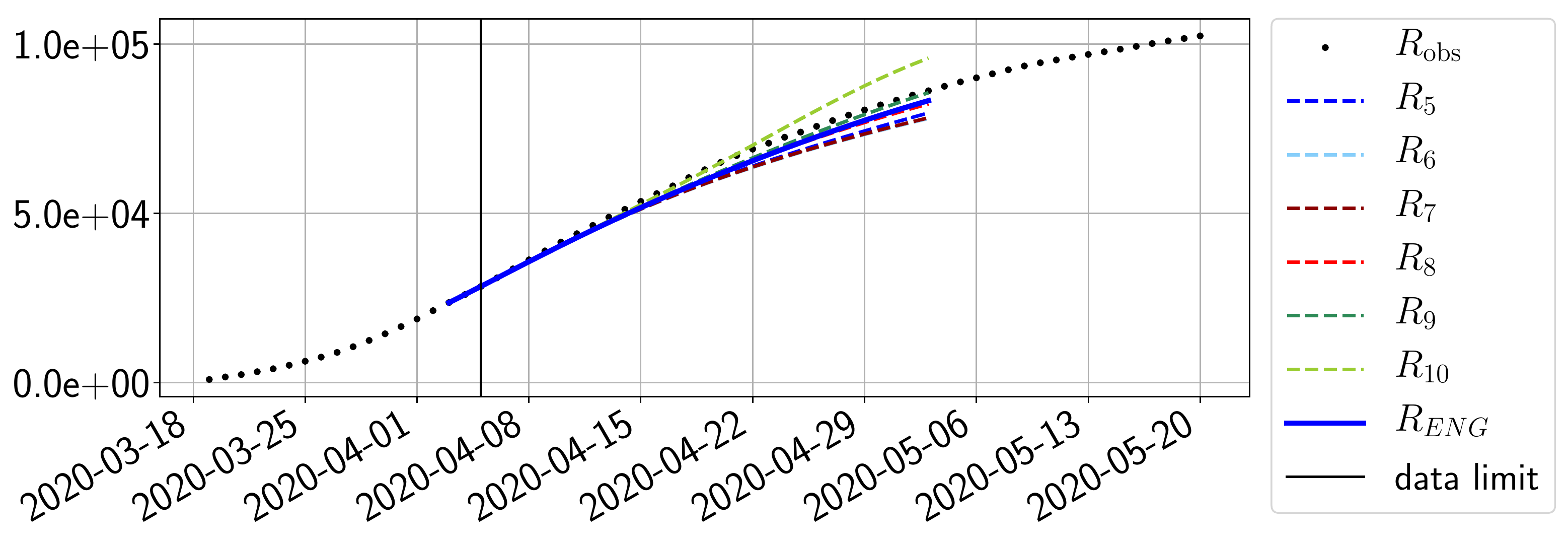}
\caption{Removed}
\end{subfigure}
\caption{ENG forecast from $T=05/04$}
\label{fig:forecast_28_days_modes_0504}
\end{figure}

\begin{figure}[H]
\centering
\begin{subfigure}{.45\textwidth}
\includegraphics[width=1\textwidth]{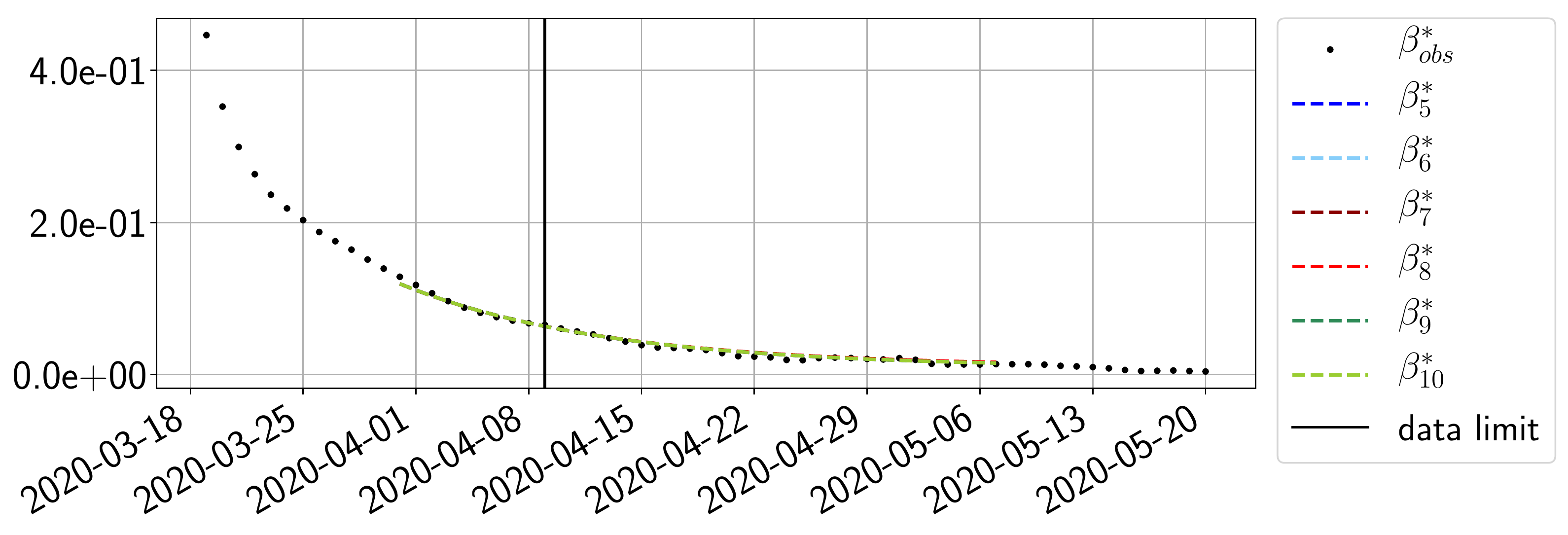}
\caption{$\beta$}
\end{subfigure}
\begin{subfigure}{.45\textwidth}
\includegraphics[width=1\textwidth]{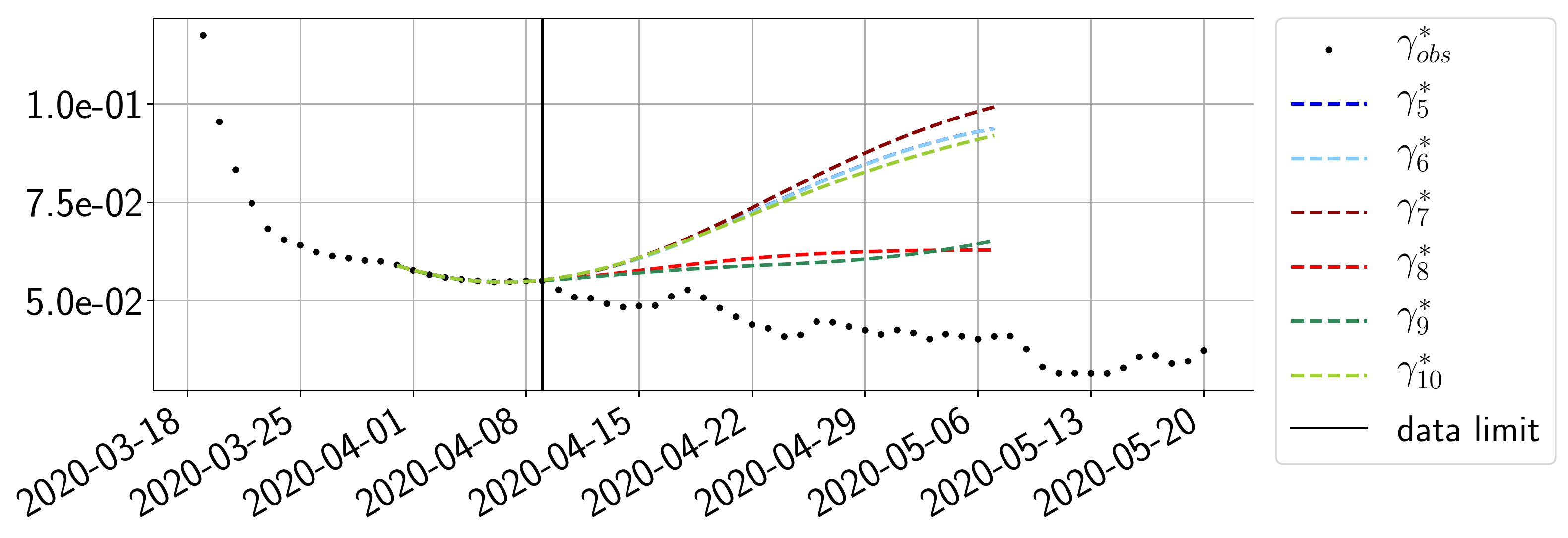}
\caption{$\gamma$}
\end{subfigure}

\vspace{0.4cm}

\begin{subfigure}{.45\textwidth}
\includegraphics[width=1\textwidth]{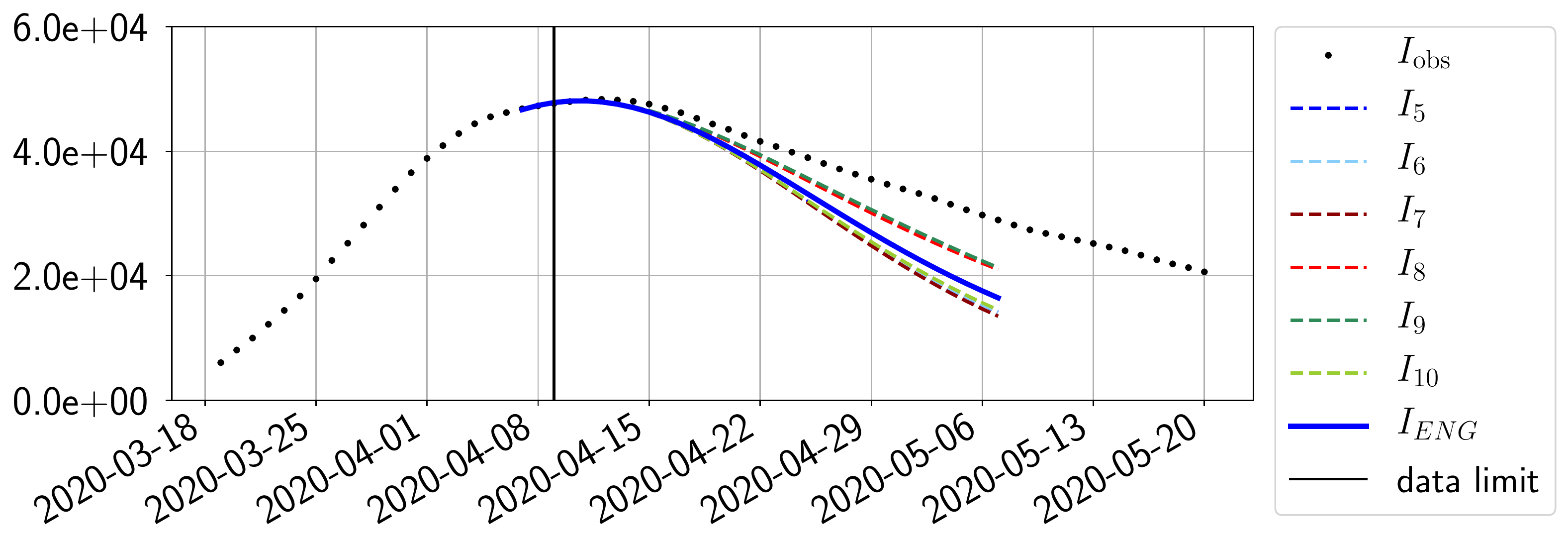}
\caption{Infected}
\end{subfigure}
\begin{subfigure}{.45\textwidth}
\includegraphics[width=1\textwidth]{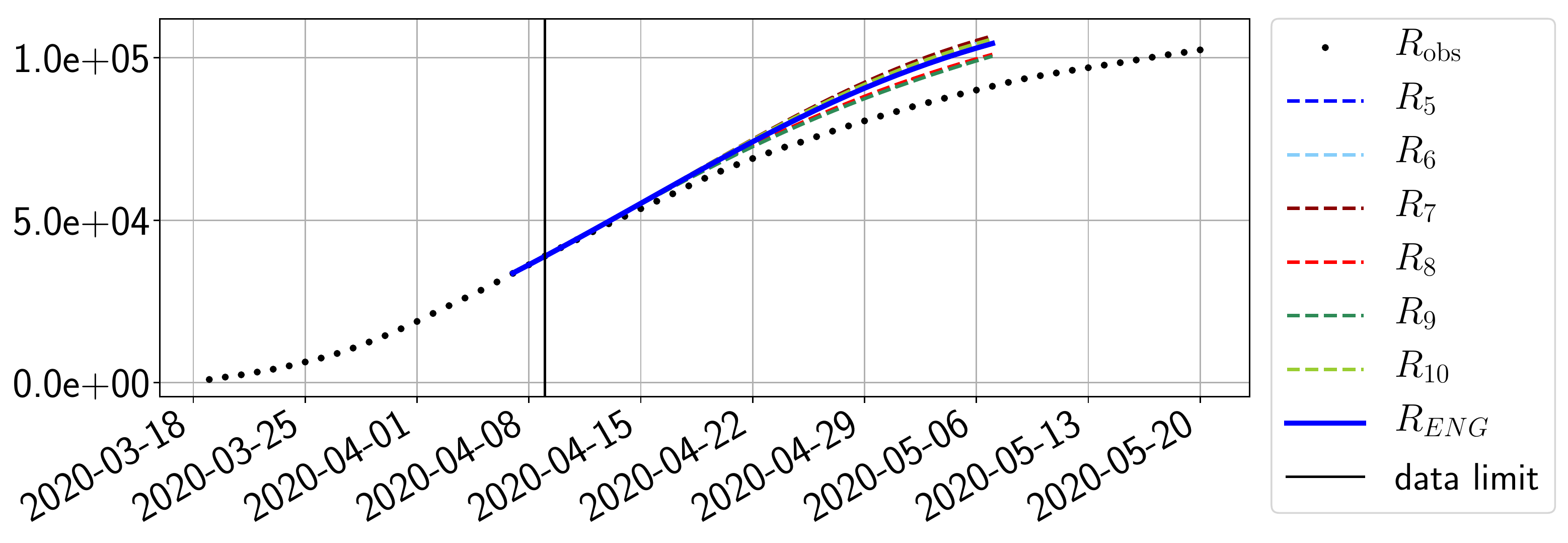}
\caption{Removed}
\end{subfigure}
\caption{ENG forecast from $T=05/04$}
\label{fig:forecast_28_days_modes_0504}
\end{figure}

\begin{figure}[H]
\centering
\begin{subfigure}{.45\textwidth}
\includegraphics[width=1\textwidth]{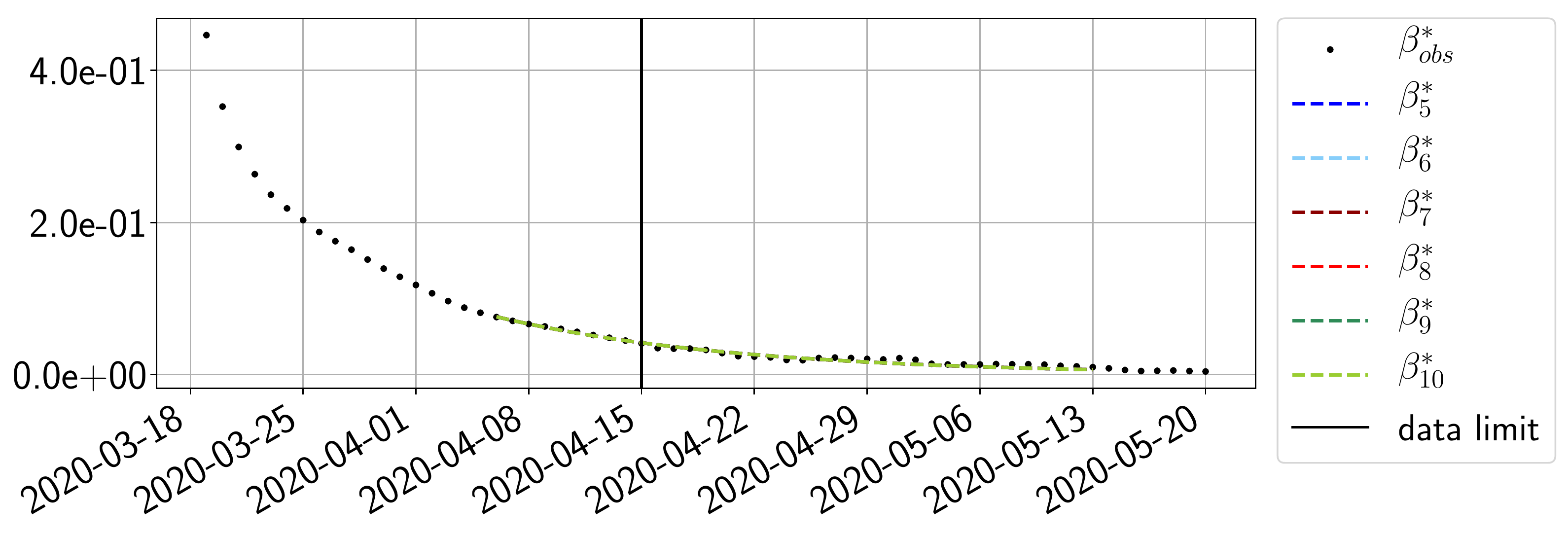}
\caption{$\beta$}
\end{subfigure}
\begin{subfigure}{.45\textwidth}
\includegraphics[width=1\textwidth]{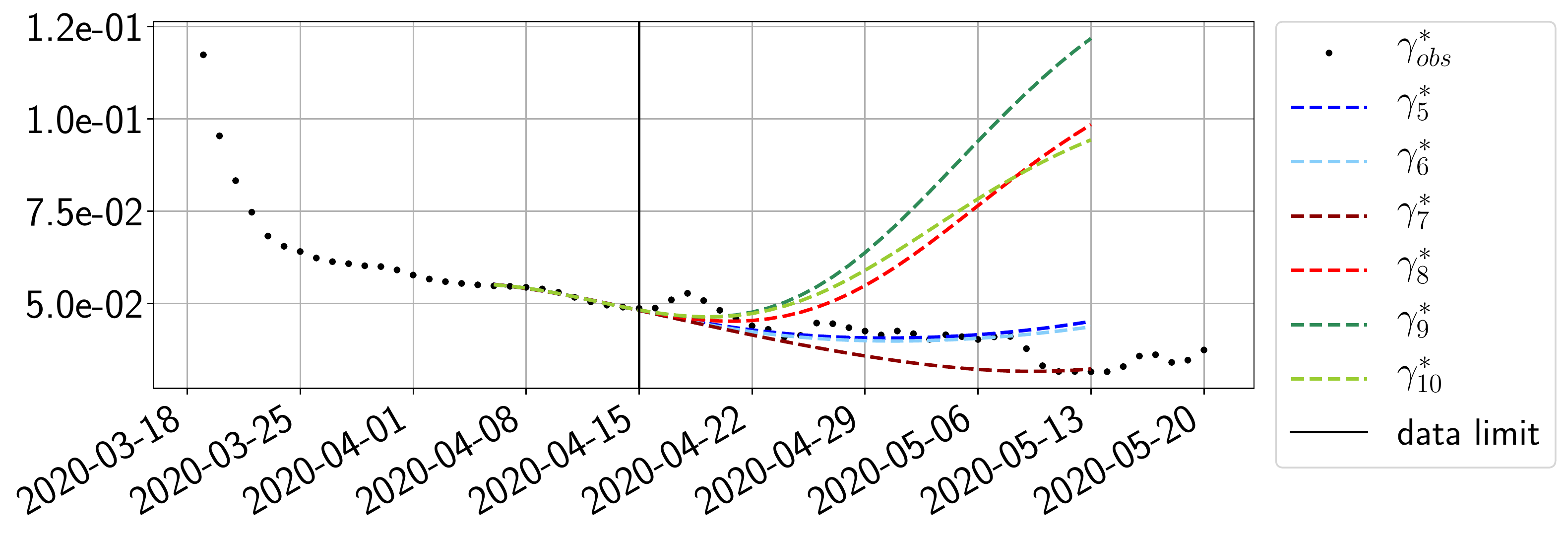}
\caption{$\gamma$}
\end{subfigure}

\vspace{0.4cm}

\begin{subfigure}{.45\textwidth}
\includegraphics[width=1\textwidth]{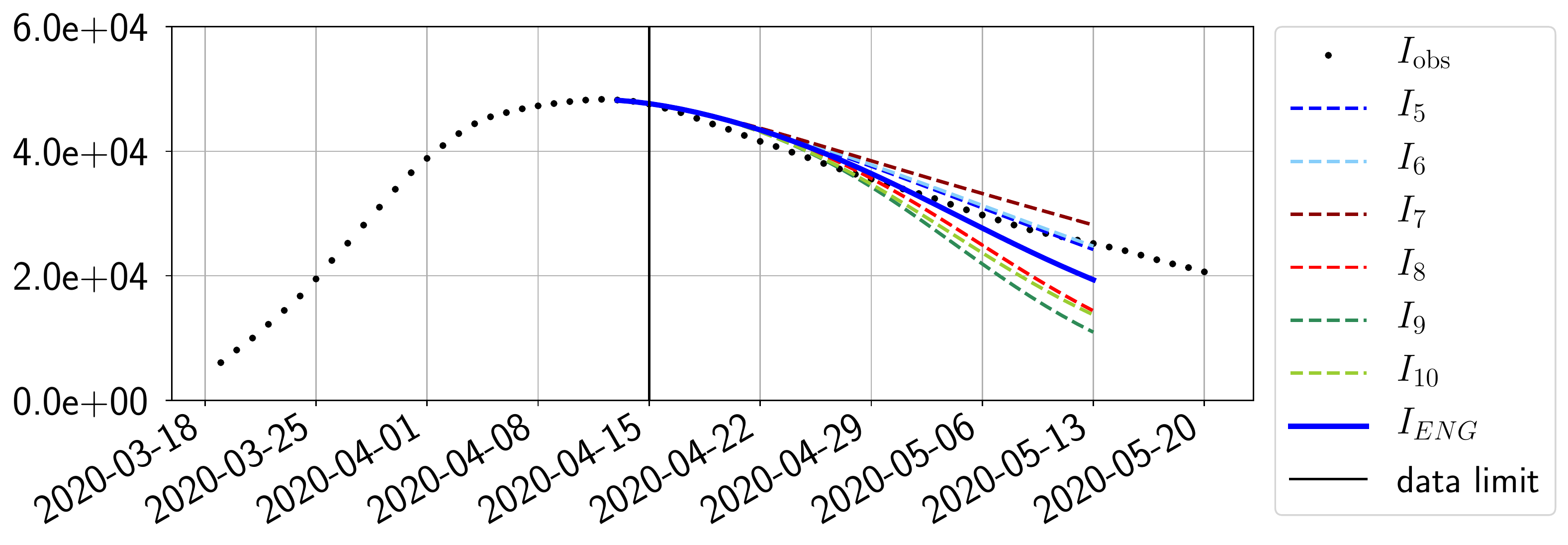}
\caption{Infected}
\end{subfigure}
\begin{subfigure}{.45\textwidth}
\includegraphics[width=1\textwidth]{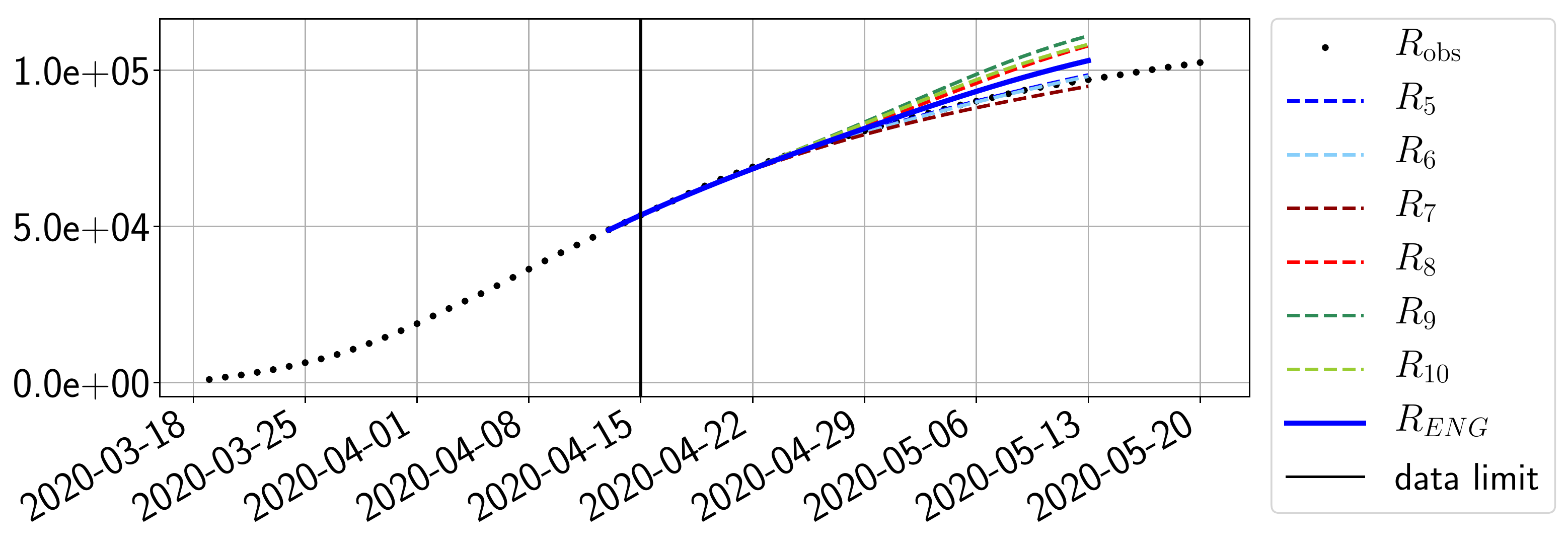}
\caption{Removed}
\end{subfigure}
\caption{ENG forecast from $T=15/04$}
\label{fig:forecast_28_days_modes_1504}
\end{figure}

\begin{figure}[H]
\centering
\begin{subfigure}{.45\textwidth}
\includegraphics[width=1\textwidth]{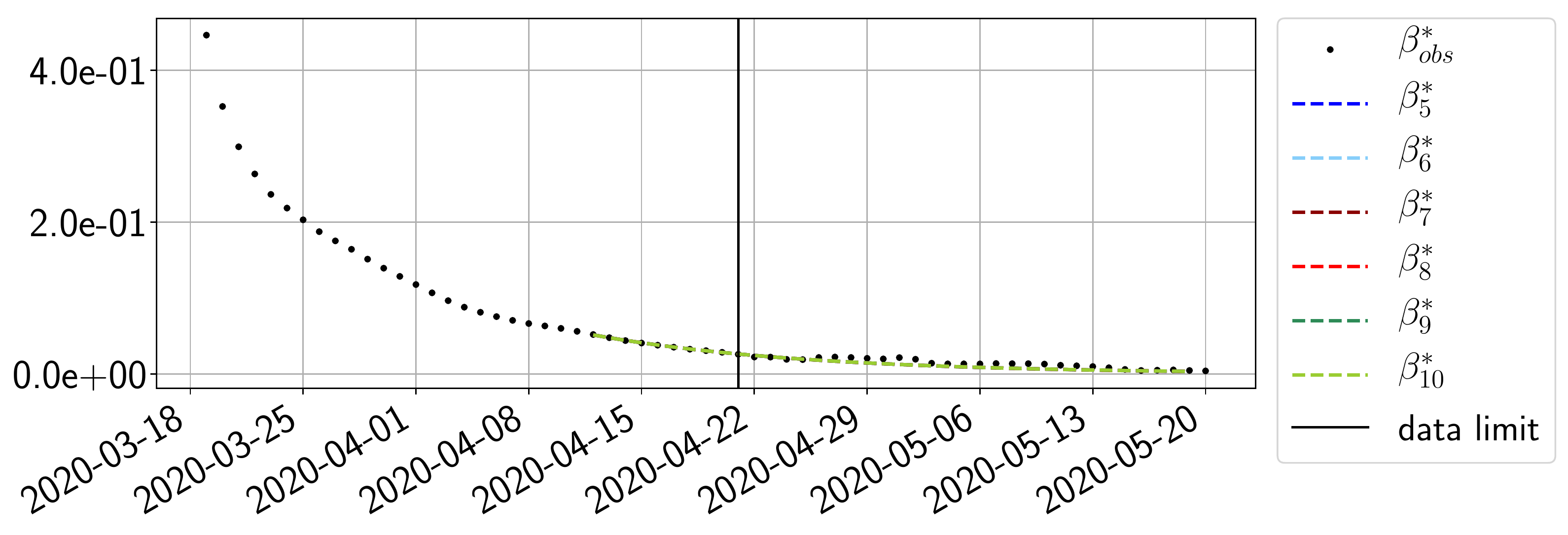}
\caption{$\beta$}
\end{subfigure}
\begin{subfigure}{.45\textwidth}
\includegraphics[width=1\textwidth]{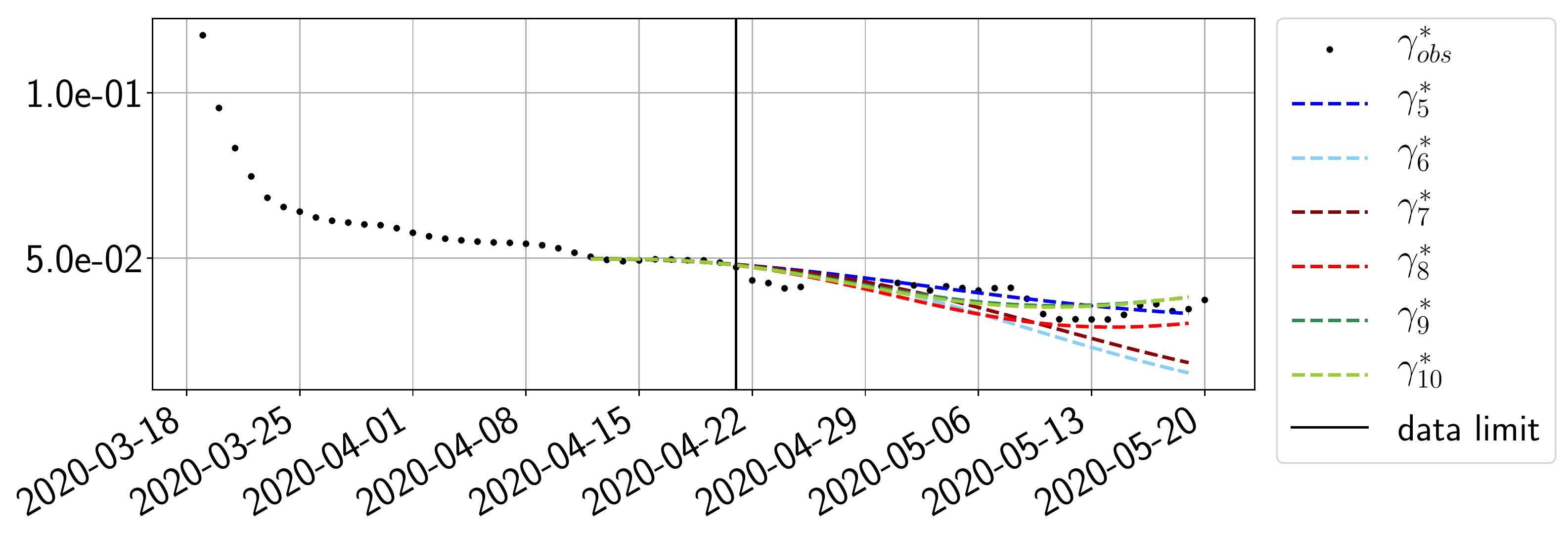}
\caption{$\gamma$}
\end{subfigure}

\vspace{0.4cm}

\begin{subfigure}{.45\textwidth}
\includegraphics[width=1\textwidth]{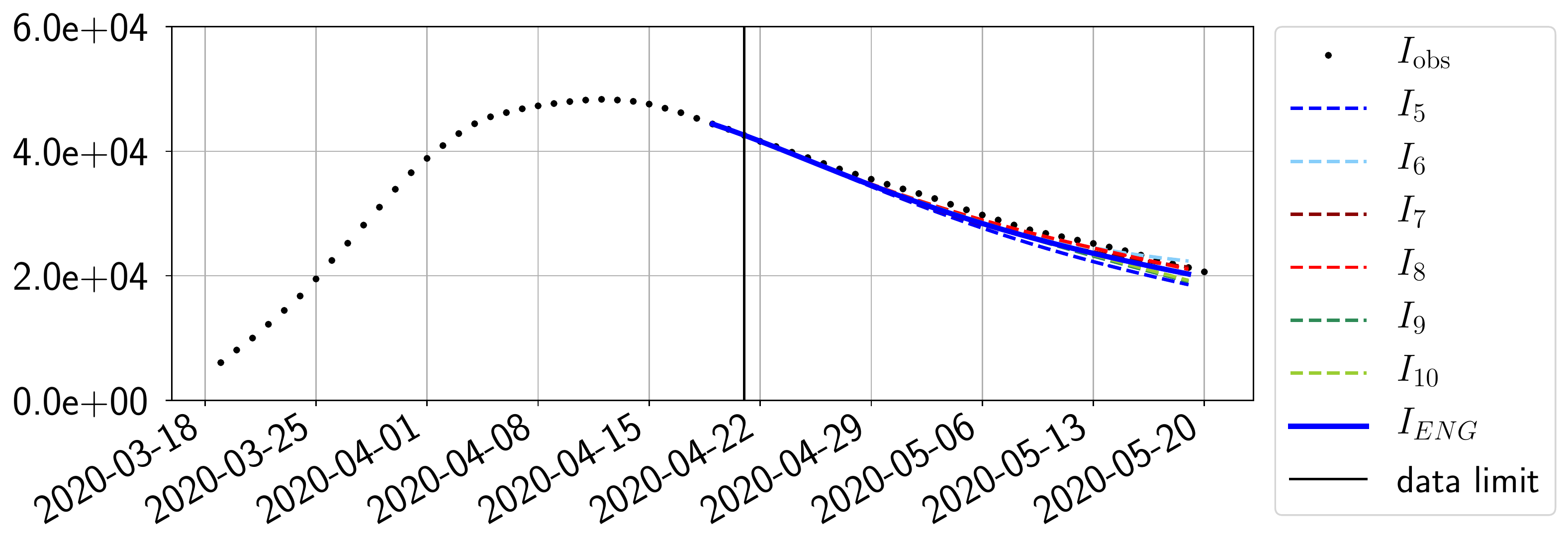}
\caption{Infected}
\end{subfigure}
\begin{subfigure}{.45\textwidth}
\includegraphics[width=1\textwidth]{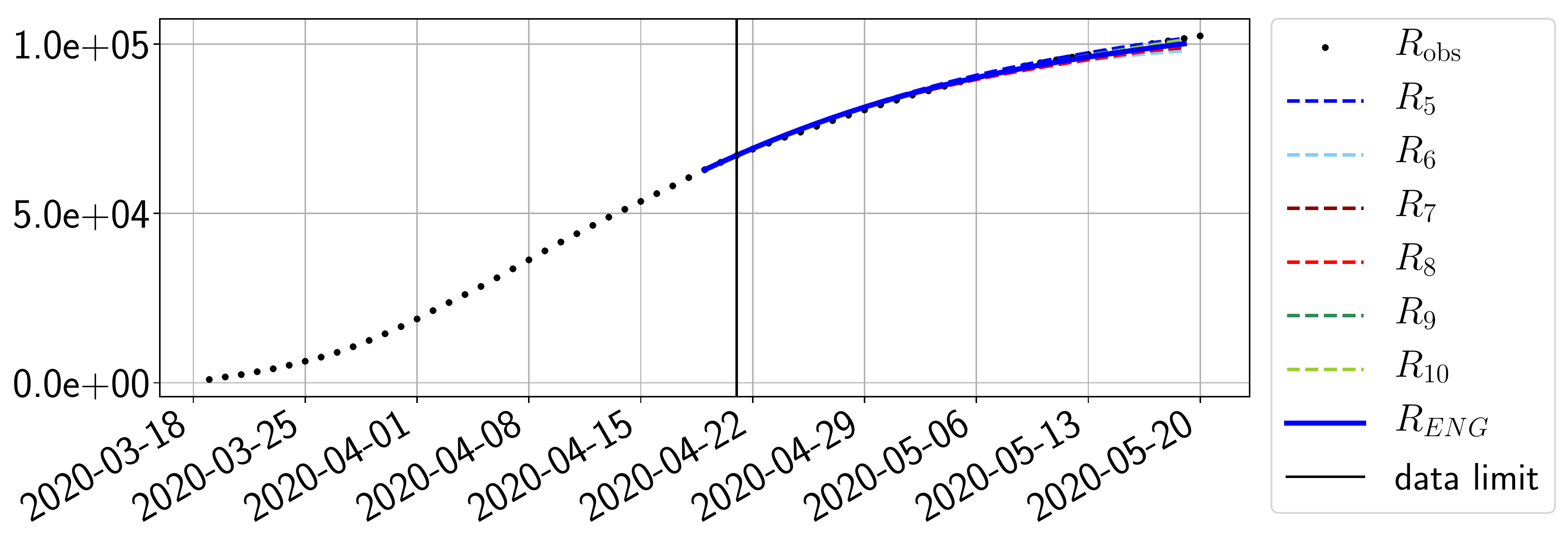}
\caption{Removed}
\end{subfigure}
\caption{ENG forecast from $T=21/04$}
\label{fig:forecast_28_days_modes_2104}
\end{figure}

\begin{figure}[H]
\centering
\begin{subfigure}{.45\textwidth}
\includegraphics[width=1\textwidth]{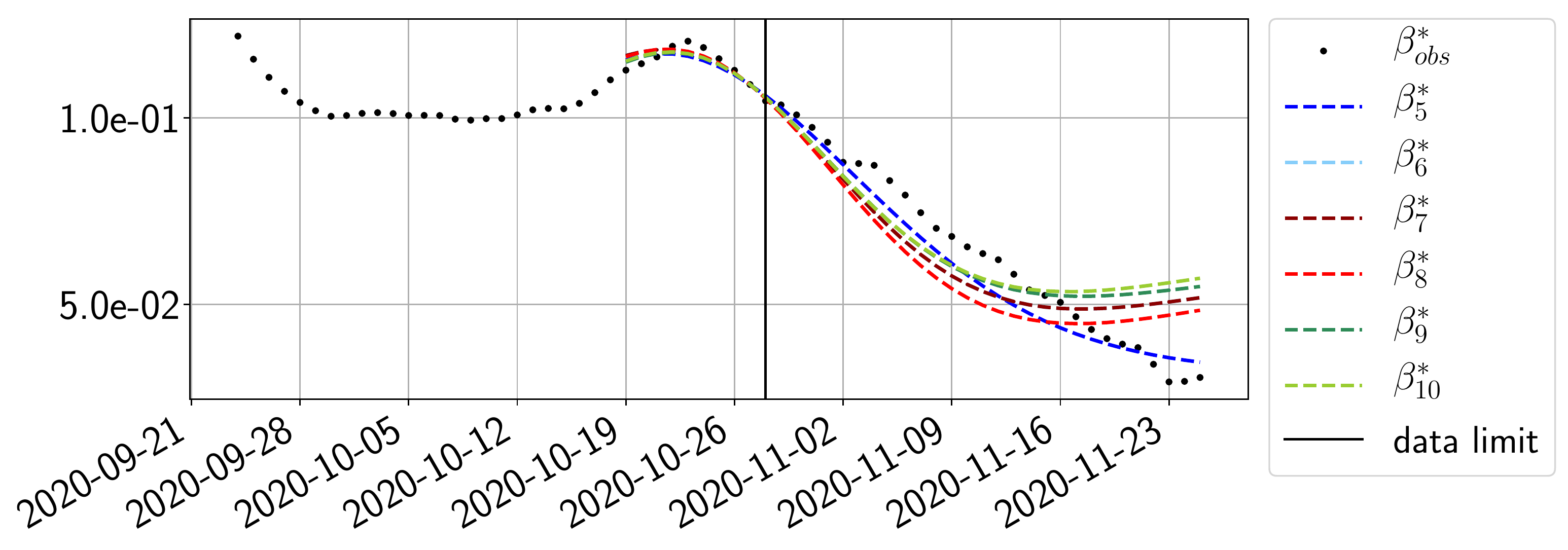}
\caption{$\beta$}
\end{subfigure}
\begin{subfigure}{.45\textwidth}
\includegraphics[width=1\textwidth]{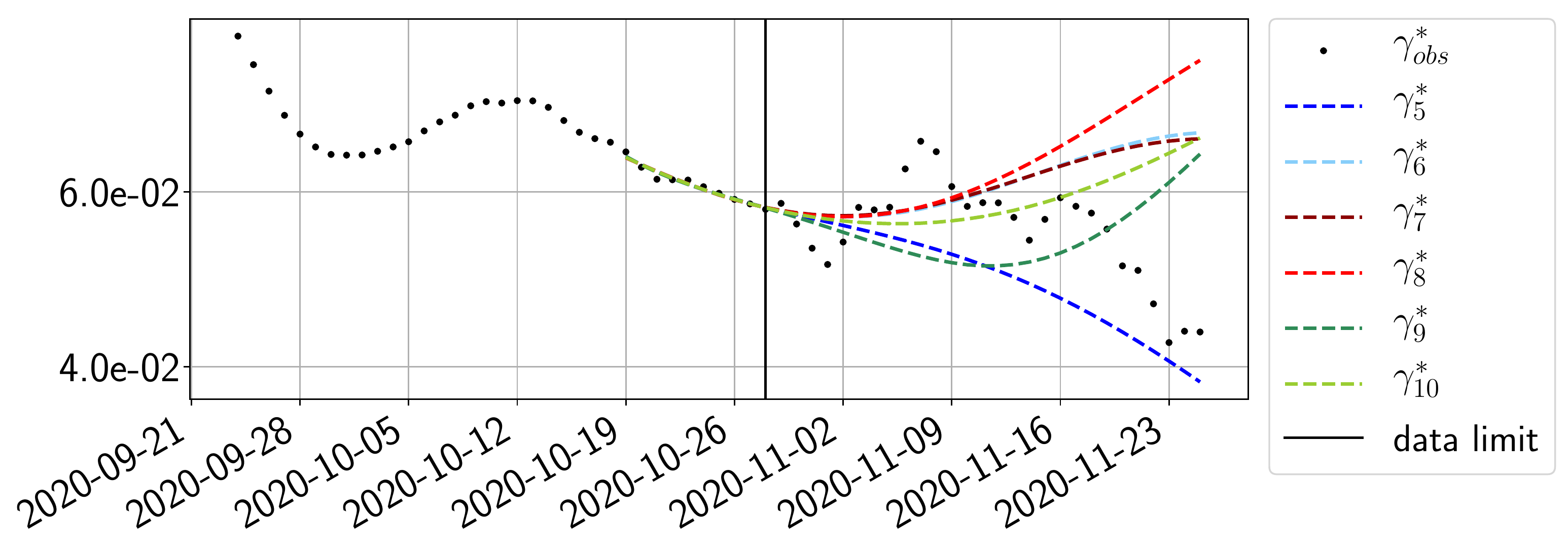}
\caption{$\gamma$}
\end{subfigure}

\vspace{0.4cm}

\begin{subfigure}{.45\textwidth}
\includegraphics[width=1\textwidth]{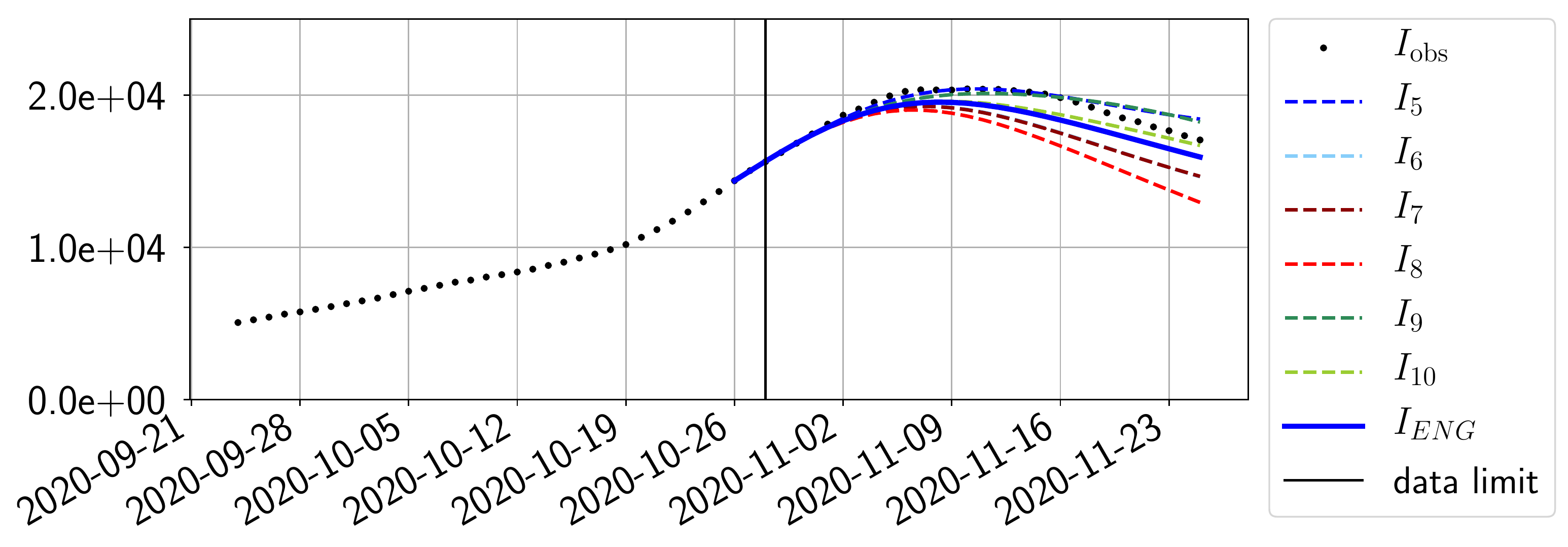}
\caption{Infected}
\end{subfigure}
\begin{subfigure}{.45\textwidth}
\includegraphics[width=1\textwidth]{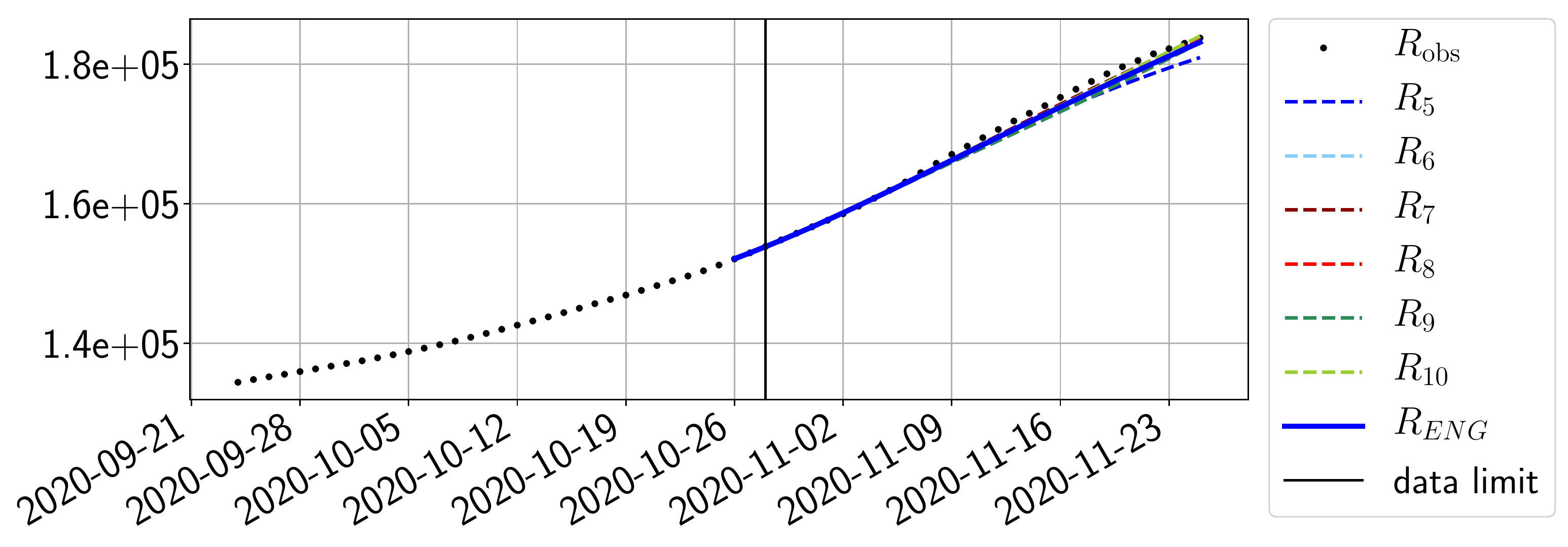}
\caption{Removed}
\end{subfigure}
\caption{ENG forecast from $T=28/10$}
\label{fig:forecast_28_days_modes_2810}
\end{figure}

\section{Conclusion}
\label{sec:conclusion} 
We have developed an epidemiological forecasting method based on reduced modeling techniques. Among the different strategies that have been explored, the one that outperforms the rest in terms of robustness and forecasting power involves reduced models that are built with an Enlarged Nonnegative Greedy (ENG) strategy. This method is novel and of interest in itself since it allows to build reduced models that preserve positivity and even other types of bounds. Despite that ENG does not have optimal fitting properties (i.e.~interpolation properties), it is well suited for forecasting since, due to the preservation of the natural constraints of the coefficients, it generates realistic dynamics with few modes. The results have been presented for a mono-regional test case, and we plan to extend the present methodology to a multi-regional setting using mobility data.

Last but not least, we would like to emphasize that the developed strategy is very general, and could be applied to the forecasting of other types of dynamics. The specificity of each problem may however require adjustments in the reduced models. This is exemplified in the present problem through the construction of reduced models that preserve positivity.

\section*{Acknowledgments and disclosure of funding}
We would like to thank our colleagues from the ``Face au virus'' initiative in PSL University: Jamal Atif, Laurent Massoulié, Olivier Cappé, and Akin Kazakcy. We also thank Gabriel Turinici for his feedback on the multiregional models. Part of this research has been supported by the Emergences project grant ``Models and Measures'' of the Paris city council, by the generous donation made available by Alkan together with the complementary funding given by the Sorbonne University Foundation. This research is done in the frame of the project Pandemia/Covidia and also the GIS Obepine, both projects aiming at a better understanding of the Covid-19 pandemic.

\appendix
\section{Analysis of model error and observation noise}
\label{appendix:noise}
In this section we study the impact of observation noise and model error in the quality and behavior of the fitting step. The elements that impact the accuracy of our procedure are the following:
\begin{enumerate}
\item \textbf{Observation noise: }The observed health data $\bfI_\obs$ and $\bfR_\obs$ presents several sources of noise: there may be some inaccuracies in the reporting of cases, and the data are then post-processed. This eventually yields to noisy time series $\bfI_\obs$ and $\bfR_\obs$ for which it is difficult to estimate the uncertainty. These noisy data are then used to produce (through finite differences) $\beta^*_\obs$ and $\gamma^*_\obs$ which are in turn also noisy.
\item \textbf{Model errors:} Two types of model errors are identified 
\begin{enumerate}
\item \textbf{Intrinsic model error on $\cB$ and $\cG$:} The families of detailed models that we use are rich but they may not cover all possible evolutions of $\bfI_\obs$ and $\bfR_\obs$. In other words, our manifolds $\cB$ and $\cG$ may not perfectly cover the real epidemiological \col{evolution}. \col{Such error motivated the introduction of the exponential functions $b_0$ and $g_0$ described in Section \ref{sec:MOR}.}
\item \textbf{Sampling error of $\cB$ and $\cG$:} The size of the training sets $\cB_\tr$ and $\cG_\tr$ and the dimension $n$ of the reduced models $\rB_n$ and $\rG_n$ also limit the approximation capabilities of the continuous target sets $\cB$ and $\cG$
\end{enumerate}
\end{enumerate}
In this section, we aim at disentangling these effects in order to give a better insight of the fitting error behavior on the real data.

\subsection{Study of the impact of the sampling error: Fitting a virtual scenario}
Our set of virtual epidemiological dynamics is $\cU$. After the collapsing step, the manifolds for $\bfbeta$ and $\bfgamma$ are $\cB$ and $\cG$. These sets are potentially infinite and we have used finite training subsets $\cB_\tr\subseteq \cB$ and $\cG_\tr\subseteq \cG$ to build the reduced models $\rB_n$ and $\rG_n$.

First we look at the eigenvalues for $\beta$ and $\gamma$ when performing an SVD decomposition for the virtual scenarios. Figure \ref{fig:decay_eigenvalues} shows a rapid decay of the eigenvalues obtained by SVD decomposition, it shows that we can obtain a very good approximation of elements of $\cB_\tr\subseteq \cB$ and $\cG_\tr\subseteq \cG$ with only a few modes.

\begin{figure}[H]
\centering
\begin{subfigure}{.4\textwidth}
\includegraphics[width=1\textwidth]{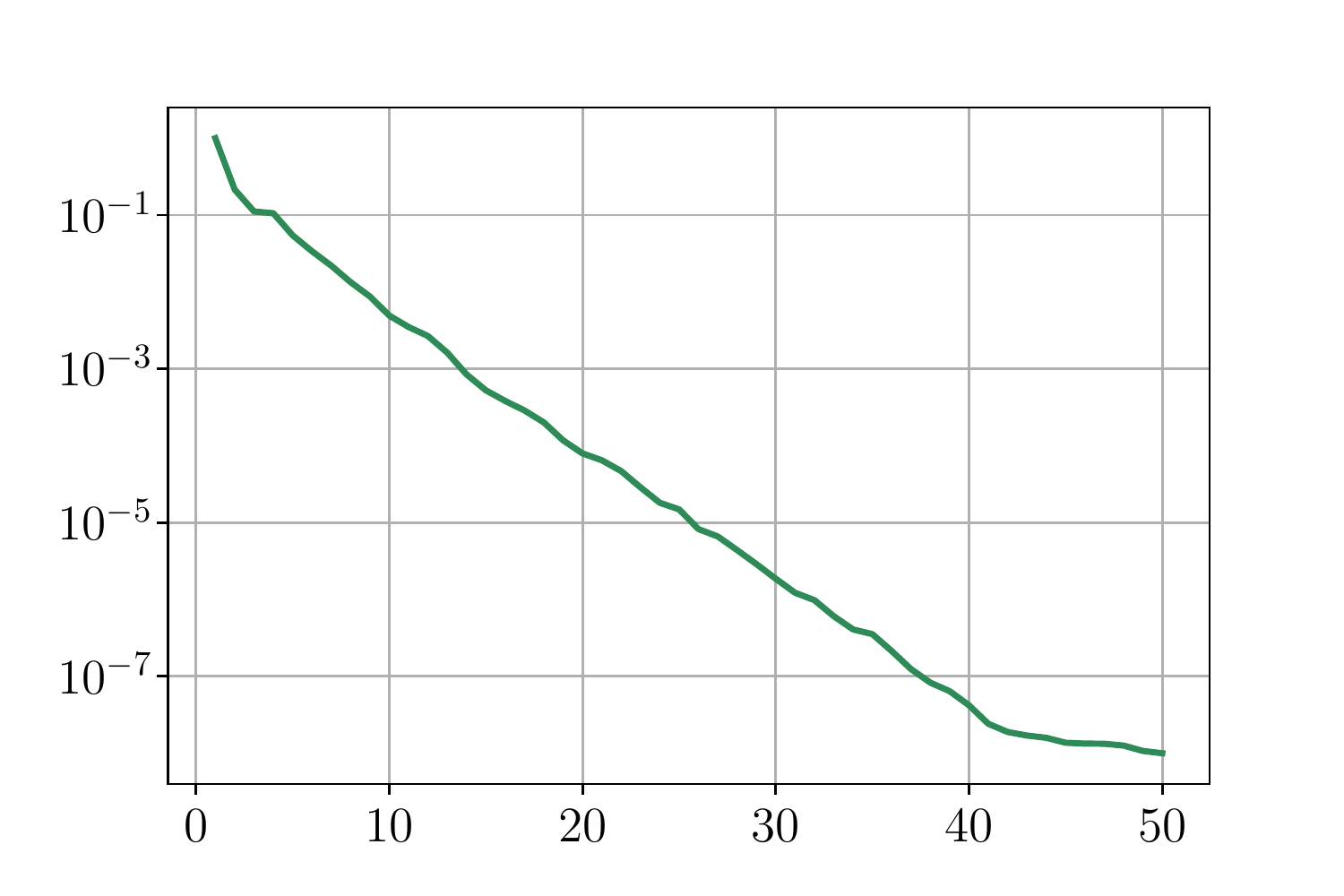}
\caption{Normalised eigenvalues for $\beta$ vs $n$}
\end{subfigure}
\begin{subfigure}{.4\textwidth}
\includegraphics[width=1\textwidth]{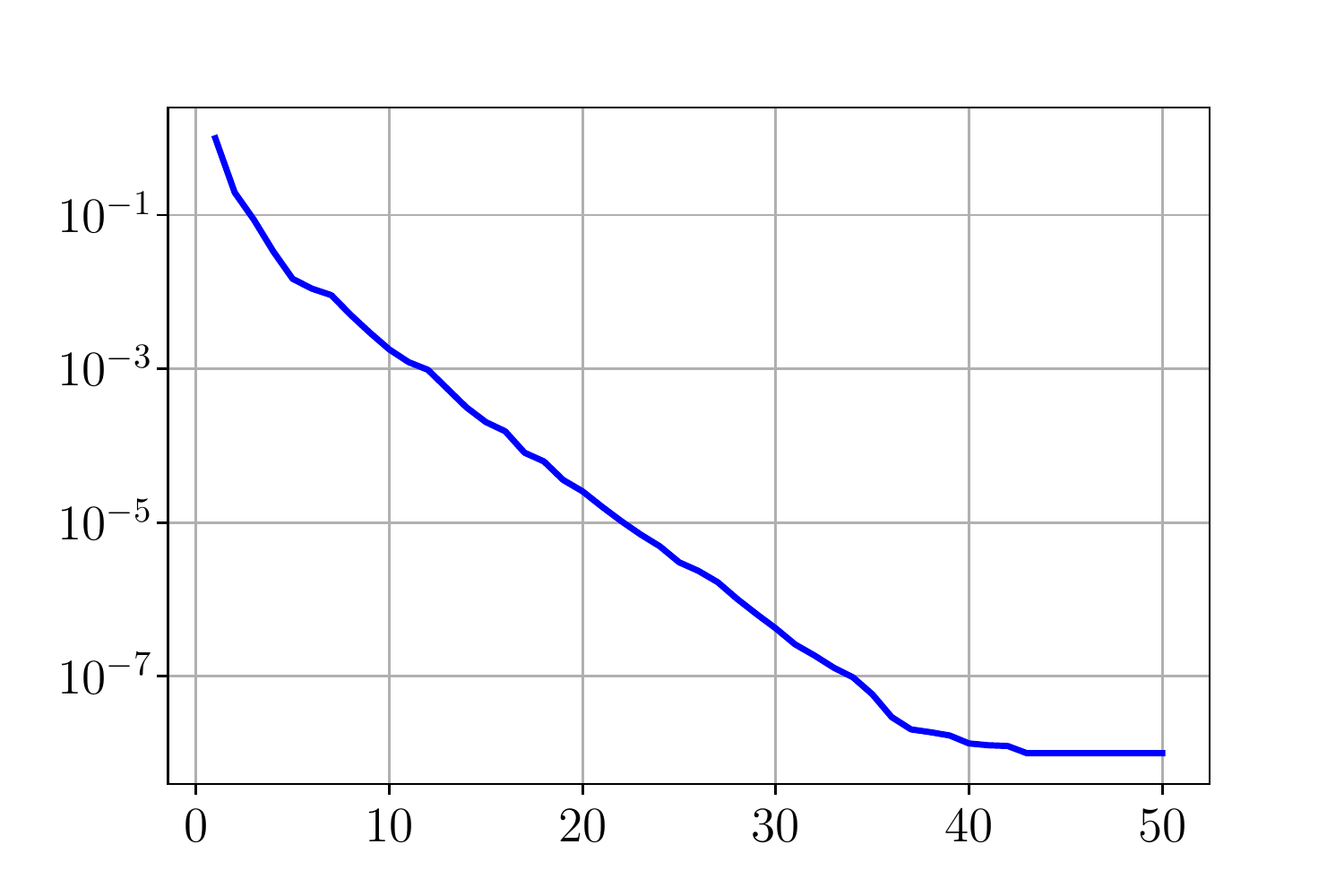}
\caption{Normalised eigenvalues for $\gamma$ vs $n$}
\end{subfigure}
\caption{Training step: Decay of SVD eigenvalues}
\label{fig:decay_eigenvalues}
\end{figure}

Among the sources of noise and model error given at the beginning of Appendix \ref{appendix:noise}, we can study the impact of the sampling error (point 2-b) as follows. \col{For this, we start by examining} the fitting error on an interval of $45$ days for two functions $\bfbeta$ and $\bfgamma$ which belong to the manifold $\cB$ and $\cG$ and are in the training sets $\cB_\tr$ and $\cG_\tr$. \col{This error will serve as a benchmark that we will use to compare fitting errors of functions that are not in the training sets.}

\col{Figure~\ref{fig:fitting_error_b_g_inbase} shows relative fitting errors $\| \beta^*_n - \beta^*_\obs\|_{L^2([t_0, T])}/\| \bar{\beta}^*_\obs\|_{L^2([t_0, T])}$ and $\| \gamma^*_n - \gamma^*_\obs\|_{L^2([t_0, T])}/\| \bar{\gamma}^*_\obs\|_{L^2([t_0, T])}$ using SVD, NMF and ENG reduced bases with $n=2,\dots,20$, where the notation $\bar{x}$ denotes the mean of the quantity $x$ over $[0,T]$.} We observe that \col{for SVD} the fitting errors behave similarly as the error decay of the training step (Figure \ref{fig:decay_eigenvalues}).

\begin{figure}[H]
\centering
\begin{subfigure}{.4\textwidth}
\includegraphics[width=1\textwidth]{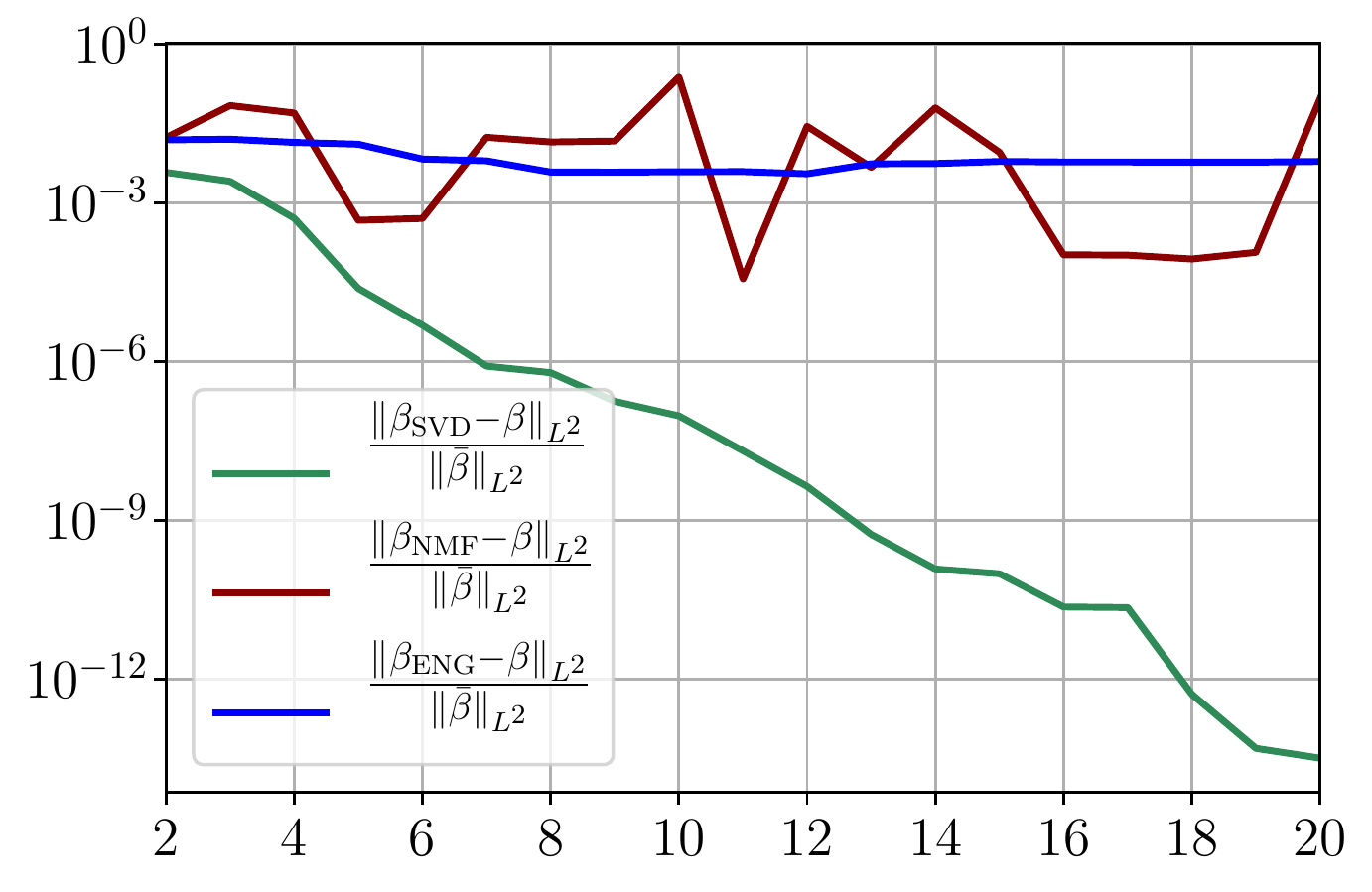}
\caption{$L^2$ relative error of $\beta$ vs $n$}
\end{subfigure}
\begin{subfigure}{.4\textwidth}
\includegraphics[width=1\textwidth]{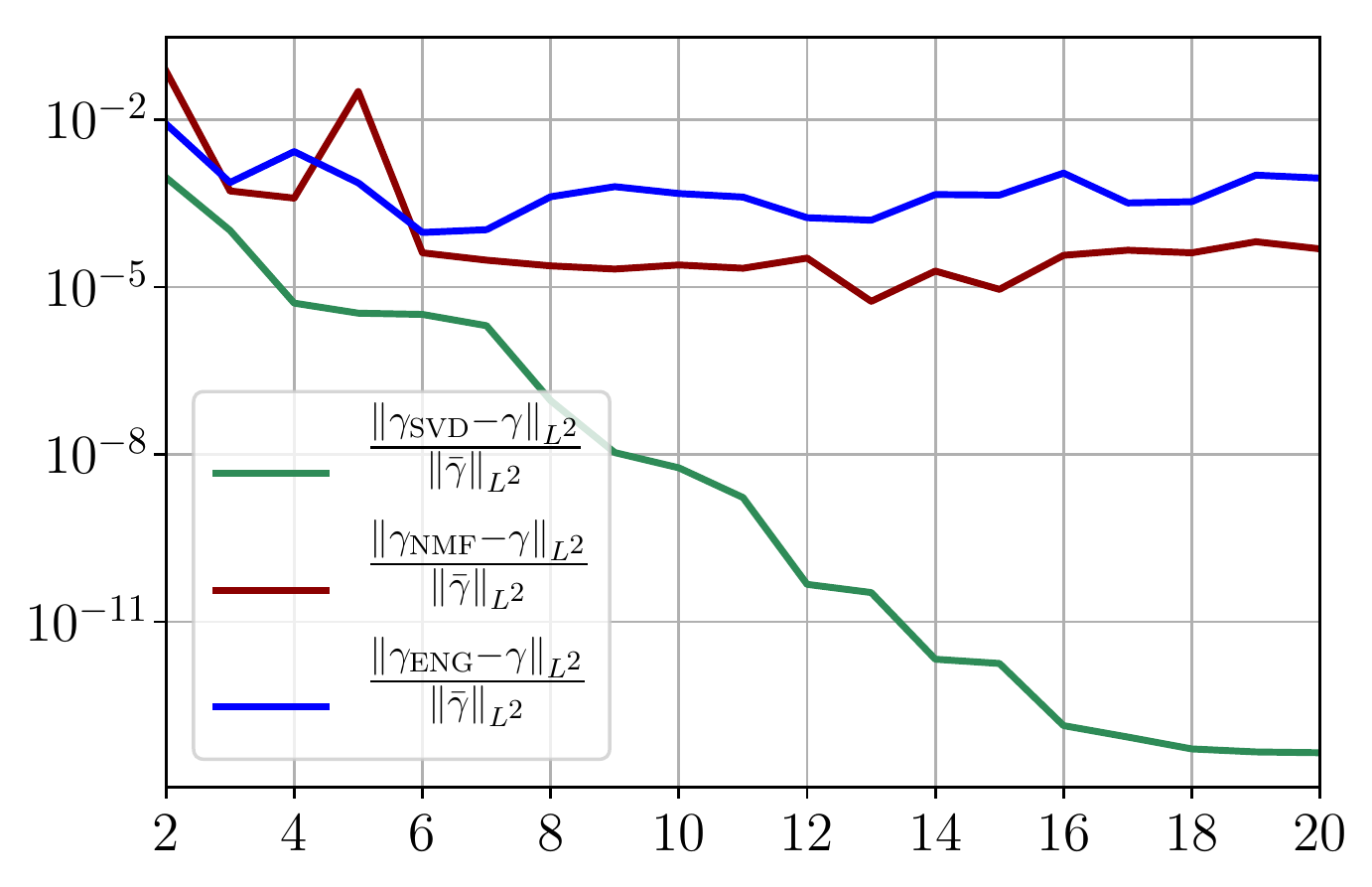}
\caption{$L^2$ relative error of $\gamma$ vs $n$}
\end{subfigure}
\caption{Study of sampling errors: Fitting errors of functions $\bfbeta$ and $\bfgamma$ from $\cB$ and $\cG$ that belong to the training sets $\cB_\tr$ and $\cG_\tr$.}
\label{fig:fitting_error_b_g_inbase}
\end{figure}

Figure \ref{fig:fitting_error_I_R_inbase} shows the $L^1$ and $L^\infty$ errors obtained after the propagation of the fittings of $\bfbeta$ and $\bfgamma$. In both metrics, the error for $\bfI$ and $\bfR$ obtained using SVD is decreasing \col{below $10^{-12}$}. When the NMF and the ENG are used, the error decreases and stagnates at $10^{-2}$ for both $\bfI$ and $\bfR$.

\begin{figure}[H]
\centering
\begin{subfigure}{.4\textwidth}
\includegraphics[width=1\textwidth]{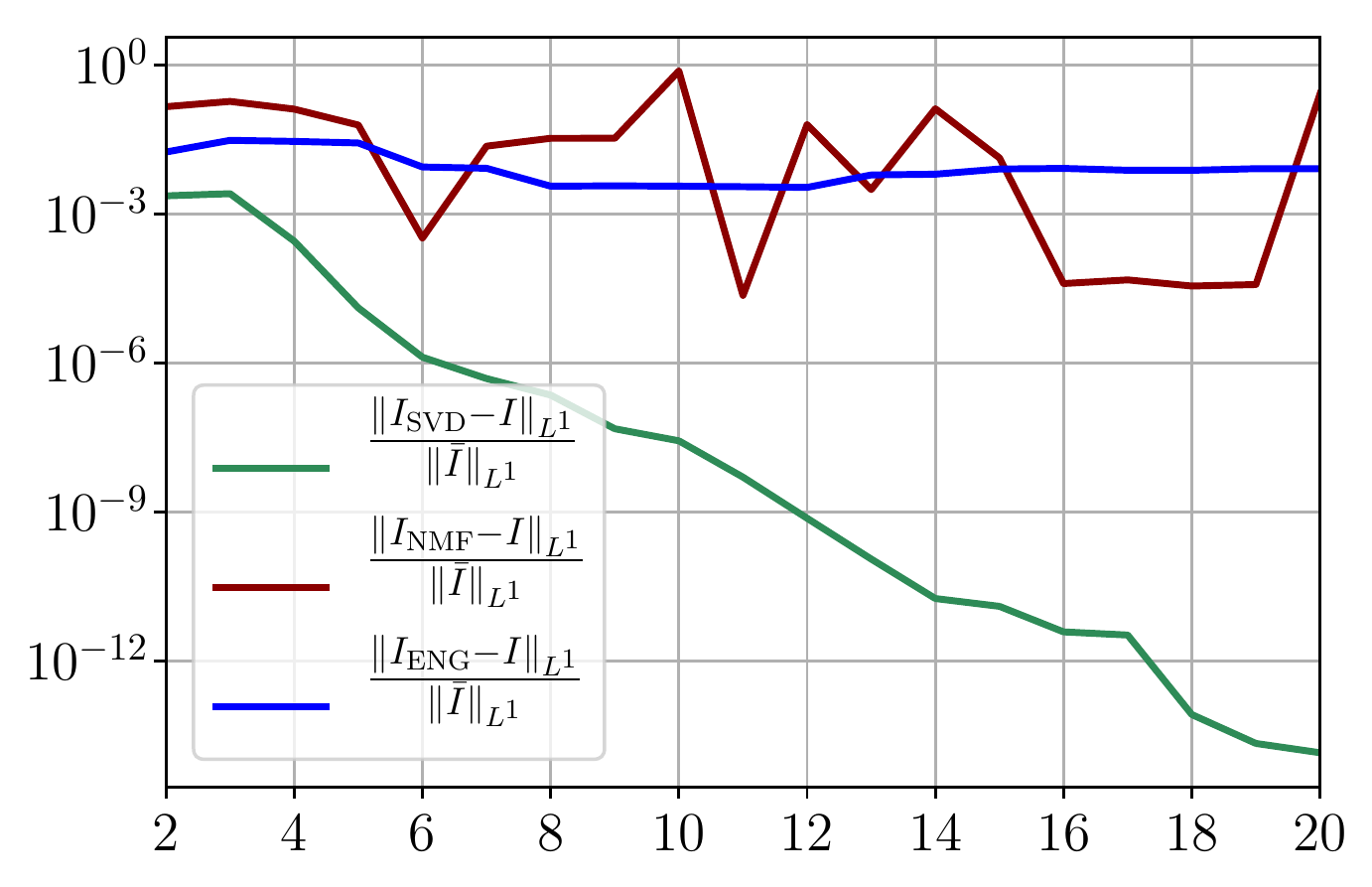}
\caption{$L^1$ relative error  of $I$ vs $n$}
\end{subfigure}
\begin{subfigure}{.4\textwidth}
\includegraphics[width=1\textwidth]{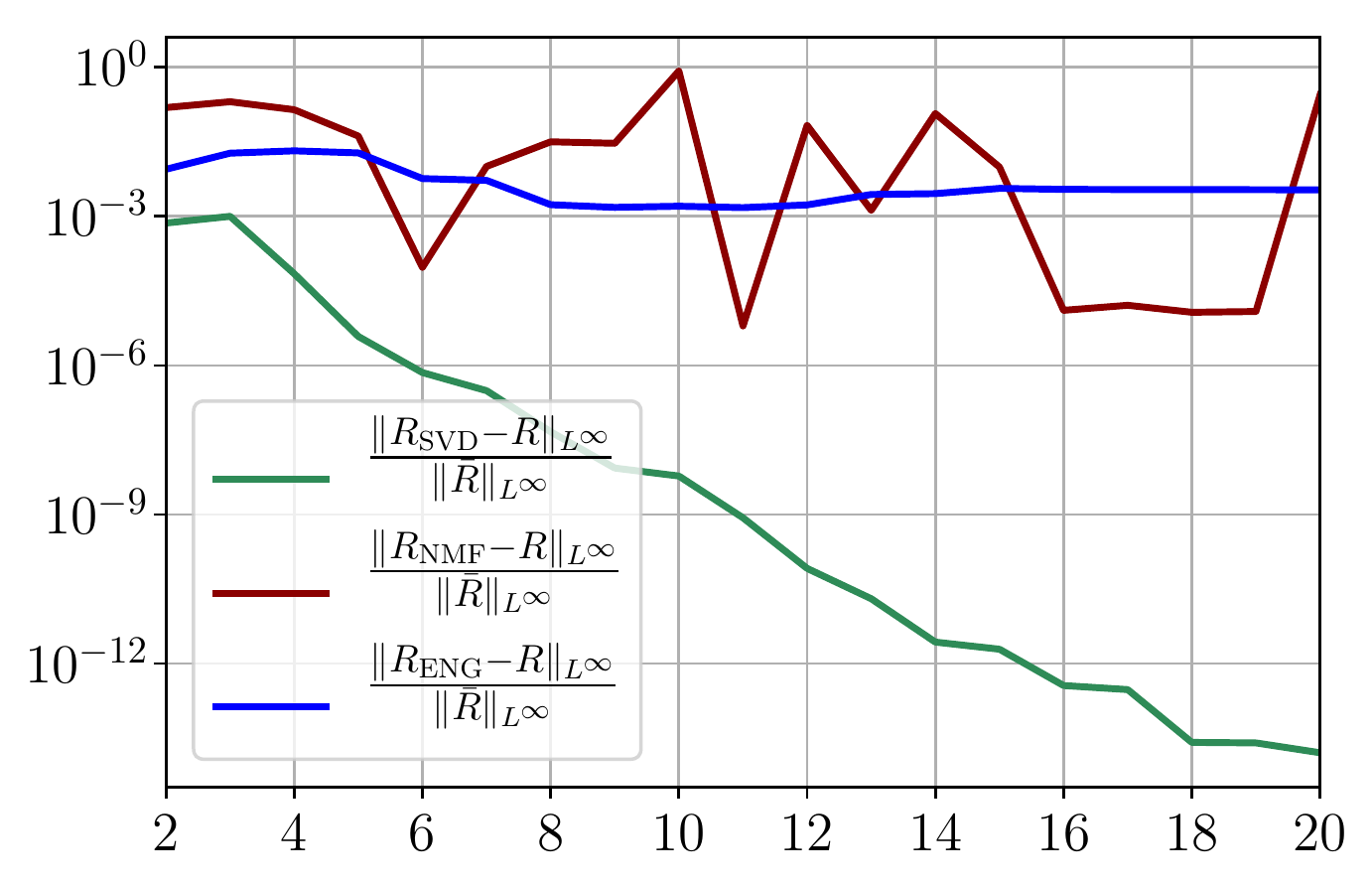}
\caption{$L^\infty$ relative error  of $R$ vs $n$}
\end{subfigure}
\caption{Study of sampling errors: Errors on $\bfI$ and $\bfR$ obtained by a SIR model using the fitted $\bfbeta$ and $\bfgamma$.}
\label{fig:fitting_error_I_R_inbase}
\end{figure}

Now we consider the fitting error for two functions $\bfbeta$ and $\bfgamma$ which belong to the manifold $\cB$ and $\cG$ but which are not in the training sets $\cB_\tr$ and $\cG_\tr$. Figure \ref{fig:fitting_error_b_g} shows the fitting error on the virtual scenario considered using SVD, NMF and ENG reduced bases for $n=2,\dots,20$. \col{We note that the quality of the fittings by each method is very similar to that of Figure \ref{fig:fitting_error_I_R_inbase} where the functions were taken in the training sets.} This illustrates that the sampling error is not playing a major role in our experiments.

\begin{figure}[H]
\centering
\begin{subfigure}{.4\textwidth}
\includegraphics[width=1\textwidth]{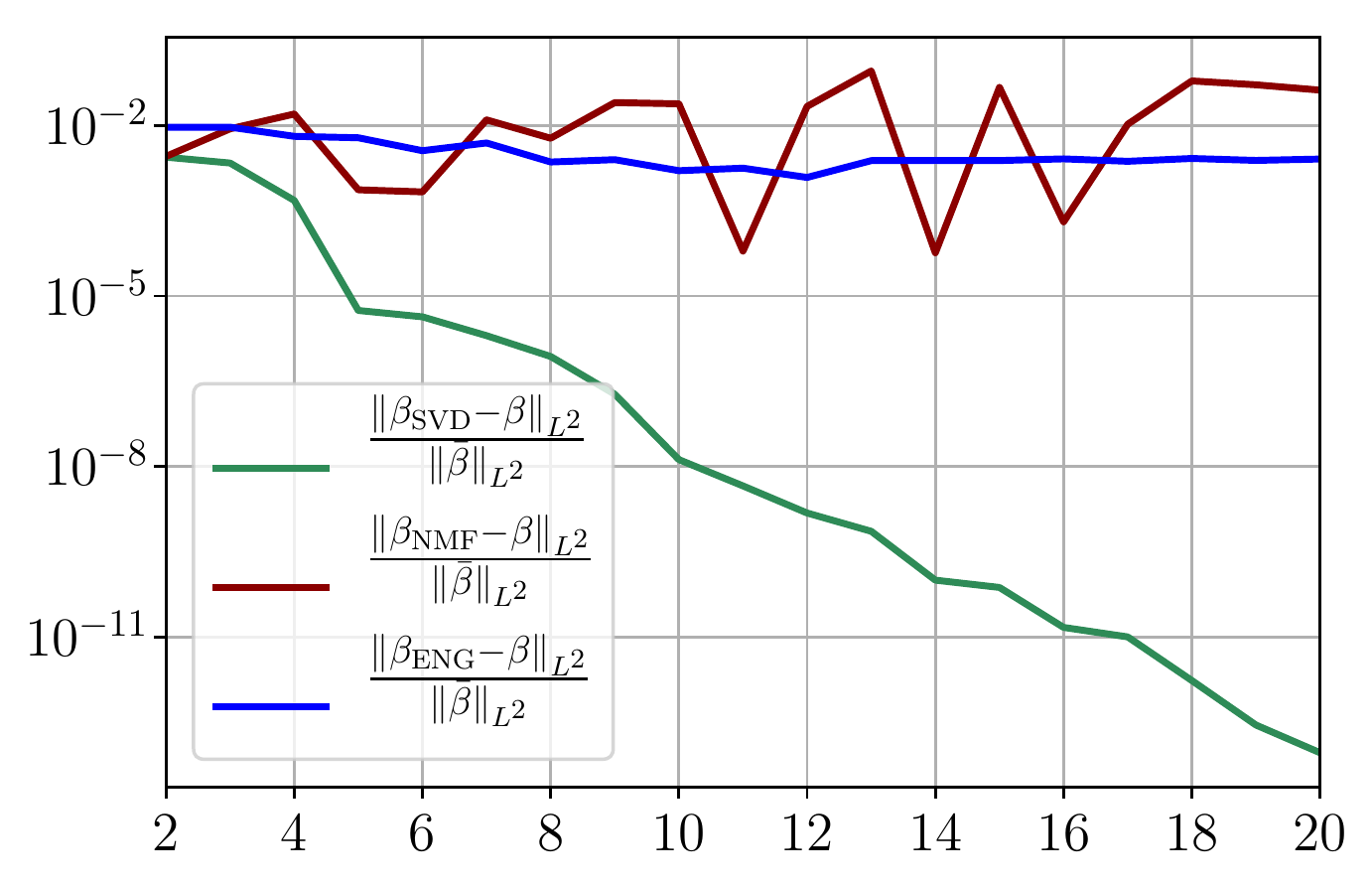}
\caption{$L^2$ relative error  of $\beta$ vs $n$}
\end{subfigure}
\begin{subfigure}{.4\textwidth}
\includegraphics[width=1\textwidth]{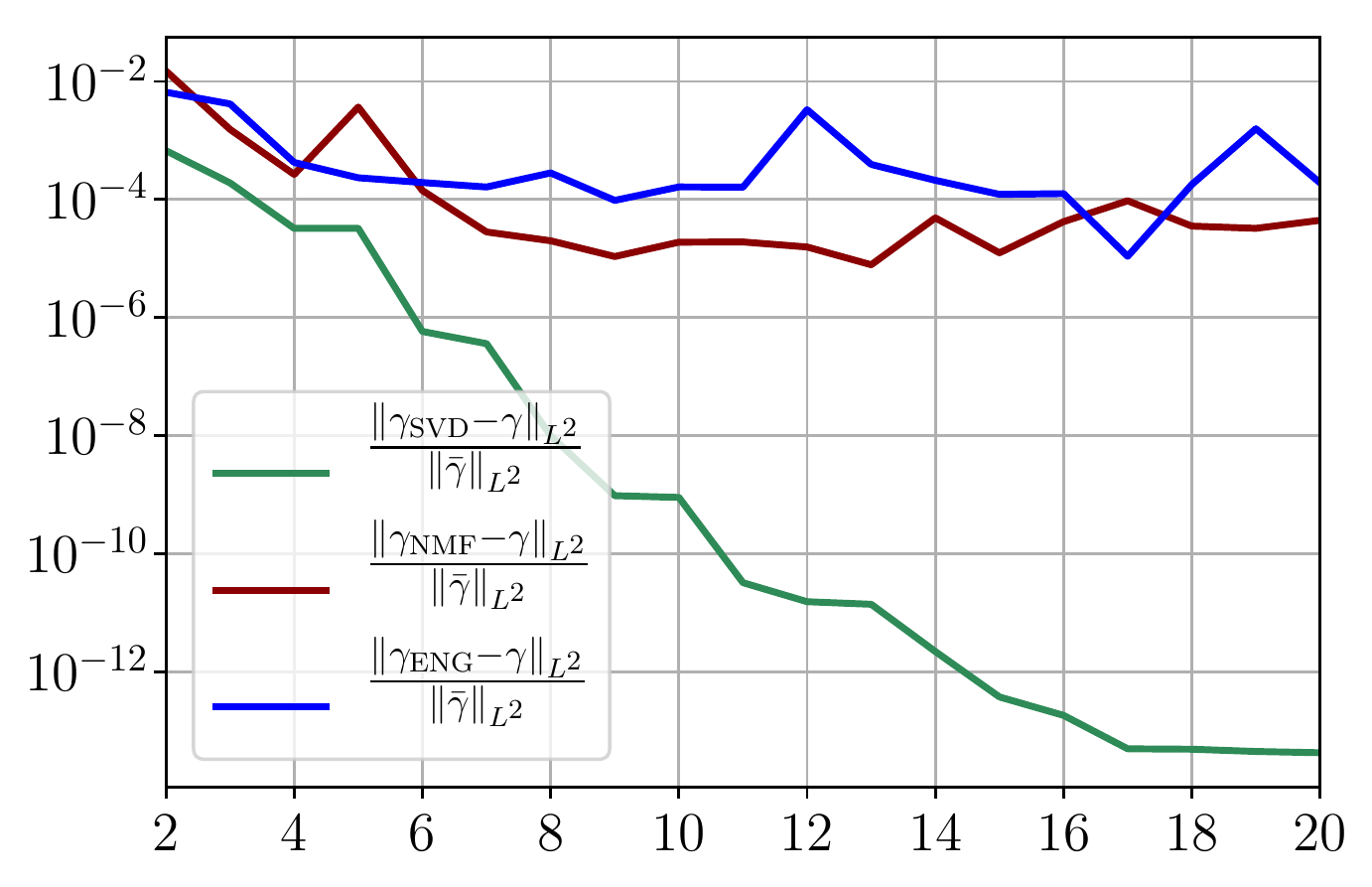}
\caption{$L^2$ relative error  of $\gamma$ vs $n$}
\end{subfigure}
\caption{Study of sampling errors: Fitting errors of functions $\bfbeta$ and $\bfgamma$ from $\cB$ and $\cG$ but which do not belong to the training sets $\cB_\tr$ and $\cG_\tr$.}
\label{fig:fitting_error_b_g}
\end{figure}

Figure \ref{fig:fitting_error_I_R} shows the $L^1$ and $L^\infty$ errors obtained after the propagation of the fittings of $\bfbeta$ and $\bfgamma$. In both metrics, the error for $\bfI$ and $\bfR$ obtained using SVD is decreasing to each  $10^{-14}$. When the NMF and the ENG are used, the error decreases and stagnates respectively at $10^{-3}$ and $10^{-4}$ for both $\bfI$ and $\bfR$.

\begin{figure}[H]
\centering
\begin{subfigure}{.4\textwidth}
\includegraphics[width=1\textwidth]{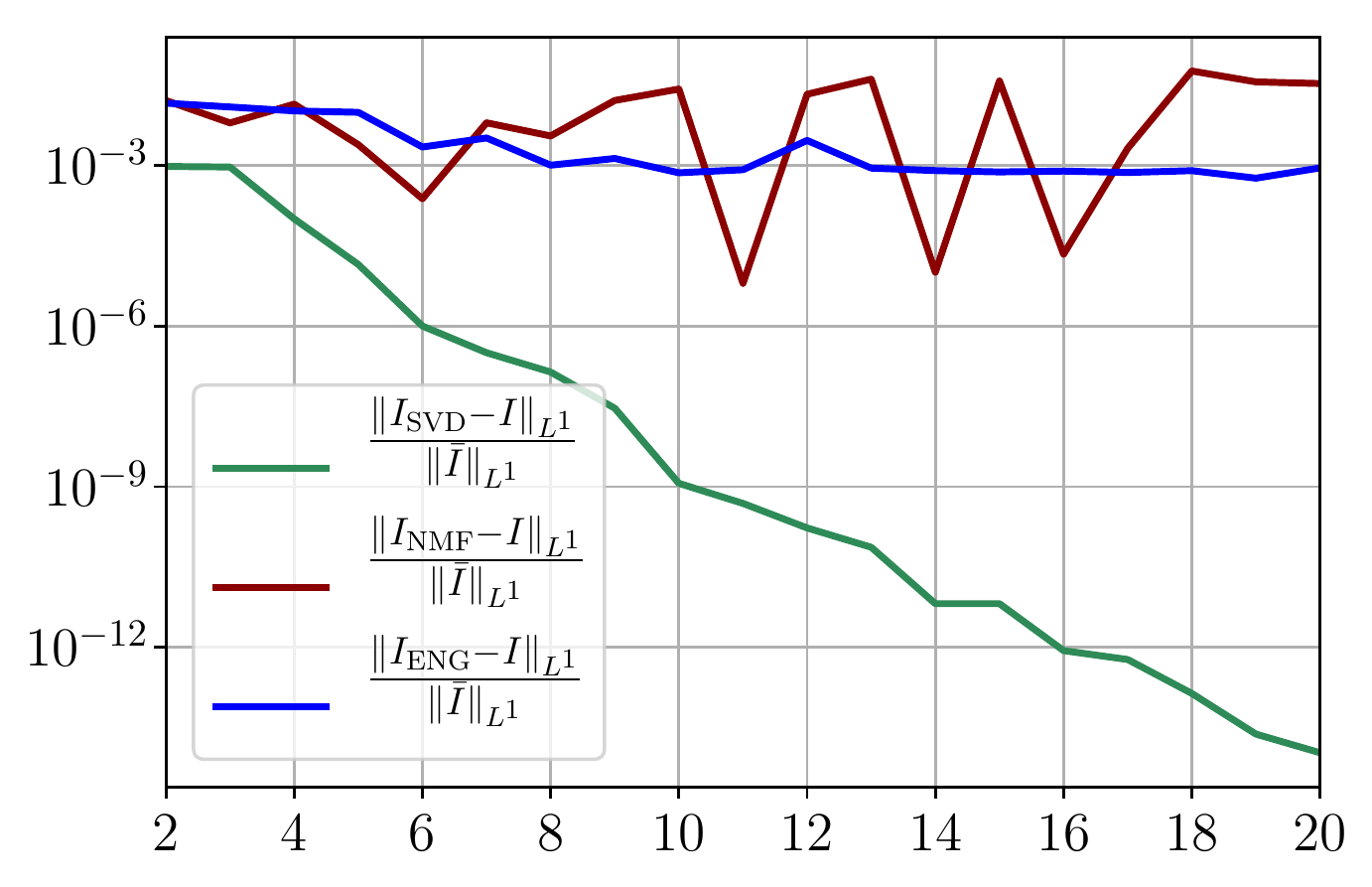}
\caption{$L^1$ relative error  of $I$ vs $n$}
\end{subfigure}
\begin{subfigure}{.4\textwidth}
\includegraphics[width=1\textwidth]{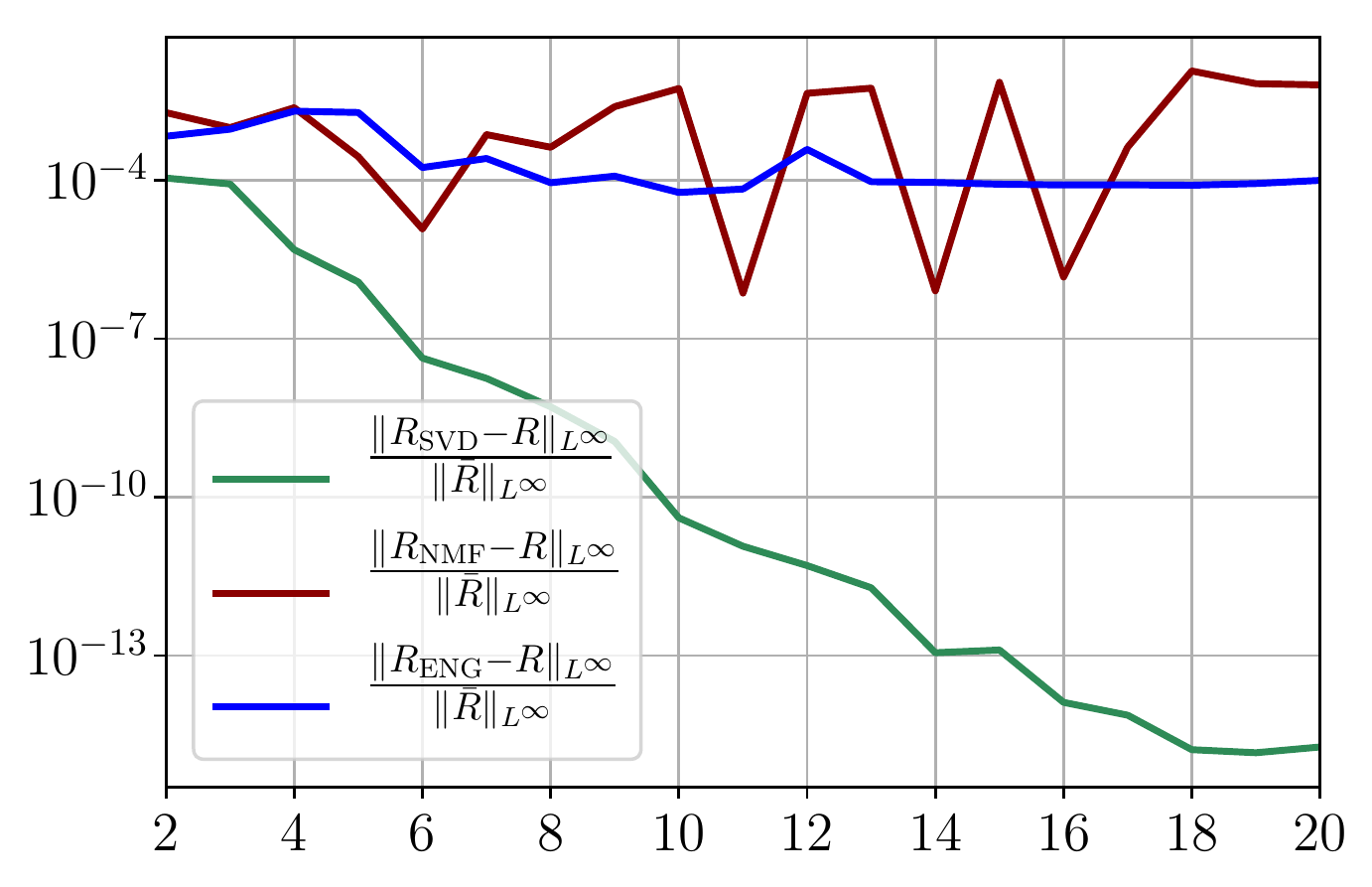}
\caption{$L^\infty$ relative error  of $R$ vs $n$}
\end{subfigure}
\caption{Study of sampling errors: Errors on $\bfI$ and $\bfR$ obtained by a SIR model using the fitted $\bfbeta$ and $\bfgamma$.}
\label{fig:fitting_error_I_R}
\end{figure}

\subsection{Study of the impact of noisy data and intrinsic model error}
To investigate the impact of noise in the observed data, we now add noise to the two previously chosen functions $\bfbeta\in \cB$ and $\bfgamma\in \cG$, and we study the fitting error on this noisy data. The level of the noise has been chosen to be of the same level as the one estimated for the real dynamics. In order to estimate the noise we performed a fitting of the real data using SVD reduced bases, the level of noise is defined as the difference between this fitting and the real data. This level of noise is then added to the virtual scenario considered here. Figure \ref{fig:fitting_error_b_g_noise} shows the fitting error on $\bfbeta$ and $\bfgamma$ using SVD, NMF and ENG reduced bases for $n=2,\dots,20$.
\begin{figure}[H]
\centering
\begin{subfigure}{.4\textwidth}
\includegraphics[width=1\textwidth]{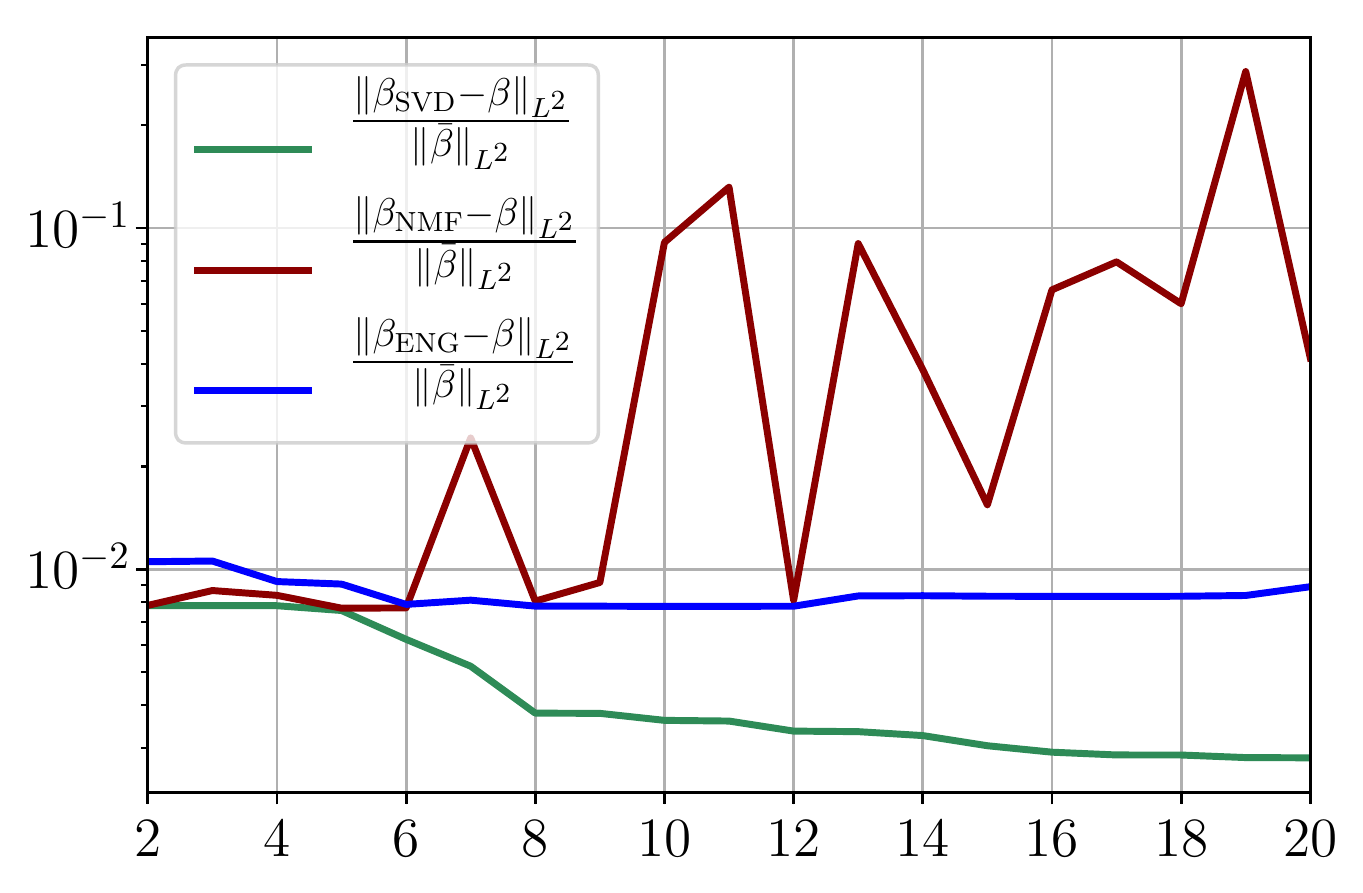}
\caption{$L^2$ relative error  of $\bfbeta$ vs $n$}
\end{subfigure}
\begin{subfigure}{.4\textwidth}
\includegraphics[width=1\textwidth]{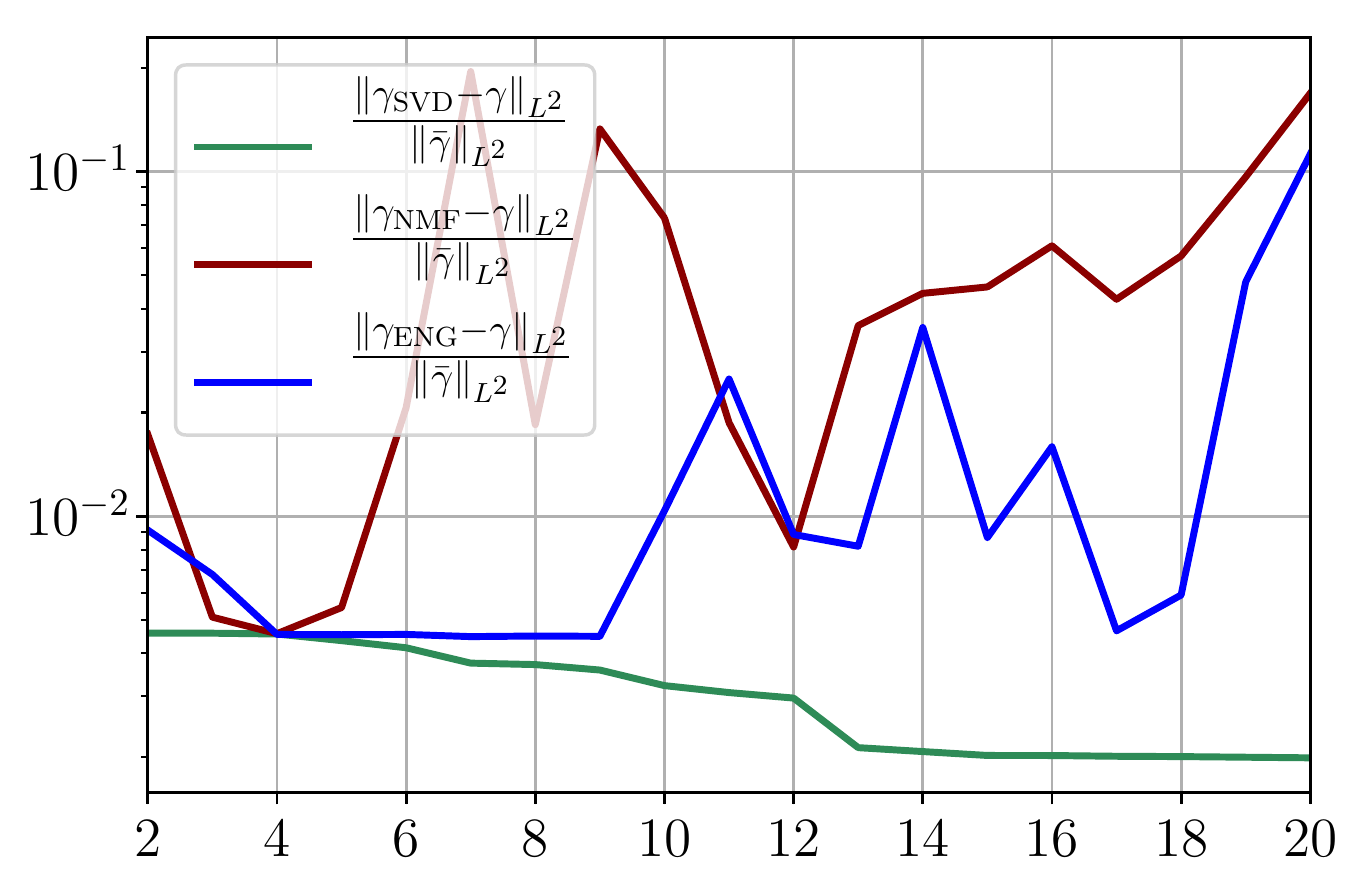}
\caption{$L^2$ relative error  of $\bfgamma$ vs $n$}
\end{subfigure}
\caption{Fitting errors of $\bfbeta$ and $\bfgamma$ on noisy virtual scenario}
\label{fig:fitting_error_b_g_noise}
\end{figure}
We observe that the noise strongly deteriorates  the fitting error obtained using NMF, and the error becomes oscillatory and very unstable. For ENG, the error remains low and consistently around $10^{-2}$ for $\bfbeta$. \col{We observe the same behaviour for $\bfgamma$ with instabilities arising for $n>10$. For SVD the error is lower than in the ENG case and slowly decreases as $n$ increases.}
Figure \ref{fig:fitting_error_I_R_noise} shows the $L^1$ and $L^\infty$ errors obtained after the propagation of the fittings of $\bfbeta$ and $\bfgamma$. In line with the observations from Figure \ref{fig:fitting_error_b_g_noise} the error obtained using NMF is very unstable. \col{Using ENG, we observe a decay from $n=2$ to $n=7$ and oscillations that remain around $10^{-2}$ for $I$ and $10^{-3}$ for $R$. The decay observed for SVD is slow and steady, the error nearly reaches $10^{-4}$ for $I$ and $10^{-5}$ for $R$ when $n=20$.}
\begin{figure}[H]
\centering
\begin{subfigure}{.4\textwidth}
\includegraphics[width=1\textwidth]{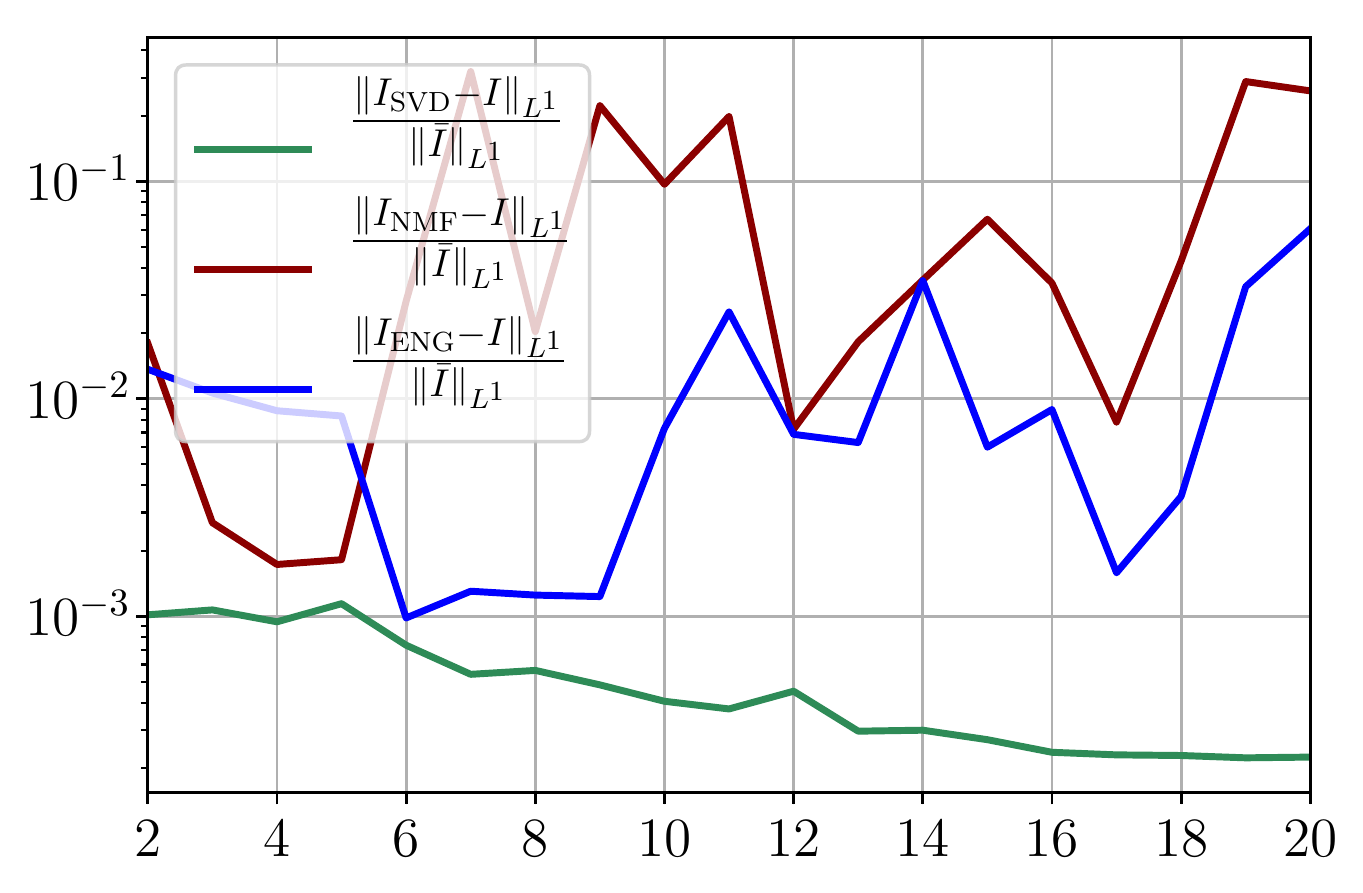}
\caption{$L^1$ relative error  of $I$ vs $n$}
\end{subfigure}
\begin{subfigure}{.4\textwidth}
\includegraphics[width=1\textwidth]{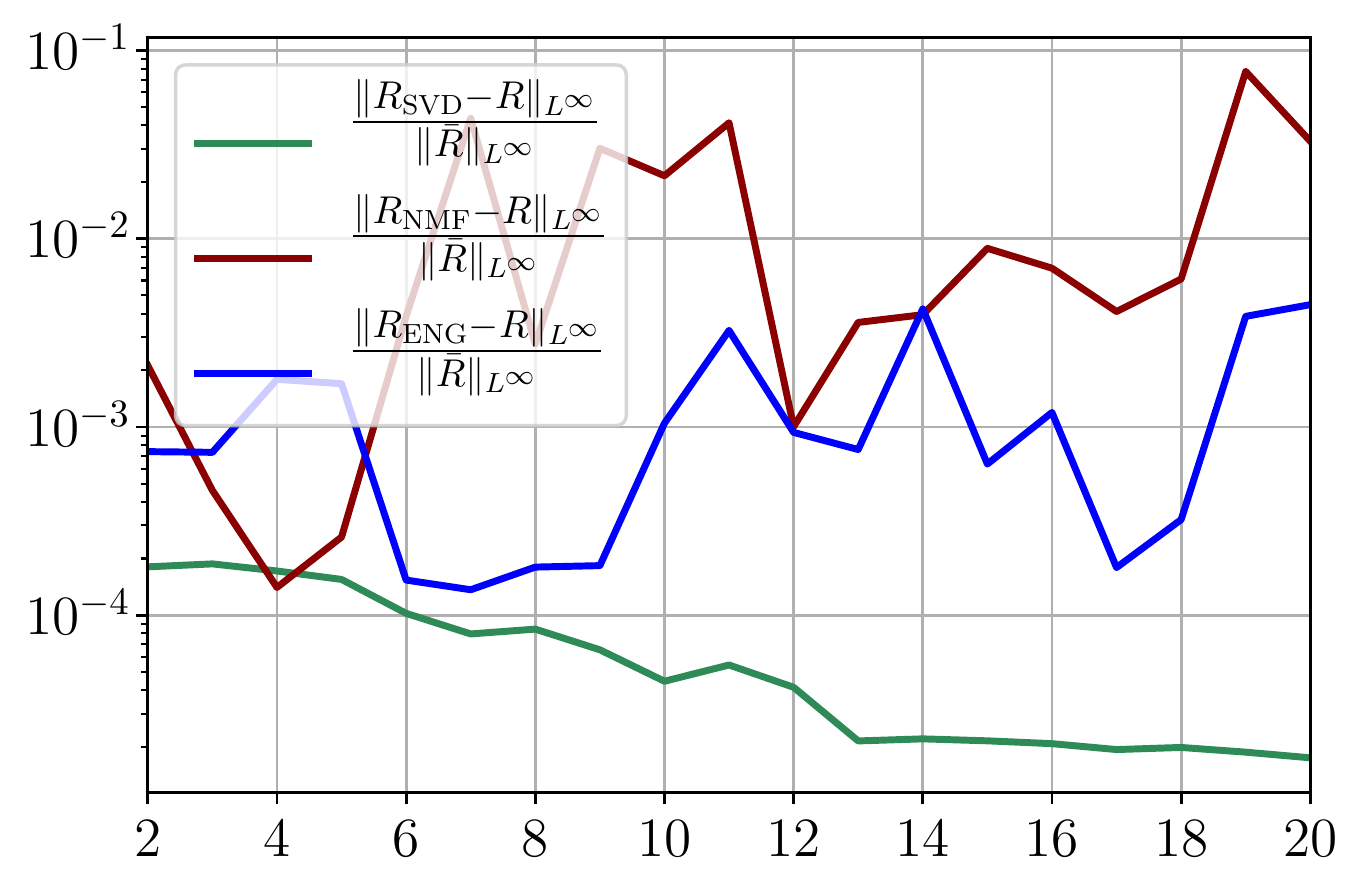}
\caption{$L^\infty$ relative error  of $R$ vs $n$}
\end{subfigure}
\caption{Fitting errors of $I$ and $R$ on noisy data}
\label{fig:fitting_error_I_R_noise}
\end{figure}

Finally, it remains to add intrinsic model error on top of the previous sampling error and observation noise. If we do so, the {main} change is that the previous fitting error plots from Figure \ref{fig:fitting_error_b_g_noise} have essentially the same behavior but the error values are increased depending on the degree of model inaccuracy.
\col{
We have therefore disentangled all the effects of model error and noisy data, and all the observations from this section give thus a better insight of the fitting on the real data.

Figure \ref{fig:fitting_error_1} summarizes the fitting results for an example fitting period taken from $t_0=19/03$ to $T=03/05$. Figures \ref{fig:fitting_error_b_g_real_beta} and \ref{fig:fitting_error_b_g_real_gamma} show the fitting error on $\bfbeta$ and $\bfgamma$ using SVD, NMF and ENG reduced bases for $n=2,\dots,20$. Figures \ref{fig:fitting_error_b_g_real_I} and \ref{fig:fitting_error_b_g_real_R} show the $L^1$ and $L^\infty$ relative errors on $\bfI_n$ and $\bfR_n$ after the propagation of the fittings of $\beta^*_n$ and $\gamma^*_n$.

\begin{figure}[h!]
\centering
\begin{subfigure}{.43\textwidth}
\includegraphics[width=1\textwidth]{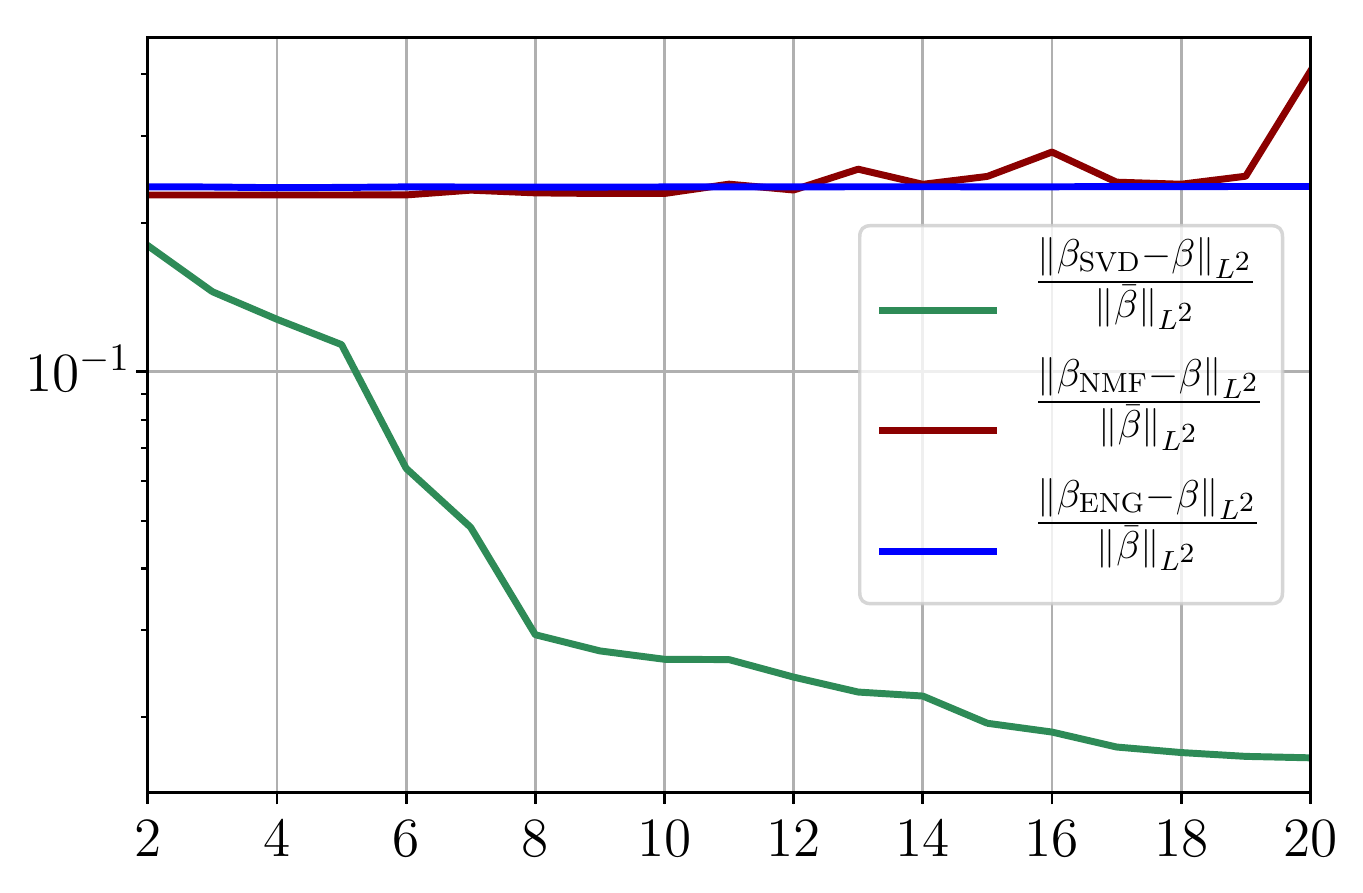}
\caption{$L^2$ relative error  of $\bfbeta$ vs $n$}
\label{fig:fitting_error_b_g_real_beta}
\end{subfigure}
\begin{subfigure}{.43\textwidth}
\includegraphics[width=1\textwidth]{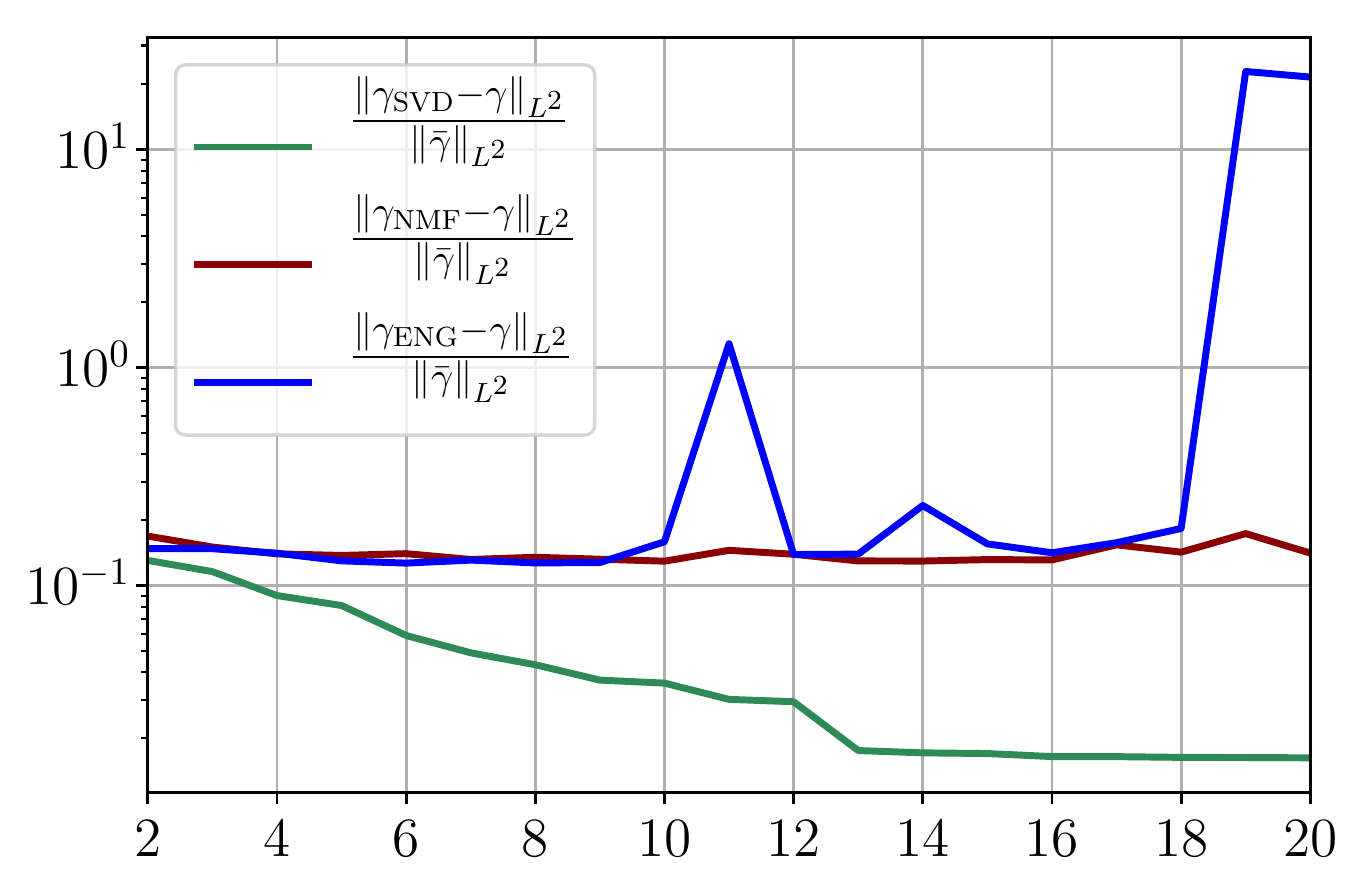}
\caption{$L^2$ relative error  of $\bfgamma$ vs $n$}
\label{fig:fitting_error_b_g_real_gamma}
\end{subfigure}
\par\bigskip 
\begin{subfigure}{.43\textwidth}
\includegraphics[width=1\textwidth]{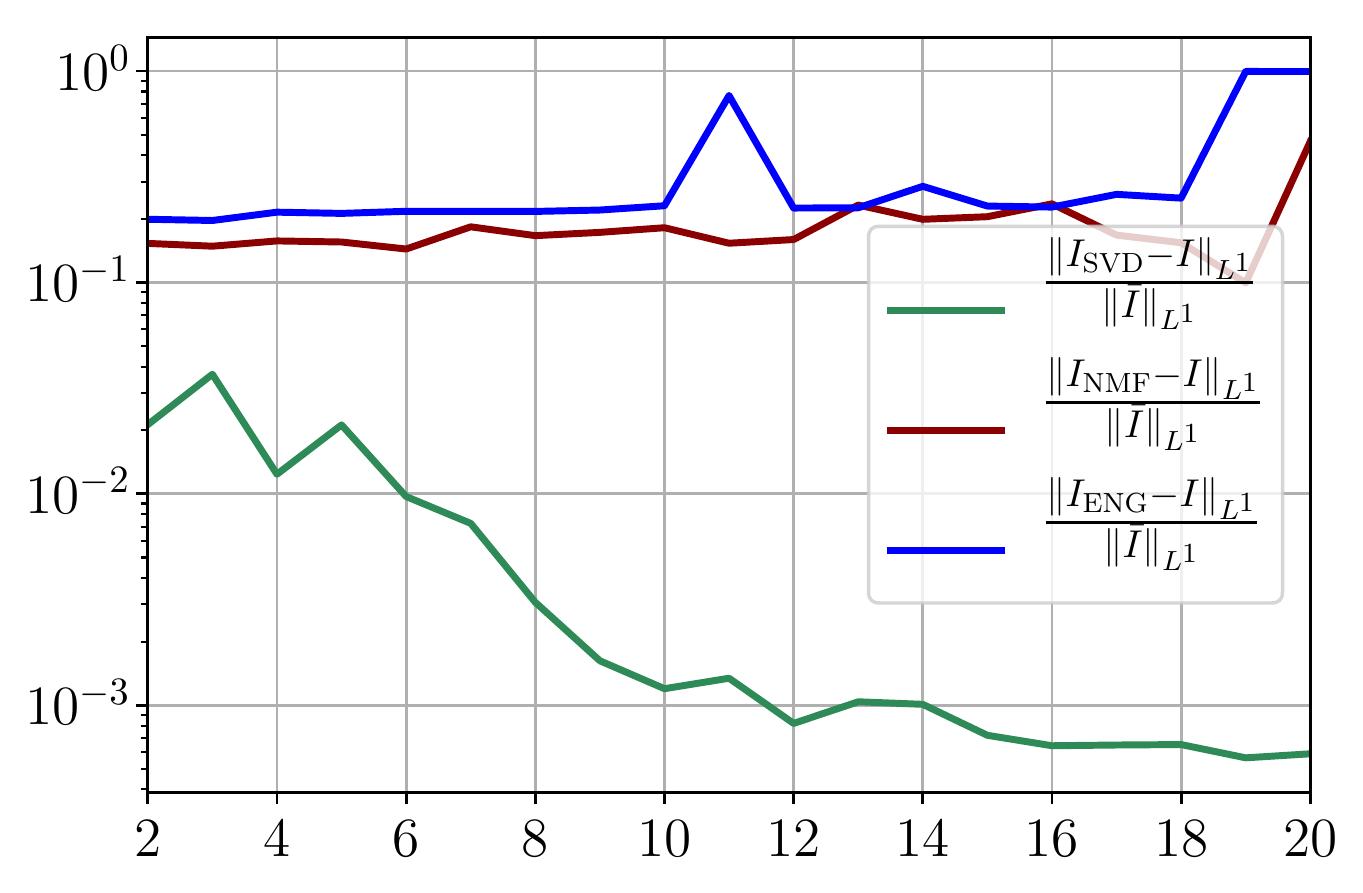}
\caption{$L^1$ relative error  of $I$ vs $n$}
\label{fig:fitting_error_b_g_real_I}
\end{subfigure}
\begin{subfigure}{.43\textwidth}
\includegraphics[width=1\textwidth]{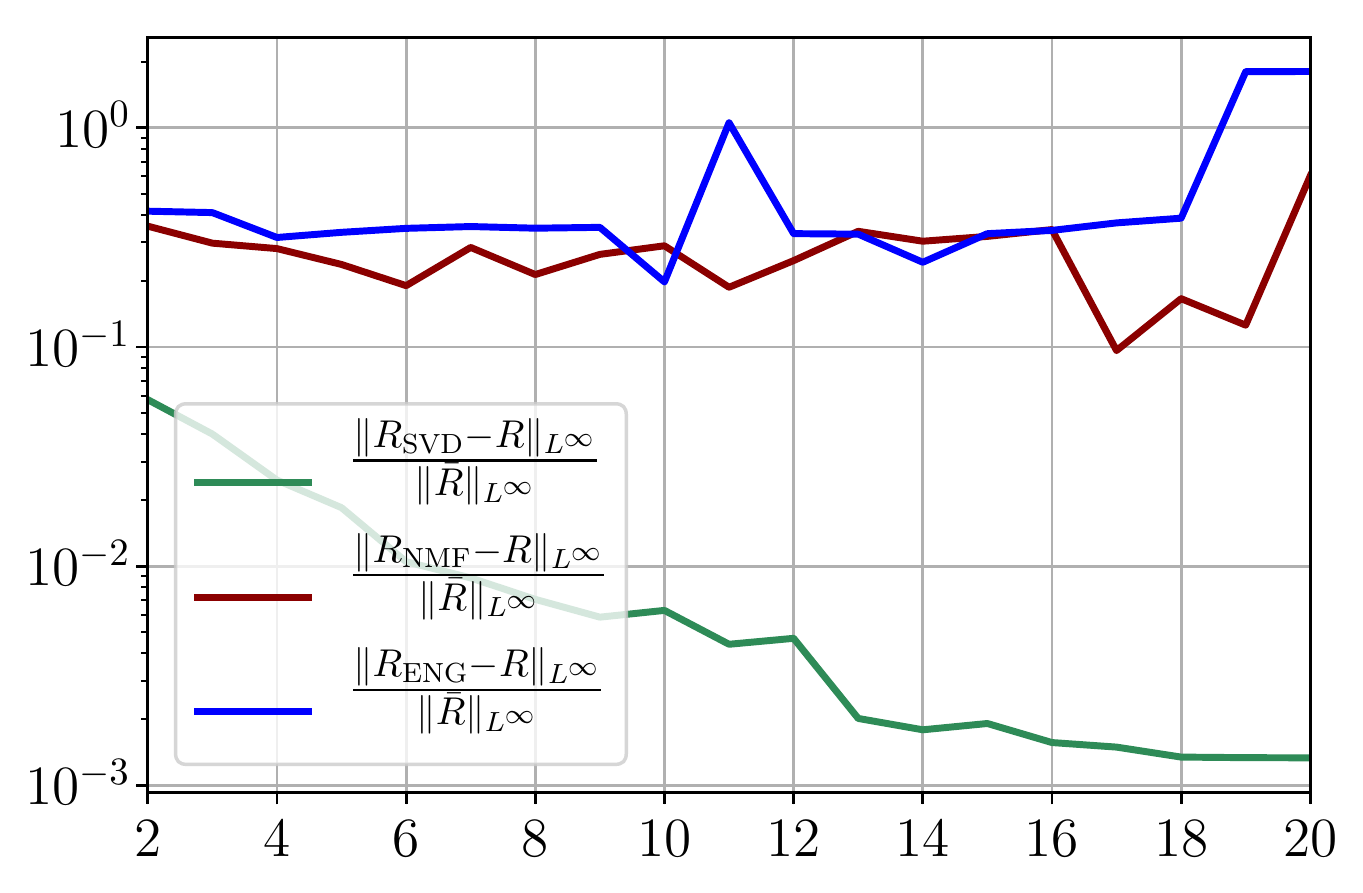}
\caption{$L^\infty$ relative error  of $R$ vs $n$}
\label{fig:fitting_error_b_g_real_R}
\end{subfigure}
\caption{\textbf{Fitting errors from $t_0=19/03/2020$ to $T=03/05/2020$}}
\label{fig:fitting_error_1}
\end{figure}

From figures \ref{fig:fitting_error_b_g_real_beta} and \ref{fig:fitting_error_b_g_real_gamma}, we observe that the fitting error with SVD decreases at a moderate rate as the dimension $n$ of the reduced basis is increased. The error with NMF and ENG does not decrease and oscillates around certain constant error value of order $10^{-1}$. Note that this value is small and yields to small errors in the approximation of $\bfI$ and $\bfR$ as figures \ref{fig:fitting_error_b_g_real_I} and \ref{fig:fitting_error_b_g_real_R} illustrate. 
}

\section{Study of forecasting errors}
\label{appendix:forecasting-error}
In this section we give a detailed study of the forecasting errors for each different fitting strategy (\fitIR, \fitbg), reduced model (SVD, NMF, ENG) and starting date $T$. We anticipate the main conclusion which we already announced in Section \ref{sec:forecast-compare}: ENG fitted with \fitbg~outperforms the other reduced models and is the most robust and accurate reduced model to use in a forecasting strategy. \col{Figure \ref{fig:forecast_error_0104} shows the relative errors of a 14 days forecast from $T=01/04$ for each forecasting method and each reduced basis. Similarly figures \ref{fig:forecast_error_0304} to \ref{fig:forecast_error_0505} show forecasts relative errors respectively from different times.}

\begin{figure}[H]
\centering
\begin{subfigure}{.45\textwidth}
\includegraphics[width=1\textwidth]{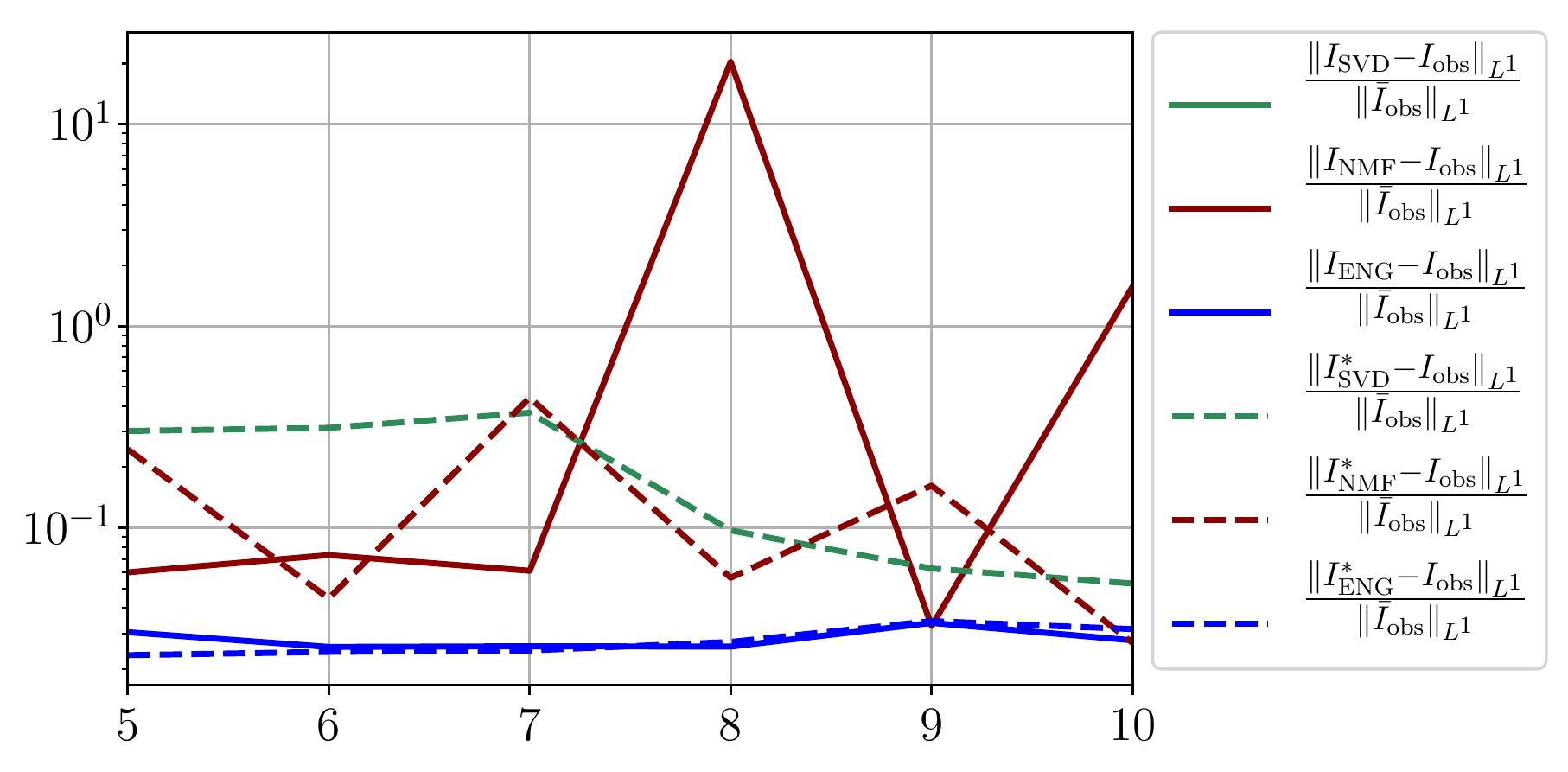}
\caption{$L^1$ relative error  of $I$ vs $n$}
\end{subfigure}
\begin{subfigure}{.45\textwidth}
\includegraphics[width=1\textwidth]{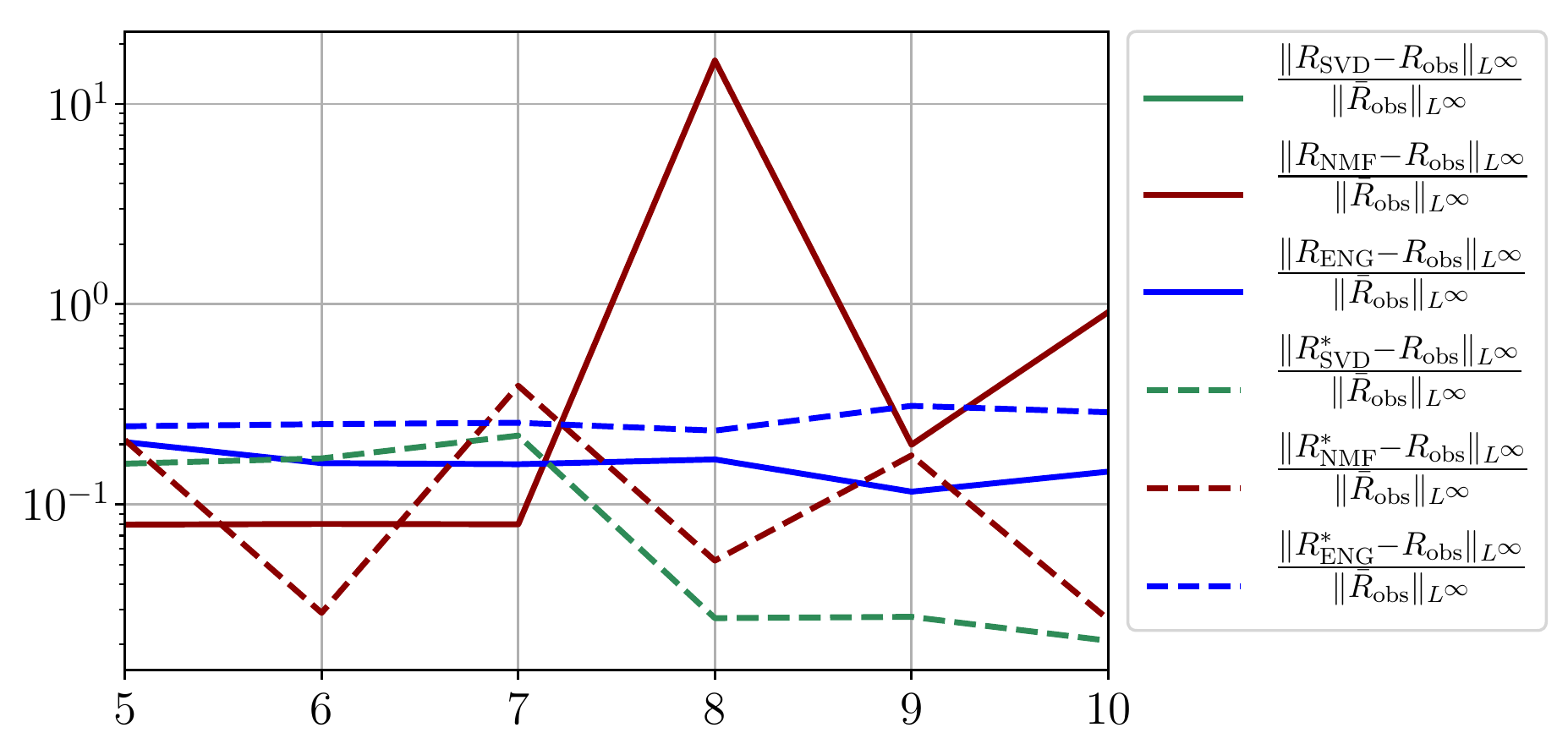}
\caption{$L^\infty$ relative error  of $R$ vs $n$}
\end{subfigure}
\caption{Forecasting errors of $I$ and $R$ (from $T=01/04$)}
\label{fig:forecast_error_0104}
\end{figure}

\begin{figure}[H]
\centering
\begin{subfigure}{.45\textwidth}
\includegraphics[width=1\textwidth]{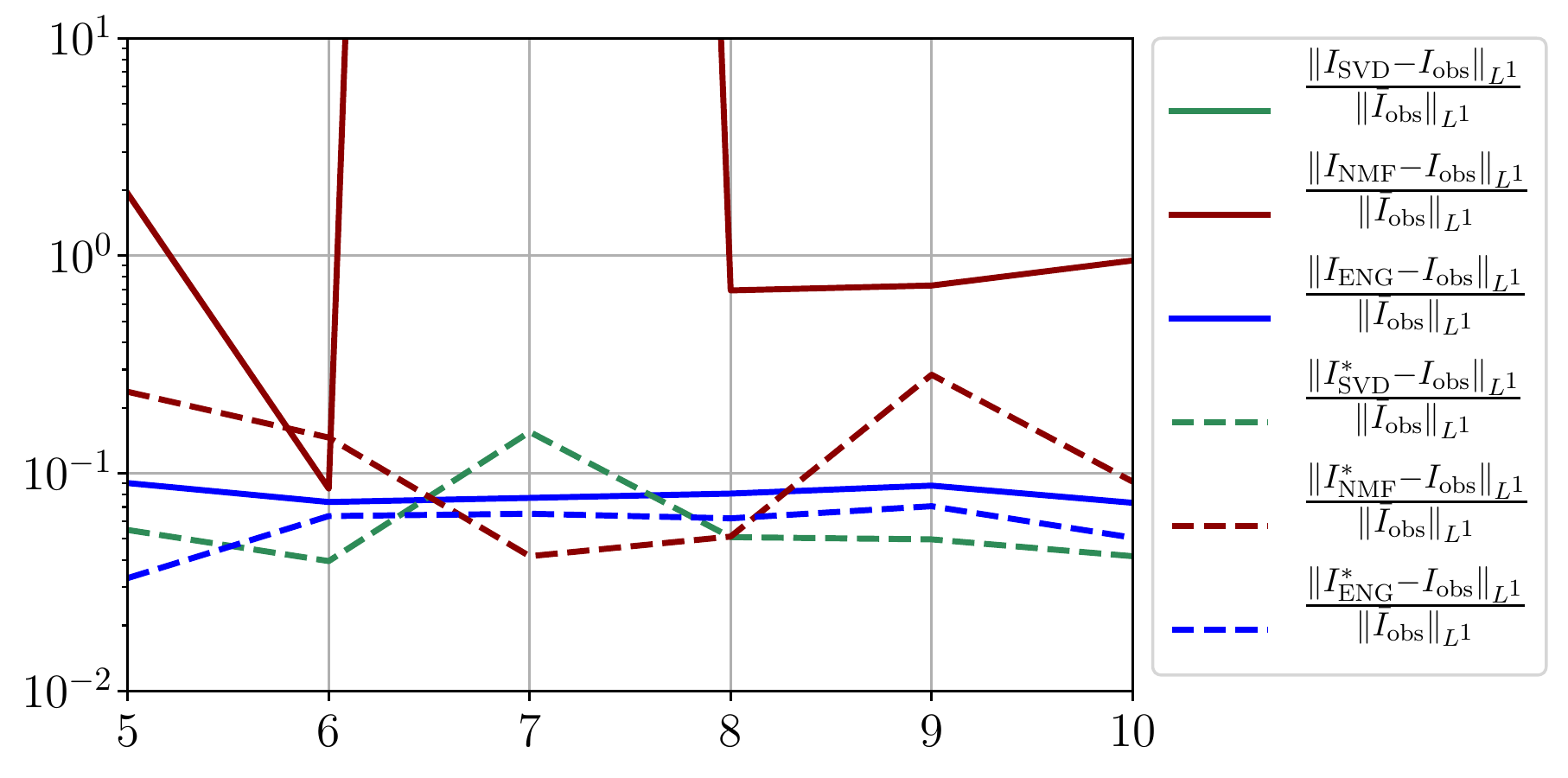}
\caption{$L^1$ relative error  of $I$ vs $n$}
\end{subfigure}
\begin{subfigure}{.45\textwidth}
\includegraphics[width=1\textwidth]{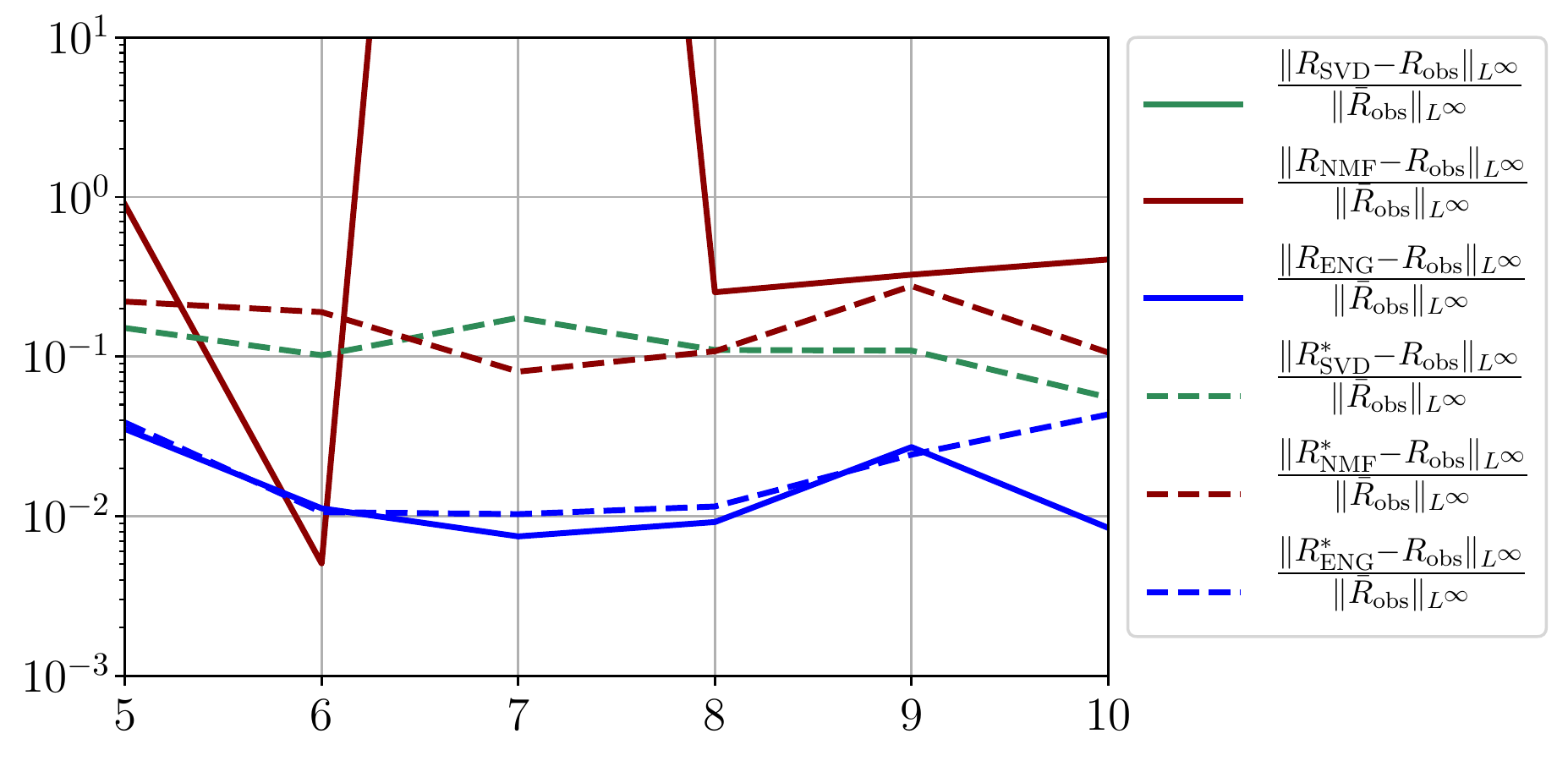}
\caption{$L^\infty$ relative error  of $R$ vs $n$}
\end{subfigure}
\caption{Forecasting errors of $I$ and $R$ (from $T=03/04$)}
\label{fig:forecast_error_0304}
\end{figure}

\begin{figure}[H]
\centering
\begin{subfigure}{.45\textwidth}
\includegraphics[width=1\textwidth]{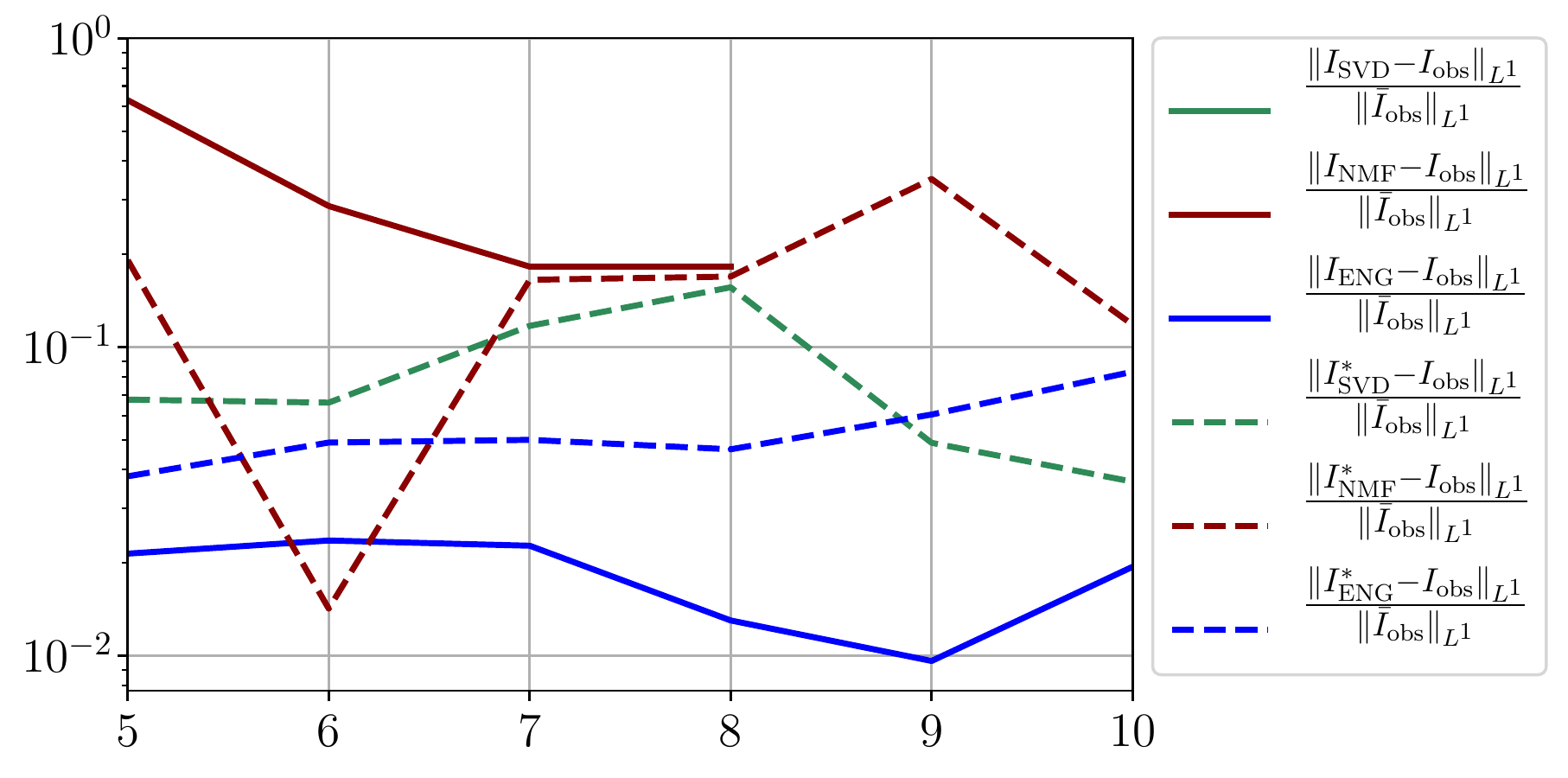}
\caption{$L^1$ relative error  of $I$ vs $n$}
\end{subfigure}
\begin{subfigure}{.45\textwidth}
\includegraphics[width=1\textwidth]{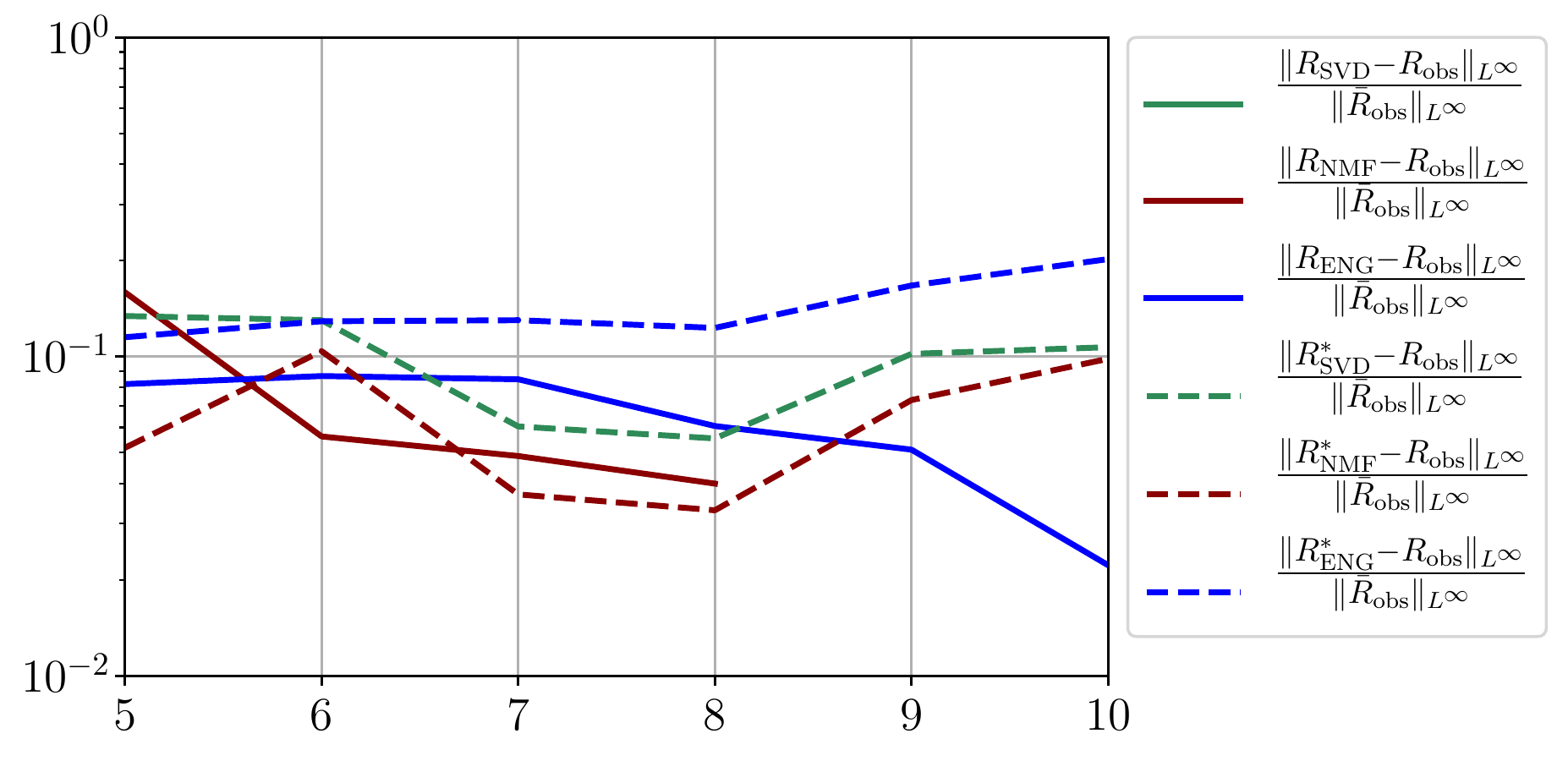}
\caption{$L^\infty$ relative error  of $R$ vs $n$}
\end{subfigure}
\caption{Forecasting errors of $I$ and $R$ (from $T=05/04$)}
\label{fig:forecast_error_0504}
\end{figure}

\begin{figure}[H]
\centering
\begin{subfigure}{.45\textwidth}
\includegraphics[width=1\textwidth]{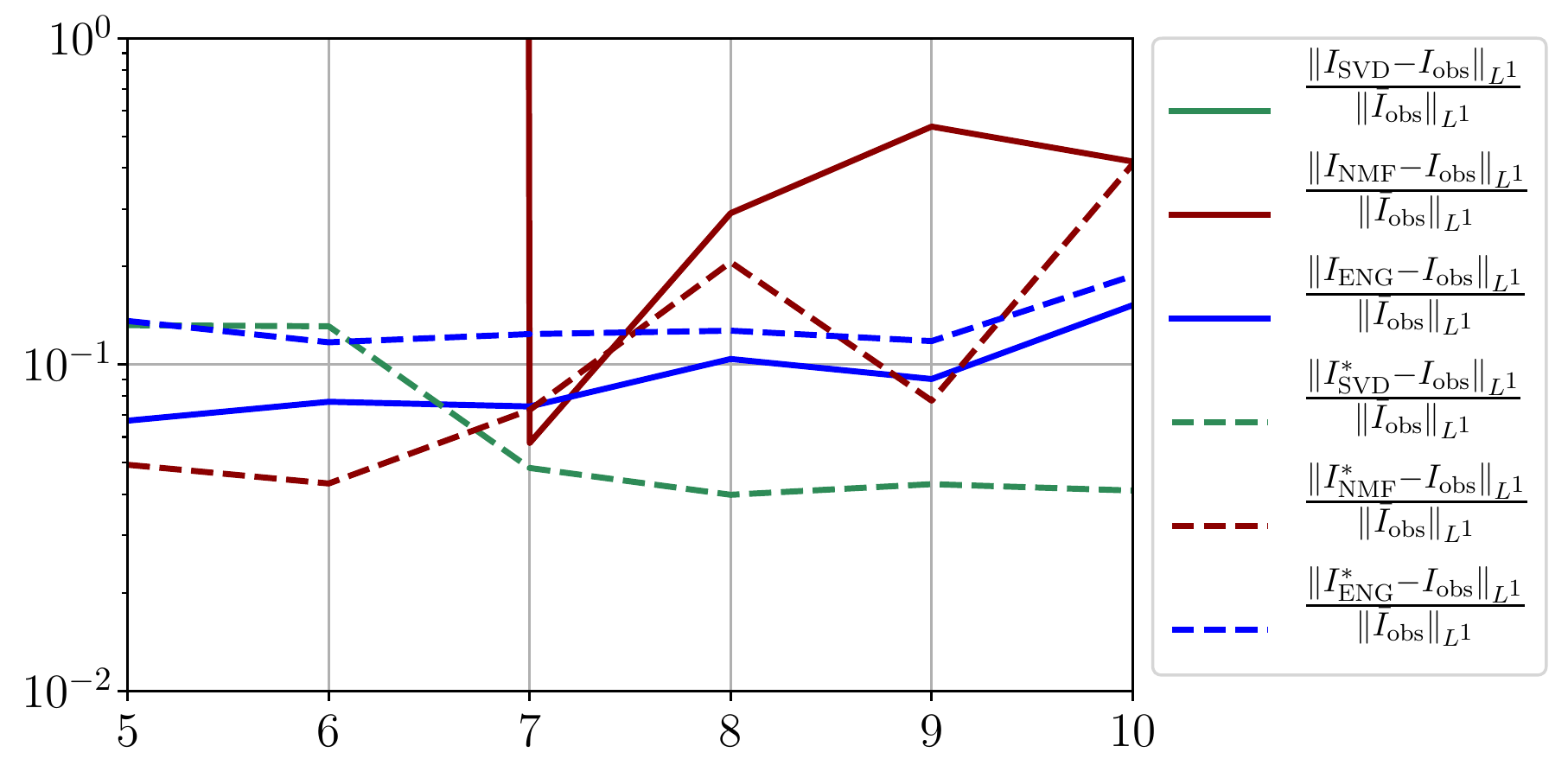}
\caption{$L^1$ relative error  of $I$ vs $n$}
\end{subfigure}
\begin{subfigure}{.45\textwidth}
\includegraphics[width=1\textwidth]{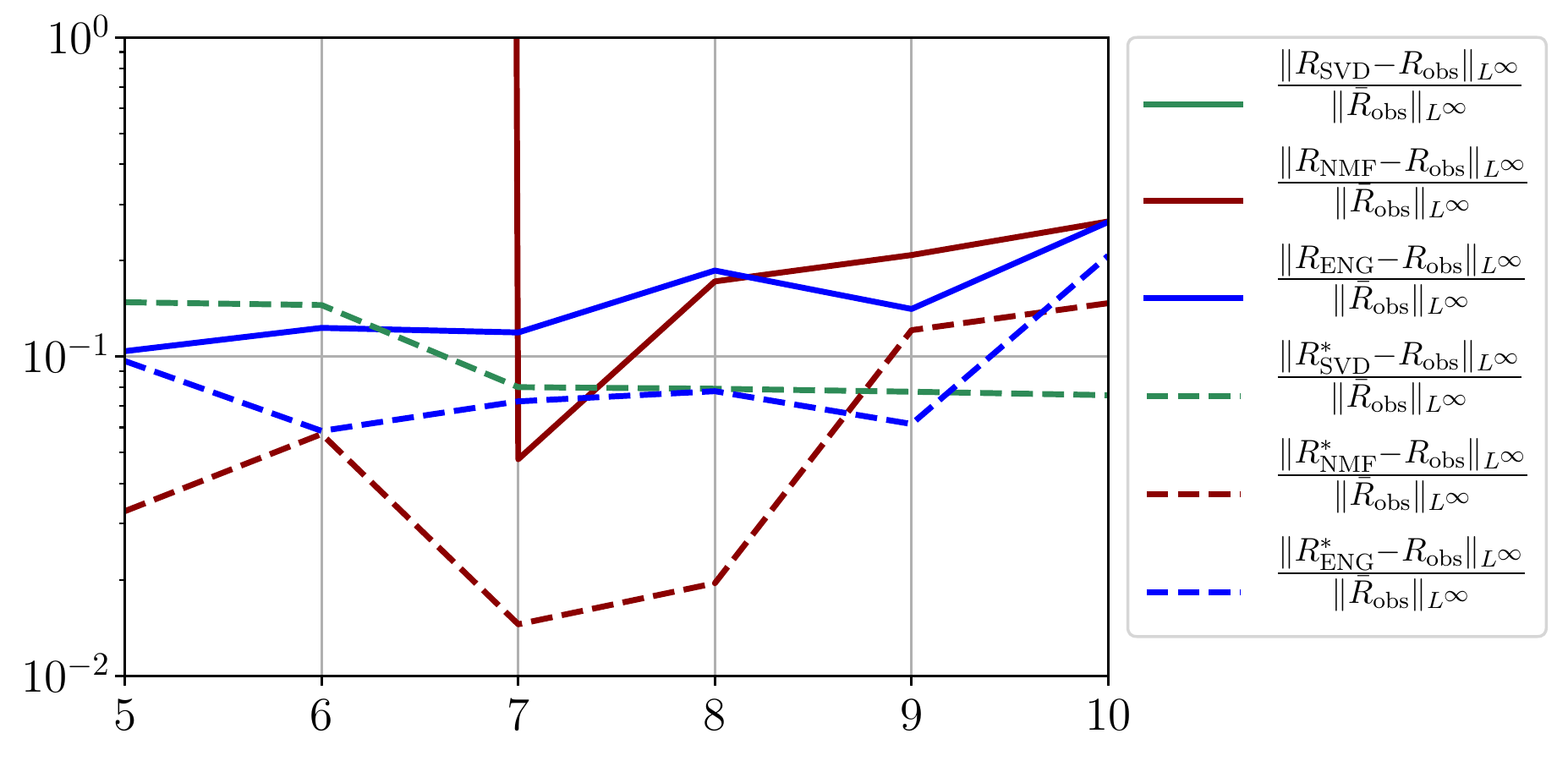}
\caption{$L^\infty$ relative error  of $R$ vs $n$}
\end{subfigure}
\caption{Forecasting errors of $I$ and $R$ (from $T=07/04$)}
\label{fig:forecast_error_0704}
\end{figure}

\begin{figure}[H]
\centering
\begin{subfigure}{.45\textwidth}
\includegraphics[width=1\textwidth]{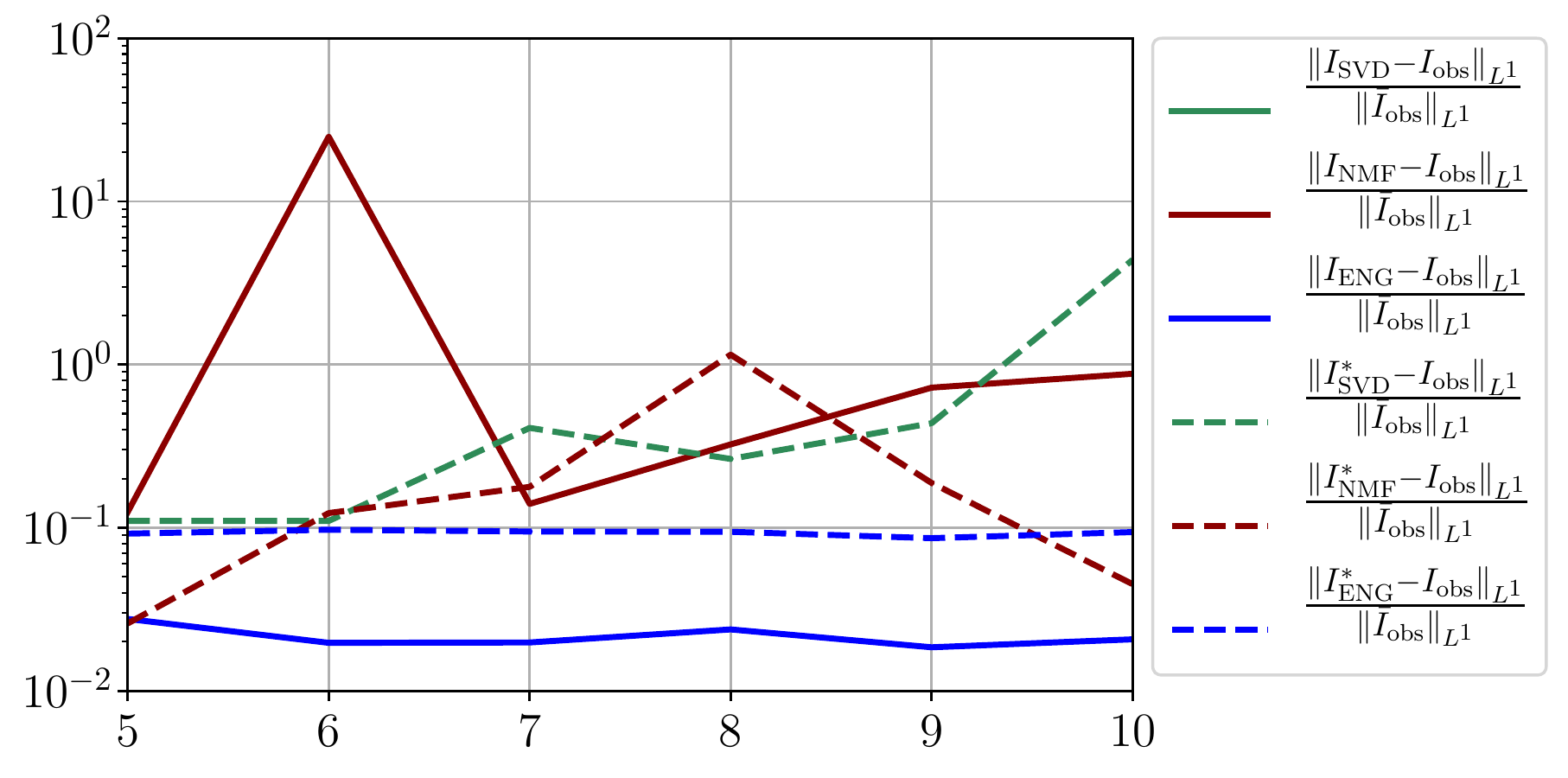}
\caption{$L^1$ relative error  of $I$ vs $n$}
\end{subfigure}
\begin{subfigure}{.45\textwidth}
\includegraphics[width=1\textwidth]{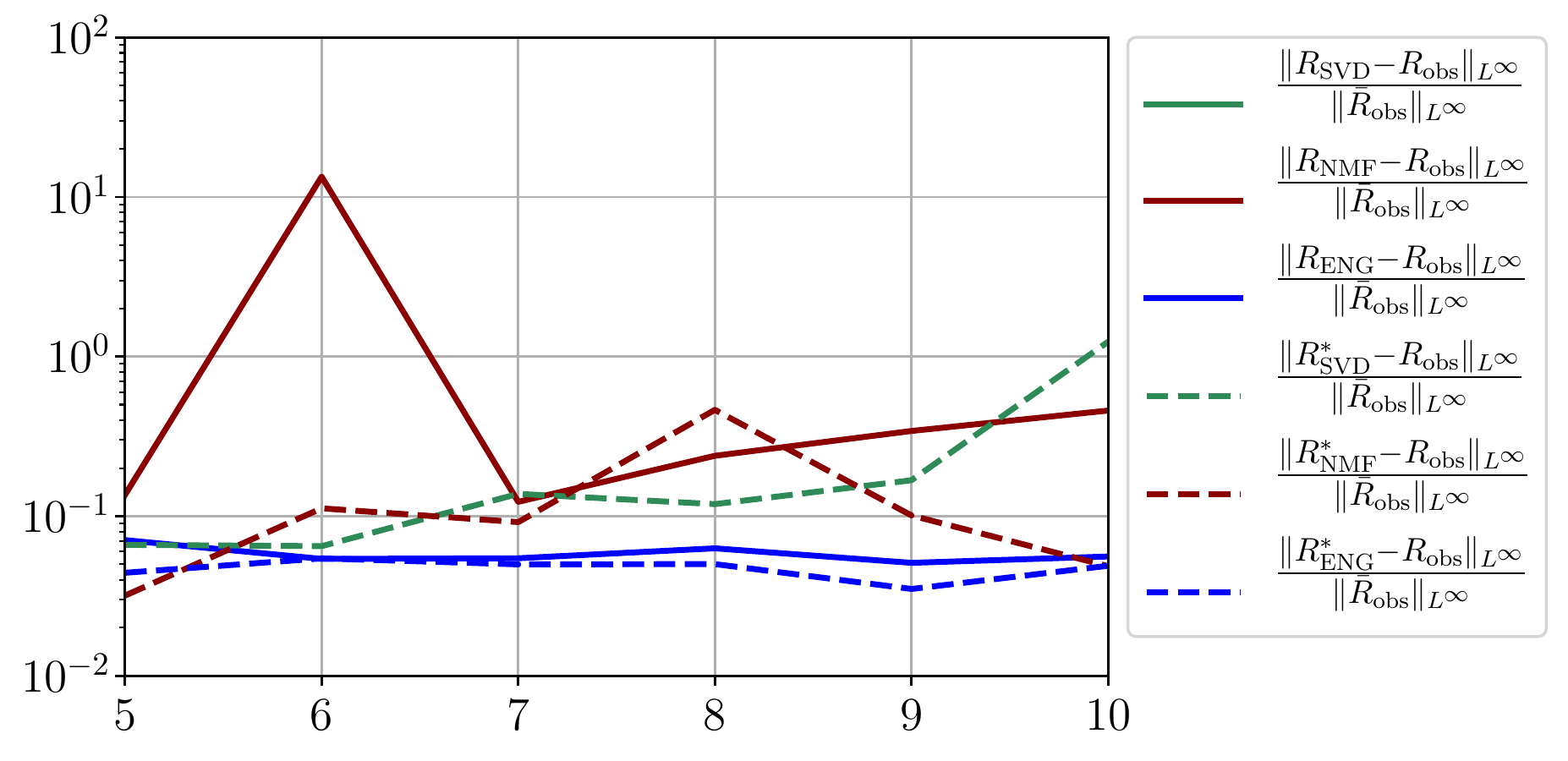}
\caption{$L^\infty$ relative error  of $R$ vs $n$}
\end{subfigure}
\caption{Forecasting errors of $I$ and $R$ (from $T=09/04$)}
\label{fig:forecast_error_0904}
\end{figure}

\begin{figure}[H]
\centering
\begin{subfigure}{.45\textwidth}
\includegraphics[width=1\textwidth]{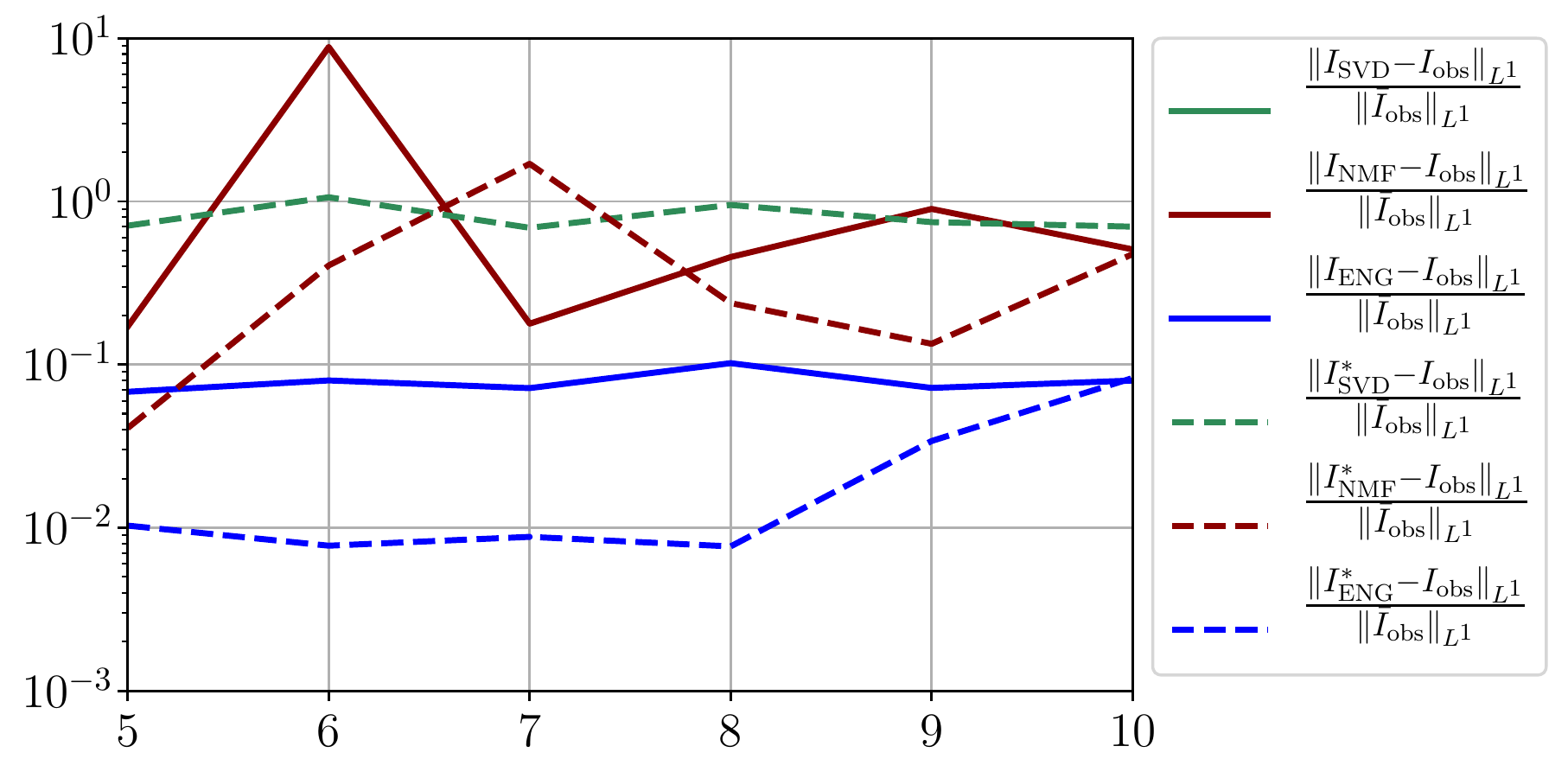}
\caption{$L^1$ relative error  of $I$ vs $n$}
\end{subfigure}
\begin{subfigure}{.45\textwidth}
\includegraphics[width=1\textwidth]{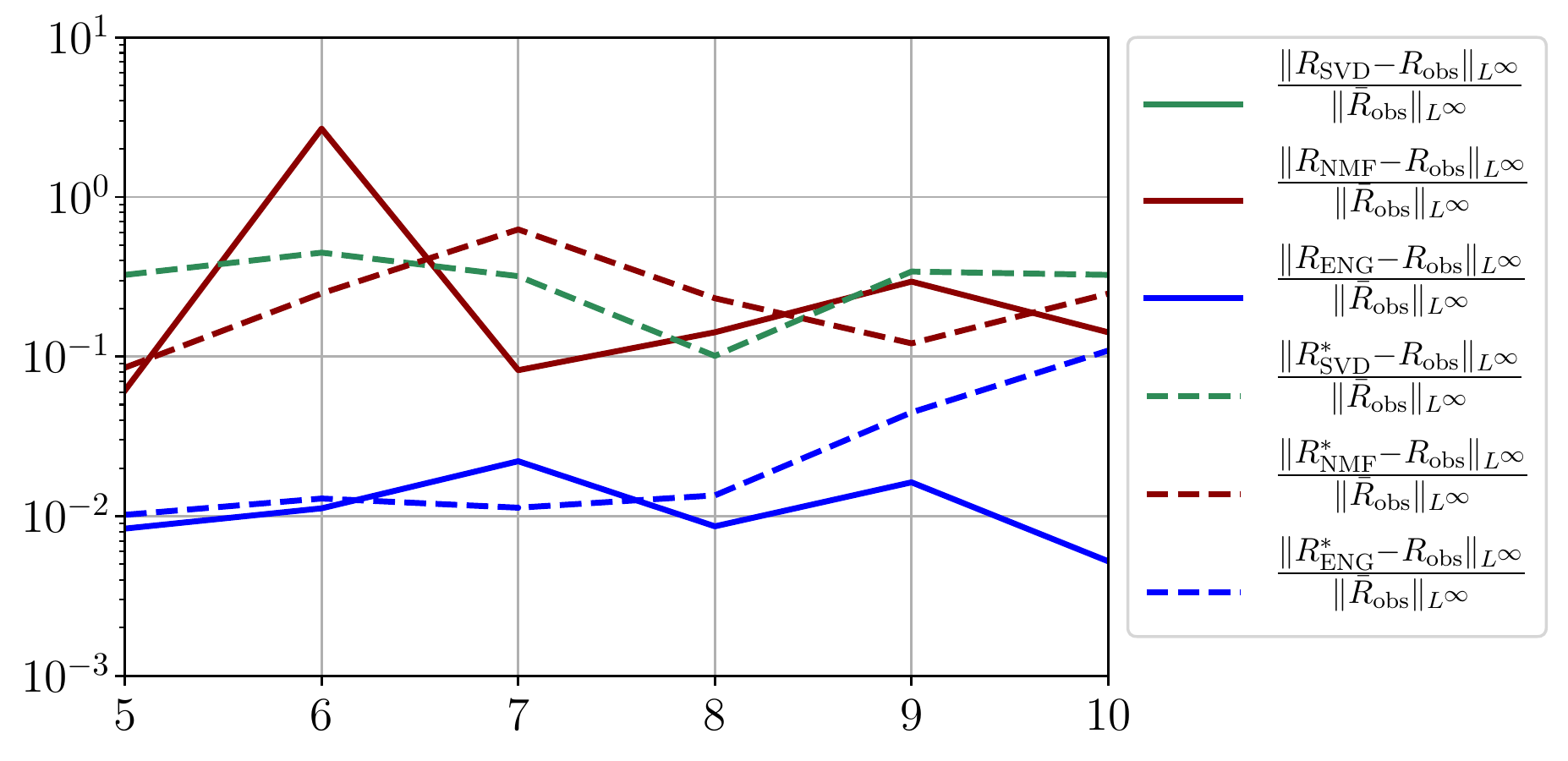}
\caption{$L^\infty$ relative error  of $R$ vs $n$}
\end{subfigure}
\caption{Forecasting errors of $I$ and $R$ (from $T=11/04$)}
\label{fig:forecast_error_1104}
\end{figure}

\begin{figure}[H]
\centering
\begin{subfigure}{.45\textwidth}
\includegraphics[width=1\textwidth]{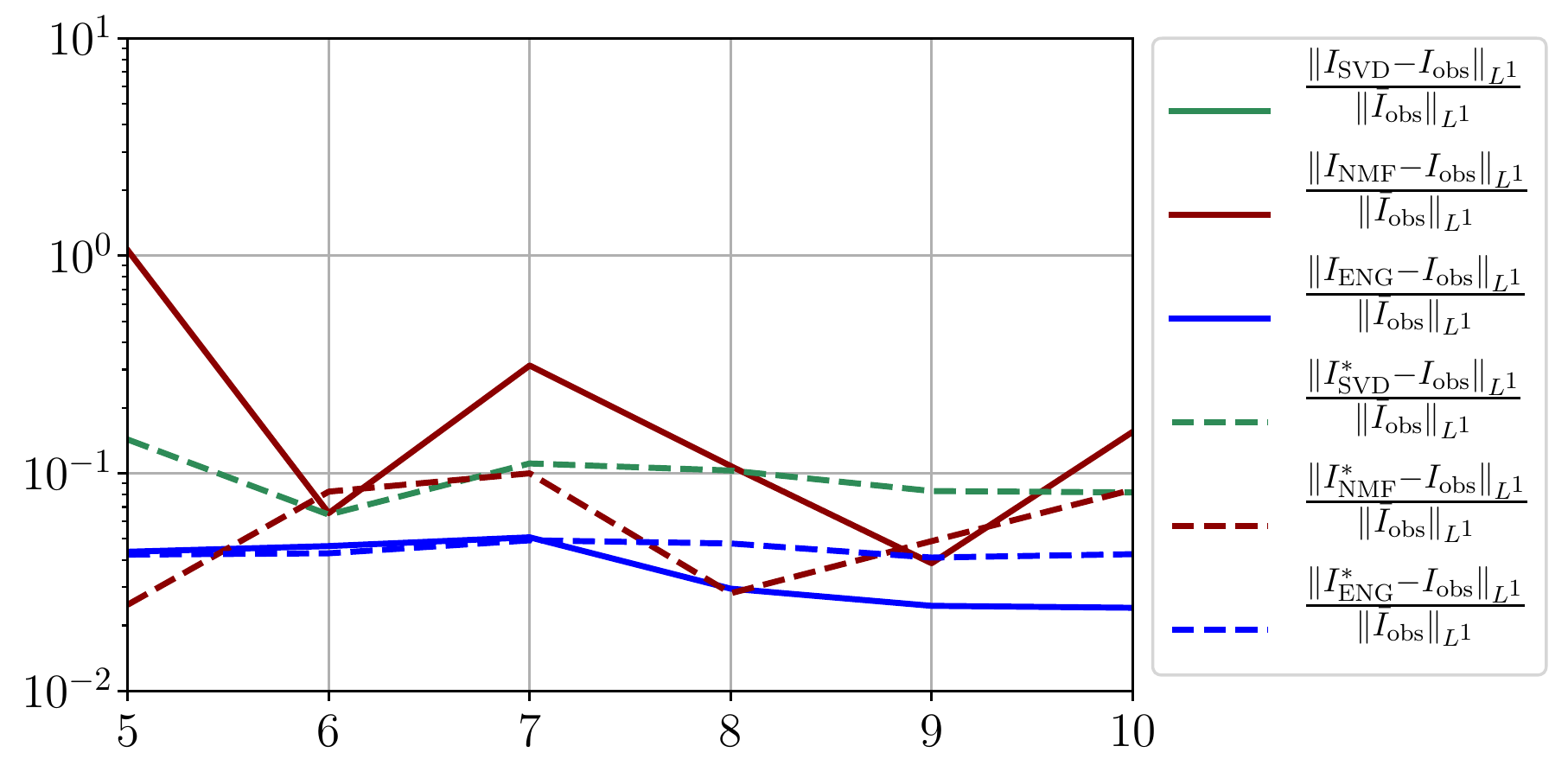}
\caption{$L^1$ relative error  of $I$ vs $n$}
\end{subfigure}
\begin{subfigure}{.45\textwidth}
\includegraphics[width=1\textwidth]{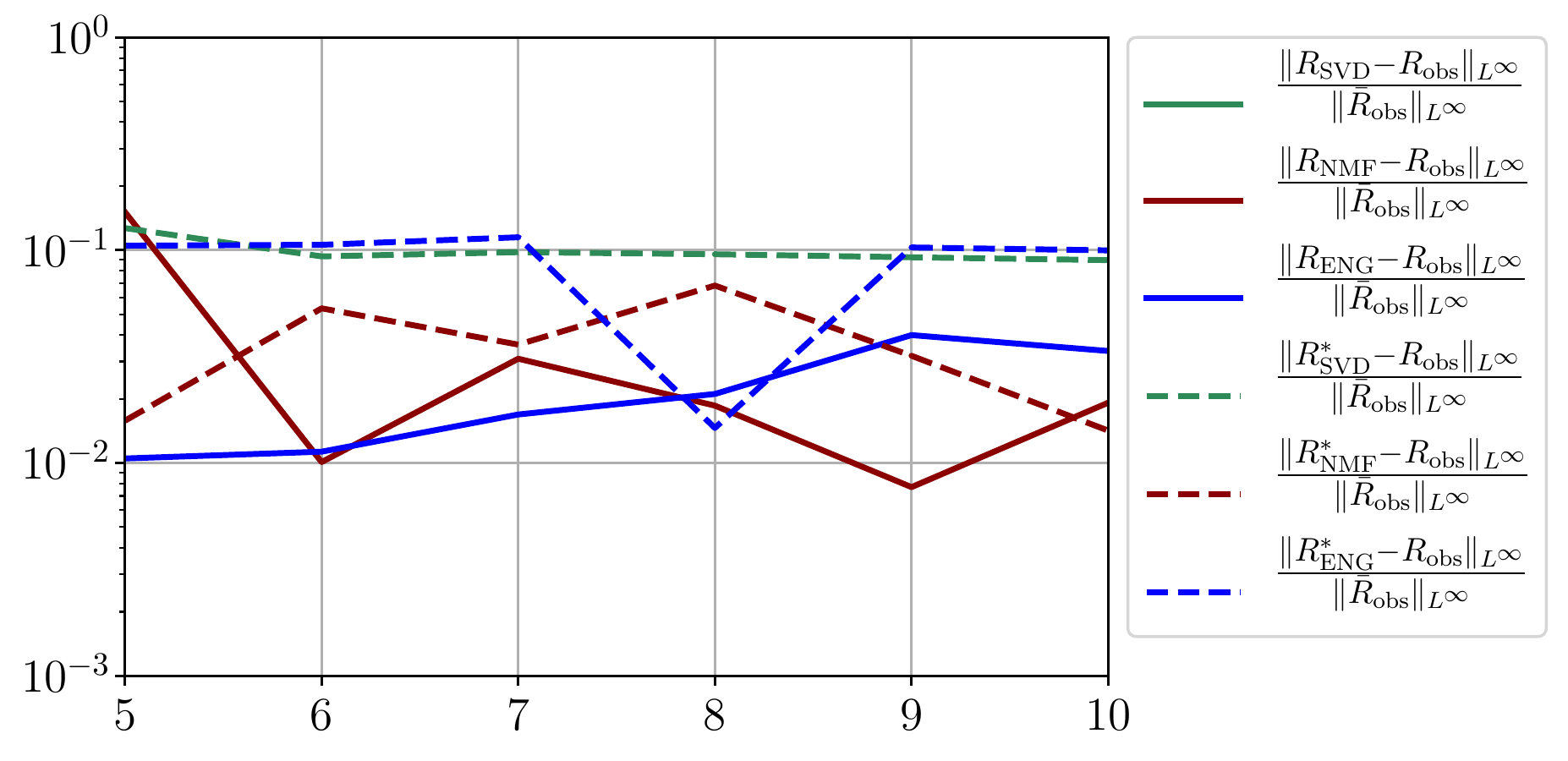}
\caption{$L^\infty$ relative error  of $R$ vs $n$}
\end{subfigure}
\caption{Forecasting errors of $I$ and $R$ (from $T=15/04$)}
\label{fig:forecast_error_1504}
\end{figure}

\begin{figure}[H]
\centering
\begin{subfigure}{.45\textwidth}
\includegraphics[width=1\textwidth]{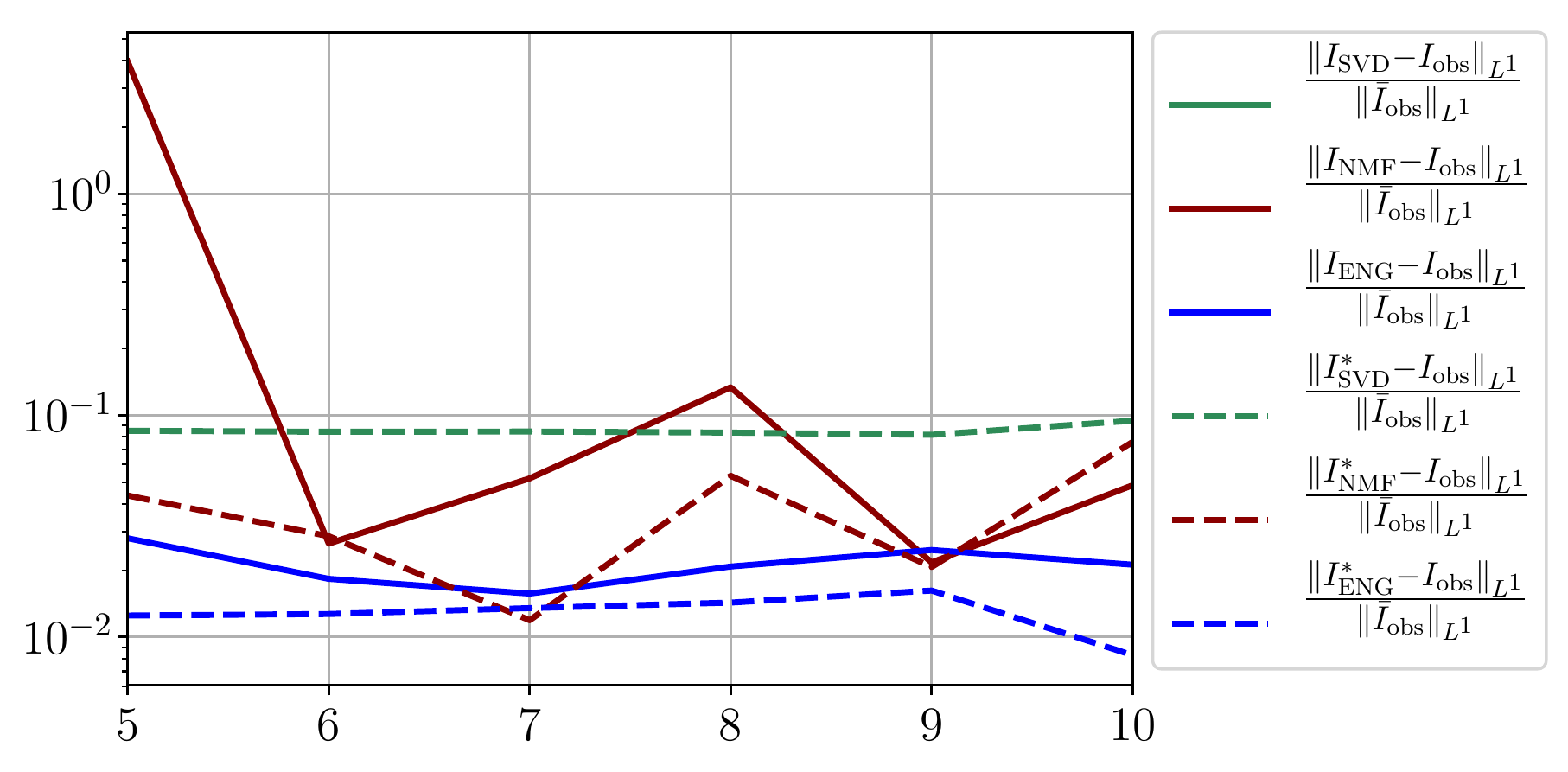}
\caption{$L^1$ relative error  of $I$ vs $n$}
\end{subfigure}
\begin{subfigure}{.45\textwidth}
\includegraphics[width=1\textwidth]{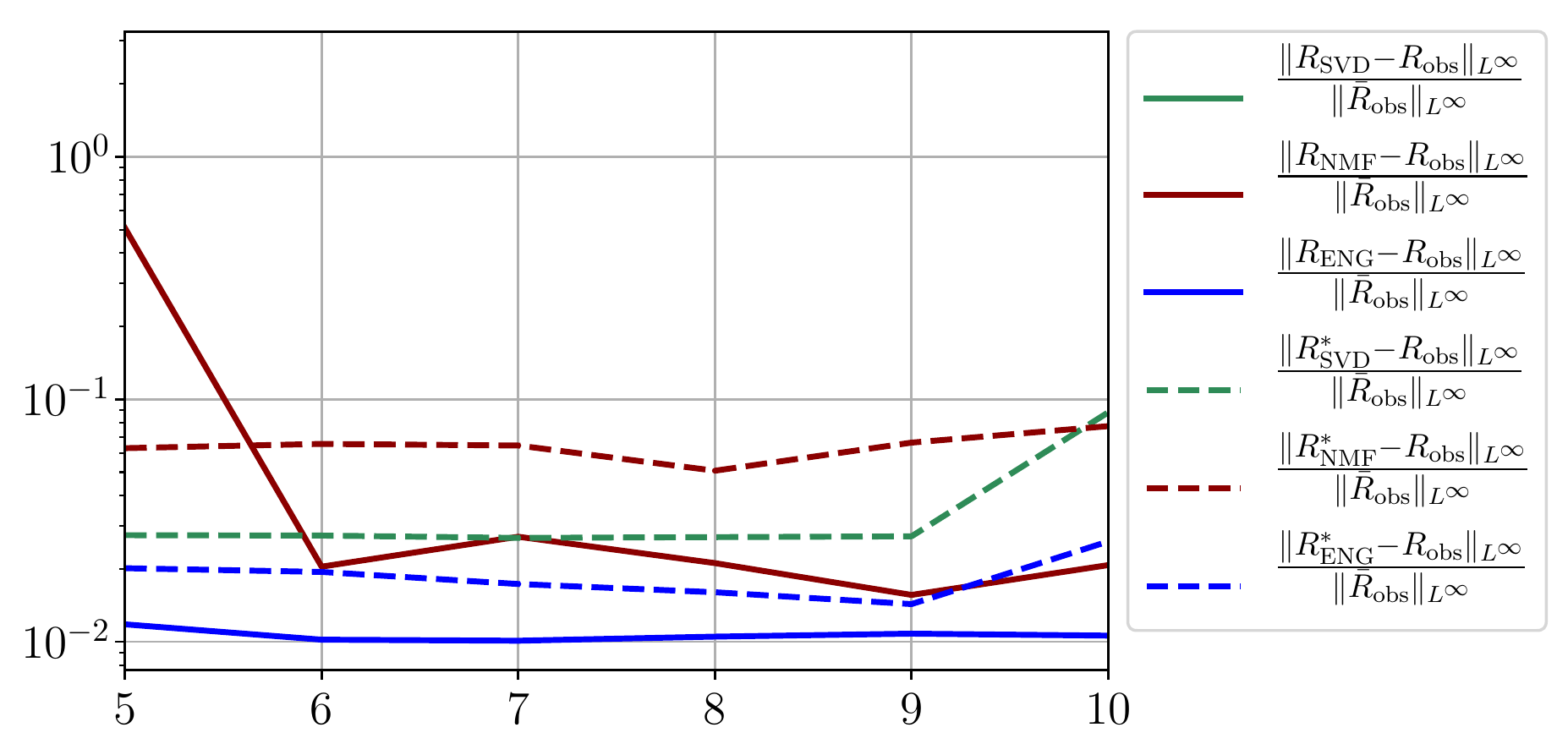}
\caption{$L^\infty$ relative error  of $R$ vs $n$}
\end{subfigure}
\caption{Forecasting errors of $I$ and $R$ (from $T=21/04$)}
\label{fig:forecast_error_2104}
\end{figure}

\begin{figure}[H]
\centering
\begin{subfigure}{.45\textwidth}
\includegraphics[width=1\textwidth]{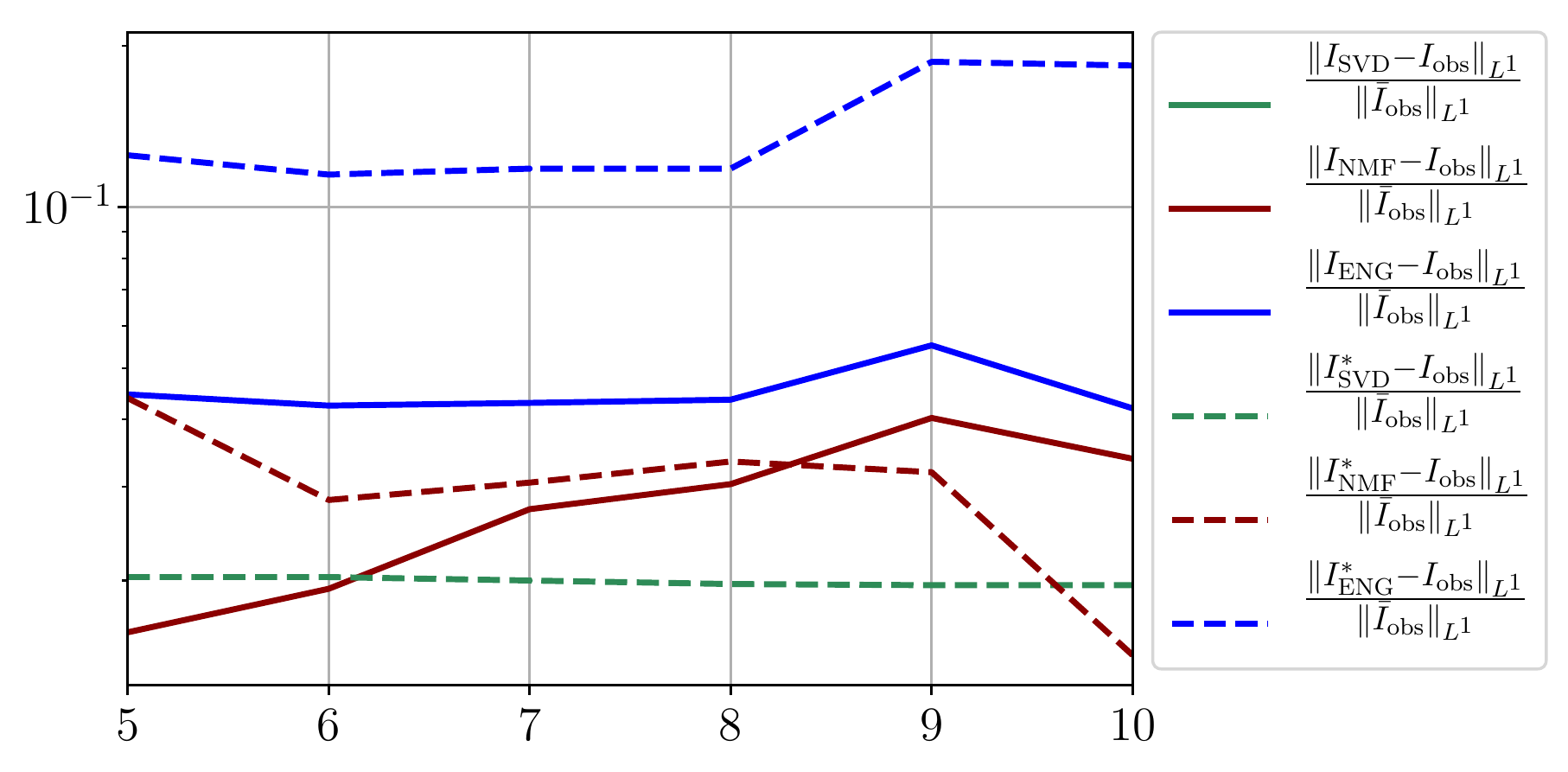}
\caption{$L^1$ relative error  of $I$ vs $n$}
\end{subfigure}
\begin{subfigure}{.45\textwidth}
\includegraphics[width=1\textwidth]{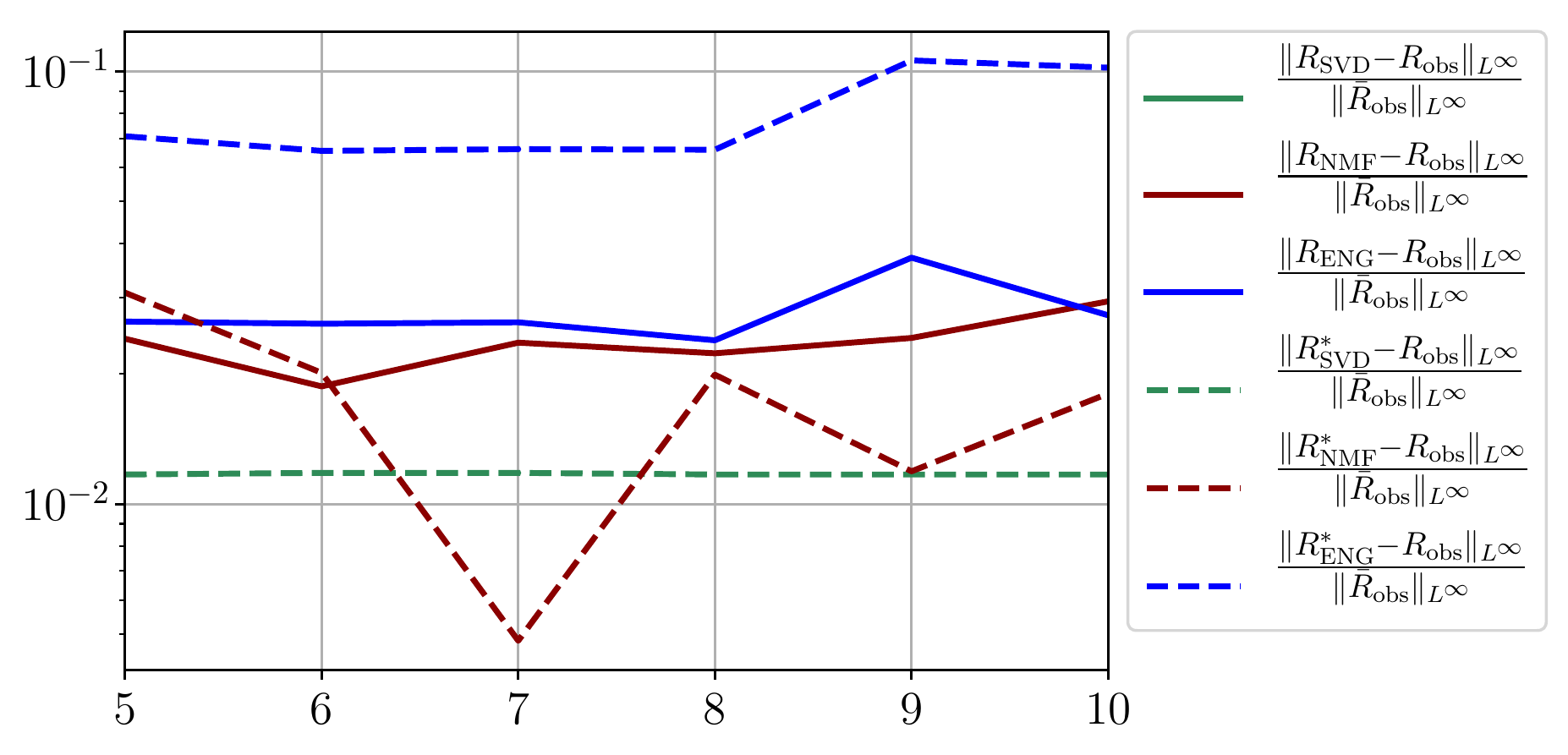}
\caption{$L^\infty$ relative error  of $R$ vs $n$}
\end{subfigure}
\caption{Forecasting errors of $I$ and $R$ (from $T=05/05$)}
\label{fig:forecast_error_0505}
\end{figure}

\col{
We observe that the quality of the forecast depends on the reduced basis but also strongly on the starting day $T$ from which the forecast is done. The forecasts using \fitbg~with SVD and NMF are not accurate and most of the time exploding. When \fitIR~is used with SVD and NMF, the forecasts are more robust as there is no explosion of the error observed. Reduced bases from ENG consistently lead to the the best forecast obtained using either \fitbg~and \fitIR ; by inspecting the error on figures \ref{fig:forecast_error_0104} to \ref{fig:forecast_error_0505} and the averaged forecasts obtained in Section \ref{sec:forecast-compare} we conclude that \fitbg~with ENG reduced bases provide slightly better forecasts.
}

\bibliographystyle{bibstyle}
\bibliography{refs}
\end{document}